\documentclass[USletter,10pt]{article}
\usepackage[dvips]{geometry}
\usepackage{graphicx,cancel}
\usepackage{graphics}
\usepackage{hyperref}
\usepackage{color, soul}
\usepackage{amsmath, amsthm, slashed}
\usepackage{amssymb}
\usepackage{rotating}
\usepackage{mathrsfs}
\usepackage{epsfig}
\usepackage{verbatim}
\usepackage[affil-it]{authblk}
\usepackage{accents}
\newtheorem{definition}{Definition}[section]
\newtheorem{theorem}{Theorem}[section]
 
\newtheorem*{conjecture*}{Conjecture}
\newtheorem{corollary}{Corollary}[section]
\newtheorem*{theorem*}{Theorem}
\newtheorem*{corollary*}{Corollary}
\newtheorem{proposition}{Proposition}[subsection]
\newtheorem{lemma}{Lemma}[subsection]
\newtheorem{remark}{Remark}[section]

\newtheoremstyle{named}{}{}{\itshape}{}{\bfseries}{.}{.5em}{\thmnote{#3 }#1}
\theoremstyle{named}

\newcommand{\bea}{\begin{eqnarray}}
\newcommand{\eea}{\end{eqnarray}}
\newcommand{\beaa}{\begin{eqnarray*}}
\newcommand{\eeaa}{\end{eqnarray*}}
\newcommand{\bsplit}{\begin{split}}

\newcommand{\les}{\lesssim}

\newcommand{\lot}{\mbox{l.o.t.}}
\newcommand{\dual}{{\,^\star \mkern-2mu}}
\newcommand{\tr}{\mbox{tr}}
\newcommand{\nabb}{\nab\mkern-13mu /\,}

\newcommand{\nab}{\nabla}
\renewcommand{\div}{\mbox{div }}
\newcommand{\curl}{\mbox{curl }}
\newcommand{\divv}{\mbox{div}\mkern-19mu /\,\,\,\,}
\newcommand{\lapp}{\mbox{$\bigtriangleup  \mkern-13mu / \,$}}
\newcommand{\curll}{\mbox{curl}\mkern-19mu /\,\,\,\,}
\newcommand{\pr}{\partial}

\newcommand{\dk}{\mathfrak{d}}
\newcommand{\DD}{{\mathcal D}}
\newcommand{\DDs}{ \, \DD \hspace{-2.4pt}\dual    \mkern-16mu /}
\newcommand{\DDd}{ \, \DD \hspace{-2.4pt}    \mkern-8mu /}

\renewcommand{\a}{\alpha}
\renewcommand{\b}{\beta}
\newcommand{\de}{\delta}
\newcommand{\ep}{\epsilon}

\newcommand{\Si}{\Sigma}
\newcommand{\om}{\omega}

\renewcommand{\th}{\theta}

\newcommand{\ka}{\kappa}
\newcommand{\ze}{\zeta}
\newcommand{\Up}{\Upsilon}

\newcommand{\vphi}{{\varphi}}

\renewcommand{\aa}{\protect\underline{\a}}
\newcommand{\bb}{\protect\underline{\b}}
\newcommand{\omb}{{\underline{\om}}}
\newcommand{\Lb}{{\underline{L}}}

\newcommand{\chib}{\underline{\chi}}
\newcommand{\xib}{\underline{\xi}}
\newcommand{\etab}{\underline{\eta}}
\newcommand{\kab}{\underline{\kappa}}
\newcommand{\chih}{\widehat{\chi}}
\newcommand{\chibh}{\widehat{\chib}}

\newcommand{\bF}{\,^{(F)} \hspace{-2.2pt}\b}
\newcommand{\bbF}{\,^{(F)} \hspace{-2.2pt}\bb}
\newcommand{\rhoF}{\,^{(F)} \hspace{-2.2pt}\rho}
\newcommand{\sigmaF}{\,^{(F)} \hspace{-2.2pt}\sigma}

\newcommand{\MM}{{\mathcal M}}

\newcommand{\D}{{\bf D}}
\newcommand{\F}{{\bf F}}

\newcommand{\R}{{\bf R}}
\newcommand{\W}{{\bf W}}
\newcommand{\g}{{\bf g}}
\newcommand{\pf}{\frak{p}}
\newcommand{\qf}{\frak{q}}
\newcommand{\ff}{\frak{f}}
\newcommand{\ffb}{\protect\underline{\ff}}

\newcommand{\fb}{\underline{f}}

\newcommand{\gS}{ g \mkern-8.5mu/\,}

\newcommand{\hot}{\widehat{\otimes}}
\renewcommand{\c}{\cdot}

\newcommand{\Omegab}{\underline{\Omega}}
\newcommand{\vsi}{\varsigma}
\newcommand{\ov}{\overline}

\newcommand{\omblin}{\omb^{\tiny (1)}}
\newcommand{\kalin}{\ka^{\tiny (1)}}
\newcommand{\kablin}{\kab^{\tiny (1)}}
\newcommand{\rhoFlin}{\rhoF^{\tiny (1)}}
\newcommand{\sigmaFlin}{\sigmaF^{\tiny (1)}}
\newcommand{\rlin}{\rho^{\tiny (1)}}
\newcommand{\Klin}{K^{\tiny (1)}}
\newcommand{\flin}{f^{\tiny (1)}}
\newcommand{\Omegablin}{\Omegab^{\tiny (1)}}
\newcommand{\vsilin}{\vsi^{\tiny (1)}}
\newcommand{\trglin}{\tr g^{\tiny (1)}}
\newcommand{\slin}{\sigma^{\tiny (1)}}

\title{The linear stability of Reissner-Nordstr{\"o}m spacetime \\ for small charge}

\author{Elena Giorgi \footnote{Gravity Initiative, Princeton University, egiorgi@princeton.edu}}

\begin{document}

\maketitle

\begin{abstract}
In this paper, we prove the linear stability to gravitational and electromagnetic perturbations of the Reissner-Nordstr{\"o}m family of charged black holes with small charge. Solutions to the linearized Einstein-Maxwell equations around a Reissner-Nordstr\"om solution arising from regular initial data remain globally bounded on the black hole exterior and in fact decay to a linearized Kerr-Newman metric. We express the perturbations in geodesic outgoing null foliations, also known as Bondi gauge. To obtain decay of the solution, one must add a residual pure gauge solution which is proved to be itself controlled from initial data. Our results rely on decay statements for the Teukolsky system of spin $\pm2$ and spin $\pm1$ satisfied by gauge-invariant null-decomposed curvature components, obtained in earlier works. These decays are then exploited to obtain polynomial decay for all the remaining components of curvature, electromagnetic tensor and Ricci coefficients. In particular, the obtained decay is optimal in the sense that it is the one which is expected to hold in the non-linear stability problem. 
\end{abstract}

\section{Introduction}\label{intro}

  The problem of stability of   the Kerr family $(\MM, g_{M, a})$  in the context of the Einstein vacuum equations occupies a central stage  in mathematical General Relativity. Roughly speaking, the problem of stability of the Kerr metric consists in showing that all solutions of the Einstein vacuum equation 
  \bea\label{einstein-vacuum-equation}
  \operatorname{Ric}(g)=0
  \eea
  which are spacetime developments of initial data sets sufficiently close to a member of the Kerr family converge asymptotically to another member of the Kerr family.

  The problem in the generality hereby formulated remains open, but many interesting cases have been solved in the recent years. 
The only known proof of non-linear stability of the Einstein vacuum equation \eqref{einstein-vacuum-equation} with no symmetry assumption is the celebrated global stability of Minkowski spacetime \cite{Ch-Kl}. A recent work \cite{stabilitySchwarzschild} proves the non-linear stability of Schwarzschild spacetime under a restrictive symmetry class, which excludes rotating Kerr solutions as final state of the evolution. 
In the case of the Einstein equation with a positive cosmological constant, the global non-linear stability of Kerr-de Sitter spacetime for small angular momentum has been proved in \cite{hintz-vasy}. 

An important step to understand non-linear stability is proving linear stability, which means proving boundedness and decay for the linearization of the Einstein equations around the Kerr solution. The first proof of the linear stability of Schwarzschild spacetime to gravitational perturbations has been obtained in \cite{DHR}. Different results and proofs of the linear stability of the Schwarzschild spacetime have followed, using the original   Regge-Wheeler  approach  of metric perturbations \cite{mu-tao}, and using wave gauge \cite{pei-ken}, \cite{Johnson1}, \cite{Johnson2}. Steps towards the linear stabilty of Kerr solution have been made in the proof of boundedness and decay for solutions to the Teukolsky equations in Kerr in \cite{ma2} and \cite{TeukolskyDHR}. See also the recent \cite{linearkerr}, \cite{Kerr-lin2}.

In this paper we consider the above problem in the setting of Einstein-Maxwell equations for charged black holes.

   The problem of stability of charged black holes has as final goal the proof of non-linear stability of the Kerr-Newman family $(\MM, g_{M, Q, a})$ as solutions to the Einstein-Maxwell equations
\begin{equation} \label{Einsteineq}
\operatorname{Ric}(g)_{\mu\nu}=T(F)_{\mu\nu}:=2 F_{\mu \lambda} {F_\nu}^{\lambda} - \frac 1 2 g_{\mu\nu} F^{\alpha\beta} F_{\alpha\beta}
\end{equation}
where $F$ is a $2$-form satisfying the Maxwell equations 
\begin{equation} \label{Max}
D_{[\alpha} F_{\beta\gamma]}=0, \qquad D^\alpha F_{\alpha\beta}=0.
\end{equation}
In the case of the Einstein-Maxwell equations with a positive cosmological constant, the global non-linear stability of Kerr-Newman-de Sitter spacetime for small angular momentum has been proved in \cite{hintz}. 

The presence of a non-trivial right hand side in the Einstein equation \eqref{Einsteineq} and the Maxwell equations \eqref{Max} add new difficulties to the analysis of the problem, due to the coupling between the gravitational and the electromagnetic perturbations of a solution. This creates major difficulties in both the analysis of the equations and the choice of the gauge, for which the entanglement between the gravitational and the electromagnetic perturbations changes the structure of the estimates and the choice of gauge. 

An intermediate step towards the proof of non-linear stability of charged black holes is the linear stability of the simplest non-trivial solution of the Einstein-Maxwell equations, the Reissner-Nordstr{\"o}m spacetime.
   The Reissner-Nordstr{\"o}m family
  of spacetimes $(\mathcal{M},g_{M, Q})$ is most easily   expressed  in local coordinates in the form:
 \begin{equation}
  \label{RNintro}
 g_{M, Q}=-\left(1-\frac{2M}{r}+\frac{Q^2}{r^2}\right)dt^2 +\left(1-\frac{2M}{r}+\frac{Q^2}{r^2}\right)^{-1}dr^2 +r^2(d\theta^2+\sin^2\theta d\phi^2),
  \end{equation}
where $M$ and $Q$ are arbitrary parameters.  The parameters $M$ and $Q$ are interpreted as the mass and the charge of the source respectively. For physical reasons, it is normally assumed that $M > |Q|$, which excludes the case of naked singularity.  The Kerr-Newman metric $g_{M, Q, a}$ reduces to the Reissner-Nordstr\"om metric $g_{M, Q}$ for $a=0$, and Reissner-Nordstr{\"o}m reduces to Schwarzschild spacetime when $Q=0$.

 The Reissner-Nordstr{\"o}m spacetime $(\mathcal{M},g_{M, Q})$ is the simplest non-trivial solution to the
 Einstein-Maxwell equation and the unique electrovacuum spherically symmetric spacetime. It therefore plays for the Einstein-Maxwell equation the same role as the Schwarzschild metric for the Einstein vacuum equation \eqref{einstein-vacuum-equation}. It then makes sense to start the study of the stability of charged black holes from the linearized equations around the Reissner-Nordstr\"om metric. For the mode stability of the Reissner-Nordstr\"om spacetime, i.e. the proof of non exponentially growing mode solutions to the Einstein-Maxwell equation, see \cite{Moncrief1}, \cite{Moncrief2}, \cite{Moncrief3}, \cite{Chandra-RN}, \cite{Chandra-RN2}, \cite{Dotti2}. The mode stability does not imply boundedness of the linearized equations.

The purpose of the present paper is to resolve the linear stability problem to coupled gravitational and electromagnetic perturbations of the Reissner-Nordstr{\"o}m spacetime for small charge, i.e. the case of $|Q| \ll M$. This is the first result  of quantitative stability of black holes coupled with matter with no symmetry assumption, and is the electrovacuum analogue of the linear stability of the Schwarzschild solution as obtained in \cite{DHR}. In particular, the proof does not involve any decomposition in mode solutions, and is obtained in physical space.
   
 A first version of our main result can be stated as follows.
 
 \begin{theorem}(Linear stability of Reissner-Nordstr\"om: case $|Q|\ll M$ -  Rough version) All solutions to the linearized Einstein-Maxwell equations (in Bondi gauge) around Reissner-Nordstr{\"o}m with small charge arising from regular asymptotically flat initial data
 \begin{enumerate}
 \item {\bf remain uniformly bounded} on the exterior and
 \item {\bf decay} according to a specific peeling\footnote{The decay is consistent with the decay which is expected in the non-linear problem.} to a standard linearized Kerr-Newman solution 
 \end{enumerate}
 after adding a pure gauge solution which can itself be estimated by the size of the data.
 \end{theorem}

The proof of linear stability roughly consists in two steps:
\begin{enumerate}
\item obtaining decay statements for gauge-invariant quantities,
\item choosing an appropriate gauge which allows to prove decay statements for the gauge-dependent quantities.
\end{enumerate}
In the first step, we need to identify the right gauge-invariant quantities which verify wave equations which can be analyzed and for which quantitative decay statements can be obtained. We completed the resolution of this part in our \cite{Giorgi4} and \cite{Giorgi5}. We summarize the results in Section \ref{summary-gauge-inv}.

The contribution of this paper is the resolution of the second step. Once we obtain decay for gauge-independent quantities from the first step, it is crucial to understand the structure of the equations in order to choose just the right gauge conditions to obtain decay for the gauge-dependent quantities. In particular, our goal here is to obtain \textit{optimal} decay for \textit{all} quantities, where with optimal we mean decay which would be consistent with bootstrap assumptions in the case of non-linear stability of Reissner-Nordstr\"om spacetime. In particular, we obtain the same peeling decay of the bootstrap assumptions in the non-linear stability of Schwarzschild \cite{stabilitySchwarzschild}.\footnote{Such decay is not necessarily sharp as one would expect for the given equations. } Having non-linear applications in mind, we aim to obtain decay for all components, since they would all show up in the non-linear terms of the wave equations. 

In order to obtain the optimal decay for all components, we choose a particular gauge ``far away'' in time and space. This choice is inspired by the gauge choice used in \cite{stabilitySchwarzschild}, which allows for optimal decay for all components. We adapt this choice to our case, where the coupling between gravitational and electromagnetic radiation makes the equations more involved, and isolating quantities which satisfy transport equations and allow us to prove decay is a difficult part of the problem. We summarize the main difficulties and the choice of gauge in Section \ref{section-choice-gauge}.

\subsection{Comparison with \cite{DHR} and \cite{stabilitySchwarzschild}}

We briefly compare here our result to the work of linear stability of Schwarzschild spacetime \cite{DHR} and non-linear stability of axially symmetric polarized perturbations of Schwarzschild \cite{stabilitySchwarzschild}. 

The main difference with the case of gravitational and electromagnetic perturbations of  Reissner-Nordstr\"om spacetime is in the behavior of the projection to the $\ell=1$ spherical harmonics of the perturbations. Because of the presence of the electromagnetic tensor, such projection is not exhausted by gauge solutions and change in angular momentum, like in Schwarzschild, but presents decay in the form of electromagnetic radiation. See Section \ref{section-l1}.

The choice of gauge here presented is  the outgoing null geodesic foliation as used in \cite{stabilitySchwarzschild}. In our intermediate proof of boundedness of the solution we make use of a normalization at the level of initial data which is evocative of the one used in \cite{DHR}. 
 On the other hand, when proving decay for the solution, we make use of a ``far-away'' normalization, which is inspired by the ``last slice'' GCM construction in \cite{stabilitySchwarzschild}. This choice of gauge allows for the derivation of the optimal decay for all quantities which is consistent with non-linear applications. See Section \ref{section-choice-gauge} and Section \ref{section-decay-gauge} for more details.

 \subsection{The gauge-invariant quantities and the Teukolsky system}\label{summary-gauge-inv}
In the proof of linear stability of Schwarzschild \cite{DHR}, the first step is the proof of boundedness and decay for the solution of the spin $\pm 2$ Teukolsky equation. 
These are   wave  equations  verified  by  the extreme null   components of the curvature tensor  which decouple,   to second order, from all other curvature components.

 In linear  theory,   the Teukolsky equation, combined with cleverly chosen gauge conditions,  allows one to prove  what is known as mode stability, i.e the lack of exponentially growing modes for all curvature components. Extensive literature by the physics community covers these results (see \cite{Chandra}, \cite{Chandra-RN}, \cite{Chandra-RN2}).     This weak version of stability is however far from sufficient to prove boundedness and decay of the solution;  one  needs instead    to derive    sufficiently strong decay estimates to hope to apply them in  the nonlinear  framework. 
 
 The first   quantitative  decay estimates  for the Teukolsky equation in Schwarzschild were obtained in \cite{DHR}. Their approach to derive   boundedness and quantitative decay for the Teukolsky equation relies on the following ingredients:
\begin{enumerate}
\item       A    map which takes a solution to the Teukolsky equation, verified by the null curvature component $\a$, to  a solution of   a wave equation which is simpler   to analyze. In the case of Schwarzschild,  this equation is known as the \textit{Regge-Wheeler  equation}.   The first such transformation was discovered by  Chandrasekhar \cite{Chandra}  in the context of mode decompositions and generalized by Wald \cite{Wald-T}. 
The physical-space version of this transformation first appears in \cite{DHR}.
\item  A vector field method    to  get   quantitative decay for the    new wave equation.
\item A method by which one can derive estimates  for  solutions to the Teukolsky equation from 
those of solutions to the transformed  Regge -Wheeler  equation.
\end{enumerate}

Similarly, in the case of charged black holes, a key step towards the proof of linear stability of Reissner-Nordstr{\"o}m spacetime is to find an analogue of the Teukolsky equation and understand the behavior of its solution. The gauge-independent quantities involved, analogous to $\alpha$ or $\underline{\alpha}$ in vacuum, as well as the structure of the equations that they verify, were identified in our earlier work \cite{Giorgi4}. We relied on the following ingredients:

\begin{enumerate}
\item Computations in physical space which derived the Teukolsky-type equations verified by the extreme null curvature components in Reissner-Nordstr{\"o}m spacetime. We obtained a system of two coupled Teukolsky-type equations.
\item       A    map which takes solutions to the above equations to  solutions of a coupled Regge-Wheeler-type  equations.
\item  A vectorfield method    to  get   quantitative decay for the system.  The analysis was strongly affected by the fact that we were dealing with a system, as opposed to a single equation. 
\item A method by which we derived estimates  for  solutions to the Teukolsky-type system from 
those of solutions to the transformed  Regge-Wheeler-type system.
\end{enumerate}

In \cite{Giorgi4}, we derived the spin $\pm2$ Teukolsky-type system verified by two gauge-independent curvature components of the gravitational and electromagnetic perturbations of Reissner-Nordstr{\"o}m. In addition to the Weyl curvature component $\alpha$, we introduce a new gauge-independent electromagnetic component $\mathfrak{f}$, which appears as a coupling term to the Teukolsky-type equation for $\alpha$. It is remarkable that such $\mathfrak{f}$ verifies itself a Teukolsky-type equation coupled back to $\alpha$, as shown in \cite{Giorgi4}. The quantities $\alpha$ and $\mathfrak{f}$ verify a system of the schematic form:
\begin{equation}\label{systemTeu}
\begin{cases}
\Box_{g_{M, Q}}\alpha+c_1(r) \underline{L}(\alpha)+c_2(r) L (\alpha)+ \tilde{V_1}(r) \alpha = Q \cdot c_3(r) L (\mathfrak{f}), \\
\Box_{g_{M, Q}}\mathfrak{f}+d_1(r) \underline{L}(\mathfrak{f})+d_2(r) L (\mathfrak{f})+ \tilde{V_2} (r)\mathfrak{f} = -Q \cdot d_3(r)\underline{L} (\alpha)
\end{cases}
\end{equation} 
where $L$ and $\underline{L}$ are outgoing and ingoing null directions, $c_i$, $d_i$ and $\tilde{V}_i$ are smooth functions of the radial function $r$, and $Q$ is the charge of the spacetime. The presence of the first order terms $\underline{L}(\alpha),  L (\alpha)$ and $\underline{L}(\mathfrak{f}), L (\mathfrak{f})$ in \eqref{systemTeu} prevents one from getting quantitative estimates to the system \eqref{systemTeu} directly.

 In order to derive appropriate decay estimates the system, new quantities $\mathfrak{q}$ and $\mathfrak{q}^\F$ were defined in \cite{Giorgi4}, at the level of two and one derivative respectively of $\alpha$ and $\mathfrak{f}$. They correspond to physical space versions of the Chandrasekhar transformations mentioned earlier. This transformation has the remarkable property of turning the system of Teukolsky type equations into a system of Regge-Wheeler-type equations. More precisely, it transforms the system \eqref{systemTeu} into the following schematic system:
 \begin{equation}\label{systemRW}
\begin{cases}
\Box_{g_{M, Q}}\mathfrak{q}+ V_1(r)\mathfrak{q} = Q \cdot a_1(r) D^{\leq 2}\mathfrak{q}^\F, \\
\Box_{g_{M, Q}}\mathfrak{q}^\F+ V_2(r) \mathfrak{q}^\F = Q \cdot a_2(r) \mathfrak{q}
\end{cases}
\end{equation}
where $D^{\leq 2}\mathfrak{q}^\F$ denotes a linear expression in terms of up to two derivative of $\mathfrak{q}^\F$, and $V_i$, $a_i$ are smooth functions of the radial coordinate $r$.
In the case of zero charge, system \eqref{systemRW} reduces to the first equation, i.e. the Regge-Wheeler equation analyzed in \cite{DHR}.

Boundedness and decay for $\mathfrak{q}$ and $\mathfrak{q}^\F$, and therefore for $\alpha$ and $\mathfrak{f}$, was obtained in \cite{Giorgi4}. We derived estimates for the system \eqref{systemRW}, by making use of the smallness of the charge to absorb the right hand side through a combined estimate for the two equations. Particularly problematic is the absorption in the trapping region, where the Morawetz bulk terms in the estimates are degenerate. The specific structure of the terms appearing in the system is exploited in order to obtain cancellation in this region.

Observe that the quantities $(\alpha, \mathfrak{f})$ are symmetric-traceless two-tensors transporting gravitational radiation, and therefore supported in $\ell\geq 2$ spherical harmonics.

\subsection{New feature in Reissner-Nordstr{\"o}m: the projection to $\ell=1$ spherical harmonics}\label{section-l1}

 In the linear stability of the Schwarzschild spacetime to gravitational perturbations \cite{DHR}, the decay for $\a$ implies specific decay estimates for all the other curvature components and Ricci coefficients supported in $\ell\geq2$ spherical harmonics, once a gauge condition is chosen. 
 In addition, an intermediate step of the proof is the following theorem: solutions of the linearized gravity system around Schwarzschild supported only on $\ell=0,1$ spherical harmonics are a linearized Kerr plus a pure gauge solution\footnote{Pure gauge solutions correspond to coordinate transformations, while linearized Kerr solutions correspond to a change in mass and angular momentum.}.

In the setting of linear stability of Reissner-Nordstr{\"o}m to coupled gravitational and electromagnetic perturbations, we expect to have electromagnetic radiation supported in $\ell \geq 1$ spherical modes, as for solutions to the Maxwell equations in Schwarzschild (see \cite{BlueMax}, \cite{Federico}).

On the other hand, the decay for the two tensors $\a$ and $\ff$ obtained in \cite{Giorgi4} will not give any decay information about the $\ell=1$ spherical mode of the perturbations. It turns out that, in the case of solutions to the linearized gravitational and electromagnetic perturbations around Reissner-Nordstr{\"o}m spacetime, the projection to the $\ell=0,1$ spherical harmonics is not exhausted by the linearized Kerr-Newman\footnote{Linearized Kerr-Newman solutions correspond to a change in mass, charge and angular momentum. See Definition \ref{def:kerrl}. } and the pure gauge solutions. Indeed, the presence of the Maxwell equations involving the extreme curvature component of the electromagnetic tensor, which is a one-form, transports electromagnetic radiation supported in $\ell\geq 1$ spherical harmonics. The gauge-independent quantities involved in the electromagnetic radiation in Reissner-Nordstr{\"o}m were identified in our earlier work \cite{Giorgi5}.

In \cite{Giorgi5}, we identified a new gauge-independent one-form $\tilde{\beta}$, which is a mixed curvature and electromagnetic component.  Such $\tilde{\beta}$ verifies a spin $\pm1$ Teukolsky-type equation, with non-trivial right hand side, which can be schematically written as 
\bea\label{Teukolksy-spin1}
\Box_{g_{M, Q}} \tilde{\beta}+c_1(r) L(\tilde{\beta})+ c_2(r) \underline{L}(\tilde{\beta})+ V_1(r) \tilde{\beta}= R.H.S.
\eea
where $c_i$, $V_1$ are smooth functions of $r$, and the right hand side involves curvature components, electromagnetic components and Ricci coefficients. 

By applying the Chandrasekhar transformation, we obtain a derived quantity $\mathfrak{p}$ at the level of one derivative of $\tilde{\beta}$. Similar physical space versions of the Chandrasekhar transformation were introduced in \cite{Federico} in the case of Schwarzschild. This transformation has the remarkable property of turning the Teukolsky-type equation \eqref{Teukolksy-spin1} into a Fackerell-Ipser-type  equation, with right hand side which vanishes in $\ell=1$ spherical harmonics. Indeed, $\pf$ verifies an equation of the schematic form:
 \begin{equation}\label{systeml1}
\Box_{g_{M, Q}}\mathfrak{p}+ V\mathfrak{p} =  Q \cdot \divv\qf^\F
\end{equation}
where the right hand side is supported in $\ell\geq 2$ spherical harmonics. 

Projecting equation \eqref{systeml1} in $\ell=1$ spherical harmonics, we obtain a scalar wave equation with vanishing right hand side, for which techniques developed in \cite{redshift}, \cite{lectures}, \cite{rp} can be straightforwardly applied. This proves boundedness and decay for the projection of $\mathfrak{p}$, and therefore $\tilde{\beta}$, to the $\ell=1$ spherical mode, in Reissner-Nordstr{\"o}m spacetimes. 
The boundedness and decay for the projection of $\tilde{\b}$ into $\ell\geq 2$ is implied by using the result for the spin $\pm2$ Teukolsky equation in \cite{Giorgi4}, for small charge.

The Main Theorems in  \cite{Giorgi4} and \cite{Giorgi5} provide decay for the three quantities $\alpha$, $\mathfrak{f}$, $\tilde{\beta}$, and their negative spin equivalent $\aa$, $\underline{\ff}$ and $\tilde{\bb}$. Since these quantities are gauge-independent, the above decay estimates do not depend on the choice of gauge.

The aim of this paper is to show that the decay rates for these gauge-invariant quantities imply boundedness and specific decay rates for all the remaining quantities in the linear stability for coupled gravitational and electromagnetic perturbations of Reissner-Nordstr{\"o}m spacetime for small charge. Observe that the proof of decay for the gauge-dependent quantities here obtained does not need assumptions on the smallness of the charge \textit{once the gauge-invariant quantities are controlled}. Since the estimates for $\qf$, $\qf^\F$ and $\pf$ in \cite{Giorgi4} and \cite{Giorgi5} are only valid for $|Q| \ll M$, the final proof of the linear stability of Reissner-Nordstr\"om in this paper only holds for arbitrarily small charge. Nevertheless, could the assumption on the smallness of the charge be relaxed for the estimates for $\qf$, $\qf^\F$ and $\pf$, this paper could be straightforwardly applied to a different range of charge.

 The optimal decay for the gauge-dependent quantities we are aiming to can be obtained only through a specific choice of gauge. Such a choice of gauge is a crucial step, and will be discussed in the next section.

\subsection{Choice of gauge}\label{section-choice-gauge}
In the linear stability of Schwarzschild \cite{DHR}, the perturbations of the metric are restricted to the form of double null gauge. This choice still allows for residual gauge freedom which in linear theory appears as the existence of pure gauge solutions. Those are obtained from linearizing the families of metrics from applying coordinate transformations which preserve the double null gauge of the metric.

In this work we use the Bondi gauge, inspired by the recent work on the non-linear stability of Schwarzschild \cite{stabilitySchwarzschild}. In particular, we consider metric perturbations on the outgoing null geodesic gauge, of the form
\beaa
  \g&=&-2 \vsi du dr+\vsi^2\Omegab du^2+\slashed{g}_{AB} \left(d\th^A-\frac 1 2 \vsi \underline{b}^A du \right) \left(d\th^B-\frac 1 2\vsi \underline{b}^B du   \right)
  \eeaa
As in \cite{DHR}, this choice still allows for residual gauge freedom, corresponding to pure gauge solutions. 
The residual gauge freedom allows us to further impose gauge conditions which are fundamental for the derivation of the specific decay rates of the gauge-dependent quantities we want to achieve. 

We make two choices of gauge-normalization: an \textit{initial data normalization} and a \textit{far-away normalization}. The motivation for the two choices of gauge-normalization is different, and can be explained as follows.

\subsubsection{The initial data normalization and proof of boundedness}

The initial data normalization consists of normalizing the solution on initial data by adding an appropriate pure gauge solution which is explicitly computable from the original solution's initial data. This normalization allows to obtain boundedness statements for the solution which is initial data normalized. Those are obtained through integration forward from initial data.

Using the integration forward from initial data, there are components of the solution which do not decay in $r$, and this would be a major obstacle in extending this result to the non-linear case. We call this type of decay for the gauge-dependent components \textit{weak decay}.

We point out here that because of the choice of Bondi gauge, which does not cover the horizon, the estimates obtained through integration forward from the initial ingoing null cone cannot be obtained in terms of inverse powers of $u$ close to the horizon. In particular, those estimates are only valid in a region $\{ r\geq r_1\} \cap \{ u_0 \leq u \leq u_1\}$ where $r_1 > r_{\mathcal{H}}$, $\{u=u_0\}$ is the outgoing initial cone and $u_1$ is the $u$-coordinate of the sphere of intersection between $\underline{C}_0$ and $\{ r=r_1\}$.

\subsubsection{The far-away normalization and proof of decay}

Suppose now that a solution to the linearized Einstein-Maxwell equation around Reissner-Nordstr\"om is bounded. In this case, we can define a normalization far-away which is the correct one to obtain the optimal decay we want to achieve for each component of the solution. We call this type of decay for the gauge-dependent components \textit{strong decay}.  

This far-away normalization is inspired by the gauge choice done in \cite{stabilitySchwarzschild}. More precisely, the normalization is realized on an ingoing null hypersurface, with big $r$ and $u$, obtained as the past incoming null cone of a sphere $S_{U, R}$. We should think of this null hypersurface as a bounded version of null infinity, from which optimal decay for all the components can be derived in the past of it. The hard of the matter in the proof of decay is to show that those decays are independent of the chosen far-away position of the null hypersurface. 

As before the estimates in terms of inverse powers of $u$ cannot be obtained close to the horizon, and therefore the integration backward holds up to a timelike hypersurface $\{ r=r_2\}$ with $r_2 > r_{\mathcal{H}}$. 
In the region close to the horizon, one has to integrate forward and obtain decay in $v$-coordinate. 

We now explain in more detail how one obtains the boundedness and the decay using the initial-data and the far-away normalization respectively.

\subsection{Decay of the gauge-dependent components}\label{section-decay-gauge}
Using the initial data normalization, we obtain by construction that some components of the solution do not decay, or even grow, in $r$. More precisely, using initial data normalization we obtain for instance\footnote{See Section \ref{linearized-equations-chapter} for the definition of the quantities.} the following weak decay:
\beaa
|\xib|+|\check{\omb}|&\leq C  ,  \qquad |\underline{b}|+|\check{\Omegab}|\leq C  r , \qquad | \check{\tr_\gamma \slashed{g}}| \leq C r^2 
\eeaa
where $r$ is the coordinate of the underlying Reissner-Nordstr\"om metric. 
The growth in $r$ of these components is intrinsic to the initial data normalization. Indeed, the transport equation\footnote{Equation  \eqref{nabb-4-check-omb-ze-eta}.} for $\check{\omb}$ does not improve in powers of $r$ in the integration forward, so any integration forward starting from a bounded region of the spacetime will not give any decay in $r$. Similarly for $\xib$.\footnote{This issue is present also in the linear stability of Schwarzschild spacetime in \cite{DHR}, where the component $\omblin$ does not decay in $r$ for the same reason. }

Strictly speaking, in linear theory this would not be an issue: it just proves a weaker result. On the other hand, if we consider the linearization of the Einstein-Maxwell equations as a first step towards the understanding of the non-linear stability of black holes, we should obtain a decay which is consistent with bootstrap assumptions in the non-linear case. Having growth in $r$ in some components will not allow to close the analysis of the non-linear terms in the wave equations, and in the remaining decay estimates.

In order to obtain the strong decay for all the components of the solution, we define the  normalization in the far-away hypersurface, inspired by the construction of the ``last slice'' in \cite{stabilitySchwarzschild} and their choice of gauge. In all spherical harmonics, the gauge is chosen so that the traces of the two null second fundamental forms vanish on said far-away hypersurface, as in \cite{stabilitySchwarzschild}. 

In addition, we define two new scalar functions, called \textit{charge aspect function} and \textit{mass-charge aspect function}, respectively denoted $\check{\nu}$ (defined in \eqref{definition-of-nu}) and $\check{\mu}$ (defined in \eqref{definition-of-mu}), which generalize the properties of the known mass-aspect function in the case of the Einstein vacuum equation. Our generalization is essential to obtain the optimal decay for all the components of the solution, and are needed to obtain decay which is independent on the chosen far-away hypersurface.  These quantities are related to the Hawking mass and the quasi local charge of the spacetime and verify good transport equations with integrable right hand sides\footnote{Here $\partial_r$ with respect to coordinates $(u, r, \th, \vphi)$ is outgoing null. See \eqref{null-frame-RN}. }:
\beaa
\partial_r(\check{\nu}_{\ell=1})&=& 0 \\
\partial_r(\check{\mu})&=& O(r^{-1-\de} u^{-1+\de})
\eeaa
In order to make use of these integrable transport equations, we impose these functions to vanish along the ``last slice". They can be chosen to vanish there as solutions to an elliptic system on the far-away hypersurface.

The strong decay can be divided into an optimal decay in $r$ (which would be relevant in regions far-away in the spacetime) and an optimal decay in $u$ (relevant in regions far in the future). In particular, our goal is to prove two different sets of estimates, which have different decay rates in $u$ and $r$, and the estimates with better decay in $r$ have necessarily worse weight in $u$, and vice-versa. If we only allow for decay in $u$ as $u^{-1/2}$, then the optimal decay in $r$ is easier to obtain, because there are quite a few transport equations which are integrable in $r$ from far-away. For example, the transport equation for $\chih$ \eqref{nabb-4-chih-ze-eta} can be written as:
\beaa
\partial_r(r^2 \chih)&=& -r^2\a
\eeaa
  Since $\alpha$ decays as $r^{-3-\de} u^{-1/2+\de}$, we see that the right hand side is integrable in $r$. On the other hand, $\a$ only decays as $r^{-2-\de} u^{-1+\de}$, which would give a non-integrable right hand side in the above transport equation. 
  
  To circumvent this difficulty, we identify a quantity $\Xi$ (defined in \eqref{definition-Xi}) which verifies a transport equation with integrable right hand side:
  \beaa
  \partial_r(\Xi)&=& O(r^{-1-\de} u^{-1+\de})
  \eeaa
   The quantity $\Xi$ is a combination of curvature, electromagnetic and Ricci coefficient terms and generalizes a quantity (also denoted $\Xi$) in \cite{stabilitySchwarzschild} which serves the same purpose. Observe that, as opposed to the charge aspect function or the mass-charge aspect function, $\Xi$ decays fast enough along the``last slice", so that we do not need to impose its vanishing along it. 
  
  Combining the decay of the above quantities we can prove that all the remaining components verify the optimal decay in $r$ and $u$ which is consistent with non-linear applications. In particular, we obtain the same decay as the bootstrap assumptions used in \cite{stabilitySchwarzschild} in the case of Schwarzschild.  More precisely, using the far-away normalization we obtain for instance the following strong decay (see Theorem \ref{linear-stability} for the complete decay rates for all the components):
\beaa
|\xib|+|\check{\omb}|&\leq C r^{-1} u^{-1+\de},  \qquad |\underline{b}|+|\check{\Omegab}|+| \check{\tr_\gamma \slashed{g}}| \leq C   u^{-1+\de}
\eeaa
where $r$ and $u$ are the coordinates of the underlying Reissner-Nordstr\"om metric. 
Comparing with the above decay rate, we see that the far-away normalization significantly improve the rate of decay and is the appropriate result to applications for non-linear theory. For a more detailed outline of the proof of the main theorem, see Section \ref{statement-theorem}.

We now briefly describe how to combine the proof of boundedness of the initial-data normalization with the proof of decay of the far-away normalization. 

Recall that in the construction of the far-away normalization we need to know that the solution is bounded up to the far-away sphere. More precisely, let $S_{U, R}$ be the far-away sphere and consider the intersection of the past outgoing null cone of $S_{U, R}$, i.e. $\{ u=U\}$, and the incoming initial null cone $\underline{C}_0$. Such intersection is a sphere with a certain radius $\tilde{r}$. Consider now a radius $r_1$ such that $r_{\mathcal{H}} < r_1 < \tilde{r}$. Applying the proof of boundedness for such $r_1$ integrating forward from initial data implies boundedness of the initial-data normalized solution in the region $\{ r\geq r_1\} \cap \{ u_0 \leq u \leq u_1\}$ where $\{u=u_0\}$ is the outgoing initial cone and $u_1$ is the $u$-coordinate of the sphere of intersection between $\underline{C}_0$ and $\{ r=r_1\}$. By construction this region contains the far-away sphere $S_{U, R}$, and therefore the solution is bounded there.

By using the far-away normalized solution and the proof of decay, integrating backward one obtains estimates in terms of inverse powers of $u$ up to a hypersurface $\{r=r_2\}$ for $\tilde{r} < r_2 < R$, i.e. in the region  $\{ r \geq r_2\} \cap \{ u \geq u_0 \} \cap I^{-}(S_{U, R})$.  Finally, in the region close to the horizon the estimates ought to be expressed in terms of the advanced time $v$ in the ingoing Eddington Finkelstein coordinates. As in \cite{stabilitySchwarzschild}, the decay obtained in the region $\{ r \geq r_2\}$ can be translated into a decay in $v$ along the time-like hypersurface $\{ r = r_2\}$, and integrating forward from this hypersurface along the ingoing null direction implies decay as $v^{-1+\de}$ for all the components. 
\begin{figure}[h]
\centering
\includegraphics[width=7cm]{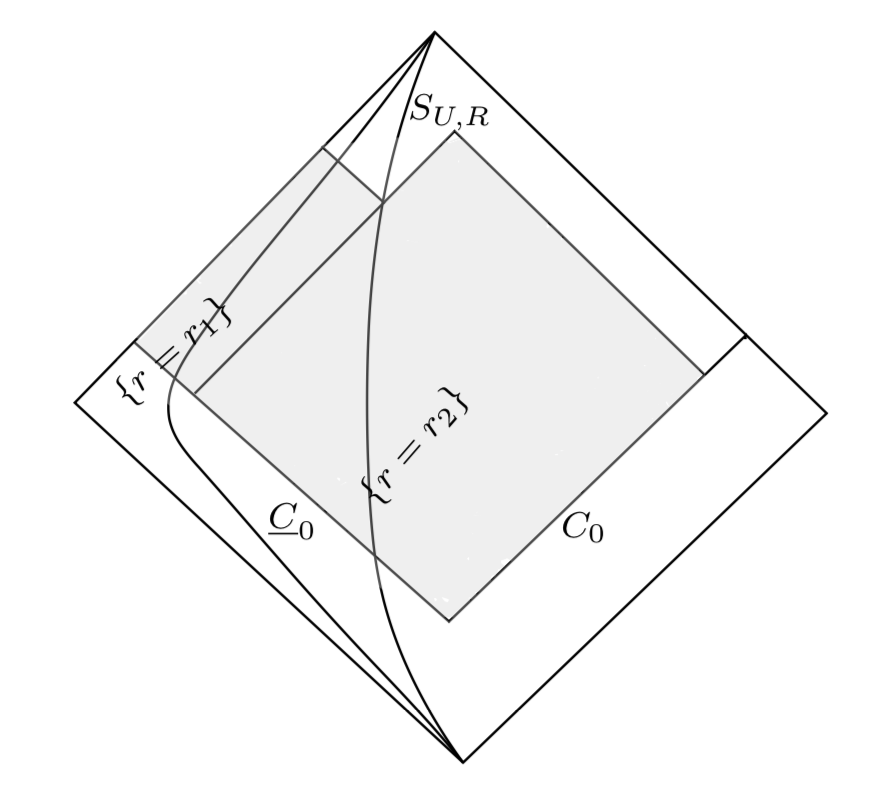}
\caption{Penrose diagram of Reissner-Nordstr{\"o}m spacetime. For every $S_{U, R}$, the gray area describes the area where the estimates are obtained.}
\end{figure}

Since the estimates in the proof of decay do not depend on $U$ and $R$, in the above construction $U$ and $R$ can be taken to be arbitrarily large. As $U$ and $R$ vary to arbitrarily large values, the regions constructed as above cover the whole exterior region in the future of $C_0 \cup \underline{C}_0$, and therefore the above proves the validity of the estimates in the whole exterior region. 
See Section \ref{conclusions} for more details about the conclusion of the proof.

We also point out here that the proof of boundedness of the solution using the initial-data normalization could have been avoided in the final proof of decay by performing a bootstrap argument. By choosing the maximal $U$ and $R$ for which the solution was assumed to be bounded, and then adding the far-away normalization, one can prove that the solution was bounded (and actually decays) in the full exterior. On the other hand, the proof of boundedness here presented has the advantage of making clear the necessity of the far-away normalization in order to obtain a decay consistent with non-linear applications, and also helps to draw a clear analogy with \cite{DHR}.

\subsection{The identification of the final parameters $M$, $Q$, and $a$}

The decay described above corresponds to the gravitational and electromagnetic radiation emitted by the perturbation of the Reissner-Nordstr\"om black hole as it settles down to a member of the Kerr-Newman family with final parameters $M$, $Q$, and $a$. 

The projection to the $\ell=0$ spherical harmonics of the perturbation contains the spherically symmetric solutions, which correspond to other Reissner-Nordstr\"om spacetimes with modified mass $M$ and charge $Q$. 
The projection to the $\ell=1$ spherical harmonics contains instead the axially symmetric solutions, which correspond to a slowly rotating Kerr-Newman spacetime with small angular momentum $a$. These special solutions are called linearized Kerr-Newman solutions, as in Definition \ref{def:kerrl}.

This change in parameters can be identified in linear theory from the projection to the $\ell=0$ and $\ell=1$ spherical harmonics at the initial data, similarly to \cite{DHR}.

  \subsection{Outline}
  
  We outline here the structure of the paper.
  
  In Section \ref{VEeqDNGsec}, we derive the general form of the Einstein-Maxwell equations written with respect to a local null frame.  In Section \ref{double-null-gauge}, we introduce our choice of gauge, the Bondi gauge, to be used in the linear perturbations of Reissner-Nordstr\"om spacetime. 
  
  In Section \ref{RN-chapter}, we describe the Reissner-Nordstr\"om spacetime, which is the background solution around which we perform the gravitational and electromagnetic perturbations. 
   In Section \ref{linearized-equations-chapter}, we derive the linearized Einstein-Maxwell equations around the Reissner-Nordstr\"om solution. We denote a linear gravitational and electromagnetic perturbation of Reissner-Nordstr\"om a set of components which is a solution to those equations. 
  In Section \ref{sec:specialsol}, we present special solutions to the linearized Einstein-Maxwell equations around the Reissner-Nordstr\"om: pure gauge solutions and linearized Kerr-Newman solutions. 
  
  In Section \ref{statement-section}, we state the precise version of our main Theorem.
   In Section \ref{Teukolsky-eqs}, we summarize the results on the boundedness and decay for the solutions to the Teukolsky system as proved in \cite{Giorgi4} and \cite{Giorgi5}.

  In Section \ref{initial-data-well-posedness-chapter}, we present the characteristic initial problem and the well-posedness of the linearized Einstein-Maxwell equations. 
 In Section \ref{gauge-normalized-solutions-chapter} we describe the two gauge normalization which we will use and the final Kerr-Newman parameters. 
  
  In Section \ref{chapter-linear-stability}, we prove boundedness of the solution using the initial data normalization and in Section \ref{decay-chapter} we finally prove decay for all the gauge-dependent components of the solution, therefore obtaining the proof of quantitative linear stability.

  In the Appendix \ref{appendix}, we present explicit computations.

  \bigskip

{\bf{Acknowledgements}}  The author would like to thank Sergiu Klainerman and Mu-Tao Wang for their guidance and support. The author is also grateful to J\'er\'emie Szeftel, Pei-Ken Hung and Federico Pasqualotto for helpful discussions. The author is grateful to the anonymous referee for several helpful suggestions.

\section{The Einstein-Maxwell equations in null frames}
\label{VEeqDNGsec}
In this section, we derive the general form of the Einstein-Maxwell equations \eqref{Einsteineq} and \eqref{Max} written with respect to a local null frame attached to a general foliation of a Lorentzian manifold. In this section, we do not restrict to a specific form of the metric and derive the main equations in their full generality. It is these equations we shall linearize in Section \ref{linearized-equations-chapter} to obtain the equations for a linear gravitational and electromagnetic perturbation of a spacetime. 

We begin in Section \ref{sec:genmfld} with preliminaries, recalling the notion of local null frame and tensor algebra. In Section \ref{sec:rccc}, we define Ricci coefficients, curvature and electromagnetic components of a solution to the Einstein-Maxwell equations. Finally, we present the Einstein-Maxwell equations in Section \ref{nseq}. 

\subsection{Preliminaries} \label{sec:genmfld}
Let $(\MM, \g)$ be a $3+1$-dimensional Lorentzian manifold, and let $\D$ be the covariant derivative associated to $\g$.

Suppose that the the Lorentzian manifold $(\MM, \g)$ can be foliated by spacelike $2$-surfaces $(S,\slashed{g}) $, where $\slashed{g}$ is the pullback of the metric $\g$ to $S$.  To each point of $\MM$, we can associate a null frame $\mathscr{N}=\left\{e_A, e_3, e_4\right\}$, with $\{ e_A \}_{A=1,2}$ being tangent vectors to $(S,\slashed{g}) $, such that the following relations hold:
\bea\label{relations-null-frame}
\begin{split}
\g\left(e_3,e_3\right) &= 0, \qquad \g\left(e_4,
e_4 \right) = 0, \qquad \g\left(e_3,e_4\right) = -2\\ 
\g\left(e_3,
e_A \right) &= 0 \ \ \ , \ \ \  \g\left(e_4, e_A\right) = 0 \ \ \ , \ \ \ \g \left(e_A, e_B \right) = \slashed{g}_{AB} \, .
\end{split}
\eea
The surfaces $S$ will be identified in Chapter \ref{double-null-gauge} as intersections of two specified hypersurfaces. Similarly, after a choice of gauge, the frame $\mathscr{N}$ can be identified explicitly in terms of coordinates. See Section \ref{local-bondi} for the identification of the null frame in the Bondi gauge.

In the following section, we will express the Ricci coefficients, curvature and electromagnetic components with respect to a null frame $\mathscr{N}$ associated to a foliation of surfaces $S$. The objects we shall define are therefore $S$-tangent tensors. We recall here the standard notations for operations on $S$-tangent tensors. (See \cite{Ch-Kl} and \cite{DHR})

We recall the definition of the projected covariant derivatives and the angular operator on $S$-tensors. We denote $\nabb_3=\nabb_{e_3}$ and $\nabb_4=\nabb_{e_4}$ the projection to $S$ of the spacetime covariant derivatives $\D_{e_3}$ and $\D_{e_4}$ respectively. We denote by $\underline{D}\chi$ and $D\chi$ the projected Lie derivative with respect to $e_3$ and $e_4$. The relations between them are the following:
\bea\label{covariant-Lie}
\begin{split}
Df&=\nabb_4(f), \\
D\xi_A&=\nabb_4 \xi_A +\chi_{AB}\xi^B, \\
D\th_{AB}&=\nabb_4 \th_{AB}+\chi_{AC} {\th^C}_{B}+\chi_{BC} {\th_A}^C
\end{split}
\eea
and similarly for $e_3$ replacing $\chi$ by $\chib$, where $\chi$ and $\chib$ are defined in \eqref{def1}.

We recall the following angular operators on $S$-tensors. 

Let $\xi$ be an arbitrary one-form and $\th$ an arbitrary symmetric traceless $2$-tensor on $S$. 
\begin{itemize}
\item $\nabb$ denotes the covariant derivative associated to the metric $\slashed{g}$ on $S$.
\item $\DDd_1$ takes $\xi$ into the pair of functions $(\divv \xi, \curll \xi)$, where $$\divv \xi=\slashed{g}^{AB} \nabb_A \xi_B, \qquad \curll\xi=\ep^{AB}\nabb_A \xi_B$$
\item $\DDs_1$ is the formal $L^2$-adjoint of $\DDd_1$, and takes any pair of functions $(\rho, \sigma)$ into the one-form $-\nabb_A \rho+\ep_{AB} \nabb^B \sigma$.
\item  $\DDd_2$ takes $\th$ into the one-form $\DDd_2\th=(\divv \th)_C=\slashed{g}^{AB}\nabb_A \th_{BC}$.
\item $\DDs_2$ is the formal $L^2$-adjoint of $\DDd_2$, and takes $\xi$ into the symmetric traceless two tensor $$(\DDs_2\xi)_{AB}=-\frac 12 \left( \nabb_B\xi_A+\nabb_A\xi_B-(\divv \xi)\slashed{g}_{AB}\right)$$
\end{itemize}

We can easily check that $\DDs_k$ is the formal adjoint of $\DDd_k$, i.e.
\beaa
\int_{S} (\DDd_k f) g =\int_{S} f (\DDs_k g )
\eeaa

We recall the following $L^2$ elliptic estimates. (see \cite{Ch-Kl} or \cite{stabilitySchwarzschild}).

\begin{proposition}\label{L^2-estimates} Let $(S, \slashed{g})$ be a compact surface with Gauss curvature $K$. 
\begin{enumerate}
\item The following identity holds for a pair of functions $(\rho, \sigma)$ on $S$:
\bea
\int_{S} \left(|\nabb \rho|^2+ |\nabb\sigma|^2\right) &=& \int_{S} |\DDs_1(\rho, \sigma) |^2 \label{third-elliptic-estimate}
\eea
\item The following identity holds for $1$-forms $\xi$ on $S$:
\bea
\int_{S} \left(|\nabb \xi|^2+K |\xi|^2\right) &=& \int_{S} \left(|\divv \xi|^2+|\curll \xi|^2 \right)=\int_{S} |\DDd_1 \xi|^2 \label{first-elliptic-estimate} \\
\int_{S} \left(|\nabb \xi|^2-K |\xi|^2\right) &=& 2\int_{S} |\DDs_2 \xi|^2 
\eea
\item The following identity holds for symmetric traceless $2$-tensors $\theta$ on $S$:
\bea
\int_{S} \left(|\nabb \theta|^2+2K |\theta|^2\right) &=& 2\int_{S} |\divv \theta|^2=2\int_{S} |\DDd_2 \theta|^2 \label{second-elliptic-estimate}
\eea
\item Suppose that the Gauss curvature is bounded away from zero. Then there exists a constant $C>0$ such that the following estimate holds for all vectors $\xi$ on $S$ orthogonal to the kernel of $\DDs_2$:
\bea
\int_{S} \frac{1}{r^2} |\xi|^2 \leq C \int_{S} |\DDs_2 \xi|^2\label{second-elliptic-estimate}
\eea
where $r$ is the radial area function, i.e. $|S|=4\pi r^2$. 
\end{enumerate}
\end{proposition}

Given $f$ a $S$-tensor, we define $\lapp f =\slashed{g}^{AB} \nabb_A \nabb_B f$. We recall the relations between the angular operators and the laplacian $\lapp$ on $S$:
\bea\label{angular-operators}
\begin{split}
\DDd_1 \DDs_1&=-\lapp_0, \qquad \DDs_1 \DDd_1= -\lapp_1+K, \\
\DDd_2 \DDs_2&= -\frac 1 2 \lapp_1-\frac 1 2 K, \qquad \DDs_2 \DDd_2= -\frac 1 2 \lapp_2+K
\end{split}
\eea
where $\lapp_0$, $\lapp_1$ and $\lapp_2$ are the Laplacian on scalars, on $1$-forms and on symmetric traceless $2$-tensors respectively, and $K$ is the Gauss curvature of the surface $S$.

Let $f$ be a scalar on $\MM$. We define its $S$-average, and denote it by $\ov{f}$ as
\bea \label{def-S-average}
\ov{f}:&=&\frac{1}{|S|} \int_{S} f
\eea
where $|S|$ denotes the volume of $(S, \slashed{g})$.
 We define the derived scalar function $\check{f}$ as
\bea \label{def-check-f}
 \check{f}:= f-\ov{f}.
\eea
It follows from the definition that, for two functions $f$ and $g$,
\bea\label{ov-f-g}
\ov{fg}&=& \ov{f} \ov{g}+\ov{\check{f}\check{g}}, \\
\check{fg}&=& fg-\ov{fg}=\check{f}\ov{g}+\ov{f}\check{g}+(\check{f}\check{g}-\ov{\check{f}\check{g}})
\eea

\subsection{Ricci coefficients, curvature and electromagnetic components} \label{sec:rccc}
We now define the Ricci coefficients, curvature and electromagnetic components associated to the metric $\g$ with respect to the null frame $\mathscr{N}=\left\{e_A, e_3, e_4\right\}$,  where the indices $A,B$ take values $1,2$. We follow the standard notations in \cite{Ch-Kl}.

\subsubsection{Ricci coefficients}\label{ricci-coefficients}

 We define the Ricci coefficients associated to the metric $\g$ with respect to the null frame $\mathscr{N}$:
   \bea\label{def1}
   \begin{split}
   \chi_{AB}:&=\g(\D_A  e_4, e_B), \qquad  \chib_{AB}:=\g(\D_A  e_3, e_B)\\
   \eta_A:&=\frac 1 2 \g(\D_3 e_4, e_A), \qquad \etab_A:=\frac 1 2 \g(\D_4 e_3, e_A),\\
    \xi_A:&=\frac 1 2 \g(\D_4 e_4, e_A), \qquad \xib_A:=\frac 1 2 \g(\D_3 e_3, e_A)\\
     \om:&=\frac 1 4 \g(\D_4 e_4, e_3), \qquad \omb:=\frac 1 4\g(\D_3 e_3, e_4)\\
   \ze_A:&=\frac 1 2 \g( \D_A e_4, e_3), \qquad 
   \end{split}
   \eea

   We decompose the $2$-tensor $\chi_{AB}$ into its tracefree part $\chih_{AB}$, a symmetric traceless 2-tensor on $S$, and its trace.   We define 
   \bea\label{definition-ka-kab}
   \ka:=\tr\chi \qquad \kab:=\tr\chib
   \eea
In particular we write $ \chi_{AB}=\frac 1 2 \ka \ \slashed{g}_{AB}+\chih_{AB}$, with $\slashed{g}^{AB} \chih_{AB}=0$ and $\ka=\slashed{g}^{AB} \chi_{AB}$. Similarly for $\chib_{AB}$.

It follows from \eqref{def1} that we have the following relations for the covariant derivatives of the null frame:
\bea\label{derivatives-}
\begin{split}
\D_4 e_4 &=-  2\om e_4 +2\xi^A e_A , \qquad   \D_3 e_3 =-2\omb e_3 +2\xib^A e_A, \\
 \D_4 e_3 &= 2\om e_3 +2\etab^A e_A, \qquad \D_3 e_4=2\omb  e_4+2\eta^A e_A,\\
\D_4 e_A &=\nabb_4 e_A+ \etab_A  e_4+\xi_A e_3,  \qquad \D_3 e_A =\nabb_3 e_A+\eta_A e_3+\xib_A e_4, \\
 \D_A  e_4 &= -\ze_A e_4 +\chi_{AB} e^B, \qquad \D_A e_3=\ze_A e_3 +\chib_{AB} e^B,\\
\D_A e_B &=\nabb_A e_B+  \frac{1}{2}\chib_{AB} e_4+\frac{1}{2}\chi_{AB} e_3.
\end{split}
\eea
The following relations for the commutators of the null frame also follow from \eqref{def1}:
\bea\label{commutators}
\begin{split}
\big[e_3, e_A\big]&= \nabb_3 e_A+(\eta_A-\ze_A) e_3+\xib_A e_4-\chib_{AB} e^B, \\
\big[e_4, e_A\big]&= \nabb_4 e_A+(\etab_A+\ze_A) e_4+\xi_A e_3-\chi_{AB} e^B, \\
\big[e_3, e_4\big]&= -2\om e_3+2\omb  e_4+2(\eta^A -\etab^A) e_A
\end{split}
\eea

\subsubsection{Curvature components}
Let $\W$ denote the Weyl curvature of $\g$ and let $\dual \W$ denote the Hodge dual on $(\MM, \g)$ of $\W$, defined by $\dual \W_{\a\b\gamma\delta}=\frac 1 2 \ep_{\a\b\mu\nu} {\W^{\mu\nu}}_{\gamma\delta}$. 

We define the null curvature components:  
   \bea\label{def3}
\begin{split}
\a_{AB}:&=\W(e_A, e_4, e_B, e_4), \qquad \aa_{AB}:= \W(e_A, e_3, e_B, e_3) \\
\b_A:&=\frac 1 2 \W(e_A, e_4, e_3, e_4), \qquad  \bb_{A} :=\frac 1 2 \W(e_A, e_3, e_3, e_4) \\
\rho:&=\frac 14 \W(e_3, e_4, e_3, e_4) \qquad \sigma:=\frac 1 4  \dual \W (e_3, e_4, e_3, e_4)
\end{split}
\eea
The remaining components of the Weyl tensor are given by 
\beaa
\W_{AB34}=2\sigma \ep_{AB}, \qquad \W_{ABC3}=\ep_{AB} \dual \bb_C, \qquad \W_{ABC4}=-\ep_{AB} \dual \b_C, \\
\W_{A3B4}=-\rho \slashed{g}_{AB}+\sigma \ep_{AB}, \qquad \W_{ABCD}=-\ep_{AB}\ep_{CD} \rho
\eeaa
Note that in this formula the star is the Hodge dual on spheres. Observe that when interchanging $e_3$ with $e_4$, the one form $\b$ becomes $-\bb$, the scalar $\sigma$ changes sign, while $\rho$ remains unchanged.

\subsubsection{Electromagnetic components}
Let $\F$ be a $2$-form in $(\MM, \g)$, and let $\dual \F$ denote the Hodge dual on $(\MM, \g)$ of $\F$, defined by $\dual \F_{\a\b}=\frac 1 2 \ep_{\mu\nu\a\b}\F^{\mu\nu}$. 

We define the null electromagnetic components:
\bea\label{decomposition-F}
\begin{split}
\bF_{A}:&=\F(e_A, e_4), \qquad \bbF_{A}:= \F(e_A, e_3) \\
 \rhoF:=& \frac 1 2 \F(e_3, e_4), \qquad \sigmaF:=\frac 1 2 \dual \F(e_3, e_4) 
\end{split}
\eea
The only remaining component of $\F$ is given by $\F_{AB}=-\ep_{AB}\sigmaF$.

Observe that when interchanging $e_3$ with $e_4$, the scalar $\rhoF$ changes sign, while $\sigmaF$ remains unchanged.

\subsection{The Einstein-Maxwell equations} \label{nseq}
If $(\MM, \g)$ satisfies the Einstein-Maxwell equations

\bea \label{Einstein-Maxwell-eq}
\R_{\mu\nu}&=&2 \F_{\mu \lambda} {\F_\nu}^{\lambda}- \frac 1 2 \g_{\mu\nu} \F^{\alpha\beta} \F_{\alpha\beta}, \label{Einstein-1}\\
\D_{[\alpha} \F_{\beta\gamma]}&=&0, \qquad \D^\alpha \F_{\alpha\beta}=0. \label{Maxwell}
\eea
the Ricci coefficients, curvature and electromagnetic components defined in \eqref{def1}, \eqref{def3} and \eqref{decomposition-F} satisfy a system of equations, which is presented in this section.

\subsubsection{Decomposition of Ricci and Riemann curvature}
The Ricci curvature of $(\MM, \g)$ can be expressed in terms of the electromagnetic null decomposition according to Einstein equation \eqref{Einstein-1}. We compute the following components of the Ricci tensor.
\beaa
\R_{A3}&=& 2 \F_{A \lambda} {\F_3}^{\lambda}=2\slashed{g}^{BC}\F_{AB}\F_{3C} - \F_{A3}\F_{34}=2\sigmaF{\ep_{A}}^C\bbF_C - 2\rhoF\bbF_A, \\
\R_{A4}&=& 2\sigmaF{\ep_{A}}^C\bF_C + 2\rhoF\bF_A, \\
\R_{33}&=&2 \g^{\lambda\mu}\F_{3 \lambda} \F_{3\mu}=2 \slashed{g}^{AB}\F_{3A} \F_{3B}= 2 \bbF \cdot \bbF, \\
\R_{34}&=&\g^{34}(\F_{34})^2+\slashed{g}^{AB} \slashed{g}^{CD} \F_{AC} \F_{DB}=2\rhoF^2 +2\sigmaF^2 \\
\R_{44}&=& 2 \bF\cdot  \bF, \\
\R_{AB}&=&2\F_{A\lambda}{\F_B}^{\lambda} - \slashed{g}_{AB} {\F^D}_{\lambda} {\F^\lambda}_D+\frac 1 2 \slashed{g}_{AB} \R_{34}=-2(\bF \hot \bbF)_{AB}+  (\rhoF^2 +\sigmaF^2)\slashed{g}_{AB}
\eeaa
We denote $\hot$ the symmetric traceless tensor product. 
Observe that as a consequence of \eqref{Einstein-1}, the scalar curvature of the metric $\g$ is zero, therefore we have $\slashed{\g}^{AC}\R_{AC}=\R_{34}$.

Using the decomposition of the Riemann curvature in Weyl curvature and Ricci tensor:
\bea
\label{WeylRiemanngeneral}
\R_{\a\b\gamma\delta}=\W_{\a\b\gamma\delta}+\frac 1 2 (\g_{\b\delta}\R_{\a\gamma}+\g_{\a\gamma}\R_{\b\delta}-\g_{\b\gamma}\R_{\a\delta}-\g_{\a\delta}\R_{\b\gamma}),
\eea
we can express the full Riemann tensor of $(\MM, \g)$ in terms of the above decompositions. We compute the following components of the Riemann tensor.
\beaa
\R_{A33B}&=&\W_{A33B}-\frac 1 2 \slashed{g}_{AB} \R_{33}=-\aa_{AB}- (\bbF\cdot  \bbF)\slashed{g}_{AB}, \\
\R_{A34B}&=&\W_{A34B}+\R_{AB}-\frac 1 2 \slashed{g}_{AB}\R_{34}=\rho \ \slashed{g}_{AB}-\sigma \ep_{AB}-(\bF \hot \bbF)_{AB}, \\
 \R_{A334}&=& \W_{A334}- \R_{A3}=2\bb_A -2\sigmaF{\ep_{A}}^C\bbF_C +2\rhoF\bbF_A, \\
\R_{3434}&=&\W_{3434}+2\R_{34}=4\rho + 4\rhoF^2 +4\sigmaF^2, \\
\R_{A3CB}&=&\W_{A3CB}+\frac 1 2 (\slashed{g}_{AC} \R_{3B}-\slashed{g}_{AB} \R_{3C})\\
&=&\ep_{CB} \dual \bb_A+ \slashed{g}_{AC} (\sigmaF{\ep_{B}}^C\bbF_C - \rhoF\bbF_B)-\slashed{g}_{AB} (\sigmaF{\ep_{C}}^D\bbF_D - \rhoF\bbF_C), \\
\R_{ABCD}&=&\W_{ABCD}+\frac 1 2 (\g_{BD}\R_{AC}+\g_{AC}\R_{BD}-\g_{BC}\R_{AD}-\g_{AD}\R_{BC})
\eeaa
We will use the above decompositions of Ricci and Riemann curvature in the derivation of the equations in Sections \ref{sec:nse}-\ref{bieq}.

\subsubsection{The null structure equations} \label{sec:nse}

The first equation for $\chi$ and $\chib$ is given by
\beaa
\nabb_3 \chib_{AB}+\chib_{A}^C \chib_{CB}+2\omb \chib_{AB}  &=& 2\nabb_B \xib_A+2\eta_B \xib_A+2\etab_A \xib_B-4\ze_B\xib_A+ \R_{A33B} \\
 \nabb_4 \chi_{AB}+\chi_{A}^C \chi_{CB}+2\om \chi_{AB}  &=& 2\nabb_B \xi_A+2\eta_B \xi_A+2\eta_A \xi_B+4\ze_B\xi_A+\R_{A44B}
\eeaa
We separate them in the symmetric traceless part, the trace part and the antisymmetric part. We obtain respectively:
\bea\label{nabb-3-chibh-general}
\begin{split}
\nabb_3 \chibh+\kab \ \chibh+2\omb \chibh&=&-2\DDs_2\xib -\aa+2(\eta+\etab-2\ze)\hot \xib  \\
\nabb_4 \chih+\ka \ \chih+2\om \chih&=&-2\DDs_2\xi -\a +2(\eta+\etab+2\ze)\hot \xi
\end{split}
\eea
\bea\label{nabb-3-kab-general}
\begin{split}
\nabb_3\kab+\frac 1 2 \kab^2+2\omb \kab=2\divv\xib-\chibh\cdot \chibh +2(\eta+\etab-2\ze)\c \xib-2 \bbF\cdot  \bbF\\
\nabb_4\ka+\frac 1 2 \ka^2+2\om \ka=2\divv\xi-\chih\cdot \chih +2(\eta+\etab+2\ze)\c \xi-2 \bF\c  \bF
\end{split}
\eea
\bea\label{curl-xib-general}
\begin{split}
\curll\xib&=& \xib \wedge (\eta+\etab-2\ze) \\
\curll\xi&=& \xi \wedge (\eta+\etab+2\ze)
\end{split}
\eea
where $\xi\wedge \eta=\ep_{AB}\xi^A \eta^B$.

The second equation for $\chi$ and $\chib$ is given by 
\beaa
\nabb_4\chib_{AB}  &=&  2\nabb_B \etab_A +2\om \chib_{AB} - \chi_B^C \chib_{AC}  +2(\xi_B \xib_A+ \etab_B \etab_A) + \R_{A34B} \\
\nabb_3\chi_{AB}  &=&  2\nabb_B \eta_A +2\omb \chi_{AB} - \chib_B^C \chi_{AC}  +2(\xib_B \xi_A+ \eta_B \eta_A) + \R_{A43B} 
\eeaa
We separate them in the symmetric traceless part, the trace part and the antisymmetric part. We obtain respectively
\bea\label{nabb-3-chih-general}
\begin{split}
\nabb_3\chih+\frac 12  \kab \,\chih-2 \omb \chih   &=&  -2 \slashed{\mathcal{D}}_2^\star \eta-\frac 1 2 \ka \chibh  + 2\eta\hot \eta+ 2\xib\hot \xi -2\bF \hot \bbF \\
\nabb_4\chibh+\frac 12  \ka \,\chibh -2 \om \chibh &=&  -2 \slashed{\mathcal{D}}_2^\star \etab -\frac 1 2 \kab \chih  + 2\etab\hot \etab+ 2\xib\hot \xi -2\bF \hot \bbF
\end{split}
\eea
\bea\label{nabb-3-ka-general}
\begin{split}
\nabb_3 \ka +\frac 12  \ka\, \kab-2\omb\,\ka &=&2 \divv \eta  -\chih\c\chibh+2\xi\c\xib +2\eta\c  \eta+2\rho \\
\nabb_4 \kab +\frac 12  \ka\,\kab-2\om\,\kab  &=&2 \divv \etab -\chih\c  \chibh+2\xi\c\xib +2\etab\c  \etab+2\rho 
\end{split}
\eea 
\bea\label{curl-eta-general}
\begin{split}
\slashed{\curl}\eta &=& -\frac 1 2\chi \wedge \chib  +\sigma,  \\ \slashed{\curl}\etab &=& \frac 1 2\chi \wedge \chib   -\sigma
\end{split}
\eea
The equations for $\ze$ are given by 
\beaa
\nabb_3 \ze&=&-2 \nabb \omb -\chib \c (\ze+\eta)+2\omb(\ze-\eta)+\chi \c \xib+2\om \xib - \frac 1 2 \R_{A334}, \\
-\nabb_4 \ze&=&-2 \nabb \om -\chi \c (-\ze+\etab)+2\om(-\ze-\etab)+\chib \c \xi+2\omb \xi - \frac 1 2 \R_{A443}
\eeaa
and therefore reducing to
\bea\label{nabb-3-ze-general}
\begin{split}
\nabb_3 \ze&=&-2 \nabb \omb -\chib \c (\ze+\eta)+2\omb(\ze-\eta)+\chi \c \xib+2\om \xib - \bb +\sigmaF{\ep_{A}}^C\bbF_C -\rhoF\bbF, \\
\nabb_4 \ze&=&2 \nabb \om +\chi \c (-\ze+\etab)+2\om(\ze+\etab)-\chib \c \xi-2\omb \xi -\b -\sigmaF \ep \c \bF -\rhoF\bF 
\end{split}
\eea
The equations for $\xi$ and $\xib$ are given by 
\beaa
\nabb_4\xib- \nabb_3 \etab &=&4\om\xib -\chib\c(\eta-\etab) -\frac 1 2 \R_{A334}, \\
\nabb_3\xi- \nabb_4 \eta &=&4\omb\xi+ \chi\c(\eta-\etab) -\frac 1 2 \R_{A443}
\eeaa
and therefore reducing to
\bea\label{nabb-4-xib-general}
\begin{split}
\nabb_4\xib- \nabb_3 \etab &=& -\chib\c(\eta-\etab)+4\om\xib -\bb +\sigmaF \ep \c \bbF -\rhoF\bbF, \\
\nabb_3\xi-  \nabb_4 \eta &=& \chi\c(\eta-\etab)+4\omb\xi +\b +\sigmaF \ep \c \bF +\rhoF\bF \label{nabb-3-xi-general}
\end{split}
\eea
The equation for $\om$ and $\omb$ is given by
\beaa
\nabb_4 \omb+\nabb_3\om&=& 4\om\omb+\xi\c\xib +\ze\c(\eta-\etab) -\eta\c\etab+\frac 1 4 \R_{3434}
\eeaa
and therefore reducing to
\bea\label{nabb-4-omb-general}
\nabb_4 \omb+\nabb_3\om&=& 4\om\omb +\ze\c(\eta-\etab)+\xi\c\xib -\eta\c\etab+\rho + \rhoF^2 +\sigmaF^2
\eea
The spacetime equations that generate Codazzi equations are
\beaa
\nabb_C \chib_{AB} + \ze_B \chib_{AC} &=& \nabb_B \chib_{AC} + \ze_C \chib_{AB} + \R_{A3CB}, \\
\nabb_C \chi_{AB} - \ze_B \chi_{AC} &=& \nabb_B \chi_{AC} - \ze_C \chi_{AB} + \R_{A4CB}
\eeaa
Taking the trace in $C,A$ we obtain
\bea\label{codazzi-general}
\begin{split}
\divv \chibh_B&=&(\chibh\c \ze)_B-\frac 1 2 \kab \ze_B+\frac 1 2 (\nabb_B \kab) +\bb_B+\sigmaF{\ep_{B}}^C\bbF_C -\rhoF\bbF_B, \\
\divv \chih_B&=&-(\chih\c \ze)_B+\frac 1 2 \ka \ze_B+\frac 1 2 (\nabb_B \ka) -\b_B+\sigmaF{\ep_{B}}^C\bF_C +\rhoF\bF_B
\end{split}
\eea
The spacetime equation that generates Gauss equation is 
\beaa
\gS^{AC} \gS^{BD} \R_{ABCD} = 2K + \frac 1 2 \ka\kab - \chih \c \chibh
\eeaa
where $K$ is the Gauss equation of the surface orthogonal to $e_3$ and $e_4$. It therefore reduces to
\bea\label{Gauss-general}
K=- \frac 1 4 \ka\kab + \frac 1 2 (\chih, \chibh)-\rho+\rhoF^2 +\sigmaF^2  
\eea

\subsubsection{The Maxwell equations}\label{Maxwell-general}
For completeness, we derive here the null decompositions of Maxwell equations \eqref{Maxwell}.

The equation $\D_{[\alpha} \F_{\beta\gamma]}=0$ gives three independent equations. The first one is obtained in the following way:
\beaa
0&=& \D_A \F_{3 4}+\D_3 \F_{4A}+\D_4 \F_{A3}\\
&=& \nabb_A\F_{3 4}-\F(\ze_A e_3 +\chib_{AB} e^B, e_4)-\F(e_3, -\ze_A e_4+\chi_{AB} e^B)+\nabb_3 \F_{4A}- \F(2\omb e_4+2\eta^B e_B, e_A)\\
&&- \F(e_4, \eta_A e_3+\xib_A e_4)+\nabb_4 \F_{A3}- \F(\etab_A e_4+\xi_A e_3, e_3)- \F(e_A, 2\om e_3+2\etab^B e_B)\\
&=& 2\nabb_A\rhoF-\frac 1 2 \kab \bF_A-(\chibh \c \bF)_A+\frac 1 2 \ka \bbF_A+(\chih \c \bbF)_A-\nabb_3 \bF_A+2\omb \bF_A-2\om \bbF_A\\
&&-2(\eta^B-\etab^B)\ep_{AB} \sigmaF+2 (\eta_A+\etab_A)\rhoF+\nabb_4 \bbF_A
\eeaa
which reduces to
\bea\label{Maxwell-1}
\begin{split}
\nabb_3 \bF_A-\nabb_4 \bbF_A&= -\left(\frac 1 2 \kab-2\omb\right) \bF_A+\left(\frac 1 2 \ka -2\om\right) \bbF_A+ 2\nabb_A\rhoF+2 (\eta_A+\etab_A)\rhoF\\
&-2(\eta^B-\etab^B)\ep_{AB} \sigmaF+(\chih \c \bbF)_A-(\chibh \c \bF)_A
\end{split}
\eea
The second and third equations are obtained in the following way:
\beaa
0&=& \D_A \F_{B 3}+\D_B \F_{3A}+\D_3 \F_{AB}\\
&=& \nabb_A \F_{B 3}- \F(e_B, \ze_A e_3+\chib_{AC} e^C)+\nabb_B \F_{3A}- \F(\ze_B e_3+\chib_{BC} e^C, e_A)+\nabb_3 \F_{AB}\\
&&-\F(\eta_A e_3+\xib_A e_4, e_B) -\F(e_A, \eta_B e_3+\xib_B e_4)\\
&=& \nabb_A \bbF_B-\nabb_B \bbF_A-(\ze_A-\eta_A)\bbF_B+(\ze_B-\eta_B) \bbF_A+\xib_A \bF_B-\xib_B \bF_A\\
&&+ (\chib_{AC}{\ep^C}_B+\chib_{BC} {\ep^C}_A) \sigmaF-\ep_{AB}\nabb_3 \sigmaF
\eeaa
Contracting with $\ep^{AB}$ we obtain
\bea\label{nabb-3-sigmaF-general}
\begin{split}
\nabb_3 \sigmaF+\kab \ \sigmaF&=& \slashed{\curl}\bbF-(\ze-\eta)\wedge \bbF+\xib \wedge \bF, \\
 \nabb_4 \sigmaF+\ka \ \sigmaF&=& \slashed{\curl}\bF+(\ze+\etab)\wedge \bF+\xi \wedge \bbF
 \end{split}
\eea
The equation $\D^\mu \F_{\mu\nu}=\slashed{\g}^{BC} \D_B \F_{C\nu}-\frac 1 2 \D_4 \F_{3\nu}-\frac 1 2 \D_3\F_{4\nu}=0$ gives three additional independent equations. The first one is obtained in the following way:
\beaa
0&=& \slashed{\g}^{BC} \D_B \F_{CA}-\frac 1 2 \D_4 \F_{3A}-\frac 1 2 \D_3\F_{4A}\\
&=&\slashed{\g}^{BC} (-\slashed{\ep}_{CA}\nabb_B \sigmaF-\F(\frac 1 2 \chib_{BC} e_4+\frac 1 2 \chi_{BC} e_3, e_A)-\F(e_C, \frac 1 2 \chib_{AB} e_4+\frac 1 2 \chi_{AB} e_3))\\
&&+\frac 1 2 \nabb_4 \bbF_A+\frac 1 2  \F(2\om e_3+2\etab^C e_C, e_A)+\frac 1 2  \F(e_3, \etab_A e_4+\xi_A e_3)\\
&&+\frac 1 2 \nabb_3\bF_A+\frac 1 2 \F(2\omb e_4+2\eta^C e_C, e_A)+\frac 1 2 \F(e_4, \eta_A e_3+\xib_A e_4)\\
&=& \ep_{AC}\nabb^C \sigmaF+\frac 1 4 \kab \bF_A+\frac 1 4 \ka \bbF_A-\frac 1 2 (\chibh \c \bF)_A-\frac 1 2 (\chih \c \bbF)_A+(-\eta_A+\etab_A)\rhoF\\
&&+\frac 1 2 \nabb_4 \bbF_A-\om \bbF_A-\omb \bF_A+\frac 1 2 \nabb_3 \bF_A+(\eta^B+\etab^B)\ep_{AB} \sigmaF
\eeaa
which reduces to
\bea\label{Maxwell-2}
\begin{split}
\nabb_3 \bF_A+\nabb_4 \bbF_A&=-\left(\frac 1 2 \kab -2\omb \right)\bF_A-\left(\frac 1 2 \ka -2\om \right)\bbF_A+2(\eta_A-\etab_A)\rhoF\\
&-2\ep_{AC}\nabb^C \sigmaF-2(\eta^B+\etab^B)\ep_{AB} \sigmaF+ (\chibh \c \bF)_A+ (\chih \c \bbF)_A
\end{split}
\eea
Summing and subtracting \eqref{Maxwell-1} and \eqref{Maxwell-2} we obtain
\bea\label{nabb-3-bF-general}
\begin{split}
\nabb_3 \bF_A+\left(\frac 1 2 \kab-2\omb\right) \bF_A&=& -\DDs_1(\rhoF, \sigmaF) +2\eta_A\rhoF-2\eta^B\ep_{AB} \sigmaF+(\chih \c \bbF)_A\\
\nabb_4 \bbF_A+\left(\frac 1 2 \ka -2\om \right)\bbF_A&=&\DDs_1(\rhoF, -\sigmaF)-2\etab_A\rhoF-2\etab^B\ep_{AB} \sigmaF+ (\chibh \c \bF)_A
\end{split}
\eea
The last two equations are given by 
\beaa
0&=& \slashed{\g}^{BC} \D_B \F_{C4}-\frac 1 2 \D_4 \F_{34}\\
&=&\slashed{\g}^{BC} (\nabb_B \bF_C-\F(\frac 1 2 \chib_{BC} e_4+\frac 1 2 \chi_{BC} e_3, e_4)-\F(e_C, -\ze_B e_4+\chi_{BA} e^A))\\
&&-\frac 1 2 (2\nabb_4 \rhoF-\F(2\om e_3+2\etab^A e_A, e_4)-\F(e_3, -2\om e_4+2\xi^A e_A))\\
&=&\divv\bF-\chibh_{AB} \ep^{AB} \sigmaF- \ka \rhoF- \nabb_4 \rhoF+((\ze+\etab) \c \bF)-\xi \c \bbF
\eeaa
which reduces to
\bea\label{nabb-4-rhoF-general}
\begin{split}
\nabb_4 \rhoF+ \ka \rhoF&=& \divv\bF+(\ze+\etab) \c \bF-\xi \c \bbF, \\
\nabb_3 \rhoF+ \kab \rhoF&=& -\divv\bbF+(\ze-\eta) \c \bbF-\xib \c \bF
\end{split}
\eea

\subsubsection{The Bianchi equations} \label{bieq}
The Bianchi identities for the Weyl curvature are given by 
\beaa
 \D^\a \W_{\a\b\gamma\delta}&=&\frac 1 2 (\D_\gamma \R_{\b\delta}-\D_{\delta}\R_{\b\gamma})=:J_{\b\gamma\delta} \\
 \D_{[\sigma }\W_{\gamma\delta] \a\b}&=&\g_{\delta \b}J_{\a\gamma\sigma}+\g_{\gamma \a}J_{\b\delta\sigma}+\g_{\sigma \b}J_{\a\delta\gamma}+\g_{\delta \a}J_{\b\sigma\gamma}+\g_{\gamma \b}J_{\a\sigma\delta}+\g_{\sigma \a}J_{\b\gamma\delta}:= \tilde{J}_{\sigma\gamma\delta\a\b}
\eeaa 
The Bianchi identities for $\a$ and $\aa$ are given by
\beaa
\nabb_3\a_{AB}+\frac 1 2 \kab\,\a_{AB}-4\omb \a_{AB}&=&-2 (\DDs_2\, \b)_{AB} -3(\chih_{AB} \rho+\dual\, \chih_{AB} \sigma)+((\ze+4\eta)\hat{\otimes} \b )_{AB}+\\
&&+\frac 1 2 (\tilde{J}_{3A4B4}+\tilde{J}_{3B4A4}+J_{434}\slashed{\g}_{AB})\\
\nabb_4\aa_{AB}+\frac 1 2 \ka\,\aa_{AB}-4\om \aa_{AB}&=&2 (\DDs_2\, \bb)_{AB} -3(\chibh_{AB} \rho+\dual\, \chibh_{AB} \sigma)-((-\ze+4\etab)\hat{\otimes} \bb )_{AB}+\\
&&+\frac 1 2 (\tilde{J}_{4A3B3}+\tilde{J}_{4B3A3}+J_{343}\slashed{\g}_{AB})
\eeaa
Using that $\tilde{J}_{3A4B4}=-\slashed{\g}_{A B}J_{4 34}+2J_{B A4}$, it is reduced to
\bea\label{nabb-3-a-general}
\begin{split}
\nabb_3\a_{AB}+\frac 1 2 \kab\,\a_{AB}-4\omb\a_{AB}&=&-2 (\DDs_2\, \b)_{AB} -3(\chih_{AB} \rho+\dual\, \chih_{AB} \sigma)+((\ze+4\eta)\hat{\otimes} \b )_{AB}+\\
&&+J_{B A4}+J_{AB4}-\frac 1 2 \slashed{g}_{A B}J_{4 34}, \\
\nabb_4\aa_{AB}+\frac 1 2 \ka\,\aa_{AB}-4\om \aa_{AB}&=&2 (\DDs_2\, \bb)_{AB} -3(\chibh_{AB} \rho+\dual\, \chibh_{AB} \sigma)-((-\ze+4\etab)\hat{\otimes} \bb )_{AB}+\\
&&+J_{B A3}+J_{AB3}-\frac 1 2 \slashed{g}_{A B}J_{3 43}\end{split}
\eea 
The Bianchi identities for $\b$ and $\bb$ are given by 
\bea\label{Bianchi-b-1}
\begin{split}
\nabb_4 \b_A+ 2\ka  \b_A+2\om \b_A &=&\divv\a_A +((2\ze+\etab)\c \a)_A+3(\xi_A \rho+\dual \xi_A \sigma) -J_{4A4}, \label{nabb-4-b-general}\\
\nabb_3 \bb_A+ 2\kab  \bb_A+2\omb \bb_A &=&-\divv\aa_A +((2\ze-\eta)\c \aa)_A-3(\xib_A \rho+\dual \xib_A \sigma) +J_{3A3}
\end{split}
\eea
and 
\bea\label{Bianchi-b-2}
\begin{split}
\nabb_3 \b_A+\kab \b_A-2\omb\, \b_A &=&\DDs_1(-\rho, \sigma)_A+2(\chih \c\bb)_A +\xib \c \a+3(\eta_A \rho +\dual\eta_A\,   \sigma) + J_{3A4}, \\
\nabb_4 \bb_A+\ka \bb_A-2\om\, \bb_A &=&\DDs_1(\rho, \sigma)_A+2(\chibh \c\b)_A-\xi \c \aa -3(\etab_A \rho -\dual\etab_A\,   \sigma) - J_{4A3} 
\end{split}
\eea
The Bianchi identity for $\rho$ is given by 
\bea\label{nabb-4-rho-general}
\begin{split}
\nabb_4 \rho+\frac 3 2 \ka \rho&=&\divv\b+(2\etab+\ze)\c\b  -\frac 1 2 (\chibh \c \a )-2\xi \c \bb -\frac 1 2 J_{434}, \\
\nabb_3 \rho+\frac 3 2 \kab \rho&=&-\divv\bb-(2\eta-\ze)\c\bb  +\frac 1 2 (\chih \c \aa )+2\xib \c \b -\frac 1 2 J_{343}
\end{split}
\eea
The Bianchi identity for $\sigma$ is given by 
\beaa
\nabb_4 \sigma+\frac 3 2 \ka \sigma&=&-\slashed{\curl}\b-(2\etab+\ze)\wedge\b  +\frac 1 2 \chibh \wedge \a  -\frac 1 2 \dual J_{434}, \\
\nabb_3 \sigma+\frac 3 2 \kab \sigma&=&-\slashed{\curl}\bb-(2\eta-\ze)\wedge\b  -\frac 1 2 \chih \wedge \aa  +\frac 1 2 \dual J_{343}
\eeaa
and writing $\dual J_{434}=\frac 1 2 J_{4\mu\nu}{\ep^{\mu\nu}}_{34}=- J_{4AB}\ep_{AB}=(J_{AB4}-J_{BA4})\ep^{AB}$, we obtain
\bea\label{nabb-4-sigma-general}
\begin{split}
\nabb_4 \sigma+\frac 3 2 \ka \sigma&=&-\slashed{\curl}\b-(2\etab+\ze)\wedge\b  +\frac 1 2 \chibh \wedge \a  -\frac 1 2 (J_{AB4}-J_{BA4})\ep^{AB}, \\
\nabb_3 \sigma+\frac 3 2 \kab \sigma&=&-\slashed{\curl}\bb-(2\eta-\ze)\wedge\bb  -\frac 1 2 \chih \wedge \aa  +\frac 1 2 (J_{AB3}-J_{BA3})\ep^{AB}
\end{split}
\eea

\section{The Bondi gauge}\label{double-null-gauge}

In this section, we introduce the choice of gauge we use throughout the paper to perform the perturbation of the solution to the Einstein-Maxwell equations. This choice of gauge introduces a restriction on the form of the metric $\g$ on $\MM$, which nevertheless does not saturate the gauge freedom of the Einstein-Maxwell equations.\footnote{The gauge freedom remaining will be exploited later by the pure gauge solutions (see Section \ref{sec:ssgauge}).}

Our choice of gauge is the outgoing geodesic foliation, also called Bondi gauge  \cite{Bondi}. This choice of coordinates is particularly suited to exploit properties of decay towards null infinity, which we will take advantage of. Another advantage of the Bondi gauge is that all metric and connection coefficients are regular near the horizon $\mathcal{H}^+$.

 We begin in Section \ref{local-bondi} with the definition of local Bondi gauge. In Section \ref{equations-implied-by-Bondi}, we derive the equations for the metric components and the Ricci coefficients implied by such a choice of gauge. These equations will be added to the set of Einstein-Maxwell equations derived in Section \ref{nseq}. In Section \ref{average-quantities-general}, we derive the equations for the average quantities in a Bondi gauge, which are used later in the derivation of the linearized equations for scalars.
 
 \subsection{Local Bondi gauge}\label{local-bondi}

Let $(\MM, \g)$ be a $3+1$ dimensional Lorentzian manifold. 

\subsubsection{Local Bondi form of the metric}
In a neighborhood of any point $p \in \MM$, we can introduce local coordinates $(u, s, \th^1, \th^2)$ such that the metric can be expressed in the  following \textit{Bondi form}  \cite{Bondi}: 
\bea\label{double-null-metric}
  \g&=&-2 \vsi du ds+\vsi^2\Omegab du^2+\slashed{g}_{AB} \left(d\th^A-\frac 1 2 \vsi \underline{b}^A du \right) \left(d\th^B-\frac 1 2\vsi \underline{b}^B du   \right)
  \eea
for two spacetime functions $\Omegab, \vsi :\MM\to \mathbb{R}$, with $\vsi\neq0$, a $S_{u,s}$-tangent vector $\underline{b}^A$ and a $S_{u,s}$-tangent covariant symmetric $2$-tensor $\slashed{g}_{AB}$. Here $S_{u,s}$ denotes the two-dimensional Riemannian manifold (with metric $\slashed{g}$) obtained as intersection of the hypersurfaces of constant $u$ and $s$.

Note that $\{ u=\text{constant} \}$ are outgoing null hypersurfaces for $\g$.

\subsubsection{Local normalized null frame}

We define a normalized outgoing geodesic null frame $\mathscr{N}=\left\{e_A, e_3, e_4\right\}$ associated to the above coordinates as follows. We define
\bea\label{null-frame-double-null}
e_3=2\vsi^{-1}\partial_u+\underline{\Omega}\partial_s+\underline{b}^A \partial_{\th^A}, \qquad e_4=\partial_s, \qquad e_A=\partial_{\th^A}
\eea

Observe that relations \eqref{relations-null-frame} hold. In particular, notice that the surfaces $S_{u,s}$ define a foliation of the spacetime of the type described in Section \ref{sec:genmfld}, therefore the decomposition in null frame of Ricci coefficients, curvature and electromagnetic components described above can be applied to this case.

To the foliation $S_{u,s}$ we can associate a scalar function $r(u,s)$ defined by 
\bea\label{r-general}
|S_{u,s}|=4\pi r(u,s)^2
\eea
where $|S_{u,s}|$ is the area of $(S_{u,s}, \slashed{g})$.

\subsection{Relations in the Bondi gauge}\label{equations-implied-by-Bondi}

The restriction to perturbations of the metric of the form \eqref{double-null-metric} verifying the Einstein-Maxwell equations gives additional relations between the Ricci coefficients as defined in Section \ref{ricci-coefficients}. We summarize them in the following lemma.

\begin{lemma}\label{double-null-Ricci} The Ricci coefficients associated to a metric $\g$ of the form \eqref{double-null-metric} with respect to the null frame \eqref{null-frame-double-null} verify:
\bea
\xi_A=0, \qquad \om=0, \qquad \etab_A=-\ze_A\label{substitute-xi-xib-ze}
 \eea
 The metric components satisfy
 \bea
 \nabb_A \vsi&=&\eta_A-\ze_A, \label{nabb-4-vsi}\\ 
 \nabb_4\vsi&=&0 \label{nabb-4-vsi-1}\\
 \nabb_A\Omegab&=&-\xib_A+\Omegab \left(\ze_A-\eta_A\right), \label{nabb-A-Omegab}\\
  \nabb_4 \Omegab&=&-2\omb, \label{additional-metric-0}\\
\nabb_4\underline{b}_A-\chi_{AB}\underline{b}^B&=&-2\left(\eta_A+\ze_A\right), \label{additional-metric-1}\\
\partial_s (\slashed{g}_{AB})&=&2 \chih_{AB}+\ka \slashed{g}_{AB}, \label{additional-metric-2}\\
2\vsi^{-1}\partial_u(\slashed{g}_{AB})+\Omegab \partial_s(\slashed{g}_{AB})&=&  2 \chibh_{AB}+2(\DDs_2 \underline{b})_{AB}+(\kab -\divv \underline{b}) \slashed{g}_{AB}\label{additional-metric-3}
\eea
\end{lemma}
\begin{proof} The vectorfield $\pr_s$ is geodesic, i.e. $\D_{e_4}e_4=0$. Using \eqref{derivatives-}, this implies $\om=0$ and $\xi_A=0$.

Since $e_3(u)=2\vsi^{-1}$ and $e_4(u)=e_A(u)=0$, we can apply $\big[e_3, e_A\big]$ and $[e_3, e_4]$ to $u$ and using \eqref{commutators} we obtain
\beaa
\big[e_3, e_A\big]u&=& \nabb_3 \nabb_A(u)-\nabb_A(\nabb_3(u))=-2\nabb_A(\vsi^{-1})=2\vsi^{-2}\nabb_A(\vsi)\\
\big[e_3, e_A\big]u&=&(\eta_A-\ze_A) e_3(u)+\xib_A e_4(u)-\chib_{AB} e^B(u)=2(\eta_A-\ze_A) \vsi^{-1} 
\eeaa
\beaa
\big[e_3, e_4\big]u&=& e_3 e_4(u)-e_4 e_3(u)=-2\nabb_4(\vsi^{-1})=2\vsi^{-2}\nabb_4(\vsi)\\
\big[e_3, e_4\big]u&=&2\omb  e_4(u)+2(\eta^B -\etab^B) e_B(u)=0
\eeaa
Since $e_3(s)=\Omegab$, $e_4(s)=1$, $e_A(s)=0$, we can apply $\big[e_4, e_A\big]$, $\big[e_3, e_A\big]$ and $\big[e_3, e_4 \big]$ to $s$, using \eqref{commutators}, and obtain
\beaa
\big[e_4, e_A\big]s&=& e_4 e_A(s)-e_A(e_4(s))=0\\
\big[e_4, e_A\big]s&=&(\etab_A+\ze_A) e_4(s)-\chi_{AB} e^B(s)=\etab_A+\ze_A\\
\big[e_3, e_A\big]s&=& e_3 e_A(s)-e_A(e_3(s))=-\nabb_A \Omegab\\
\big[e_3, e_A\big]s&=&(\eta_A-\ze_A) e_3(s)+\xib_A e_4(s)-\chib_{AB} e^B(s)=(\eta_A-\ze_A) \Omegab+\xib_A, \\
\big[e_3, e_4\big]s&=& e_3 e_4(s)-e_4 e_3(s)=-\nabb_4 \Omegab\\
\big[e_3, e_4\big]s&=&2\omb  e_4(s)+2(\eta^B -\etab^B) e_B(s)=2\omb
\eeaa
Since $e_3(\th^A)=\underline{b}^A$, $e_4(\th_A)=0$, $e_A(\th_B)=\de_{AB}$ we can apply $\big[e_3, e_4\big]$ to $\th^A$, and obtain
\beaa
\big[e_3, e_4\big]\th^A&=& e_3 e_4(\th^A)-e_4 e_3(\th^A)=-D\underline{b}^A\\
\big[e_3, e_4\big]\th^A&=&2\omb  e_4(\th^A)+2(\eta^B -\etab^B) e_B(\th^A)=2(\eta^A +\ze^A) 
\eeaa
Using \eqref{covariant-Lie}, we obtain the desired relation. 
We now derive the equation for the metric $\slashed{g}$. Using \eqref{covariant-Lie}, we obtain
\beaa
\underline{D}\slashed{g}_{AB}&=&\nabb_3 \slashed{g}_{AB}+\chib_{AC}\slashed{g}^C_B+\chib_{BC}\slashed{g}_A^C= 2 \chib_{AB}=2 \chibh_{AB}+ \kab \slashed{g}_{AB}, \\
D\slashed{g}_{AB}&=& 2 \chih_{AB}+ \ka \slashed{g}_{AB}
\eeaa
In view of the formula for the projected Lie-derivative and the null frame \eqref{null-frame-double-null}, 
\beaa
\underline{D}\slashed{g}_{AB}&=&2\vsi^{-1}\partial_u(\slashed{g}_{AB})+\Omegab \partial_s(\slashed{g}_{AB})+(\nabb_A \underline{b}_B+\nabb_B \underline{b}_A) \\
D\slashed{g}_{AB}&=&\partial_s(\slashed{g}_{AB})
\eeaa
Combining the above, we obtain the desired relations.
\end{proof}

In considering solutions $(\MM, \g)$ to the Einstein-Maxwell equations of the form \eqref{double-null-metric}, we use relations \eqref{substitute-xi-xib-ze} to set $\xi$ and $\om$ to vanish and to substitute $\etab$ in terms of $\ze$ in the equations. We will also add equations \eqref{nabb-4-vsi}-\eqref{additional-metric-3} to the set of linearized equations (see Section \ref{section-linearized-metric}).

\subsection{Transport equations for average quantities}\label{average-quantities-general}

Recall the definition of $S$-average given in \eqref{def-S-average}. We specialize here to the foliation in surfaces given by the Bondi gauge, and we derive the transport equations for average quantities. They shall be used in Chapter \ref{linearized-equations-chapter} to derive the linearized equations for the scalars involved in the perturbation.

To simplify the notation, we denote in the following $S=S_{u,s}$ and $r=r(u,s)$.

\begin{proposition}[Proposition 2.2.9 in \cite{stabilitySchwarzschild}] \label{nabb-4-int-f} For any scalar function $f$, we have 
\beaa
\nabb_4 \left(\int_S f\right)&=&\int_S( \nabb_4f+\ka f), \\
\nabb_3 \left(\int_S f\right)&=&\int_S (\nabb_3f+\kab f)+\text{Err}[\nabb_3(\int_{S} f)]
\eeaa
where the error term is given by the formula
\beaa
\text{Err}[\nabb_3(\int_{S} f)]:&=&-\vsi^{-1}\check{\vsi}\int_{S}(\nabb_3f+\kab f) +\vsi^{-1}\int_{S} \check{\vsi}(\nabb_3 f+\kab f)+ (\check{\Omegab}+\vsi^{-1}\ov{\Omegab} \check{\vsi}) \int_S( \nabb_4f+\ka f)\\
&&-\vsi^{-1}\ov{\Omegab}\int_{S} \check{\vsi}(\nabb_4f+\ka f)-\vsi^{-1} \int_{S} \check{\Omegab} \vsi(\nabb_4 f+\ka f)
\eeaa
In particular, we have 
\beaa
\nabb_4 r=\frac r 2  \ov{\ka}, \qquad \nabb_3 r=\frac r 2 \left(\ov{\kab}+\underline{A}\right)
\eeaa
where 
\beaa
\underline{A}:= -\vsi^{-1}\ov{\kab}\check{\vsi}+\ov{\ka}(\check{\Omegab}+\vsi^{-1}\ov{\Omegab}\check{\vsi})+\vsi^{-1}\ov{\check{\vsi}\check{\kab}}-\vsi^{-1}\ov{\Omegab}\ov{\check{\vsi}\check{\ka}}-\vsi^{-1}\ov{\check{\Omegab}\vsi\ka}
\eeaa
\end{proposition}
\begin{proof} Recalling that $e_4= \partial_s$, we compute 
\beaa
\partial_s\left(\int_S f\right)&=&\int_S (\partial_s f +\slashed{g}(\D_A \partial_s, e^A)f)
\eeaa
We have $\D_A \partial_s=\D_A e_4$ and using the relations \eqref{derivatives-}, we obtain 
\beaa
\slashed{g}(\D_A \partial_s, e^A)&=&\slashed{g}(-\ze_A e_4+\chi_{AC} e^C, e^A)=\ka 
\eeaa
We easily deduce the desired relation along $e_4$. The formula for derivative along $e_3$ is obtained in a similar way. See \cite{stabilitySchwarzschild}.

The equality for $r$ follows by applying the Lemma to $f=1$.
\end{proof}

\begin{corollary}[Corollary 2.2.11 in \cite{stabilitySchwarzschild}]\label{average-check} For any scalar function $f$ we have
\beaa
\nabb_4(\ov{f})&=& \ov{\nabb_4f}+\ov{\check{\ka}\check{f}}, \\
\nabb_4(\check{f})&=& \check{\nabb_4 f}-\ov{\check{\ka}\check{f}}
\eeaa
and
\beaa
\nabb_3(\ov{f})&=& \ov{\nabb_3f}+\text{Err}[\nabb_3(\ov{f})] \\\nabb_3(\check{f})&=&\check{\nabb_3 f}-\text{Err}[\nabb_3(\ov{f})]
\eeaa
where 
\beaa
\text{Err}[\nabb_3(\ov{f})]:&=&-\vsi^{-1}\check{\vsi}(\ov{\nabb_3f+\kab f}-\ov{\kab}\ov{f}) +\vsi^{-1}(\ov{ \check{\vsi}(\nabb_3 f+\kab f)}-\ov{\check{\vsi}\check{\kab}} \ov{f})\\
&&+ (\check{\Omegab}+\vsi^{-1}\ov{\Omegab} \check{\vsi}) ( \ov{\nabb_4f+\ka f}-\ov{\ka}\ov{f})-\vsi^{-1}\ov{\Omegab}(\ov{ \check{\vsi}(\nabb_4f+\ka f)}-\ov{\check{\vsi}\check{\ka}}\ov{f})\\
&&-\vsi^{-1} (\ov{\check{\Omegab} \vsi(\nabb_4 f+\ka f)}-\ov{\check{\Omegab}\vsi\ka}\ov{f})+\ov{\check{\kab}\check{f}}
\eeaa
\end{corollary}

\section{The Reissner-Nordstr{\"o}m spacetime}\label{RN-chapter}

In this section, we introduce the Reissner-Nordstr{\"o}m exterior metric, as well as relevant background structure. For completeness, we collect here standard coordinate transformations relevant to the study of Reissner-Nordstr{\"o}m spacetime (see for example \cite{Exact}), even if not directly used in our proof. 

 We first fix in Section \ref{diff-structure} an ambient manifold-with-boundary $\MM$ on which we define the Reissner-Nordstr{\"o}m exterior metric $ \g_{M, Q}$ with parameters $M$ and $Q$ verifying $|Q| < M$. We shall then pass to more convenient sets of coordinates, like double null coordinates, outgoing and ingoing Eddington-Finkelstein coordinates, and we shall show how these sets of coordinates relate to the standard form of the metric as given in \eqref{RNintro}.

 In Section \ref{Outgoing-coords}, we show that the Reissner-Nordstr{\"o}m metric admits a Bondi form as described in the previous chapter.  We then describe the null frames associated to such coordinates and the values of Ricci coefficients, curvature and electromagnetic components. 
 
Finally, in Section \ref{sec:commutation} we recall the symmetries of Reissner-Nordstr\"om spacetime and present the main operators and commutation formulae. We also recall the main properties of decomposition in spherical harmonics in Reissner-Nordstr{\"o}m spacetime. 

We will follow closely Section 4 of \cite{DHR}, where the main features of the Schwarzschild metric and differential structure are easily extended to the Reissner-Nordstr{\"o}m solution. 

\subsection{Differential structure and metric}\label{diff-structure}

We define in this section the underlying differential structure and metric in terms of the Kruskal coordinates. 

\subsubsection{Kruskal coordinate system}

Define the manifold with boundary
\begin{align} \label{SchwSchmfld}
\mathcal{M} := \mathcal{D} \times S^2 := \left(-\infty,0\right] \times \left(0,\infty\right) \times S^2
\end{align}
with coordinates $\left(U,V,\theta^1,\theta^2\right)$.
We will refer to these coordinates as \emph{Kruskal coordinates}.
 The boundary 
\[
\mathcal{H}^+:=\{0\}  \times \left(0,\infty\right) \times S^2
\]
will be referred to as the \emph{horizon}. 
We denote by $S^2_{U,V}$ the $2$-sphere $\left\{U,V\right\} \times S^2 \subset \mathcal{M}$ in $\mathcal{M}$.

\subsubsection{The Reissner-Nordstr{\"o}m metric}

We define the Reissner-Nordstr{\"o}m metric on $\mathcal{M}$ as follows.
Fix two parameters $M>0$ and $Q$, verifying $|Q|<M$. Let the function $r : \mathcal{M} \rightarrow \left[M+\sqrt{M^2-Q^2},\infty\right)$ be 
given implicitly as a function of the coordinates $U$ and $V$ by
\begin{align}\label{implicit-definition-r-U-V}
-UV = \frac{4r_{+}^4}{(r_{+}-r_{-})^2} \Big|\frac{r-r_{+}}{r_{+}}\Big| \Big|\frac{r_{-}}{r-r_{-}}\Big|^{\left(\frac{r_{-}}{r_{+}}\right)^2} \exp\Big(\frac{r_{+}-r_{-}}{r_{+}^2} r\Big) \, ,
\end{align}
where 
\bea\label{definiion-rpm}
r_{\pm}=M\pm \sqrt{M^2-Q^2}
\eea
 We will also denote 
\bea\label{def-rH}
r_{\mathcal{H}}=r_+=M+\sqrt{M^2-Q^2}
\eea
Define also
\beaa
\Up_K \left(U,V\right) &=& \frac{r_{-}r_{+}}{4r(U,V)^2} \Big( \frac{r(U,V)-r_{-}}{r_{-}}\Big)^{1+\left(\frac{r_{-}}{r_{+}}\right)^2}\exp\Big(-\frac{r_{+}-r_{-}}{r_{+}^2} r(U,V)\Big) \\
 \gamma_{AB} &=& \textrm{standard metric on $S^2$} \, .
\eeaa
Then the Reissner-Nordstr{\"o}m metric $\g_{M, Q}$ with parameters $M$ and $Q$ is defined to be the metric:
\begin{align} \label{sskruskal}
\g_{M, Q} = -4 \Up_K \left(U,V\right) d{U} d{V} +   r^2 \left(U,V\right) \gamma_{AB} d{\theta}^A d{\theta}^B.
\end{align}

 The Reissner-Nordstr{\"o}m family
  of spacetimes $(\mathcal{M},\g_{M, Q})$ is the unique electrovacuum spherically symmetric spacetime. It is a static and asymptotically flat spacetime. The parameter $Q$ may be interpreted as the charge of the source. This metric clearly reduces to the Schwarzschild metric when $Q=0$, therefore $M$ can be interpreted as the mass of the source.

Note that the horizon $\mathcal{H}^+=\partial\mathcal{M}$ 
is a null hypersurface with respect to $\g_{M, Q}$. We will use the standard
spherical coordinates $(\theta^1,\theta^2)=(\theta, \phi)$, in which case
the metric $\gamma$ takes the explicit form
\begin{equation}
\label{gammaexplicit}
\gamma= d\theta^2+\sin^2\theta d\phi^2.
\end{equation}
The above metric \eqref{sskruskal} can be extended to define the maximally-extended Reissner-Nordstr{\"o}m solution on the ambient manifold $(-\infty, \infty) \times (\infty, \infty) \times S^2$. In this paper, we will only consider the manifold-with-boundary $\MM$, corresponding to the exterior of the spacetime.

Using definition \eqref{sskruskal}, the metric $\g_{M, Q}$ is manifestly smooth in the maximally extended Reissner-Nordstr\"om. We will now describe different sets of coordinates which do not cover maximally extended Reissner-Nordstr\"om, but which are nevertheless useful for computations. 

\subsubsection{Double null coordinates $u$, $v$}
\label{EFdnulldef}

We define another double null coordinate system that covers  the interior of $\mathcal{M}$, modulo the degeneration of the 
angular coordinates. This coordinate system, 
$\left(u,v,\theta^1, \theta^2\right)$, is called  \emph{double null coordinates} and are defined via the relations
\begin{align} \label{UuVv}
U = -\frac{2r_{+}^2}{r_{+}-r_{-}}\exp\left(-\frac{r_{+}-r_{-}}{4r_{+}^2} u\right) \ \ \ \textrm{and} \ \ \ \ V = \frac{2r_{+}^2}{r_{+}-r_{-}}\exp \left(\frac{r_{+}-r_{-}}{4r_{+}^2} v\right) \, .
\end{align}
Using (\ref{UuVv}), we obtain the Reissner-Nordstr{\"o}m metric on the interior of $\MM$ in $\left(u,v,\theta^1, \theta^2\right)$-coordinates:
\begin{align} \label{ssef}
\g_{M, Q} =  - 4 \Up \left(u,v\right) \, d{u} \, d{v} +   r^2 \left(u,v\right) \gamma_{AB} d{\theta}^A d{\theta}^B  
\end{align}
with
\begin{align}
\label{officialOmegadef}
\Up:= 1-\frac{2M}{r}+\frac{Q^2}{r^2}
\end{align}
and the function $r: \left(-\infty,\infty\right) \times \left(-\infty,\infty\right) \rightarrow \left(M+\sqrt{M^2-Q^2},\infty\right)$ defined implicitly via the relations between $(U,V)$ and $(u,v)$. In $\left(u,v,\theta^1, \theta^2\right)$-coordinates, the horizon $\mathcal{H}^+$ can still be formally parametrised by $\left(\infty, v,\theta^1,\theta^2\right)$ with $v \in \mathbb{R}$, $\left(\theta^1,\theta^2\right) \in S^2$.

Note that $u, v$ are  regular optical functions. Their corresponding   null geodesic generators are
\bea\label{definition-L-Lb}
\Lb:=-g^{ab}\pr_a v  \pr_b=\frac{1}{\Up} \pr_u,         \qquad  L:=-g^{ab}\pr_a u  \pr_b=\frac{1}{\Up} \pr_v,
\eea
They verify
\beaa
g(L, L)=g(\Lb,\Lb)=0, \quad  g(L, \Lb)=-2\Up^{-1},\qquad D_L L=D_\Lb \Lb=0.
\eeaa

\subsubsection{Standard coordinates $t$, $r$}
Recall the form of the metric \eqref{ssef} in double null coordinates. Setting $$t=u+v$$ we may rewrite the above metric in coordinates $(t, r, \th, \phi)$ in the usual form \eqref{RNintro}:
\bea\label{RNintro-1}
 \g_{M, Q}=-\Up(r) dt^2 +\Up(r)^{-1}dr^2 +r^2(d\theta^2+\sin^2\theta d\phi^2),
\eea
which covers the interior of $\MM$. Observe that  
\beaa
\Up(r):=1-\frac{2M}{r}+\frac{Q^2}{r^2}=\frac{(r-r_{-})(r-r_{+})}{r^2}
\eeaa
where $r_{-}$ and $r_+$ are defined in \eqref{definiion-rpm}. 
The null vectors $L$ and $\Lb$ defined in \eqref{definition-L-Lb}, in $(t, r)$ coordinates can be written as
\bea
\Lb=\Up^{-1}\pr_t -\pr_r,         \qquad  L=\Up^{-1}\pr_t +\pr_r,
\eea

\subsubsection{ Ingoing  Eddington-Finkelstein coordinates $v$, $r$}\label{ingoing-coords}
We define another coordinate system that covers the interior of $\MM$. This coordinate system, $(v, r, \th, \phi)$ is called \emph{ingoing  Eddington-Finkelstein coordinates} and makes use of the above defined functions $v$ and $r$.  The Reissner-Nordstr{\"o}m metric on the interior of $\MM$ in $\left(v,r, \th, \phi\right)$-coordinates is given by 
\bea
\label{Schw:EF-coordinates}
\g_{M,Q}=-\Up(r) dv^2 + 2 dv dr + r^2(d\theta^2+\sin^2\theta d\phi^2). 
\eea

\subsection{The Bondi form of the Reissner-Nordstr\"om metric}\label{Outgoing-coords}
We define here another coordinate system that covers the topological interior of the manifold $\MM$, and which achieves the Bondi form of the Reissner-Nordstr\"om metric as described in Chapter \ref{double-null-gauge}. These coordinate system covers therefore the open exterior of the Reissner-Nordstr\"om black hole spacetime. 

Recall the function $r$ implicitly defined by \eqref{implicit-definition-r-U-V} and the function $u$ defined by \eqref{UuVv}. 
 In the coordinate system $(u, r, \th, \phi)$, called \emph{outgoing  Eddington-Finkelstein coordinates},  the Reissner-Nordstr{\"o}m metric on the interior of $\MM$ is given by 
\bea
\label{Schw:EF-coordinates-out}
\g_{M, Q}=- 2 du  dr-\Up(r) du^2  + r^2(d\theta^2+\sin^2\theta d\phi^2). 
\eea
Notice that this metric is of the Bondi form \eqref{double-null-metric} with the coordinate function\footnote{Notice that $r(u,s)=s$ verifies the definition given by \eqref{r-general}, since at $u=constant$ and $s=constant$, the metric $\slashed{g}$ induced on $S_{u,s}$ is given by $r^2\left( d\th^2+\sin^2 \th d\phi^2\right)$ which verifies $|S_{u,s}|=4\pi r^2$.}  $s=r$ and 
\bea
\vsi=1, \qquad \Omegab=-\Up, \qquad \underline{b}_A=0, \qquad \slashed{g}_{AB}=r^2 \gamma_{AB}
\eea
The normalized outgoing geodesic null frame $\mathscr{N}$ associated to the above is given by 
\begin{align}\label{null-frame-RN}
e_3 = 2 \partial_u +\Omegab \partial_r, \qquad e_4 =  \partial_r  
\end{align}
together with a local frame field $(e_1, e_2)$ on $S_{u,r}$. 
This frame does not extend regularly to the horizon $\mathcal{H}^+$, while the rescaled null frame 
\beaa
\mathscr{N}_*= \{ \Omegab^{-1} e_3,\  \Omegab e_4 \}
\eeaa
extends regularly to a non-vanishing null frame on $\mathcal{H}^+$. 

We will always compute with respect to the normalized null frame $\mathscr{N}$, but nevertheless passing to $\mathscr{N}_*$ will be useful to understand which quantities are regular on the horizon.

\subsubsection{Ricci coefficients and curvature components}
\label{hereforSchconcur}
We recall here the connection coefficients, curvature and electromagnetic components with respect to the null frame \eqref{null-frame-RN}.

The Ricci coefficients are given by 
\bea\label{ricci-coefficients-RN-0}
\chih_{AB}=\chibh_{AB}=0, \qquad \eta=\etab=\xi=\xib=\ze=0 \qquad \om=0
\eea
\bea
 \ka=\frac{2}{r}, \qquad \kab=\frac{2\Omegab}{r}=-\frac{2}{r}\left( 1-\frac{2M}{r}+\frac{Q^2}{r^2}\right),  \qquad \omb=\frac{M}{r^2}-\frac{Q^2}{r^3}
\eea

\begin{remark}\label{remark-Bondi} As opposed to the Ricci coefficients in double null gauge used in \cite{DHR}, in the Bondi gauge all the quantities are regular near the horizon $\mathcal{H}^+$. 
\end{remark}

The electromagnetic components are given by
\bea\label{electromagnetic-components-RN}
\bF=\bbF=0, \qquad \sigmaF=0, \qquad \rhoF=\frac{Q}{r^2}
\eea

The curvature components are given by 
\bea\label{curvature-components-RN}
\a=\aa=0, \qquad \b=\bb=0, \qquad \sigma=0, \qquad \rho= -\frac{2M}{r^3}+\frac{2Q^2}{r^4} 
\eea

 We also have that
\begin{align} \label{GCdef}
K = \frac{1}{r^2}
\end{align}
for the Gauss curvature of the round $S^2$-spheres.

Recalling the definition for a scalar function \eqref{def-check-f}, the above values in particular imply that for the scalar functions which do not vanish in Reissner-Nordstr\"om, their ``checked'' values always vanish:
\bea\label{no-check}
\check{\ka}=\check{\kab}=\check{\omb}=\check{\rhoF}=\check{\rho}=\check{K}=0
\eea

\subsection{Reissner-Nordstr{\"o}m symmetries and operators} \label{sec:commutation}
In this section, we recall the symmetries of the Reissner-Nordstr\"om metric, and specialize the operators discussed in Section \ref{sec:genmfld} to the Reissner-Nordstr{\"o}m metric in the Bondi form \eqref{Schw:EF-coordinates-out}.

\subsubsection{Killing fields of the Reissner-Nordstr{\"o}m metric}

We discuss the Killing fields associated to the metric $\g_{M, Q}$. Notice that the Reissner-Nordstr{\"o}m metric possesses the same symmetries as the ones possessed by Schwarzschild spacetime.

We define the vectorfield $T$ to be the timelike Killing vector field $\pr_t$ of the $(t, r)$ coordinates in \eqref{RNintro-1}. In outgoing Eddington-Finkelstein coordinates it is given by 
\beaa
T=\pr_u
\eeaa
 The vector field extends to a smooth Killing field on the horizon $\mathcal{H}^+$, which is moreover null and tangential to the null generator of $\mathcal{H}^+$. 
In terms of the null frames defined above, the Killing vector field $T$ can be written as 
\bea\label{definition-T}
T=\frac 12 (\Up e_3^*+e_4^*)=\frac 1 2 (e_3+\Up e_4)
\eea
Notice that at on the horizon, $T$ corresponds up to a factor with the null vector of $\mathscr{N}^*$ frame, $T=\frac 1 2 e_4^*$. 

We can also define a basis of angular momentum operators $\Omega_i$, $i=1,2,3$. Fixing standard spherical coordinates on $S^2$, we have 
\beaa
\Omega_1=\partial_\phi, \qquad \Omega_2=-\sin \phi \partial_\th-\cot \th \cos \phi \partial_\phi, \qquad \Omega_3=\cos \phi \partial_\th-\cot \th \sin \phi \partial_\phi
\eeaa
 The Lie algebra of Killing vector fields of $\g_{M,Q}$ is then generated by $T$ and $\Omega_i$, for $i=1,2,3$.

\subsubsection{The $S_{u,r}$-tensor algebra in Reissner-Nordstr{\"o}m}
We now specialize the general definitions of the projected Lie and covariant differential operators of Section \ref{sec:genmfld} to the Reissner-Nordstr{\"o}m metric with null directions given by \eqref{null-frame-RN}. 

If $\xi$ is a $S_{u,r}$ tensor of rank $n$ on $(\mathcal{M}, \g_{M, Q})$ we have in components
\bea\label{lie-derivative-RN}
(D \xi)_{A_1, \dots A_n}=\pr_r (\xi_{A_1, \dots A_n}), \qquad (\underline{D} \xi)_{A_1, \dots A_n}=2 \pr_u(\xi_{A_1, \dots A_n})+\Omegab \pr_r (\xi_{A_1, \dots A_n})
\eea
 
 Since $\chi, \chib$ only have a trace-component in Reissner-Nordstr{\"o}m, one can specialize formulas \eqref{covariant-Lie} as 
 \bea
 (\nabb_4 \xi)_A=\pr_r (\xi_A) -\frac 1 2 \ka \xi_A, \qquad  (\nabb_3 \xi)_A=2\pr_u (\xi_A)+ \Omegab \pr_r (\xi_A) -\frac 1 2 \kab \xi_A \\
  (\nabb_4 \xi)^A=\pr_r (\xi^A) +\frac 1 2 \ka \xi^A, \qquad  (\nabb_3 \xi)^A=2\pr_u (\xi^A)+ \Omegab \pr_r (\xi^A) +\frac 1 2 \kab \xi^A \label{Liederivative-cov}
 \eea
 for $1$-forms and $1$-vectors and 
 \bea\label{Lie-covariant-2}
 (\nabb_4 \th)_{AB}=\pr_r (\th_{AB})-\ka \th_{AB}, \qquad (\nabb_3 \th)_{AB}=2\pr_u (\th_{AB})+ \Omegab \pr_r (\th_{AB}) - \kab \th_{AB} \\
  (\nabb_4 \th)^{AB}=\pr_r (\th^{AB})+\ka \th_{AB}, \qquad (\nabb_3 \th)^{AB}=2\pr_u (\th_{AB})+ \Omegab \pr_r (\th_{AB}) + \kab \th_{AB} 
 \eea
for symmetric traceless $2$-tensors.

\subsubsection{Commutation formulae in Reissner-Nordstr{\"o}m}

Adapting the commutation formulae \eqref{commutators} to the Reissner-Nordstr{\"o}m metric, we obtain the following commutation formulae. For projected covariant derivatives for $\xi = \xi_{A_1...A_n}$ any $n$-covariant $S^2_{u,r}$-tensor in Reissner-Nordstr{\"o}m metric $\left(\mathcal{M},\g_{M,Q}\right)$ in Bondi gauge we have
\begin{align}\label{commutation-formulas}
\nabb_3 \slashed{\nabla}_B \xi_{A_1...A_n} - \slashed{\nabla}_B \slashed{\nabla}_3 \xi_{A_1...A_n} &= - \frac{1}{2} \kab  \slashed{\nabla}_B \xi_{A_1...A_n} \, , \nonumber \\
\slashed{\nabla}_4 \slashed{\nabla}_B \xi_{A_1...A_n} - \slashed{\nabla}_B \slashed{\nabla}_4 \xi_{A_1...A_n} &= - \frac{1}{2} \ka  \slashed{\nabla}_B \xi_{A_1...A_n} \, , \\
\slashed{\nabla}_3 \slashed{\nabla}_4 \xi_{A_1...A_n} - \slashed{\nabla}_4 \slashed{\nabla}_3 \xi_{A_1...A_n} &=2\omb \slashed{\nabla}_4 \xi_{A_1...A_n} \, . \nonumber
\end{align}
In particular, we have
\begin{align}\label{commutator-rnabla}
\left[\slashed{\nabla}_4, r \slashed{\nabla}_A \right] \xi = 0 \ \ \ , \ \ \ \left[\slashed{\nabla}_3, r \slashed{\nabla}_A \right] \xi = 0  \, .
\end{align}
We summarize here the commutation formulae for the angular operators defined in Section \ref{sec:genmfld}. Let $\rho, \sigma$ be scalar functions, $\xi$ be a $1$-tensor and $\th$ be a symmetric traceless $2$-tensor on the Reissner-Nordstr{\"o}m manifold. Then:
\bea
\left[ \nabb_4, \DDd_1\right]\xi &=&-\frac 1 2 \ka \DDd_1\xi, \qquad \left[ \nabb_3, \DDd_1\right]\xi =-\frac 1 2 \kab \DDd_1\xi\label{commutator-nabb-4-divv} \\
\left[ \nabb_4, \DDs_1\right](\rho,\sigma)&=&-\frac 1 2 \ka\DDs_1(\rho,\sigma), \qquad \left[ \nabb_3, \DDs_1\right](\rho,\sigma)=-\frac 1 2 \kab\DDs_1(\rho,\sigma)\label{commutator-nabb-4-DDs-1}, \\
\left[ \nabb_4, \DDd_2\right]\th &=&-\frac 1 2 \ka \DDd_2\th, \qquad \left[ \nabb_3, \DDd_2\right]\th =-\frac 1 2 \kab \DDd_2\th\label{commutator-nabb-4-DDd_2} \\
\left[ \nabb_4, \DDs_2\right]\xi&=&-\frac 1 2 \ka\DDs_2\xi, \qquad
\left[ \nabb_3, \DDs_2\right]\xi=-\frac 1 2 \kab\DDs_2\xi\label{commutator-nabb-4-DDs}
\eea

\subsubsection{The $\ell=0,1$ spherical harmonics}

We collect here some known definitions and properties of the Hodge decomposition of scalars, one forms and symmetric traceless two tensors in spherical harmonics. We also recall some known elliptic estimates. See Section 4.4 of \cite{DHR} for more details.

We denote by $\dot{Y}_m^\ell$, with $|m|\leq \ell$, the well-known spherical harmonics on the unit sphere, i.e. $$\lapp_0 \dot{Y}^\ell_m=- \ell(\ell+1) \dot{Y}^\ell_m$$ where $\lapp_0$ denotes the laplacian on the unit sphere $S^2$.  The $\ell=0,1$ spherical harmonics are given explicitly by 
\bea
\dot{Y}_{m=0}^{\ell=0}&=&\frac{1}{\sqrt{4\pi}}, \label{spherical-l=0}\\
\dot{Y}_{m=0}^{\ell=1}&=&\sqrt{\frac{3}{8\pi}} \cos \th, \qquad \dot{Y}_{m=-1}^{\ell=1}=\sqrt{\frac{3}{4\pi}} \sin\th \cos \phi, \qquad \dot{Y}_{m=1}^{\ell=1}=\sqrt{\frac{3}{4\pi}} \sin\th \sin \phi \label{spherical-l=1}
\eea
This family is orthogonal with respect to the standard inner product on the sphere, and any arbitrary function $f \in L^2(S^2)$ can be expanded uniquely with respect to such a basis. 

In the foliation of Reissner-Nordstr\"om spacetime, we are interested in using the spherical harmonics with respect to the sphere of radius $r$. For this reason, we normalize the definition of the spherical harmonics on the unit sphere above to the following.

We denote by $Y_m^\ell$, with $|m|\leq \ell$, the spherical harmonics on the sphere of radius $r$, i.e. $$\lapp Y^\ell_m=-\frac{1}{r^2} \ell(\ell+1) Y^\ell_m$$ where $\lapp$ denotes the laplacian on the sphere $S_{u,r}$ of radius $r$. Such spherical harmonics are normalized to have $L^2$ norm in $S_{u, r}$ equal to $1$, so they will in particular be given by $Y_m^\ell=\frac{1}{r}\dot{Y}_m^\ell$. We use this basis to project functions on Reissner-Nordstr\"om manifold in the following way.  

\begin{definition}\label{lemma-spherical-harmonics} We say that  a function $f$ on $\mathcal{M}$ is supported on $\ell\geq 2$ if the projections 
\beaa
\int_{S_{u, r}}  f \c Y^\ell_m=0
\eeaa
vanish for $Y_m^{\ell=1}$ for $m=-1, 0, 1$. Any function $f$ can be uniquely decomposed orthogonally as 
\bea\label{decomposition-f-spherical-harmonics}
f=c(u, r) Y^{\ell=0}_{m=0}+\sum_{i=-1}^{1}c_i(u, r) Y^{\ell=1}_{m=i}(\th, \vphi)+f_{\ell\ge 2}
\eea
where $f_{\ell\ge 2}$ is supported in $\ell\ge2$.
\end{definition} 

In particular, we can write the orthogonal decomposition
\beaa
f=f_{\ell=0}+f_{\ell=1}+f_{\ell\geq2}
\eeaa
where 
\bea\label{explicit-formula-projection=l1}
f_{\ell=0}&=& \frac{1}{4\pi r^2} \int_{S_{u,r}}  f \\
f_{\ell=1}&=&\sum_{i=-1}^{1} \left( \int_{S_{u, r}} f \c Y^{\ell=1}_{m=i} \right) Y^{\ell=1}_{m=i} 
\eea

Recall that an arbitrary one-form $\xi$ on $S_{u,r}$ has a unique representation $\xi=r \DDs_1(f, g)$, for two uniquely defined functions $f$ and $g$ on the unit sphere, both with vanishing mean, i.e. $f_{\ell=0}=g_{\ell=0}=0$. In particular, the scalars $\divv \xi$ and $\curll \xi$ are supported in $\ell\ge1$. 
\begin{definition}\label{decomposition-xi} We say that a smooth $S_{u,r}$ one form $\xi$ is supported on $\ell\ge 2$ if the functions $f$ and $g$ in the unique representation $$\xi=r \DDs_1(f, g)$$ are supported on $\ell\geq 2$. Any smooth one form $\xi$ can be uniquely decomposed orthogonally as $$\xi=\xi_{\ell=1}+\xi_{\ell\ge 2}$$ where the two scalar functions $r\DDd_1\xi=(r\divv\xi_{\ell=1}, r\curll\xi_{\ell=1})$ are in the span of \eqref{spherical-l=0} and $\xi_{\ell\ge 2}$ is supported on $\ell\ge 2$.
\end{definition}

Recall that an arbitrary symmetric traceless two-tensor $\th$ on $S_{u,r}$ has a unique representation $\th=r^2\DDs_2\DDs_1(f, g)$ for two uniquely defined functions $f$ and $g$ on the unit sphere, both supported in $\ell\ge 2$. In particular, the scalars $\divv \divv \th$ and $\curll \divv \th$ are supported in $\ell\ge 2$.

For future reference, we recall the following lemma.

\begin{lemma}[Lemma 4.4.1 in \cite{DHR}]\label{lemma-kernel-DDs2} The kernel of the operator $\mathcal{T}=r^2 \DDs_2 \DDs_1$ is finite dimensional. More precisely, if the pair of functions $(f_1, f_2)$ is in the kernel, then 
\beaa
f_1=c Y^{\ell=0}_{m=0}+\sum_{i=-1}^{1}c_i Y^{\ell=1}_{m=i}(\th, \vphi), \qquad f_2=\tilde{c} Y^{\ell=0}_{m=0}+\sum_{i=-1}^{1}\tilde{c}_i Y^{\ell=1}_{m=i}(\th, \vphi)
\eeaa
for constants $c, c_i, \tilde{c}, \tilde{c}_i$.
\end{lemma}

Consider a one-form $\xi$ on $\MM$ and its decomposition $\xi=\xi_{\ell=1}+\xi_{\ell\geq 2}$ as in Definition \ref{decomposition-xi}. Then Proposition \ref{L^2-estimates} implies the following elliptic estimate.

\begin{lemma}\label{main-elliptic-estimate} Let $\xi$ be a one-form on $\MM$. Then there exists a constant $C >0$ such that the following estimate holds:
\beaa
\int_{S} |\xi|^2 \leq C \left(\int_{S} |r\divv \xi_{\ell=1}|^2+|r\curll \xi_{\ell=1}|^2 + |r\DDs_2 \xi|^2 \right)
\eeaa
\end{lemma}
\begin{proof} Using the orthogonal decomposition of $\xi$, we have
\beaa
\int_{S} |\xi|^2&=& \int_{S} |\xi_{\ell=1}|^2+\int_{S} |\xi_{\ell\geq 2}|^2
\eeaa
Observe that, according to Lemma \ref{lemma-kernel-DDs2}, $\xi_{\ell = 1}$ is in the kernel of $\DDs_2$.  
Applying \eqref{first-elliptic-estimate} to $\xi_{\ell=1}$ and \eqref{second-elliptic-estimate} to $\xi_{\ell\geq 2}$ we obtain the desired estimate. 
\end{proof}

Here we collect the useful properties associated to the decomposition in average and check quantities.

\begin{lemma}\label{average-check-spherical} Any scalar function $f:\MM \to \mathbb{R}$ verifies 
\bea
\ov{f}_{\ell\ge 1}=0, \qquad \check{f}_{\ell=0}=0
\eea
Therefore $f=\ov{f}_{\ell=0}+\check{f}_{\ell\ge1}$.
\end{lemma}
\begin{proof} Using \eqref{def-S-average} and \eqref{decomposition-f-spherical-harmonics}, we compute
\beaa
|S| \ov{f}&=&  \int_S \left( c(u, r) Y^{\ell=0}_{m=0}+\sum_{i=-1}^{\ell=1}c_i(u, r) Y^{\ell=1}_{m=i}(\th, \vphi)+f_{\ell\ge 2} \right) \sin\th d\th d\vphi\\
&=&  |S| c(u,r) Y^{\ell=0}_{m=0}+\sum_{i=-1}^{\ell=1}c_i(u, r)\int_S\left( Y^{\ell=1}_{m=i}(\th, \vphi)+f_{\ell\ge 2} \right) \sin\th d\th d\vphi
\eeaa
and recalling that, by orthogonality of the spherical harmonics, $\int_SY^{\ell}_{m}(\th, \vphi)Y^{\ell'}_{m'}(\th, \vphi)\sin\th d\th d\vphi=\delta_{\ell \ell'} \delta_{m m'}$, the integral on the right hand side vanishes. Therefore $\ov{f}_{\ell\ge 1}=0$ and $f_{\ell=0}=\ov{f}_{\ell=0}$. On the other hand, 
\beaa
\check{f}&=& f-\ov{f}=c(u, s) Y^{\ell=0}_{m=0}+\sum_{i=-1}^{\ell=1}c_i(u, s) Y^{\ell=1}_{m=i}(\th, \vphi)+f_{\ell\ge 2}-c(u, s) Y^{\ell=0}_{m=0}\\
&=&\sum_{i=-1}^{\ell=1}c_i(u, s) Y^{\ell=1}_{m=i}(\th, \vphi)+f_{\ell\ge 2}
\eeaa
therefore $\check{f}$ is supported in $\ell\ge 1$.
\end{proof}

 We derive the transport equation for the projection to the $\ell=1$ spherical harmonics of a function $f$ on $\MM$.
 
 \begin{lemma}\label{commutation-projection-l1} Let $f$ be a scalar function on $\MM$. Then
 \beaa
 \nabb_4(f_{\ell=1})&=& (\nabb_4f)_{\ell=1} \\
 \nabb_3(f_{\ell=1})&=& (\nabb_3f)_{\ell=1}
 \eeaa
 \end{lemma}
 \begin{proof} Applying $\nabb_4=\pr_r$ to the expression for the projection to the $\ell=1$ spherical harmonics given by \eqref{explicit-formula-projection=l1}, we obtain
 \beaa
 \nabb_4(f_{\ell=1})&=&\sum_{i=-1}^{1}  \nabb_4\Big(\left(\int_{S} f \c Y^{\ell=1}_{m=i}\right) Y^{\ell=1}_{m=i} \Big)
 \eeaa
 Recall that the normalized spherical harmonics are defined as $Y^{\ell=1}_m=\frac{1}{r} \dot{Y}^{\ell=1}_m$, where $\dot{Y}^{\ell=1}_m$ are given by \eqref{spherical-l=1}, and therefore $\nabb_4(\dot{Y}^{\ell=1}_m)=0$. This implies
 \beaa
 \nabb_4(Y^{\ell=1}_m)=\nabb_4\left(\frac{1}{r} \dot{Y}^{\ell=1}_m\right)=\nabb_4\left(\frac{1}{r}\right) \dot{Y}^{\ell=1}_m=-\frac{1}{2r}\ov{\ka}  \dot{Y}^{\ell=1}_m=-\frac{1}{2}\ov{\ka} Y^{\ell=1}_m
 \eeaa
 where we used Proposition \ref{nabb-4-int-f}. 
 Using again Proposition \ref{nabb-4-int-f}, the computation gives
 \beaa
 \nabb_4(f_{\ell=1})&=&\sum_{i=-1}^{1}  \nabb_4\left(\int_{S} f \c Y^{\ell=1}_{m=i}\right) Y^{\ell=1}_{m=i} +\sum_{i=-1}^{1} \left(\int_{S} f \c Y^{\ell=1}_{m=i}\right) \nabb_4Y^{\ell=1}_{m=i} \\
 &=&\sum_{i=-1}^{1}  \left(\int_{S}\nabb_4( f \c Y^{\ell=1}_{m=i})+\ka  f \c Y^{\ell=1}_{m=i}\right) Y^{\ell=1}_{m=i} +\sum_{i=-1}^{1} \left(\int_{S} f \c Y^{\ell=1}_{m=i}\right) (-\frac{1}{2}\ov{\ka} Y^{\ell=1}_{m=i}) \\
 &=&\sum_{i=-1}^{1}  \left(\int_{S}\nabb_4( f) \c Y^{\ell=1}_{m=i}+f \c \nabb_4(Y^{\ell=1}_{m=i})+\frac 1 2 \ka  f \c Y^{\ell=1}_{m=i}\right) Y^{\ell=1}_{m=i}  \\
  &=&\sum_{i=-1}^{1}  \left(\int_{S}\nabb_4( f) \c Y^{\ell=1}_{m=i}\right) Y^{\ell=1}_{m=i} = (\nabb_4f)_{\ell=1}
 \eeaa
 as desired. Similarly for $\nabb_3 f$. 
 \end{proof}

\section{The linearized gravitational and electromagnetic perturbations around Reissner-Nordstr{\"o}m}\label{linearized-equations-chapter}
In this section, we present the equations of linearized gravitational and electromagnetic perturbations around Reissner-Nordstr{\"o}m.
In Section \ref{formal-derivation} we describe the procedure to the linearization of the equations of Section \ref{nseq}. In Section \ref{all-equations} we summarize the complete set of equations describing the dynamical evolution of a linear perturbation of Reissner-Nordstr\"om spacetime. 

\subsection{A guide to the formal derivation}\label{formal-derivation}
We give in this section a formal derivation of the system from the equations of Section \ref{nseq} and of Section \ref{equations-implied-by-Bondi}.

\subsubsection{Preliminaries}
We identify the general manifold $\MM$ and its Bondi coordinates $(u, s, \th^1, \th^2)$ of Section \ref{local-bondi} with the interior of the Reissner-Nordstr{\"o}m spacetime in its Bondi form in Section \ref{Outgoing-coords}.

On $\MM$, we consider a one-parameter family of Lorentzian metrics $\g(\ep)$ of the form \eqref{double-null-metric}. More precisely:
\bea\label{metric-g-epsilon}
\begin{split}
  \g(\ep)&=-2\vsi(\ep) du ds+\vsi(\ep)^2\underline{\Omega}(\ep)du^2\\
  &+\slashed{g}_{AB}(\ep) \left(d\th^A-\frac 1 2\vsi(\ep)\underline{b}(\ep)^A du \right) \left(d\th^B-\frac 1 2\vsi(\ep)\underline{b}(\ep)^B du  \right)
  \end{split}
  \eea
  such that $\g(0)=\g_{M,Q}$ expressed in the outgoing Eddington-Finkelstein coordinates \eqref{Schw:EF-coordinates-out}, i.e.
  \beaa
\vsi(0)=1, \qquad  \Omegab(0)= -\left(1-\frac{2M}{r}+\frac{Q^2}{r^2}\right), \qquad \underline{b}_A(0)=0, \qquad \slashed{g}_{AB}(0)=r^2 \gamma_{AB}
  \eeaa
 
 In view of the general discussion in Section \ref{local-bondi}, associated to the metric \eqref{metric-g-epsilon} there is an associated family of normalized frames of the form 
\beaa
e_3=2\vsi^{-1}(\ep)\partial_u+\underline{\Omega}(\ep)\partial_s+\underline{b}(\ep)^A \partial_{\th^A}, \qquad e_4=\partial_s, \qquad e_A=\partial_{\th^A}
\eeaa
Note that this frame does not extend smoothly to the event horizon $\mathcal{H}^+$. On the other hand, the rescaled null frame 
\beaa
\Omegab^{-1}(\ep) e_3, \qquad \Omegab(\ep) e_4
\eeaa
 is smooth up to the horizon.

\subsubsection{Outline of the linearization procedure}

We now linearize the smooth one-parameter family of metrics \eqref{metric-g-epsilon} in terms of $\epsilon$. We linearize the full system of equations obtained in Section \ref{nseq} around the values of the connection coefficients and curvature components in Reissner-Nordstr\"om obtained in Section \ref{hereforSchconcur}. We outline of the procedure in few different cases.

From \eqref{ricci-coefficients-RN-0}, \eqref{electromagnetic-components-RN}, \eqref{curvature-components-RN}, we notice that all the one-forms and symmetric traceless $2$-tensors appearing in the Einstein-Maxwell equations of Section \ref{nseq} vanish in Reissner-Nordstr{\"o}m. 
 Formally, we have 
\beaa
\chih(\ep)&=& 0+ \chih, \qquad \eta(\ep)= 0+\eta, \qquad \ze(\ep)= 0+\ze \\
\chibh(\ep)&=& 0+ \chibh, \qquad \xib(\ep)=0+\xib \\
\bF(\ep)&=& 0+\bF, \qquad \b(\ep)= 0+\b, \qquad \a(\ep)= 0+\a \\
\bbF(\ep)&=& 0+\bbF, \qquad \bb(\ep)= 0+\bb, \qquad \aa(\ep)= 0+\aa  
\eeaa
The linearization of the equations involving the above tensors simply consists in discarding terms containing product of those, while keeping the other terms. In doing so, we will make sure to include the information obtained by the equation \eqref{substitute-xi-xib-ze} for the Bondi form of the metric.

To give an example, consider equations \eqref{nabb-3-chibh-general}:
\beaa
\nabb_3 \chibh+\kab \ \chibh+2\omb \chibh&=&-2\DDs_2\xib -\aa+2(\eta+\etab-2\ze)\hot \xib  \\
\nabb_4 \chih+\ka \ \chih+2\om \chih&=&-2\DDs_2\xi -\a +2(\eta+\etab+2\ze)\hot \xi
\eeaa
In linearizing them, we observe that the term $2(\eta+\etab-2\ze)\hot \xib$ is quadratic, and $\om=\xi=0$ by \eqref{substitute-xi-xib-ze}. Their linearization therefore give
\beaa
\nabb_3 \chibh+\kab \ \chibh+2\omb \chibh&=&-2\DDs_2\xib -\aa  \\
\nabb_4 \chih+\ka \ \chih&=& -\a 
\eeaa
 which are \eqref{nabb-3-chibh-ze-eta} and \eqref{nabb-4-chih-ze-eta}.
 In this way we linearize \eqref{nabb-3-chibh-general}, \eqref{curl-xib-general}, \eqref{nabb-3-chih-general}, \eqref{curl-eta-general}, \eqref{nabb-3-ze-general}, \eqref{nabb-4-xib-general}, \eqref{codazzi-general}, \eqref{nabb-3-bF-general}, \eqref{nabb-3-a-general}, \eqref{nabb-4-b-general}.

Recall the non-vanishing values of the scalars $\ka$, $\kab$, $\omb$, $\rhoF$, $\rho$ and $K$ in Reissner-Nordstr\"om given by \eqref{ricci-coefficients-RN-0}, \eqref{electromagnetic-components-RN}, \eqref{curvature-components-RN}, \eqref{GCdef}. Nevertheless, by \eqref{no-check} all check-quantities vanish. Moreover, the scalars $\sigma$ and $\sigmaF$ vanish.
We take advantage of this fact by using the decomposition into average and check as defined in \eqref{def-check-f}.
For instance, we write
\beaa
\ka(\ep)&=& \frac{2}{r}+ \left(\ov{\ka(\ep)}-\frac{2}{r} \right) +\check{\ka(\ep)} = \frac{2}{r}+ \kalin(\ep) +\check{\ka(\ep)}
\eeaa
 where we define $\kalin(\ep)= \ov{\ka(\ep)}-\frac{2}{r}$. 
 
 We therefore define two scalar functions for each non-vanishing scalar: the average to which we subtract the value in Reissner-Nordstr\"om (denoted by a superscript $(1)$) and the check quantity.

In particular,  we define 
\beaa
\kalin(\ep)&=&\ov{\ka(\ep)}-\frac{2}{r}, \\
\kablin(\ep)&=&\ov{\kab(\ep)}+\frac{2}{r}\left(1-\frac{2M}{r}+\frac{Q^2}{r^2} \right), \\
\omblin (\ep)&=&\ov{\omb(\ep)}-  \left(\frac{M}{r^2}-\frac{Q^2}{ r^3}\right), \\
\rhoFlin(\ep)&=& \ov{\rhoF(\ep)}-\frac{Q}{r^2} \\
\sigmaFlin(\ep)&=& \ov{\sigmaF(\ep)} \\
\rlin(\ep)&=& \ov{\rho(\ep)}+\frac{2M}{r^3}-\frac{2Q^2}{r^4} \\
\sigma(\ep)&=& \ov{\sigma(\ep)} \\
\Klin(\ep)&=& \ov{K(\ep)}-\frac{1}{r^2} 
\eeaa
The check quantities linearize in the obvious way.

In linearizing the equations for $\ka$, we will obtain equations for the quantities $\kalin$ and $\check{\ka}$. To simplify the notation, we can therefore denote $\ka$ the value of the quantity in Reissner-Nordstr\"om. This gives $\kalin=\ov{\ka}-\ka$. Similarly for all the other quantities.

In Section \ref{average-quantities-general}, we computed the transport equations of average quantities. Using those, we compute the equations for the linearized quantities above. 
For instance, consider equation \eqref{nabb-3-kab-general}:
\beaa
\nabb_4\ka(\ep)+\frac 1 2 \ka(\ep)^2=-\chih(\ep)\cdot \chih(\ep) -2 \bF(\ep)\cdot \bF(\ep)
\eeaa
The right hand side is quadratic, therefore in linearizing we have 
\beaa
\nabb_4\ka(\ep)+\frac 1 2 \ka(\ep)^2=0
\eeaa
We can use Corollary \ref{average-check}, to compute $\nabb_4(\ov{\ka}(\ep))$:
\beaa
\nabb_4(\ov{\ka}(\ep))&=& \ov{\nabb_4 \ka(\ep)}=\ov{-\frac 1 2 \ka(\ep)^2}=-\frac 1 2 \ov{\ka(\ep)}^2=-\frac 1 2 (\kalin(\ep)+\ka)^2=- \ka\kalin(\ep)-\frac 1 2 \ka^2
\eeaa
On the other hand, using Proposition \ref{nabb-4-int-f}, we have 
\beaa
\nabb_4\left( \frac{2}{r} \right)&=& -\frac{2}{r^2} \nabb_4 r=-\frac{1}{r}  \ov{\ka}(\ep) 
\eeaa
Therefore, writing $\ka=\ov{\ka}(\ep)-\kalin(\ep)$, we have 
\beaa
\nabb_4(\kalin(\ep))&=& \nabb_4(\ov{\ka}(\ep))-\nabb_4\left( \frac{2}{r} \right)\\
&=&-\ka\kalin(\ep)-\frac 1 2 \ka^2+\frac{1}{r}  \ov{\ka}(\ep) \\
&=&-\ka \kalin(\ep)-\frac 1 2 \ka (\ov{\ka}(\ep)-\kalin(\ep))+\frac{1}{r}  \ov{\ka}(\ep) \\
&=&-\frac 1 2 \ka \kalin(\ep)+\left(-\frac 1 2 \ka+\frac{1}{r}\right)  \ov{\ka} (\ep)=-\frac 1 2 \ka\kalin(\ep)
\eeaa
which gives equation \eqref{nabb-4-kalin-ze-eta}. 
All the other equations for the average quantities are obtained in a similar manner.
The equation for the check part for a scalar quantity is obtained applying again Corollary \ref{average-check}. In this way, we linearize \eqref{nabb-3-kab-general}, \eqref{nabb-3-ka-general}, \eqref{nabb-4-omb-general}, \eqref{Gauss-general}, \eqref{nabb-3-sigmaF-general}, \eqref{nabb-4-rhoF-general}, \eqref{nabb-4-rho-general}, \eqref{nabb-4-sigma-general}.

It seems that we have doubled the equations involving scalar quantities. In reality, the separation between $\flin$ and $\check{f}$ for a scalar quantity $f$ reflects the projection into spherical harmonics.
Indeed, we have that $\flin(\ep)_{\ell\ge1}=0$ and $\check{f}(\ep)_{\ell=0}=0$, where the projections are intended to be with respect to the Reissner-Nordstr{\"o}m metric. This is proved in the following way. Since $\ov{f}(\ep)$ is constant on the spheres of constant $r$ determined by the metric $\g(\ep)$, we have $\DDs_1(\ep)(\ov{f}(\ep), 0)=0$ and therefore
\beaa
0=\DDs_1(\ep)(\ov{f}(\ep), 0)=(\DDs_1(0)+\ep)(\flin(\ep)+f_{M,Q}, 0)=\DDs_1(0)(\flin(\ep), 0)+O(\ep^2)
\eeaa
where $\DDs_1(0)$ is the angular operator $\DDs_1$ in Reissner-Nordstr{\"o}m. Similarly, in taking the mean of $\check{f}(\ep)$ we see that it has vanishing mean with respect to the Reissner-Nordstr\"om spacetime, modulo quadratic terms.

We now outline the linearization of the metric coefficients in Bondi form, verifying the equations given by Lemma \ref{double-null-Ricci}. The metric coefficients are $\vsi(\ep)$, $\Omegab(\ep)$, $\underline{b}(\ep)$ and $\slashed{g}(\ep)$. 

We decompose the scalar functions $\vsi(\ep)$ and $\Omegab(\ep)$ as above. We define 
\beaa
\Omegablin(\ep)&=& \ov{\Omegab(\ep)}+\left(1-\frac{2M}{r}+\frac{Q^2}{r^2} \right)\\
\vsilin(\ep)&=&\ov{\vsilin(\ep)}-1
\eeaa
and define $\check{\vsi}(\ep)=\vsi(\ep)-\ov{\vsi(\ep)}$, $\check{\Omegab}(\ep)=\Omegab(\ep)-\ov{\Omegab(\ep)}$.

The vector $\underline{b}$ vanishes on Reissner-Nordstr\"om, therefore the linearization of \eqref{additional-metric-1} is straightforward.
We now show how to linearize the equations for the metric $\slashed{g}(\ep)$ \eqref{additional-metric-2} and \eqref{additional-metric-3}.

Since $\slashed{g}(0)=r^2 \gamma_{AB}$, we decompose $\slashed{g}$ into:
\bea\label{linearisation-metric-1}
\slashed{g}_{AB}=\frac 1 2 (\tr_{\gamma} \slashed{g})  \gamma_{AB}+\hat{\slashed{g}}_{AB}
\eea
where the trace and the traceless part are computed in terms of the round sphere metric, i.e.
\beaa
\tr_{\gamma} \slashed{g}= \gamma^{AB} \slashed{g}_{AB}, \qquad  \gamma^{AB} \hat{\slashed{g}}_{AB}=0
\eeaa
Plugging in the decomposition \eqref{linearisation-metric-1} in the equations for the metric \eqref{additional-metric-2}, we obtain 
\bea\label{computation-linearization-metric}
\partial_s (\frac 1 2 (\tr_{\gamma} \slashed{g}) \gamma_{AB}+\hat{\slashed{g}}_{AB})=2 \chih_{AB}+\ka \left(\frac 1 2 (\tr_{\gamma} \slashed{g}) \gamma_{AB}+\hat{\slashed{g}}_{AB} \right)
\eea
Recalling that $\pr_s( \gamma_{AB})=0$ in Reissner-Nordstr{\"o}m background, the left hand side of \eqref{computation-linearization-metric} becomes
\beaa
\partial_s (\frac 1 2 (\tr_{\gamma} \slashed{g})  \gamma_{AB}+\hat{\slashed{g}}_{AB})&=& \frac 1 2 \partial_s(\tr_{\gamma} \slashed{g})  \gamma_{AB}+\pr_s\hat{\slashed{g}}_{AB}
\eeaa
Observe that $\pr_s \hat{\slashed{g}}_{AB}$ is traceless with respect to $\gamma$, since $$0=\pr_s( ( \gamma)^{AB} \hat{\slashed{g}}_{AB})=( \gamma)^{AB} \pr_s\hat{\slashed{g}}_{AB}$$ 
The right hand side of \eqref{computation-linearization-metric} is given by $\frac 1 2 \ka(\tr_{\gamma} \slashed{g})  \gamma_{AB}+2 \chih_{AB}+\ka \hat{\slashed{g}}_{AB}$. 
Observe that $\chih$ is traceless with respect to $\gamma$ modulo quadratic terms, therefore separating the equation into its traceless and trace part we obtain:
\beaa
\pr_s\hat{\slashed{g}}_{AB}-\ka \hat{\slashed{g}}_{AB}&=& 2 \chih_{AB}, \qquad \partial_s(\tr_{\gamma} \slashed{g})= \ka (\tr_{\gamma} \slashed{g})
\eeaa
Using \eqref{Lie-covariant-2}, we have 
\beaa
\nabb_4\hat{\slashed{g}}_{AB}&=& 2 \chih_{AB}  \qquad \nabb_4(\tr_{\gamma} \slashed{g})= \ka(\tr_{\gamma} \slashed{g})
\eeaa
 Define 
\beaa
\trglin&=& \ov{\tr_{\gamma} \slashed{g}}-2r^2 \qquad \check{\tr_{\gamma}\slashed{g}}= \tr_{\gamma}\slashed{g}-\ov{\tr_{\gamma} \slashed{g}}
\eeaa
By Corollary \ref{average-check}, we have
\beaa
\nabb_4(\trglin)&=& \nabb_4(\ov{\tr_{\gamma} \slashed{g}})-2\nabb_4(r^2)=\ov{\nabb_4(\tr_{\gamma} \slashed{g})}- 2r^2 \ov{\ka}=\ov{\ka} \ov{\tr_{\gamma} \slashed{g}}-2r^2 \ov{\ka}=\ka \ \trglin, \\
\nabb_4(\check{\tr_{\gamma}\slashed{g}})&=& \nabb_4(\tr_{\gamma} \slashed{g})-\ov{\nabb_4(\tr_{\gamma} \slashed{g})}=(\ka -\divv b)(\tr_{\gamma} \slashed{g})-\ov{\ka} \ov{\tr_{\gamma} \slashed{g}}=\ka \check{\tr_{\gamma}\slashed{g}}+2 r^2 \check{\ka} -2r^2 \divv b
\eeaa
Similarly, for $\nabb_3$ we have 
\beaa
\nabb_3\hat{\slashed{g}}_{AB}&=& 2 \chibh_{AB}+2(\DDs_2 \underline{b})_{AB}, \\
\nabb_3(\trglin)&=& \kab \trglin, \\
\nabb_3(\check{\tr_{\gamma}\slashed{g}})&=& 2r^2\check{\kab} -2r^2\divv \underline{b}-2r^2\ka \check{\Omegab}
\eeaa
Using that $r^2\nabb_3(r^{-2}f)= \nabb_3 f-\ka f$, we obtain the equations for the metric components.  

We linearize the Gauss curvature $K(\ep)$ of the metric $\slashed{g}$ as for the above scalar functions. We define 
\beaa
\Klin(\ep)&=&\ov{K(\ep)}-\frac{1}{r^2} \qquad \check{K}(\ep)= K(\ep)-\ov{K(\ep)}
\eeaa
Observe that, by Gauss-Bonnet theorem,  $\int_S K(\ep)=4\pi$, therefore $\ov{K(\ep)}=\frac{1}{4\pi r^2} \int_S K(\ep)=\frac {1}{ r^2}$, and consequently $\Klin(\ep)=0$.

In general, the linearization of the Gauss curvature for a metric $g=g_0+\de g$ is given by
\beaa
2(\de K)&=& -\frac 1 2 \lapp_{g_0} (\tr(\de g))-K \tr(\de g)+\div\div \widehat{(\de g)}
\eeaa
where $K=K_0+\de K$.
Writing \eqref{linearisation-metric-1} as 
\beaa
\slashed{g}_{AB}=r^2\gamma_{AB}+\de g=r^2 \gamma_{AB}+\frac 1 2 (\tr_{\gamma} \slashed{g}-2r^2)  \gamma_{AB}+\hat{\slashed{g}}_{AB}
\eeaa
 we obtain 
\beaa
2\left(K(\ep)-\frac {1}{r^2} \right)&=& -\frac 1 2 \lapp \left(\tr_\gamma \slashed{g}(\ep)-2r^2 \right)-\frac{1}{r^2} \left( \tr_\gamma \slashed{g}(\ep)-2r^2 \right)+\divv \divv \hat{\slashed{g}}(\ep)
\eeaa
Projecting into the $\ell=0$ mode, since $\Klin=0$, this implies $\trglin=0$. The projection to the $\ell\geq 1$ mode gives 
\bea\label{check-K-tr}
2\check{K}&=& -\frac 1 2 \lapp \left(\check{\tr_\gamma \slashed{g}}\right)-\frac{1}{r^2} \left(\check{ \tr_\gamma \slashed{g}}\right)+\divv \divv \hat{\slashed{g}}
\eea
In particular, projecting to the $\ell=1$ mode, we obtain that $\left(\frac 1 2 \lapp +\frac{1}{r^2}\right) \check{\tr_\gamma \slashed{g}}_{\ell=1}=0$ and $\divv \divv \hat{\slashed{g}}_{\ell=1}=0$, therefore 
\bea\label{check-K-l=1}
\check{K}_{\ell=1}=0
\eea 
The vanishing of the $\ell=1$ spherical harmonics of the Gauss curvature will be crucial later in the proof of linear stability for the lower mode of the perturbations.

\subsection{The full set of linearized equations}\label{all-equations}
In the following, we present the equations arising from the formal linearization outlined above.

\subsubsection{The complete list of unknowns}
The equations will concern the following set of quantities, separated into symmetric traceless $2$-tensors, one-tensors and scalar functions on the Reissner-Nordstr{\"o}m manifold $(\MM,\g_{M,Q})$.
\beaa
\mathscr{S}_2&=&\{\a_{AB},\quad \aa_{AB}, \quad \chih_{AB}, \quad \chibh_{AB}, \quad \hat{\slashed{g}}_{AB} \}\\
\mathscr{S}_1&=&\{\ze_A, \quad \eta_A,  \quad \xib_A, \quad \b_{A}, \quad \bb_{A}, \quad \bF_{A}, \quad \bbF_{A}, \quad \underline{b}_A \}\\
\mathscr{S}_0&=&\{ \check{\ka}, \quad \check{\kab},  \quad \check{\omb}, \quad \check{\rho},\quad \check{\sigma},  \quad \check{\rhoF},  \quad\check{\sigmaF},\quad \check{\tr_{\gamma} \slashed{g}}, \quad \check{\Omegab}, \quad \check{\vsi}, \quad \check{K} \\
&& \kalin, \quad \kablin,  \quad  \omblin, \quad \rlin , \quad  \rhoFlin , \quad \sigmaFlin, \quad \Omegablin, \quad \vsilin \}
\eeaa

\begin{definition}\label{linear-perturbation-definition}
We say that $\mathscr{S}=\mathscr{S}_0 \cup \mathscr{S}_1 \cup \mathscr{S}_2$ is a \textbf{linear gravitational and electromagnetic perturbation around Reissner-Nordstr{\"o}m spacetime} if the quantities in $\mathscr{S}$ satisfy the equations \eqref{nabb-4-g}-\eqref{nabb-4-check-sigma-ze-eta} below.
\end{definition}

Observe that in the definition we omitted $\slin$ (indeed, $\slin=0$ is the linearization of \eqref{curl-eta-general}), $\Klin$ and $\trglin$, as they are implied to be identically zero by the previous subsection.

In what follows, the scalar functions without any superscript or check are to be intended as quantities in the background spacetime $(\MM,\g_{M,Q})$.

\subsubsection{Equations for the linearised metric components}\label{section-linearized-metric}

The linearization of \eqref{additional-metric-2} and \eqref{additional-metric-3} for $2$-tensors are the following:
\bea
\nabb_4 \hat{\slashed{g}}&=& 2\chih, \label{nabb-4-g} \\
\nabb_3\hat{\slashed{g}}&=& 2\chibh+2\DDs_2 \underline{b} \label{nabb-3-g}
\eea
The linearization of \eqref{nabb-4-vsi}, \eqref{nabb-A-Omegab} and \eqref{additional-metric-1} are the following:
\bea
\DDs_1(\check{\vsi}, 0)&=& \ze-\eta \label{nabb-check-vsi}\\
\DDs_1(\check{\Omegab}, 0)&=& \xib+\Omegab(\eta-\ze) \label{xib-Omegab-ze-eta}\\
\nabb_4 \underline{b}-\frac 1 2 \ka \underline{b}&=&-2(\eta+\ze)\label{derivatives-b-ze-eta}
\eea
The linearization of \eqref{nabb-4-vsi-1}, \eqref{additional-metric-0}, \eqref{additional-metric-2} and \eqref{additional-metric-3} are the following:
\bea
\nabb_4 \vsilin&=& 0 \label{nabb-4-vsilin} \\
\nabb_4 \Omegablin&=& -2\omblin +\left(\frac{M}{r}-\frac{Q^2}{r^2} \right)\kalin \label{nabb-4-Omegablin}
\eea
and
\bea
\nabb_4\check{\vsi}&=& 0 \label{nabb-4-check-vsi}\\
\nabb_4\check{\underline{\Omega}}&=&-2\check{\omb},\label{nabb-4-check-Omegab}\\
\nabb_4\left(r^{-2}\check{\tr_{\gamma}\slashed{g}}\right)&=&2\check{\ka}  \label{nabb-4-check-tr}\\
\nabb_3\left(r^{-2}\check{\tr_{\gamma}\slashed{g}}\right)&=&2(\check{\kab}-\ka\check{\Omegab} ) -2\divv \underline{b} \label{nabb-3-check-tr}
\eea

\subsubsection{Linearized null structure equations}

We collect here the linearisation of the equations in Section \ref{sec:nse}. 

The linearization of \eqref{nabb-3-chibh-general} and \eqref{nabb-3-chih-general} are the following:
\bea
\nabb_3 \chibh+\left( \kab +2\omb\right) \chibh&=&-2\DDs_2\xib -\aa, \label{nabb-3-chibh-ze-eta} \\
\nabb_4 \chih+\ka \ \chih&=& -\a , \label{nabb-4-chih-ze-eta}\\
\nabb_3\chih+\left(\frac 12  \kab -2 \omb\right) \chih   &=&  -2 \slashed{\mathcal{D}}_2^\star \eta-\frac 1 2 \ka \chibh   \label{nabb-3-chih-ze-eta}\\
\nabb_4\chibh+\frac 12  \ka \,\chibh &=&  2 \slashed{\mathcal{D}}_2^\star \ze -\frac 1 2 \kab \chih \label{nabb-4-chibh-ze-eta}, 
\eea

The linearisation of \eqref{nabb-3-ze-general}, \eqref{nabb-4-xib-general} and \eqref{codazzi-general} are the following:
\bea
\nabb_3 \ze+ \left(\frac 1 2 \kab-2\omb \right)\ze&=&2 \DDs_1( \check{\omb}, 0)-\left(\frac 1 2 \kab+2\omb \right)\eta +\frac 1 2 \ka \xib - \bb  -\rhoF\bbF, \label{nabb-3-ze-ze-eta} \\
\nabb_4 \ze+  \ka\ze&=& - \b  -\rhoF\bF, \label{nabb-4-ze-ze-eta} \\
\nabb_4\xib +\frac 1 2 \ka \xib &=&-2 \DDs_1( \check{\omb}, 0)+2\omb (\eta-\ze), \label{nabb-4-xib-ze-eta}\\
\nabb_4 \eta+\frac 1 2 \ka \eta&=&-\frac 1 2 \ka\ze  -\b-\rhoF \bF , \label{nabb-4-eta-ze-eta} \\
\divv \chibh&=&-\frac 1 2 \kab \ze-\frac 1 2 \DDs_1( \check{\kab}, 0) +\bb-\rhoF\bbF, \label{Codazzi-chib-ze-eta} \\
\divv \chih&=&\frac 1 2 \ka \ze-\frac 1 2 \DDs_1(\check{\ka}, 0) -\b +\rhoF\bF \label{Codazzi-chi-ze-eta}
\eea

The linearization of \eqref{nabb-3-kab-general}, \eqref{nabb-3-ka-general} and \eqref{nabb-4-omb-general} are the following:
\bea
\nabb_3 \kalin+\frac 1 2 \kab\kalin&=& 2\omb\ \kalin+\frac{4}{r}\omblin+2\rlin, \label{nabb-3-kalin-ze-eta}\\
\nabb_4 \kalin+\frac 1 2 \ka \kalin&=&0, \label{nabb-4-kalin-ze-eta}\\
\nabb_3\kablin+\frac 1 2 \kab\kablin&=& -2  \kab\omblin, \label{nabb-3-kablin-ze-eta}\\
\nabb_4\kablin+\frac 12  \ka\,\kablin &=& 2\omb \ \kalin+2\rlin, \label{nabb-4-kablin-ze-eta} \\
\nabb_4 \omblin&=& \rlin + \frac{2Q}{r^2}\rhoFlin+\left(\frac{M}{r^2}-\frac{3Q^2}{2r^3}\right)\kalin \label{nabb-4-omblin-ze-eta}
\eea
and 
\bea
\nabb_4 \check{\ka}+\ka\check{\ka}&=& 0, \label{nabb-4-check-ka-ze-eta}\\
\nabb_3\check{\ka}+\left(\frac 1 2 \kab-2\omb\right)\check{\ka}&=&-\frac 1 2\ka \left(\check{\kab}-\ka\check{\underline{\Omega}}\right)+ 2\ka\check{\omb} +2 \divv \eta +2\check{\rho} \label{nabb-3-check-ka-ze-eta},\\
\nabb_4\check{\kab}+\frac 12 \ka\check{\kab}&=& -\frac 1 2 \kab\check{\ka}-2 \divv \ze +2\check{\rho}, \label{nabb-4-check-kab-ze-eta}\\
\nabb_3\check{\kab}+\left(\kab+2\omb\right)\check{\kab}&=&-2\kab\check{\omb}+2\divv\xib+\left(\frac 1 2 \ka\kab-2\rho\right)\check{\underline{\Omega}}\label{nabb-3-check-kab-ze-eta}\\
\nabb_4\check{\omb}&=& \check{\rho} + 2\rhoF\check{\rhoF}, \label{nabb-4-check-omb-ze-eta}
\eea

The linearization of \eqref{curl-xib-general}, \eqref{curl-eta-general} and \eqref{Gauss-general} are the following:
\bea
0&=&-\frac 1 4 \ka \kablin- \frac 1 4 \kab \kalin -\rlin+2\rhoF \rhoFlin \label{Gauss-lin}
\eea
and
\bea
 \curll\xib&=&0 \label{curl-xib} \\
 \check{\sigma}&=& \curll \ze \label{sigma-curl-ze-eta}\\
 \curll\left(\ze-\eta\right)&=& 0 \label{curl-ze-eta} \\
\check{K}&=&- \frac 1 4 \kab \check{\ka}- \frac 1 4 \ka \check{\kab} -\check{\rho}+2\rhoF\check{\rhoF}  \label{Gauss-check}
\eea

\subsubsection{Linearized Maxwell equations}
We collect here the linearisation of the equations in Section \ref{Maxwell-general}.

The linearization of the equations \eqref{nabb-3-bF-general} are the following:
\bea
\nabb_3 \bF+\left(\frac 1 2 \kab-2\omb\right) \bF&=& -\DDs_1(\check{\rhoF}, \check{\sigmaF}) +2\rhoF \eta, \label{nabb-3-bF-ze-eta}\\
\nabb_4 \bbF+\frac 1 2 \ka \bbF&=&\DDs_1(\check{\rhoF}, -\check{\sigmaF})+2\rhoF \ze \label{nabb-4-bbF-ze-eta}
\eea

The linearization of \eqref{nabb-3-sigmaF-general} and \eqref{nabb-4-rhoF-general} are the following:
\bea
\nabb_3\sigmaFlin+\kab \ \sigmaFlin&=&0, \label{nabb-3-sigmaFlin} \\
\nabb_4\sigmaFlin+\ka \ \sigmaFlin&=&0, \label{nabb-4-sigmaFlin}\\
\nabb_3 \rhoFlin+ \kab \  \rhoFlin &=&0\label{nabb-3-rhoFlin}\\
\nabb_4 \rhoFlin+ \ka \  \rhoFlin &=&0\label{nabb-4-rhoFlin}
\eea
and
\bea
\nabb_3\check{ \rhoF}+ \kab \check{\rhoF}&=&- \rhoF\left(\check{\kab}-\ka\check{\Omegab}\right)  -\divv\bbF \label{nabb-3-check-rhoF-ze-eta}\\
\nabb_4\check{ \rhoF}+ \ka \check{\rhoF}&=&- \rhoF\check{\ka} +\divv\bF\label{nabb-4-check-rhoF-ze-eta} \\
\nabb_3\check{\sigmaF}+\kab \ \check{\sigmaF}&=& \slashed{\curl}\bbF\label{nabb-3-check-sigmaF-ze-eta}\\
\nabb_4\check{\sigmaF}+\ka \ \check{\sigmaF}&=& \slashed{\curl}\bF\label{nabb-4-check-sigmaF-ze-eta}
\eea

\subsubsection{Linearized Bianchi identities}\label{linearized-Bianchi-ids}
We collect here the linearisation of the equations in Section \ref{bieq}. 

The linearization of equations \eqref{nabb-3-a-general} are the following:
\bea
\nabb_3\a+\left(\frac 1 2 \kab-4\omb \right)\a&=&-2 \DDs_2\, \b-3\rho\chih -2\rhoF \ \left(\DDs_2\bF +\rhoF\chih \right) \label{nabb-3-a-ze-eta}\\
 \nabb_4\aa+\frac 1 2 \ka \aa&=&2 \DDs_2\, \bb-3\rho\chibh +2 \rhoF \ \left(\DDs_2\bbF -\rhoF\chibh \right)  \label{nabb-4-aa-ze-eta}
 \eea

 The linearisation of equations \eqref{nabb-4-b-general} and \eqref{Bianchi-b-2} are the following:
 \bea
 \nabb_3 \b+\left(\kab -2\omb\right) \b &=&\DDs_1(-\check{\rho}, \check{\sigma}) +3\rho  \eta +\rhoF\left(-\DDs_1(\check{\rhoF}, \check{\sigmaF})  - \ka\bbF-\frac 1 2 \kab \bF    \right) , \label{nabb-3-b-ze-eta}\\
 \nabb_4 \bb+\ka \bb &=&\DDs_1(\check{\rho}, \check{\sigma}) +3 \rho \ze +\rhoF\left(\DDs_1(\check{\rhoF}, -\check{\sigmaF})  - \kab\bF-\frac 1 2 \ka \bbF    \right) , \label{nabb-4-bb-ze-eta}\\
 \nabb_3 \bb+ \left( 2\kab  +2\omb\right) \bb &=&-\divv\aa-3\rho \xib  +\rhoF\left(\nabb_3\bbF+2\omb \bbF+2\rhoF \ \xib\right), \label{nabb-3-bb-ze-eta}\\
 \nabb_4 \b+ 2\ka \b &=&\divv\a +\rhoF \nabb_4\bF \label{nabb-4-b-ze-eta}
 \eea
 
 The linearization of \eqref{nabb-4-rho-general} and \eqref{nabb-4-sigma-general} are the following:
 \bea
\nabb_3 \rlin +\frac 3 2 \kab \rlin&=& -2\kab\rhoF \rhoFlin \label{nabb-3-rlin}\\
\nabb_4 \rlin +\frac 3 2 \ka \rlin&=& -2\ka\rhoF \rhoFlin \label{nabb-4-rlin}
\eea
and 
\bea
\nabb_3\check{\rho}+\frac 3 2 \kab \check{\rho}&=& -\left(\frac 3 2  \rho+\rhoF^2\right)(\check{\kab}-\ka \check{\Omegab})-2\kab\rhoF\check{\rhoF}-\divv\bb-\rhoF \ \divv\bbF \label{nabb-3-check-rho-ze-eta} \\
\nabb_4\check{\rho}+\frac 3 2 \ka \check{\rho}&=&-\left(\frac 3 2 \rho+\rhoF^2\right) \check{\ka}-2\ka\rhoF\check{\rhoF} +\divv\b+\rhoF \ \divv\bF \label{nabb-4-check-rho-ze-eta} \\
\nabb_3 \check{\sigma}+\frac 3 2 \kab \check{\sigma}&=&-\slashed{\curl}\bb - \rhoF \ \slashed{\curl}\bbF \label{nabb-3-check-sigma-ze-eta}\\
\nabb_4 \check{\sigma}+\frac 3 2 \ka \check{\sigma}&=&-\slashed{\curl}\b - \rhoF  \ \slashed{\curl}\bF \label{nabb-4-check-sigma-ze-eta}
\eea

The above equations \eqref{nabb-4-g}-\eqref{nabb-4-check-sigma-ze-eta} exhaust all the equations governing the dynamics of linear electromagnetic and gravitational perturbations of Reissner-Nordstr\"om spacetime in Bondi gauge.

\section{Special solutions: pure gauge and linearized Kerr-Newman} \label{sec:specialsol}
In this section, we consider two special linear gravitational and electromagnetic perturbations around Reissner-Nordstr{\"o}m spacetime: the pure gauge solutions and the linearized Kerr-Newman solutions.

These solutions are of fundamental importance in the proof of linear stability. The convergence of a linear gravitational and electromagnetic perturbation around Reissner-Nordstr{\"o}m spacetime only holds modulo a certain additional gauge freedom and modulo the convergence to a linearized Kerr-Newman solution. 

We describe here in general such solutions and we will specialize in Section \ref{gauge-normalized-solutions-chapter} to the actual choice of gauge and Kerr-Newman parameters in the linear stability. 
We begin in Section \ref{sec:ssgauge} with a discussion of pure gauge solutions to the linearized Einstein-Maxwell equations, followed by the description of a $6$-dimensional family of linearized Kerr-Newman solutions in Section \ref{6dimfam}.

\subsection{Pure gauge solutions $\mathscr{G}$}\label{sec:ssgauge}
Pure gauge solutions to the linearized Einstein-Maxwell equations are those derived from
linearizing the families of metrics that arise from applying to Reissner-Nordstr{\"o}m
smooth coordinate transformations which
preserve the Bondi form of the metric \eqref{double-null-metric}. We will classify such solutions here, making a connection between coordinate transformations and null frame transformations which preserve the Bondi form.

\subsubsection{Coordinate and null frame transformations}

In order to obtain pure gauge solutions in the setting of linearized Einstein-Maxwell equations we can equivalently consider coordinate transformations applied to the metric, or null frame transformations applied to the null frame associated to the metric. For completeness, we make here a connection between these two approaches.

Consider four functions $g_1$, $g_2$, $g_3$, $g_4$ on the Reissner-Nordstr\"om manifold, and consider a smooth one-parameter family of coordinates defined by 
\beaa
\tilde{u}&=& u+\ep g_1(u, r, \th, \phi), \quad \tilde{r}= r+\ep g_2(u, r, \th, \phi), \quad \tilde{\th}= \th+\ep g_3(u, r, \th, \phi), \quad \tilde{\phi}= \phi+\ep g_4(u, r, \th, \phi) 
\eeaa
If we express the Reissner-Nordstr\"om metric in the form \eqref{Schw:EF-coordinates-out} with respect to $\tilde{u}$, $\tilde{r}$, $\tilde{\th}$, $\tilde{\phi}$:
\bea
\label{Schw:EF-coordinates-out-gauge}
g_{M, Q}=- 2 d\tilde{u}  d\tilde{r}+\underline{\Omega}(\tilde{r}) d\tilde{u}^2  + \tilde{r}^2(d\tilde{\theta}^2+\sin^2\tilde{\theta} d\tilde{\phi}^2). 
\eea
then this defines with respect to the original coordinates $u$, $r$, $\th$, $\phi$ a one-parameter family of metrics. We can classify the coordinate transformations which preserve the Bondi form of the metric \eqref{double-null-metric}.
 \begin{lemma}\label{lemma-coords}
The general coordinate transformation that preserves the Bondi form of the metric is given by
  \beaa
\tilde{u}&=& u+\ep g_1(u, \th, \phi) \\
\tilde{r}&=& r+\ep \left(r \cdot w_1(u, \th, \phi) +w_2(u, \th, \phi) \right) \\
\tilde{\th}&=& \th+\ep \left(-\frac{1}{r}(g_1)_\th(u, \th, \phi)+j_3(u, \th, \phi)\right) \\
\tilde{\phi}&=& \phi+\ep \left(-\frac{1}{r\sin^2\th}(g_1)_\phi(u, \th, \phi)+j_4(u, \th, \phi)\right)
\eeaa
for any function $g_1(u, \th, \phi)$, $w_1(u, \th, \phi)$, $w_2(u, \th, \phi)$, $j_3(u, \th, \phi)$, $j_4(u, \th, \phi)$.
\end{lemma}
\begin{proof} See Appendix \ref{app-sec}.
\end{proof}

Null frame transformations, i.e. linear transformations which take null frames into null frames, can be thought of as pure gauge transformations, which correspond to a change of coordinates.
We recall here the classification of null frame transformations.
\begin{lemma}[Lemma 2.3.1 in \cite{stabilitySchwarzschild}] A linear null frame transformation, i.e. a map from null frames to null frames modulo quadratic terms in $f$, $\underline{f}$ and $\log \lambda$, can be written in the form 
\beaa
e_4'&=&\lambda \left(e_4+f^A e_A \right), \\
e_3'&=& \lambda^{-1} \left(e_3+\underline{f}^A e_A \right), \\
e_A'&=& {O_A}^B e_B+\frac 1 2 \underline{f}_A e_4+\frac 1 2 f_A e_3
\eeaa
where $\lambda$ is a scalar function, $f$ and $\underline{f}$ are $S_{u,s}$-tensors and ${O_A}^B$ is an orthogonal transformation of $(S_{u,s}, \slashed{g})$, i.e. ${O_A}^C {O_B}^D\slashed{g}_{CD}=\slashed{g}_{AB}$.
\end{lemma}
Observe that the identity transformation is given by $\lambda=1$, $f_A=\underline{f}_A=0$ and ${O_A}^B=\de_A^B$. Therefore, a linear perturbation of a null frame is a one for which $\log(\lambda)=f_A=\underline{f}_A=O(\ep)$ and ${O_A}^B=\de_A^B+O(\ep)$.

Writing the transformation for the Ricci coefficients and curvature components under a general null transformation of this type, we have for example (see Proposition 2.3.4. in \cite{stabilitySchwarzschild}):
\beaa
\xi_A'&=& \lambda^2\left( \xi_A+\frac 1 2 \lambda^{-1} e_4'(f_A)+\om f_A +\frac 1 4 \ka f_A \right) \\
\ze_A'&=& \ze_A-e_A'(\log \lambda)+\frac 1 4 (-\kab f_A+\ka \fb_A)+\om\fb_A-\omb f_A \\
\etab_A'&=& \etab_A+\frac 1 2 \lambda^{-1} e_4'(\fb_A)+\frac 1 2 \kab f_A -\om \fb_A \\
\om'&=& \lambda \left(\om-\frac 1 2 \lambda^{-1} e_4'(\log\lambda) \right)
\eeaa
If the metric is in Bondi gauge then it verifies \eqref{substitute-xi-xib-ze}, i.e. $\xi_A=0$,  $\om=0$ and $\etab_A+\ze_A=0$. This means that a null frame transformation which preserves the Bondi form has to similarly verify $\xi_A'=0$,  $\om'=0$ and $\etab_A'+\ze_A'=0$. This translates into conditions for $e_4 f$, $e_4 \fb$ and $e_4 \lambda$. In particular we have the following

\begin{lemma}\label{null-frame-Bondi} The general null frame transformation that preserves the Bondi form of the metric is given by  a transformation verifying
\bea
\nabb_4 \lambda&=&0 \label{condition-lambda-null}\\
\nabb_4 f+\frac 1 2 \ka f&=&0 \label{condition-null-f}\\
 \nabb_4 \underline{f} +\frac 1 2 \ka\underline{f}    &=& 2\omb f-2\DDs_1(\lambda, 0) \label{condition-null-ff}
\eea
\end{lemma}
\begin{proof} Straightforward computation from the above formulas for the change of null frame. 
\end{proof}

Given a coordinate transformation which preserves the Bondi metric as in Lemma \ref{lemma-coords}, we can associate a null frame transformation between the null frames canonically defined in terms of the coordinates vector fields by \eqref{null-frame-double-null}. In particular, we can explicitly write the terms which determine the null frame transformation $f$, $\fb$, $\lambda$, $O_{AB}$ in terms of the coordinate transformations $g_1$, $w_1$, $w_2$, $j^A$. 
We summarize the relation in the following lemma.

\begin{lemma}\label{lemma-2} Given a coordinate transformation which preserves the Bondi form of the metric as in Lemma \ref{lemma-coords} of the form 
  \beaa
\tilde{u}&=& u+\ep g_1(u, \th, \phi) \\
\tilde{r}&=& r+\ep \left(r \cdot w_1(u, \th, \phi) +w_2(u, \th, \phi) \right) \\
\tilde{\th^A}&=& \th^A+\ep \left(\DDs_1(g_1, 0)(u, \th, \phi)+j^A(u, \th, \phi)\right)
\eeaa
 then the null frame transformation which brings the associated null frame $\{ \tilde{e_4}, \tilde{e_3}, \tilde{e_A}\}$ into $\{ e_4, e_3, e_A\}$ is determined by
\beaa
\lambda&=&1+\ep w_1 \\
f &=& -\ep \DDs_1(g_1, 0) \\
\underline{f}&=& \ep\left(-2r\DDs_1( w_1, 0)-2\DDs_1(w_2, 0) +\underline{\Omega}(r)\DDs_1(g_1, 0)  \right) \\
O_A^B&=& \de_A^B +\ep \left(\nabb_A \nabb^B g_1+ \nabb_A(j^B)  \right)
\eeaa

\end{lemma}

Using the above Lemma, the conditions imposed to preserve the Bondi metric in terms of coordinate transformations or in terms of null frame transformations become manifest. They are the following:
 \begin{itemize}
 \item The condition for $g_1$ which gives $\partial_r(g_1)= 0$ translates into
   \beaa
   \nabb_4f &=& -\ep \nabb_4\DDs_1(g_1, 0)=-\ep \DDs_1(\nabb_4g_1, 0)+\frac 1 2 \ka \ep \DDs_1(g_1, 0)=-\frac 1 2 \ka f
   \eeaa
    which is the condition for the frame coming from imposing $\xi=0$, i.e. \eqref{condition-null-f}. 
   \item The condition for $g_2$ which gives $g_2(u, r, \th, \phi)=r \cdot w_1(u, \th, \phi) +w_2(u, \th, \phi)$ translates into
   \beaa
  \nabb_4 \underline{f}&=& \ep\left(-r\ka \DDs_1( w_1, 0) -2\nabb_4\DDs_1(w_2, 0) +\nabb_4\underline{\Omega}(r)\DDs_1(g_1, 0) +\underline{\Omega}(r)\nabb_4\DDs_1(g_1, 0)  \right)\\
  &=& \ep\left(\ka\DDs_1(w_2, 0) -2\omb\DDs_1(g_1, 0) -\frac 1 2 \ka \underline{\Omega}(r)\DDs_1(g_1, 0)  \right) \\
   &=& -2\DDs_1(\lambda, 0) -\frac 1 2 \ka(\underline{f}) +2\omb f
   \eeaa
    which is the condition for the frame as a result of imposing $\etab+\ze=0$, i.e. \eqref{condition-null-ff}.
    \item The condition for $g_2$ which gives $g_2(u, r, \th, \phi)=r \cdot w_1(u, \th, \phi) +w_2(u, \th, \phi)$ also translates into
    \beaa
\nabb_4 \lambda&=& \nabb_4(1+\ep w_1)=0
\eeaa
which is the condition for the frame coming from imposing $\om=0$, i.e. \eqref{condition-lambda-null}. 
\item Writing $j=-r\DDs_1(q_1, q_2)$ for two functions $q_1$, $q_2$ with vanishing mean, the conditions for $j^A$ which give $j^A=j^A(u, \th, \phi)$
translate into 
\beaa
\nabb_4 q_1=0, \qquad \nabb_4 q_2=0
\eeaa
 \end{itemize}

In the next subsections, we will look at the explicit pure gauge solutions produced by null frame or coordinate transformations preserving the Bondi form of the metric, and we separate them into 
\begin{enumerate}
\item pure gauge solutions arising from setting $j^{A}=0$: Lemma \ref{pure-gauge-solution}
\item pure gauge solutions arising from setting $\log\lambda=f=\underline{f}=0$ (or equivalently $g_1=w_1=w_2=0$): Lemma \ref{pure-gauge-2}
\end{enumerate}
In view of linearity, the general pure gauge solution can be obtained from summing solutions in the two above cases. 

\subsubsection{Pure gauge solutions with $j^A=0$}

The following is the explicit form of the pure gauge solution arising from a null transformation with $j^A=0$.
Define 
\beaa
h&=& -\ep g_1 \\
\underline{h}&=& \ep \left(-2r w_1-2w_2+\Omegab g_1 \right)\\
a&=&\ep w_1
\eeaa
then according to Lemma \ref{lemma-2}, the null frame components can be simplified to
\beaa
f=\DDs_1(h, 0) \qquad \underline{f}=\DDs_1(\underline{h}, 0) \qquad  \lambda=e^a
\eeaa
In particular, the relations on $f$, $\fb$ and $\lambda$ given by Lemma \ref{null-frame-Bondi} translate into conditions on the derivative along the $e_4$ directions for the functions $h$, $\underline{h}$ and $a$ (conditions \eqref{transport-a}-\eqref{transport-fb}). 

Observe that the Bondi conditions given by $\xi_A=\etab_A+\ze_A=0$ do not give impose any gauge conditions on the projection to the $\ell=0$ spherical harmonics. In particular, we have additional gauge freedom at the $\ell=0$ which is translated in the functions $c$, $d$ and $l$ supported in the $\ell=0$ spherical harmonics below.

\begin{lemma}\label{pure-gauge-solution}Let $h$, $\underline{h}$, $a$ be smooth functions supported in $\ell \ge 1$ spherical harmonics. Suppose they verify the following transport equations:
\bea
\nabb_4 a&=&0 \label{transport-a}\\
\nabb_4 h&=&0 \label{transport-f}\label{transport-h}\\
\nabb_4 \underline{h}&=& 2\omb h -2a \label{transport-hb}\label{transport-fb}
\eea
Then the following  is a linear gravitational and electromagnetic perturbation around Reissner-Nordstr\"om spacetime supported in $\ell\geq 1$ spherical harmonics.
The linearized metric components are given by 
\beaa
\hat{\slashed{g}}&=& 2r\DDs_2\DDs_1(h, 0) , \qquad \underline{b}=r\nabb_3\DDs_1(h, 0)+\DDs_1(\underline{h}, 0) , \qquad  \check{\tr_{\gamma} \slashed{g}}=r^2\left( -2r\DDd_1\DDs_1(h,0)-\kab h -\ka \underline{h}\right)  \\
\check{\Omegab}&=&\frac 1 2 \nabb_3\underline{h}+\omb \underline{h}+\Omegab\left(\frac 12 \nabb_3 h -a \right) , \qquad \check{\vsi}=-\frac 12 \nabb_3 h +a
\eeaa
The Ricci coefficients are given by
\beaa
\chih&=& -\DDs_2 \DDs_1(h, 0), \qquad \chibh=-\DDs_2 \DDs_1(\underline{h}, 0), \qquad \ze= \left(-\frac 1 4 \kab-\omb\right)\DDs_1(h, 0)+\frac 1 4 \ka \DDs_1(\underline{h}, 0)+\DDs_1(a, 0) \\
\eta&=&\frac 1 2 \nabb_3 \DDs_1(h, 0)-\omb \DDs_1(h, 0)+\frac 1 4 \ka \DDs_1(\underline{h}, 0), \qquad \xib=\frac 1 2 \nabb_3\DDs_1(\underline{h}, 0)+\left(\frac{1}{4}\kab +\omb\right) \DDs_1(\underline{h}, 0) \\
\check{\ka}&=& \ka a +\DDd_1\DDs_1 (h, 0) +\left( \frac 1 4 \ka\kab\right) h+\frac 1 4 \ka^2\underline{h}, \\ 
\check{\kab}&=&-\kab a+\DDd_1\DDs_1(\underline{h}, 0)+\left( \frac 1 4 \kab^2+\omb\kab \right) h+\left( \frac 1 4 \ka \kab-\rho \right) \underline{h}, \\
 \check{\omb}&=& \frac 1 2 \nabb_3 a-  \omb a-\frac 1 2 (\nabb_3 \omb )h -\frac 1 2 \left(\rho+\rhoF^2\right) \underline{h}  , \\
  \check{K}&=&- \frac 1 4 \ka\DDd_1 \DDs_1(\underline{h},0) - \frac 1 4 \kab\DDd_1 \DDs_1(h,0) +\frac {1}{ 2r^2}\left( \kab h+\ka \underline{h}\right)
\eeaa
The electromagnetic components are given by 
\beaa
\bF&=& \rhoF \DDs_1(h, 0), \qquad \bbF= -\rhoF \DDs_1(\underline{h}, 0)\qquad  \check{\rhoF}=\frac 1 2 \rhoF \left(\kab h + \ka  \underline{h}\right) , \qquad \check{\sigmaF}=0
\eeaa
The curvature components are given by 
\beaa
\a&=& 0, \qquad \aa=0 \qquad \b=\frac 3 2 \rho \DDs_1(h, 0) , \qquad \bb=-\frac 3 2 \rho \DDs_1(\underline{h}, 0) \\
 \check{\rho}&=& \left(\frac 3 4\rho+\frac 1 2 \rhoF^2 \right) \left(\kab h  +\ka \underline{h}\right) , \qquad \check{\sigma}=0
\eeaa

Let $c$, $d$ and $l$ be smooth functions supported in $\ell=0$ spherical harmonics. Suppose they verify the following transport equations:
\bea
\nabb_4 c&=&0 \label{nabb-4-c} \\
\nabb_4d&=&0 \label{nabb-4-d} \\
\nabb_4 l &=& -\nabb_3 c +4\omb c \label{nabb-4-l}
\eea
Then the following is a (spherically symmetric) linear gravitational and electromagnetic perturbation around Reissner-Nordstr\"om spacetime.
\beaa
\vsilin&=& d, \qquad \Omegablin=l \qquad \kalin= \ka c, \qquad \kablin= -\kab c,\qquad \omblin= \frac 1 2 \nabb_3c-\omb c, \\
&&\rhoFlin= \sigmaFlin=\rlin=0,
\eeaa
\end{lemma}
\begin{proof} 
We check that the quantities defined above verify the equations in Section \ref{all-equations}. 

We verify some of the equations for metric coefficients. Equations \eqref{nabb-4-g} and \eqref{nabb-3-g} are verified using \eqref{transport-h} and the fact that $\ka=\frac 2 r $:
\beaa
\nabb_4 \hat{\slashed{g}}&=&\nabb_4 (2r\DDs_2\DDs_1(h, 0))=2\DDs_2\DDs_1(h, 0)+2r(-\ka \DDs_2\DDs_1(h, 0)) = 2\chih \\
\nabb_3\hat{\slashed{g}}&=& \nabb_3 (2r\DDs_2\DDs_1(h, 0))=r\kab \DDs_2\DDs_1(h, 0)+2r(\DDs_2(\nabb_3\DDs_1(h, 0))-\frac 1 2 \kab \DDs_2\DDs_1(h, 0)) =2\chibh+2\DDs_2\underline{b}
\eeaa
 Equation \eqref{derivatives-b-ze-eta} is verified using \eqref{transport-h} and the fact that $2\omb-r\rho=4\omb$:
\beaa
&&\nabb_4 \underline{b}-\frac 1 2 \ka \underline{b}=\nabb_3\DDs_1(h, 0)+r\nabb_4\nabb_3\DDs_1(h, 0)+\nabb_4\DDs_1(\underline{h}, 0)-\frac 1 2 \ka (r\nabb_3\DDs_1(h, 0)+\DDs_1(\underline{h}, 0) ) \\
&=&r((\frac 1 4 \ka\kab-\rho) \DDs_1(h, 0)-\frac 1 2 \ka \nabb_3\DDs_1(h, 0))+2\omb \DDs_1(h, 0)-\frac 1 2 \ka \DDs_1(\underline{h}, 0)-2\DDs_1(a, 0)-\frac 1 2 \ka (\DDs_1(\underline{h}, 0) ) \\
&=&- \nabb_3 \DDs_1(h, 0)- \ka \DDs_1(\underline{h}, 0)+\left(\frac 1 2 \kab+4\omb\right)\DDs_1(h, 0)-2\DDs_1(a, 0)=-2(\eta+\ze)
\eeaa
Observe that $\check{\kab}-\ka\check{\Omegab}=\DDd_1\DDs_1(\underline{h}, 0)+\left( \frac 1 4 \kab^2+\omb\kab \right) h+ \frac 1 4 \ka \kab  \underline{h}-\frac 1 2 \ka\nabb_3\underline{h}-\frac 12 \kab \nabb_3 h  $. Equations \eqref{nabb-4-check-tr} and \eqref{nabb-3-check-tr} are verified:
\beaa
\nabb_4(r^{-2} \check{\tr_{\gamma} \slashed{g}})&=&-2\DDd_1\DDs_1(h,0)-2r\nabb_4\DDd_1\DDs_1(h,0)-\nabb_4\kab h -\nabb_4\ka \underline{h}-\ka \nabb_4\underline{h}\\
&=&-2\DDd_1\DDs_1(h,0)-2r(-\ka\DDd_1\DDs_1(h,0))-(-\frac 1 2 \ka\kab+2\rho) h -(-\frac 1 2 \ka^2) \underline{h}-\ka (2\omb h -2a)=2\check{\ka}\\
\nabb_3(r^{-2}\check{\tr_{\gamma} \slashed{g}})&=&-r\kab\DDd_1\DDs_1(h,0)-2r(\DDd_1\nabb_3\DDs_1(h,0)-\frac 1 2 \kab  \DDd_1\DDs_1(h, 0))-\nabb_3\kab h-\kab \nabb_3h -\nabb_3\ka \underline{h}-\ka \nabb_3\underline{h}\\
&=& 2(\check{\kab}-\ka\check{\Omegab})-2\divv\underline{b}
\eeaa
We verify some of the null structure equations. Equation \eqref{nabb-3-chibh-ze-eta} is verified using \eqref{commutator-nabb-4-DDs}:
\beaa
&&\nabb_3 \chibh+\left( \kab +2\omb\right)\chibh+2\DDs_2\xib +\aa\\
&=& \nabb_3(-\DDs_2 \DDs_1(\underline{h}, 0))+\left( \kab +2\omb\right)(-\DDs_2 \DDs_1(\underline{h}, 0))+2\DDs_2(\frac 1 2 \nabb_3(\DDs_1(\underline{h}, 0))+(\frac{1}{4}\kab +\omb) \DDs_1(\underline{h}, 0))=0
\eeaa
Equation \eqref{nabb-3-kalin-ze-eta} is verified, using that $2\omb+r\rho=0$: 
\beaa
\nabb_3 \kalin+\frac 1 2 \kab\kalin&=&\nabb_3(c \ka)+\frac 1 2 c\kab \ka=\nabb_3(c)\ka= 2\omb\ \kalin+\frac{4}{r}\omblin+2\rlin
\eeaa
Equation \eqref{nabb-4-kablin-ze-eta} is verified, using \eqref{transport-a}:
\beaa
\nabb_4\kablin+\frac 12  \ov{\ka}\,\kablin &=&\nabb_4(-c \kab)+\frac 12  \ov{\ka}\,(-c\kab) =-c\left(\nabb_4(\kab)+\frac 12  \ka\,\kab\right)=-2c \rho\\
\left(\frac{2M}{r^2}-\frac{2Q^2}{r^3} \right)\kalin+2\rlin&=& \left(\frac{2M}{r^2}-\frac{2Q^2}{r^3} \right)c\ka=-2c\rho
\eeaa
Equation \eqref{nabb-4-omblin-ze-eta} is verified, using \eqref{transport-a}:
\beaa
\nabb_4 \omblin&=&\nabb_4 (-c\omb+\frac 1 2 \nabb_3(c))=-c\nabb_4\ov{\omb}=-c(\rho + \rhoF^2)\\
\left(\frac{M}{r^2}-\frac{3Q^2}{2r^3}\right)\kalin &=&\left(\frac{M}{r^2}-\frac{3Q^2}{2r^3}\right)c\ka=\left(\frac{2M}{r^3}-\frac{3Q^2}{r^4}\right)c=-c\left(-\frac{2M}{r^3}+\frac{2Q^2}{r^4}+\frac{Q^2}{r^4}\right)
\eeaa

We verify some of the Maxwell equations. Equation \eqref{nabb-3-bF-ze-eta} is verified:
\beaa
&&\nabb_3 \bF+\left(\frac 1 2 \kab-2\omb\right) \bF+\DDs_1(\check{\rhoF}, \check{\sigmaF}) -2\rhoF \eta\\
&=&-\kab \rhoF \DDs_1(h, 0)+\rhoF \nabb_3\DDs_1(h, 0)+\left(\frac 1 2 \kab-2\omb\right) (\rhoF \DDs_1(h, 0))+\DDs_1(\frac 1 2 \rhoF \left(\kab h + \ka  \underline{h}\right), 0) \\
&&-2\rhoF (\frac 1 2 \nabb_3 \DDs_1(h, 0)-\omb \DDs_1(h, 0)+\frac 1 4 \ka \DDs_1(\underline{h}, 0))=0
\eeaa

We verify some of the Bianchi identities. 
 Equation \eqref{nabb-3-bb-ze-eta} is verified:
 \beaa
 &&\nabb_3 \bb+ \left( 2\kab  +2\omb\right) \bb+\divv\aa+3\rho \xib  -\rhoF\left(\nabb_3\bbF+2\omb \bbF+2\rhoF \ \xib\right)\\
 &=&\nabb_3 (-\frac 3 2 \rho \DDs_1(\underline{h}, 0))+ \left( 2\kab  +2\omb\right)(-\frac 3 2 \rho \DDs_1(\underline{h}, 0))+3\rho (\frac 1 2 \nabb_3(\DDs_1(\underline{h}, 0))+(\frac{1}{4}\kab +\omb) \DDs_1(\underline{h}, 0)) \\
 && -\rhoF\left(\nabb_3(-\rhoF \DDs_1(\underline{h}, 0))+2\omb (-\rhoF \DDs_1(\underline{h}, 0))+2\rhoF \ (\frac 1 2 \nabb_3(\DDs_1(\underline{h}, 0))+(\frac{1}{4}\kab +\omb) \DDs_1(\underline{h}, 0))\right)\\
  &=& -\frac 3 2 (-\frac 3 2 \kab \rho-\kab \rhoF^2) \DDs_1(\underline{h}, 0)+ \left( 2\kab \right)(-\frac 3 2 \rho \DDs_1(\underline{h}, 0))+3\rho ((\frac{1}{4}\kab ) \DDs_1(\underline{h}, 0)) \\
 && -\rhoF\left(-(-\kab \rhoF) \DDs_1(\underline{h}, 0)+2\rhoF \ ((\frac{1}{4}\kab) \DDs_1(\underline{h}, 0))\right)=0
 \eeaa
The remaining equations are verified in a similar manner. 
\end{proof}

\subsubsection{Pure gauge solutions with $g_1=w_1=w_2=0$}
The following is the explicit form of the pure gauge solution arising from a null transformation for which $e_3$ and $e_4$ are unchanged, while the frame on the spheres change.  They only generate non-trivial values for the metric components, while all other quantities of the solution vanish. 

Define
\beaa 
j=-r\DDs_1(q_1, q_2)
\eeaa
 for two functions $q_1$, $q_2$ with vanishing mean.

\begin{lemma}\label{pure-gauge-2} Let $q_1$, $q_2$ be smooth functions, with $q_1$ supported in $\ell\geq 1$ and $q_2$ supported in $\ell\geq 2$ spherical harmonics.\footnote{The pure gauge solution corresponding to $q_1=0$ and $q_2=\dot{Y}^{\ell=1}_m$ generates the trivial solution.} Suppose they verify the following transport equations:
\bea
\nabb_4q_1&=& 0 \label{nabb-4-q1}\\
\nabb_4q_2&=& 0 \label{nabb-4-q2}
\eea
Then the following  is a linear gravitational and electromagnetic perturbation around Reissner-Nordstr\"om spacetime. The linearized metric components are given by
\beaa
\hat{\slashed{g}}=2 r^2 \DDs_2 \DDs_1 (q_1, q_2), \qquad \check{\tr_\gamma \slashed{g}}=-2 r^4 \DDd_1\DDs_1( q_1, 0),  \qquad \underline{b}=r^2\DDs_1(\nabb_3 q_1, \nabb_3 q_2)
\eeaa
while all other components of the solution vanish. 
\end{lemma}
\begin{proof} Equations \eqref{nabb-4-g} and \eqref{nabb-3-g} are verified using \eqref{nabb-4-q1} and \eqref{nabb-4-q1}:
\beaa
\nabb_4 \hat{\slashed{g}}&=&2\nabb_4(r^2 \DDs_2 \DDs_1 (q_1, q_2))=2r^2 \DDs_2 \DDs_1 (\nabb_4q_1, \nabb_4q_2)=0 \\
\nabb_3\hat{\slashed{g}}&=&2\nabb_3(r^2 \DDs_2 \DDs_1 (q_1, q_2))=2r^2 \DDs_2 \DDs_1 (\nabb_3q_1, \nabb_3q_2)=2\DDs_2\underline{b}
\eeaa
Equation \eqref{derivatives-b-ze-eta} is verified:
\beaa
\nabb_4 \underline{b}-\frac 1 2 \ka \underline{b}&=&\nabb_4 (r^2\DDs_1(\nabb_3 q_1, \nabb_3 q_2))-\frac 1 2 \ka r^2\DDs_1(\nabb_3 q_1, \nabb_3 q_2)\\
&=&r\DDs_1(\nabb_3 q_1, \nabb_3 q_2)+r\DDs_1(\nabb_4\nabb_3 q_1, \nabb_4\nabb_3 q_2)-\frac 1 2 \ka r^2\DDs_1(\nabb_3 q_1, \nabb_3 q_2)=0
\eeaa
 Equations \eqref{nabb-4-check-tr} and \eqref{nabb-3-check-tr} are verified:
\beaa
\nabb_4 (r^{-2}\check{\tr_\gamma \slashed{g}})&=& \nabb_4\left(-2 r^2 \DDd_1\DDs_1( q_1, 0)\right)=\left(-2 r^2 \DDd_1\DDs_1( \nabb_4q_1, 0)\right)=0 \\
\nabb_3 (r^{-2}\check{\tr_\gamma \slashed{g}})&=& \nabb_3\left(-2 r^2 \DDd_1\DDs_1( q_1, 0)\right)=\left(-2 r^2 \DDd_1\DDs_1( \nabb_3q_1, 0)\right)=-2\divv\underline{b}
\eeaa
All the other equations are trivially satisfied.
\end{proof}

We identify the solutions given by Lemma \ref{pure-gauge-solution} and Lemma \ref{pure-gauge-2} to pure gauge solutions, and we denote their linear sum as $\mathscr{G}_I$ with $I=(h, \underline{h}, a, c, d, l,  q_1, q_2)$.

\subsubsection{Gauge-invariant quantities}\label{section-gauge-inv}
We can identify quantities which vanish for any pure gauge solution $\mathscr{G}_I$. Such quantities are referred to as gauge-invariant.

The symmetric traceless two tensors $\a$, $\aa$ are clearly gauge-invariant from Lemma \ref{pure-gauge-solution}.  These curvature components are important because in the case of the Einstein vacuum equation they verify a decoupled wave equation, the celebrated Teukolsky equation, first discovered in the Schwarzschild case in \cite{Bardeen} and generalized to the Kerr case in \cite{Teukolsky}. In the Einstein-Maxwell case, the tensors $\a$ and $\aa$ verify Teukolsky equations coupled with new quantities, denoted $\ff$ and $\underline{\ff}$. 

The symmetric traceless $2$-tensors
\bea\label{definition-ff}
\ff:=\DDs_2 \bF+\rhoF \chih \qquad \text{and} \qquad  \underline{\ff}:=\DDs_2 \bbF-\rhoF \chibh
\eea
 are gauge-invariant quantities.  Indeed, using Lemma \ref{pure-gauge-solution}, we see that for every pure gauge solution 
 \beaa
 \ff=\DDs_2\bF+\rhoF \chih=\DDs_2(\rhoF \DDs_1(h, 0))+\rhoF (-\DDs_2 \DDs_1(h, 0))=0
 \eeaa
 Similarly for $\ffb$. 
 
 Notice that the quantities $\ff$ and $\ffb$ appear in the Bianchi identities for $\a$ and $\aa$. The equations \eqref{nabb-3-a-ze-eta} and \eqref{nabb-4-aa-ze-eta} can be rewritten as 
 \bea
  \nabb_3\a+\left(\frac 1 2 \kab-4\omb \right)\a&=&-2 \DDs_2\, \b-3\rho\chih -2\rhoF \ \ff, \label{nabb-3-a-ff}\\
 \nabb_4\aa+\left(\frac 1 2 \ka-4\om\right) \aa&=&2 \DDs_2\, \bb-3\rho\chibh +2 \rhoF \ \ffb\label{nabb-4-aa-ff} 
 \eea
Using the above, it is clear that $\ff$ and $\ffb$ shall appear on the right hand side of the wave equation verified by $\a$ and $\aa$. 
The quantities $\ff$ and $\ffb$ themselves verify Teukolsky-type equations, which are coupled with $\a$ and $\aa$ respectively. The equations for $\a$ and $\ff$ and for $\aa$ and $\ffb$ constitute the generalized spin $\pm 2$ Teukolsky system obtained in Section \ref{spin2}.

Observe that the extreme electromagnetic component $\bF$ and $\bbF$ are not gauge-invariant if $\rhoF$ is not zero in the background.\footnote{In the case of the Maxwell equations in Schwarzschild, the components $\bF$ and $\bbF$ are gauge-invariant, and satisfy a spin $\pm1$ Teukolsky equation, see \cite{Federico}. } On the other hand, the one-forms 
\bea\label{definition-tilde-b}
\tilde{\b}:= 2\rhoF \b-3\rho \bF \qquad \text{and} \qquad \underline{\tilde{\b}}:= 2\rhoF \bb-3\rho \bbF
\eea
 are gauge invariant.  Indeed, using Lemma \ref{pure-gauge-solution}, we see that for every gauge solution 
\beaa
\tilde{\b}=2\rhoF \b-3\rho \bF=2\rhoF \left(\frac 3 2 \rho \DDs_1(h, 0)\right)-3\rho \left(\rhoF \DDs_1(h, 0)\right)=0
\eeaa
Similarly for $\tilde{\bb}$. 

\subsection{A $6$-dimensional linearised Kerr-Newman family $\mathscr{K}$}
\label{6dimfam}
The other class of special solutions
corresponds to the family 
that arises  by linearizing one-parameter representations of 
Kerr-Newman  around Reissner-Nordstr{\"o}m. We will present
such a family here, giving first in Section~\ref{sec:linss}  a $3$-dimensional family corresponding to
Kerr-Newman with fixed angular momentum $a$ (supported in $\ell=0$ spherical harmonics) and then in Section~\ref{sec:nontrivkerr}, a $3$-dimensional family
corresponding to Kerr-Newman with fixed mass $M$ and charge $Q$ (supported in $\ell=1$ spherical harmonics).

\subsubsection{Linearized Kerr-Newman solutions with no angular momentum} \label{sec:linss}
Reissner-Nordstr{\"o}m spacetimes are obviously solutions to the nonlinear Einstein-Maxwell equation. Therefore, linearization around the parameters $M$ and $Q$ give rise to solution of the linearized system of gravitational and electromagnetic perturbations, which can be interpreted as the solution converging to another Reissner-Nordstr{\"o}m solution with a small change in the mass or in the charge. 

In addition to those, there is a family of solutions with non-trivial magnetic charge, which can arise as solution of the Einstein-Maxwell equations. Indeed, the following expression gives stationary solutions to the Maxwell system on Reissner-Nordstr{\"o}m:
\beaa
\F=\frac{\mathfrak{b}}{r^2} \epsilon_{AB}+\frac{\mathfrak{Q}}{r^2} dt \wedge dr
\eeaa
where $\mathfrak{b}$ and $\mathfrak{Q}$ are two real parameters, respectively the magnetic and the electric charge. 

We summarize these solutions in the following Proposition. 

\begin{proposition} \label{lem:linss}
For every $\mathfrak{M}, \mathfrak{Q}, \mathfrak{b} \in \mathbb{R}$, the following is a (spherically symmetric) solution of the system of gravitational and electromagnetic perturbations in $\mathcal{M}$. The non-vanishing quantities are
\beaa
\rlin&=& \left(-\frac{2\frak{M}}{r^3}+\frac{4Q\frak{Q}}{r^4}\right) , \qquad \rhoFlin= \frac{\frak{Q}}{r^2} \qquad \sigmaFlin= \frac{\frak{b}}{r^2} \\
\kablin&=& \left(\frac{4\frak{M}}{r^2}-\frac{4Q\frak{Q}}{r^3}\right),\qquad \omblin= \left(\frac{\frak{M}}{r^2}-\frac{2Q\frak{Q}}{r^3}\right), \qquad \Omegablin= \left(\frac{2\frak{M}}{r}-\frac{2Q\frak{Q}}{r^2} \right) 
\eeaa
\end{proposition}

\begin{proof} We verify some of the equations in Section \ref{all-equations} which are not trivially satisfied.
Equation \eqref{nabb-4-kablin-ze-eta} is verified:
\beaa
\nabb_4\kablin+\frac 12  \ka\,\kablin &=& \nabb_4\left(\frac{4\frak{M}}{r^2}-\frac{4Q\frak{Q}}{r^3}\right)+\frac 1 2 \ka\left(\frac{4\frak{M}}{r^2}-\frac{4Q\frak{Q}}{r^3}\right)=\left(-\frac{4\frak{M}}{r^3}+\frac{8Q\frak{Q}}{r^4}\right)=2\rlin
\eeaa
Equation \eqref{nabb-4-omblin-ze-eta} reads:
\beaa
\nabb_4 \omblin&=& \nabb_4 \left(\frac{\frak{m}}{r^2}-\frac{2Q\frak{Q}}{r^3}\right)=  \left(-\frac{2\frak{m}}{r^2}+\frac{6Q\frak{Q}}{r^3}\right)= \rlin + \frac{2Q}{r^2}\rhoFlin
\eeaa
Equation \eqref{Gauss-lin} is verified:
\beaa
-\frac 1 4 \ka \kablin- \frac 1 4 \kab \kalin -\rlin+2\rhoF \rhoFlin &=& -\frac 1 4 \frac 2 r  \left(\frac{4\frak{m}}{r^2}-\frac{4Q\frak{Q}}{r^3}\right) -\left(-\frac{2\frak{m}}{r^3}+\frac{4Q\frak{Q}}{r^4}\right) +2\frac{Q}{r^2} \frac{\frak{Q}}{r^2}=0
\eeaa
The Maxwell equations \eqref{nabb-3-sigmaFlin}-\eqref{nabb-4-sigmaFlin} and \eqref{nabb-3-rhoFlin}-\eqref{nabb-4-rhoFlin} are verified: 
  \beaa
 \nabb_3\rhoFlin+\kab\rhoFlin&=&\nabb_3(\frac{\frak{Q}}{r^2})+\kab\frac{\frak{Q}}{r^2}=-2\frac{\frak{Q}}{r^3}\nabb_3r+\kab\frac{\frak{Q}}{r^2}=-2\frac{\frak{Q}}{r^2}\frac 1 2 \kab +\kab\frac{\frak{Q}}{r^2}=0
  \eeaa
The Bianchi identities \eqref{nabb-3-rlin}-\eqref{nabb-4-rlin} are verified:
\beaa
\nabb_3\rlin+\frac 3 2 \kab \rlin+\frac{2Q}{r^2}\kab\rhoFlin&=&\nabb_3( \left(-\frac{2\frak{m}}{r^3}+\frac{4Q\frak{Q}}{r^4}\right))+\frac 3 2 \kab \left(-\frac{2\frak{m}}{r^3}+\frac{4Q\frak{Q}}{r^4}\right)+\frac{2Q}{r^2}\kab\frac{\frak{Q}}{r^2}\\
&=& \left(\frac{6\frak{m}}{r^4}-\frac{16Q\frak{Q}}{r^5}\right)\nabb_3(r)+ \kab  \left(-\frac{3\frak{m}}{r^3}+\frac{8Q\frak{Q}}{r^4}\right)\\
&=& \left(\frac{6\frak{m}}{r^3}-\frac{16Q\frak{Q}}{r^4}\right)\frac 1 2 \kab + \kab \left(-\frac{3\frak{m}}{r^3}+\frac{8Q\frak{Q}}{r^4}\right)=0
\eeaa 
which proves the Proposition. 
\end{proof}

\subsubsection{Linearized Kerr-Newman solutions leaving the mass and the charge unchanged} \label{sec:nontrivkerr}

In addition to the variation of mass and change in the Reissner-Nordstr\"om solution, a variation in the angular momentum is also possible. We therefore have to take into account the perturbation into a Kerr-Newman solution for small $a$. 

We start from the Kerr-Newman metric expressed in outgoing Eddington-Finkelstein coordinates ignoring all terms quadratic or higher in $a$:
\beaa
g_{K-N}&=&- 2 dr du- \left(1-\frac{2M}{r}+\frac{Q^2}{r^2}\right) du^2  + r^2  \left(d\th^2+ \sin^2\th d\phi^2\right)\\
&&-\left(\frac{4Mr-2Q^2}{r^2} \right)a \sin^2\th du d\phi+2a\sin^2\th dr d\phi
\eeaa
Notice that these coordinates do not realize the Bondi gauge, because of the presence of the last term $(dr d\phi)$, which is not allowed in the Bondi form \eqref{double-null-metric}. 
Performing the change of coordinates $\phi'= \phi+\frac{a}{r}$, 
we obtain, ignoring all terms quadratic in $a$:
\beaa
g_{K-N}&=&- 2 dr du- \left(1-\frac{2M}{r}+\frac{Q^2}{r^2}\right) du^2  + r^2  \left(d\th^2+ \sin^2\th (d\phi')^2-2\sin^2\th\frac{a}{r^2}dr d\phi'\right)\\
&&-\left(\frac{4Mr-2Q^2}{r^2} \right)a \sin^2\th du d\phi'+2a\sin^2\th dr d\phi'\\
&=&- 2 dr du- \left(1-\frac{2M}{r}+\frac{Q^2}{r^2}\right) du^2  + r^2  \left(d\th^2+ \sin^2\th (d\phi')^2\right)-\left(\frac{4Mr-2Q^2}{r^2} \right)a \sin^2\th du d\phi'
\eeaa
which is the Kerr-Newman metric in Bondi gauge, and clearly of the form \eqref{metric-g-epsilon} with $a=\ep \mathfrak{a}$. Linearizing in $a$, we can read off the linearized metric coefficients, and notice that the only one non-vanishing is $\underline{b}$. We obtain the following solutions:

\begin{proposition} \label{lem:explicitkerr} 
Let $\dot{Y}^{\ell=1}_m$ for $m=-1,0,1$ denote the $\ell=1$ spherical harmonics on the unit sphere. 
For any $\mathfrak{a} \in \mathbb{R}$, the following is a smooth
solution of the system of gravitational and electromagnetic perturbations  on $\mathcal{M}$. 

The only non-vanishing metric coefficient is 
\beaa
 \underline{b}=\left(-\frac{8M}{r}+\frac{4Q^2}{r^2} \right) \mathfrak{a} \ \epsilon^{AB} \partial_B \dot{Y}^{\ell=1}_m 
\eeaa
The only non-vanishing Ricci coefficient is
\beaa
\ze=\eta&=& \left(\frac{-6M}{r^2}+\frac{4Q^2}{r^3} \right) \mathfrak{a} \ \epsilon^{AB} \partial_B \dot{Y}^{\ell=1}_m 
\eeaa
The non-vanishing electromagnetic components are 
\beaa
\bF&=& \frac{Q}{r}\ka\mathfrak{a} \ \epsilon^{AB} \partial_B \dot{Y}^{\ell=1}_m , \qquad \bbF= \frac{Q}{r}\kab\mathfrak{a} \ \epsilon^{AB} \partial_B \dot{Y}^{\ell=1}_m, \qquad \check{\sigmaF}= -\frac{4Q}{r^3}\mathfrak{a} \   \dot{Y}^{\ell=1}_m
\eeaa
The non-vanishing curvature components are
\beaa
\b &=&   \left(\frac{-3M}{r^2}+\frac{3Q^2}{r^3} \right) \ka\mathfrak{a} \ \epsilon^{AB} \partial_B \dot{Y}^{\ell=1}_m  \ \ \ , \ \ \  \bb= \left(\frac{-3M}{r^2}+\frac{3Q^2}{r^3} \right) \kab\mathfrak{a} \ \epsilon^{AB} \partial_B \dot{Y}^{\ell=1}_m  \\
 \check{\sigma} &=& \left(\frac{-12M}{r^4}+\frac{8Q^2}{r^5} \right) \mathfrak{a} \  \dot{Y}^{\ell=1}_m
\eeaa
\end{proposition}

Note that the above family may be parametrised by the $\ell=1$-modes of the electromagnetic component $\check{\sigmaF}$.

\begin{proof} We verify the equations in Section \ref{all-equations}.
Since $\DDs_2\underline{b},\divv \underline{b}=0$, the only non-trivial equation for the metric coefficients is \eqref{derivatives-b-ze-eta}, which is verified:
\beaa
\nabb_4 \underline{b}-\frac 1 2 \ka \underline{b}&=&\nabb_4 \Big(\left(\frac{-8M}{r}+\frac{4Q^2}{r^2} \right) \mathfrak{a}\epsilon^{AB} \partial_B \dot{Y}^{\ell=1}_m \Big)-\frac 1 2 \ka \left(\frac{-8M}{r}+\frac{4Q^2}{r^2} \right) \mathfrak{a}\epsilon^{AB} \partial_B \dot{Y}^{\ell=1}_m\\
&=&\left(\frac{8M}{r^2}-\frac{8Q^2}{r^3} \right) \mathfrak{a}\epsilon^{AB} \partial_B \dot{Y}^{\ell=1}_m- \ka \left(\frac{-8M}{r}+\frac{4Q^2}{r^2} \right) \mathfrak{a}\epsilon^{AB} \partial_B \dot{Y}^{\ell=1}_m\\
&=&\left(\frac{24M}{r^2}-\frac{16Q^2}{r^3} \right) \mathfrak{a}\epsilon^{AB} \partial_B \dot{Y}^{\ell=1}_m=-2(\eta+\ze)
 \eeaa
Since $\DDs_2\ze,\divv \ze,\DDs_2\eta,\divv\eta=0$, the only non-trivial null structure equations are \eqref{nabb-3-ze-ze-eta}-\eqref{nabb-4-ze-ze-eta}, the Codazzi equations  \eqref{Codazzi-chib-ze-eta}-\eqref{Codazzi-chi-ze-eta} and \eqref{sigma-curl-ze-eta}.
To verify \eqref{nabb-3-ze-ze-eta} and \eqref{nabb-4-ze-ze-eta}, recalling that $\ze=\eta$, we compute
\beaa
\nabb_3 \ze+ \kab \ze&=&\nabb_3 ( \left(-\frac{6M}{r^2}+\frac{4Q^2}{r^3} \right) \mathfrak{a}\epsilon^{AB} \partial_B \dot{Y}^{\ell=1}_m)+ \kab  \left(\frac{-6M}{r^2}+\frac{4Q^2}{r^3} \right) \mathfrak{a}\epsilon^{AB} \partial_B \dot{Y}^{\ell=1}_m\\
&=& \left(\frac{6M}{r^2}-\frac{6Q^2}{r^3} \right)\kab \mathfrak{a}\epsilon^{AB} \partial_B \dot{Y}^{\ell=1}_m+\frac 1 2 \kab  \left(-\frac{6M}{r^2}+\frac{4Q^2}{r^3} \right) \mathfrak{a}\epsilon^{AB} \partial_B \dot{Y}^{\ell=1}_m\\
&=& \left(\frac{3M}{r^2}-\frac{ 4Q^2}{r^3} \right)\kab \mathfrak{a}\epsilon^{AB} \partial_B \dot{Y}^{\ell=1}_m=- \bb  -\rhoF\bbF
\eeaa
To verify \eqref{sigma-curl-ze-eta}, we compute:
\beaa
\curl \ze&=& \epsilon^{BA} \partial_B \ze=  \left(-\frac{6M}{r^2}+\frac{4Q^2}{r^3} \right) \mathfrak{a} \epsilon^{BA} \partial_B\slashed{\epsilon}_{AC} \partial^C \dot{Y}^{\ell=1}_m= \left(\frac{6M}{r^4}-\frac{4Q^2}{r^5} \right) \mathfrak{a}  \lapp_{S^2}\dot{Y}^{\ell=1}_m \\
&=& \left(-\frac{12M}{r^4}+\frac{8Q^2}{r^5} \right) \mathfrak{a} \dot{Y}^{\ell=1}_m=\check{\sigma}
\eeaa

Since $\divv \bF,\divv\bbF=0$, the Maxwell equations \eqref{nabb-3-sigmaFlin}-\eqref{nabb-4-check-rhoF-ze-eta} are trivially satisfied. To verify \eqref{nabb-3-check-sigmaF-ze-eta}-\eqref{nabb-4-check-sigmaF-ze-eta}, we compute:
\beaa
\nabb_3 \check{\sigmaF}+\kab \ \check{\sigmaF}&=& \nabb_3 \left(-\frac{4Q}{r^3}\right)\mathfrak{a} \cdot \dot{Y}^{\ell=1}_m- \kab  \frac{4Q}{r^3}\mathfrak{a} \cdot \dot{Y}^{\ell=1}_m=  \frac{2Q}{r^3}\mathfrak{a} \kab\cdot \dot{Y}^{\ell=1}_m\\
 \slashed{\curl}\bbF&=&  -\frac{Q}{r^3}\kab\mathfrak{a}\lapp_{S^2} \dot{Y}^{\ell=1}_m= \frac{2Q}{r^3}\kab\mathfrak{a} \dot{Y}^{\ell=1}_m
\eeaa

Since $\DDs_2\b,\DDs_2\bb=0$ and $\divv\b,\divv\bb=0$, the Bianchi identities \eqref{nabb-3-a-ze-eta}-\eqref{nabb-4-aa-ze-eta} and \eqref{nabb-3-check-rho-ze-eta}-\eqref{nabb-4-check-rho-ze-eta} are trivially satisfied. To verify \eqref{nabb-3-check-sigma-ze-eta} and \eqref{nabb-4-check-sigma-ze-eta}, we compute 
\beaa
\nabb_3 \check{\sigma}+\frac 3 2 \kab \check{\sigma}&=& \nabb_3 \left(-\frac{12M}{r^4}+\frac{8Q^2}{r^5} \right) \mathfrak{a} \dot{Y}^{\ell=1}_m+ \frac 3 2 \kab \left(-\frac{12M}{r^4}+\frac{8Q^2}{r^5} \right) \mathfrak{a} \dot{Y}^{\ell=1}_m\\
&=& \left(\frac{6M}{r^4}-\frac{8Q^2}{r^5} \right)\kab \mathfrak{a} \dot{Y}^{\ell=1}_m\\
-\slashed{\curl}\bb - \rhoF \ \slashed{\curl}\bbF&=& \left(\frac{6M}{r^4}-\frac{6Q^2}{r^5} \right) \ka\mathfrak{a} \dot{Y}^{\ell=1}_m - \frac{2Q^2}{r^5}\kab\mathfrak{a} \dot{Y}^{\ell=1}_m
\eeaa
which proves the proposition.
 \end{proof}

 Observe that the above Proposition describes for each $\mathfrak{a}\in \mathbb{R}$ and for each $m=-1, 0, 1$ a solution to the linearized Einstein-Maxwell equations corresponding to a linearized Kerr-Newman solution. The reason why we have a $3$-dimensional family of solution varying the angular momentum is identical to the case of Schwarzschild. Indeed, linearizing the metric at a non-trivial ($a\neq 0$) member of the Kerr-Newman family creates non-trivial pure gauge solutions corresponding to a rotation of the axis. On the other hand, while linearizing the spherically symmetric Reissner-Nordstr\"om metric these rotations correspond to trivial pure gauge solutions. The $3$-dimensional family above corresponds then to the identification of the axis of symmetry (a unit vector in $\mathbb{R}^3$) and the angular momentum of the solution (the length of the vector).

We combine the 3-dimensional space of solutions of Proposition \ref{lem:linss} and the 3-dimensional space
of solutions of Proposition \ref{lem:explicitkerr} in the following definition: 
\begin{definition} \label{def:kerrl}
Let $\mathfrak{M}$, $\mathfrak{Q}$, $\mathfrak{b}$, $\mathfrak{a}_{-1}, \mathfrak{a}_0,\mathfrak{a}_1$ be six real parameters. We call the sum of the solution of Proposition~\ref{lem:linss} with parameters $\mathfrak{m}$, $\mathfrak{Q}$ and $\mathfrak{b}$ and the solution of 
Proposition~\ref{lem:explicitkerr} satisfying $\check{\sigmaF} = \sum \mathfrak{a}_{m} Y^{\ell=1}_m$ the
{\bf linearized Kerr-Newman solution} with parameters $\left(\mathfrak{M}, \mathfrak{Q}, \mathfrak{b},\mathfrak{a}_{-1}, \mathfrak{a}_0, \mathfrak{a}_1\right)$ and denote it by $\mathscr{K}_{(\mathfrak{M}, \mathfrak{Q}, \mathfrak{b},\mathfrak{a})}$ or simply $\mathscr{K}$. 
\end{definition}

\begin{remark}\label{tilde-b-0} Observe that the gauge-invariant quantities $\tilde{\b}$ and $\underline{\tilde{\b}}$ defined by \eqref{definition-tilde-b} have the additional remarkable property that they vanish for every linearised Kerr-Newman solution $\mathscr{K}_{(\mathfrak{m}, \mathfrak{Q}, \mathfrak{b}, \mathfrak{a})}$. 
Indeed, for every $\mathfrak{a}$,  
\beaa
\tilde{\b}&=& 2\rhoF \b-3\rho \bF\\
&=&2\frac{Q}{r^2} \left(\frac{-3M}{r^2}+\frac{3Q^2}{r^3} \right) \ka\mathfrak{a}\epsilon^{AB} \partial_B Y^{\ell=1}_m-3\left(-\frac{2M}{r^3}+\frac{2Q^2}{r^4} \right) \frac{Q}{r}\ka\mathfrak{a}\epsilon^{AB} \partial_B Y^{\ell=1}_m=0
\eeaa
\end{remark}

\section{Statement of the main theorem}\label{statement-section}

After the above preparation concerning a linear gravitational and electromagnetic perturbation around Reissner-Nordstr\"om spacetime and the special pure gauge and linearized Kerr-Newman solutions, we state here our main theorem.

\begin{theorem}\label{linear-stability}(Linear stability of Reissner-Nordstr\"om: case $|Q|\ll M$) Let $\mathscr{S}$ be a linear gravitational and electromagnetic perturbation around Reissner-Nordstr\"om spacetime $(\MM, g_{M, Q})$, with $|Q| \ll M$, arising from smooth asymptotically flat initial data as in Definition \ref{def-asympt-flat}. 

Then, on the exterior of $(\MM, g_{M, Q})$, $\mathscr{S}$ decays inverse polynomially to a linearized Kerr-Newman solution $\mathscr{K}^{id}$, which is explicitly determined from the initial data, after adding a pure gauge solution $\mathscr{G}$ which can itself be estimated by the size of the data. In particular, $$\mathscr{S}-\mathscr{G}-\mathscr{K}^{id}$$ verifies the following pointwise decay in $u$ and $r$ for $r \geq r_2$ for some $r_2> r_{\mathcal{H}}$:
\beaa
|\a|+|\b|&\leq& C \min\{r^{-3-\de} u^{-1/2+\de}, r^{-2-\de} u^{-1+\de}\} \\
|\check{\rho}|+| \check{\sigma}|+|\check{K}|&\leq& C\min\{ r^{-3} u^{-1/2+\de}, r^{-2} u^{-1+\de}\}\\
|\bF| &\leq& C\min\{ r^{-2-\de} u^{-1/2+\de}, r^{-1-\de} u^{-1+\de}\}\\
|\chih|+|\ze|+|\check{\rhoF}|+|\check{\sigmaF}|+|\check{\kab}|&\leq& C\min\{ r^{-2} u^{-1/2+\de}, r^{-1} u^{-1+\de}\}\\
|\hat{\slashed{g}}|+| \check{\tr_\gamma \slashed{g}}| &\leq& C \min\{ r^{-1} u^{-1/2+\de},  u^{-1+\de}\}\\
|\bb| &\leq &C  r^{-2} u^{-1+\de} \\
|\chibh|+|\bbF|+|\eta|+|\xib|+|\check{\omb}| &\leq& C  r^{-1} u^{-1+\de}\\
 |\check{\vsi}|+|\check{\Omegab}|+|\underline{b}|&\leq& C  u^{-1+\de} 
\eeaa
where the above components are defined in terms of the null frame $\mathscr{N}$, and $C$ depends on some norms of initial data. 
For $r \leq r_2$, the above norms for the components defined in terms of the regular null frame $\mathscr{N}_*$ on the horizon can all be bounded by $C v^{-1+\de}$. 

Moreover, the projection to the $\ell=0$ spherical harmonics of $\mathscr{S}-\mathscr{G}-\mathscr{K}^{id}$ vanishes (i.e. $\kalin=\kablin=\omblin=\rlin=\rhoFlin=\sigmaFlin=\Omegablin=0$). 
\end{theorem}

The decay of the norms in Theorem \ref{linear-stability} are expressed in terms of two different sets of coordinates. The norms in the spacetime region with unbounded $r$ are expressed in terms of the outgoing Eddington-Finkelstein coordinates in Reissner-Nordstr\"om $(u, r)$ as defined in Section \ref{Outgoing-coords}, which are smooth towards null infinity. The norms in the spacetime region close to the horizon (for $r\leq r_2$) are expressed in terms of the incoming Eddington-Finkelstein coordinates $(v, r)$ as defined in Section \ref{ingoing-coords}, which are smooth towards the horizon. 
 
 \begin{remark}Observe that the most delicate part of the proof of decay for the above theorem is in the region $r \geq r_2$, where optimal decay for all components is obtained. The decay in the region close to the horizon can be obtained from the decay of the quantities up to $r = r_2$ by integration toward the horizon, and since $r$ is bounded, all the quantities verify the same decay in $v$, as explained in Section \ref{section-close-horizon}. 
 \end{remark}

In the following sections, we introduce the main necessary ingredients for the proof of Theorem \ref{linear-stability}, and we defer the proof to Section \ref{decay-chapter}.

\section{The Teukolsky equations and the decay for the gauge-invariant quantities}\label{Teukolsky-eqs}

In this section, we introduce the \textit{generalized Teukolsky equations} of spin $\pm2$ and the \textit{generalized Teukolsky equation} of spin $\pm1$ which govern the gravitational and electromagnetic perturbations of Reissner-Nordstr\"om spacetime. We also introduce the \textit{generalized Regge-Wheeler system} and the \textit{generalized Fackerell-Ipser equation} and the connection between the two sets of equations, through the Chandrasekhar transformation. 
These equations have a fundamental relation to the problem of linear stability for gravitational and electromagnetic perturbations of Reissner-Nordstr\"om spacetime. Indeed, the gauge invariant quantities identified in the previous chapter verify the generalized Teukolsky equations. The Chandrasekhar transformation shall allow us to derive estimates for these quantities.

In Section \ref{spin-2-chapter} we define the generalized spin $\pm2$ Teukolsky equations and the generalized Regge-Wheeler system considering them as a second order hyperbolic PDEs for independent unknowns $\a$, $\aa$, $\ff$, $\ffb$ and $\qf$, $\qf^\F$. 
In Section \ref{spin-1-chapter} we define the generalized spin $\pm1$ Teukolsky equations and the generalized Fackerell-Ipser equation considering them as a second order hyperbolic PDEs for independent unknowns $\tilde{\b}$, $\tilde{\bb}$ and $\pf$.

In Section \ref{Chandra-transform} we recall a fundamental transformation mapping solutions of the generalized Teukolsky equations to solutions to the Regge-Wheeler/Fackerell-Ipser equations. This transformation plays an important role in deriving estimates for this equation. The Chandrasekhar transformation here defined generalizes the physical space definition given in \cite{DHR} to the case of Reissner-Nordstr\"om, and identifies one operator which is applied to all quantities involved.

In Section \ref{relation-to-linear-stability}, we explain the relation of the above PDEs with the full system of linear gravitational and electromagnetic perturbations of Reissner-Nordstr\"om. As guessed from the notation, the curvature and electromagnetic components $\a$, $\aa$, $\ff$, $\ffb$, $\tilde{\b}$, $\tilde{\bb}$ verify the Teukolsky equations respectively. 
Finally in Section \ref{decay-for-gauge-invariant-quantities} we state the main theorems in \cite{Giorgi4} and \cite{Giorgi5} which give control and quantitative decay statements of the gauge invariant quantities $\a$, $\aa$, $\ff$, $\ffb$, $\tilde{\b}$, $\tilde{\bb}$. We state here the control in $L^2$ and $L^\infty$ norms which are needed in the following for the proof of linear stability in Section \ref{decay-chapter}.  
The main results on this chapter have appeared in \cite{Giorgi4} and \cite{Giorgi5}.

\subsection{The spin $\pm2$ Teukolsky equations and the Regge-Wheeler system}\label{spin2}\label{spin-2-chapter}

In this section, we recall a generalization of the celebrated spin $\pm 2$ Teukolsky equations and the Regge-Wheeler equation, as introduced in \cite{Giorgi4}, and explain the connection between them. 

\subsubsection{Generalized spin $\pm2$ Teukolsky system}

The generalized spin $\pm2$ Teukolsky system concern symmetric traceless $2$-tensors in Reissner-Nordstr{\"o}m spacetime, which we denote $(\a, \ff)$ and $(\aa, \underline{\ff})$ respectively.

\begin{definition} Let $\a$ and $\ff$ be two  symmetric traceless 2-tensors on $S$ defined
on a subset $\mathcal{D}\subset \mathcal{M}$. We say that $(\alpha, \ff)$ satisfy the {\bf generalized Teukolsky system of spin ${\bf +2}$} if they satisfy the following coupled system of PDEs:
\beaa
\Box_\g \a&=& -4\omb \nabb_4\a+2\left( \ka+2\om\right) \nabb_3\a+\left( \frac 1 2 \ka\kab-4\rho +4\rhoF^2+2\om\,\kab-10\omb\ka-8\om\omb -4\nabb_4\omb  \right)\,\a\\
&& +4\rhoF\left(\nabb_4\ff+ \left( \ka  +2\om \right) \ff\right), \\
\Box_\g (r\ff)&=& -2\omb \nabb_4(r\ff)+\left( \ka+2\om\right) \nabb_3(r\ff)+\left(-\frac 1 2 \ka\kab-3\rho+\om\kab-3\omb\ka-2\nabb_4\omb\right) r\ff \\
&&-r\rhoF \left(\nabb_3\a+\left(\kab-4\omb\right)\a\right)
\eeaa
where $\Box_\g=\g^{\mu\nu} \D_{\mu}\D_{\nu}$ denotes the wave operator in Reissner-Nordstr{\"o}m spacetime, and the quantities appearing on the right hand side as coefficients are the underlying Reissner-Nordstr\"om values.

Let $\aa$ and $\underline{\ff}$ be two  symmetric traceless 2-tensor on $S$ defined
on a subset $\mathcal{D}\subset \mathcal{M}$. We say that $(\aa, \underline{\ff})$ satisfy the {\bf generalized Teukolsky system of spin ${\bf -2}$} if they satisfy the following coupled system of PDEs:
\beaa
\Box_\g \aa&=& -4\om \nabb_3\aa+2\left( \kab+2\omb\right) \nabb_4\aa+\left( \frac 1 2 \ka\kab-4\rho +4\rhoF^2+2\omb\,\ka-10\om\kab-8\om\omb -4\nabb_3\om  \right)\,\aa\\
&& -4\rhoF\left(\nabb_3\underline{\ff}+ \left( \kab  +2\omb \right) \underline{\ff}\right), \\
\Box_\g (r\underline{\ff})&=& -2\om \nabb_3(r\underline{\ff})+\left( \kab+2\omb\right) \nabb_4(r\underline{\ff})+\left(-\frac 1 2 \ka\kab-3\rho+\omb\ka-3\om\kab-2\nabb_3\om\right) r\underline{\ff} \\
&&+r\rhoF \left(\nabb_4\aa+\left(\ka-4\om\right)\aa\right)
\eeaa
\end{definition}
We note that the generalized Teukolsky system of spin $-2$ is obtained from that of spin $+2$ by interchanging $\nabb_3$ with $\nabb_4$ and underline quantities with non-underlined ones.

\subsubsection{Generalized Regge-Wheeler system}

The other generalized system to be defined here is the generalized Regge-Wheeler system, to be satisfied again by symmetric traceless tensors $(\qf, \qf^\F)$.

\begin{definition} \label{def:rwe}
Let $\qf$ and $\qf^\F$ be two symmetric traceless $2$-tensors on $S$.
We say that $(\qf, \qf^\F)$ satisfies the {\bf generalized Regge--Wheeler system for spin $+2$} if they satisfy the following coupled system of PDEs:
\bea\label{finalsystem}
\begin{split}
\Box_\g\qf+\left( \ka\kab-10\rhoF^2\right)\ \qf&= \rhoF\Bigg(4r\lapp_2\qf^{\F}-4r\kab \nabb_4(\qf^{\F})-4r\ka\nabb_3( \qf^\F) + r\left(6\ka\kab+16\rho+8\rhoF^2\right)\qf^{\F}\Bigg)\\
&+\rhoF(l.o.t.)_1, \\
\Box_\g\qf^{\F}+\left( \ka\kab+3\rho\right)\ \qf^{\F}&=\rhoF \Bigg(-\frac {1}{ r} \qf\Bigg) +\rhoF^2 (l.o.t.)_2
\end{split}
\eea
where $(l.o.t.)_1$ and $(l.o.t.)_2$ are lower order terms with respect to $\qf$ and $\qf^\F$. More precisely, $\lot=c_1(r)\a+c_2(r) \nabb_3 \a+c_3(r) \ff$ for smooth functions $c_i$. 

Let $\underline{\qf}$ and $\underline{\qf}^\F$ be two symmetric traceless $2$-tensors on $S$.
We say that $(\underline{\qf}, \underline{\qf}^\F)$ satisfies the {\bf generalized Regge--Wheeler system for spin $-2$} if they satisfy the following coupled system of PDEs:
\bea\label{finalsystem-2}
\begin{split}
\Box_\g\underline{\qf}+\left( \ka\kab-10\rhoF^2\right)\ \underline{\qf}&= -\rhoF\Bigg(4r\lapp_2\underline{\qf}^{\F}-4r\kab \nabb_4(\underline{\qf}^{\F})-4r\ka\nabb_3( \underline{\qf}^\F) + r\left(6\ka\kab+16\rho+8\rhoF^2\right)\underline{\qf}^{\F}\Bigg)\\
&-\rhoF(l.o.t.)_1, \\
\Box_\g\underline{\qf}^{\F}+\left( \ka\kab+3\rho\right)\ \underline{\qf}^{\F}&=\rhoF \Bigg(\frac {1}{ r} \underline{\qf}\Bigg) +\rhoF^2 (l.o.t.)_2
\end{split}
\eea
where $(l.o.t.)_1$ and $(l.o.t.)_2$ are lower order terms with respect to $\underline{\qf}$ and $\underline{\qf}^\F$. More precisely, $\lot=c_1(r)\aa+c_2(r) \nabb_3 \aa+c_3(r) \ffb$ for smooth functions $c_i$. 
\end{definition}

In Section \ref{transformation-theory}, we will recall that given a solution $(\a, \ff)$ and $(\aa, \underline{\ff})$ of the spin $\pm 2$ Teukolsky equations, respectively, we can derive two solutions $(\qf,\qf^\F)$ and $(\underline{\qf},\underline{\qf}^\F)$, respectively, of the generalized Regge-Wheeler system. 
Standard well-posedness results hold for both the generalized Teukolsky system of spin $\pm2$ and the generalized Regge-Wheeler system.  

\subsection{The spin $\pm1$ Teukolsky equation and the Fackerell-Ipser equation}\label{spin-1-chapter}\label{spin1}

In this section, we recall a generalization of the celebrated spin $\pm 1$ Teukolsky equations and the Fackerell-Ipser equation, as introduced in \cite{Giorgi5}, and explain the connection between them.

\subsubsection{Generalized spin $\pm 1$ Teukolsky equation}

The generalized spin $\pm1$ Teukolsky equation concerns $1$-tensors which we denote $\tilde{\b}$ and $\tilde{\bb}$ respectively.

\begin{definition}\label{teukolsky-1} Let $\tilde{\b}$ be a 1-tensor defined
on a subset $\mathcal{D}\subset \mathcal{M}$. We say that $\tilde{\b}$ satisfy the {\bf generalized Teukolsky equation of spin ${\bf +1}$} if it satisfies the following PDE:
\beaa
\Box_\g (r^3 \tilde{\b})&=& -2\omb \nabb_4(r^3\tilde{\b})+\left(\ka+2\om\right) \nabb_3(r^3\tilde{\b})+\left(\frac {1}{4} \ka \kab-3\omb \ka+\om\kab-2\rho+3\rhoF^2-8\om\omb+2\nabb_3 \om\right) r^3 \tilde{\b}\\
&&-2r^3\kab\rhoF^2\left(\nabb_4 \bF+ (\frac 3 2 \ka+2\om)\bF-2\rhoF \xi \right)+\mathcal{I}
\eeaa
where $\Box_\g=\g^{\mu\nu} \D_{\mu}\D_{\nu}$ denotes the wave operator in Reissner-Nordstr{\"o}m spacetime, and $\mathcal{I}$ is a $1$-tensor with vanishing projection to the $\ell=1$ spherical harmonics, i.e. $\divv \ \mathcal{I}_{\ell=1}=\curll \ \mathcal{I}_{\ell=1}=0$.

Let $\tilde{\bb}$ be a 1-tensor defined
on a subset $\mathcal{D}\subset \mathcal{M}$. We say that $\tilde{\b}$ satisfy the {\bf generalized Teukolsky equation of spin ${\bf -1}$} if it satisfies the following PDE:
\beaa
\Box_\g (r^3 \tilde{\bb})&=& -2\om \nabb_3(r^3\tilde{\bb})+\left(\kab+2\omb\right) \nabb_4(r^3\tilde{\bb})+\left(\frac {1}{4} \ka \kab-3\om \kab+\omb\ka-2\rho+3\rhoF^2-8\om\omb+2\nabb_4 \omb\right) r^3 \tilde{\bb}\\
&&-2r^3\ka\rhoF^2\left(\nabb_3 \bbF+ (\frac 3 2 \kab+2\omb)\bbF+2\rhoF \xib \right)+\underline{\mathcal{I}}
\eeaa
where $\underline{\mathcal{I}}$ is a $1$-tensor with vanishing projection to the $\ell=1$ spherical harmonics.
\end{definition}

\subsubsection{Generalized Fackerell-Ipser equation in $\ell=1$ mode}
The other generalized equation in $l=1$ to be defined here is the generalized Fackerell-Ipser equation, to be satisfied by a one tensor $\pf$.

\begin{definition} \label{def:rwe-1}
Let $\pf$ be a $1$-tensor on $\mathcal{D}\subset \mathcal{M}$.
We say that $\pf$ satisfies the {\bf generalized Fackerell-Ipser equation in $\ell=1$} if it satisfies the following PDE:
\bea\label{finalequationl1}
\Box_\g\pf+\left(\frac 1 4 \ka\kab-5\rhoF^2\right)\pf &=&\mathcal{J}
\eea
where $\mathcal{J}$ is a $1$-tensor with vanishing projection to the $\ell=1$ spherical harmonics, i.e. $\divv \ \mathcal{J}_{\ell=1}= \curll \ \mathcal{J}_{\ell=1}=0$.
\end{definition}

In Section \ref{transformation-theory-1}, we recall that given a solution $\tilde{\b}$ and $\tilde{\bb}$ of the generalized spin $\pm 1$ Teukolsky equations in $\ell=1$, respectively, we can derive two solutions $\pf$ and $\underline{\pf}$, respectively, of the generalized Fackerell-Ipser equation in $\ell=1$. 
Standard well-posedness results hold for both the generalized Teukolsky equation of spin $\pm1$ and the generalized Fackerell-Ipser equation in $\ell=1$.

\subsection{The Chandrasekhar transformation}\label{Chandra-transform}\label{transformation-theory}\label{transformation-theory-1}

We now recall the transformation theory relating solutions of the generalized Teukolsky equations defined above to solutions of the generalized Regge-Wheeler system or the Fackerell-Ipser equation. We emphasize that a physical space version of the Chandrasekhar transformation was first introduced in \cite{DHR}, for the Schwarzschild spacetime.

We introduce the following operators for a $n$-rank $S$-tensor $\Psi$:
\bea \label{operators}
\underline{P}(\Psi)&=&\frac{1}{\kab} \nabb_3(r \Psi), \qquad  P(\Psi)=\frac{1}{\ka} \nabb_4(r \Psi) 
\eea 

Given a solution  $(\alpha, \ff)$ of the generalized Teukolsky system of spin $+2$ and a solution  $\tilde{\b}$ of the generalized Teukolsky equation of spin $+1$, we defined the following \emph{derived} quantities for $(\a, \ff)$ and $\tilde{\b}$ in \cite{Giorgi4} and \cite{Giorgi5}:
\bea\label{quantities}\label{quantities-1}
\begin{split}
\psi_0 &= r^2 \kab^2 \a, \\
\psi_1&=\underline{P}(\psi_0), \\
\psi_2&=\underline{P}(\psi_1)=\underline{P}(\underline{P}(\psi_0))=:\qf, \\
\psi_3&= r^2 \kab \ \ff, \\
\psi_4&= \underline{P}(\psi_3)=:\qf^\F \\
\psi_5 &= r^4 \kab \ \tilde{\b}, \\
\psi_6&=\underline{P}(\psi_5):=\pf
\end{split}
\eea
Similarly, given a solution  $(\aa, \underline{\ff})$ of the generalized Teukolsky system of spin $-2$ and given a solution  $\underline{\tilde{\b}}$ of the generalized Teukolsky equation of spin $-1$, we defined the following \emph{derived} quantities for $(\aa, \underline{\ff})$ and $\tilde{\bb}$ in \cite{Giorgi4} and \cite{Giorgi5}:
\bea\label{quantities-2}
\begin{split}
\underline{\psi}_0 &= r^2 \ka^2 \aa, \\
\underline{\psi}_1&=P(\underline{\psi}_0), \\
\underline{\psi}_2&=P(\underline{\psi}_1)=P(P(\underline{\psi}_0))=:\underline{\qf}, \\
\underline{\psi}_3&= r^2 \ka \ \underline{\ff}, \\
\underline{\psi}_4&= P(\underline{\psi}_3)=:\underline{\qf}^\F\\
\underline{\psi}_5 &= r^4 \ka \ \tilde{\bb}, \\
\underline{\psi}_6&=P(\underline{\psi}_5):=\underline{\pf}
\end{split}
\eea

The following  proposition is proved in \cite{Giorgi4} and \cite{Giorgi5}. 
\begin{proposition} \label{prop:rwt1}
Let $(\alpha, \ff)$ be a solution of the generalized Teukolsky system of spin $+2$. Then the symmetric traceless tensors $(\qf, \qf^\F)$ as defined through \eqref{quantities} satisfy the generalized Regge-Wheeler system of spin $+2$. 
Similarly, let $(\aa, \underline{\ff})$ be a solution of the generalized Teukolsky system of spin $-2$. Then the symmetric traceless tensors $(\underline{\qf}, \underline{\qf}^\F)$ as defined through \eqref{quantities-2} satisfy the generalized Regge-Wheeler system of spin $-2$. 

Let $\tilde{\b}$ be a solution of the generalized Teukolsky equation of spin $+1$. Then the $1$-tensor $\pf$ as defined through \eqref{quantities-1} satisfies the generalized Fackerell-Ipser equation in $\ell=1$. 
Similarly, let $\tilde{\bb}$ be a solution of the generalized Teukolsky equation of spin $-1$. Then the $1$-tensor $\underline{\pf}$ as defined through \eqref{quantities-2} satisfies the generalized Fackerell-Ipser equation in $\ell=1$. 
\end{proposition}

\subsection{Relation to the gravitational and electromagnetic perturbations of Reissner-Nordstr\"om spacetime}\label{relation-to-linear-stability}

We finally relate the equations presented above to the full system of linearized
gravitational and electromagnetic perturbations of Reissner-Nordstr{\"o}m spacetime in the context of linear stability of Reissner-Nordstr{\"o}m.
The relation is summarized in the following Theorem, proved in \cite{Giorgi4} and \cite{Giorgi5}.

\begin{theorem} \label{prop:relfull}
Let $\a$, $\aa$, $\ff$, $\underline{\ff}$, $\tilde{\b}$, $\tilde{\bb}$ be the curvature components of a linear gravitational and electromagnetic perturbation around Reissner-Nordstr{\"o}m spacetime as in Section \ref{section-gauge-inv}. 

Then
\begin{itemize}
\item $(\a, \ff)$ satisfy the generalized Teukolsky system of spin $+2$, and $(\aa, \underline{\ff})$ satisfy the generalized Teukolsky system of spin $-2$. 
\item $\tilde{\b}$ satisfies the generalized Teukolsky equation of spin $+1$, and $\tilde{\bb}$ satisfies the generalized Teukolsky equation of spin $-1$. 
\end{itemize}
\end{theorem}

Using Proposition \ref{prop:rwt1}, we can therefore associate to any solution to the  linearized Einstein-Maxwell equations around Reissner-Nordstr{\"o}m spacetime, two symmetric traceless $2$-tensors which verify the generalized Regge-Wheeler system of spin $\pm 2$ and a one form which verifies the generalized Fackerell-Ipser equation in $\ell=1$.

\subsubsection{Gravitational versus electromagnetic radiation}
The linear perturbations considered in Definition \ref{linear-perturbation-definition} allow the perturbation of the Weyl curvature of the spacetime, as well as the perturbation of the Ricci curvature, in the form of the electromagnetic tensor.  We will refer to the perturbation of the Weyl tensor $\W$ as \emph{gravitational radiation}, and to the perturbation of the electromagnetic tensor $\F$ as \emph{electromagnetic radiation}. 

The radiation is always transported by the gauge-invariant versions of the extreme component of the tensor, i.e. the tensor $\aa$ represents the gravitational radiation, and the tensor $\tilde{\bb}$ represents the electromagnetic radiation. 

The combined gravitational and electromagnetic perturbations of Reissner-Nordstr{\"o}m spacetime is not solved by simply considering the separated gravitational and electromagnetic perturbations. In fact, as already clear from the Bianchi identites and the Maxwell equations above, the interaction between these two perturbations is complex, and the two radiations are very much coupled together. 

The gravitational radiation is still transported by $\aa$, but this gauge-independent quantity verifies a generalized Teukolsky equation which is coupled to another gauge independent quantity, $\underline{\ff}$, which has contributions from the electromagnetic tensor. Being both symmetric traceless $2$-tensors, they are in general supported in $\ell\geq 2$ spherical harmonics. Therefore, these two tensors, coming from both the Weyl and the Ricci perturbations, are responsible for the gravitational radiation in $\ell\geq 2$. 

The quantity $\bbF$ is not gauge-invariant in the presence of the Weyl perturbation. It is $\tilde{\bb}$, defined using both the curvature and the electromagnetic tensor, to be gauge invariant and to transport the electromagnetic perturbation. Since it is a $1$-form, it is in general supported in $\ell\geq 1$ spherical harmonics, which is where the electromagnetic radiation is supported.

In summary, the gravitational and electromagnetic perturbations of Reissner-Nordstr\"om are, by effect of the above equations, totally coupled together.

\subsection{Boundedness and decay for the gauge-invariant quantities}\label{decay-for-gauge-invariant-quantities}
The main results in \cite{Giorgi4} and \cite{Giorgi5} are the boundedness and quantitative decay statements obtained for the gauge-invariant quantities $\a$, $\aa$, $\ff$, $\underline{\ff}$, $\tilde{\b}$, $\tilde{\bb}$ verifying the above Teukolsky equations. Using Theorem \ref{prop:relfull}, we can summarize the results in the following. 

We denote $A\les B$ if $A \leq C B$ where $C$ is an universal constant depending on appropriate Sobolev norms of initial data. We define the following norms:
\beaa
|| f||_{\infty}(u, r):= ||f||_{L^{\infty}(S_{u, r})} \qquad || f||_{\infty, k}(u, r):= \sum_{i=0}^k ||\dk^i f||_{L^{\infty}(S_{u, r})}
\eeaa
where $\dk=\{ \nabb_3, r\nabb_4, r\nabb\}$. 

\begin{theorem}\label{estimates-theorem}[Main Theorem in \cite{Giorgi4} and Main Theorem in \cite{Giorgi5}] Let $\a$, $\aa$, $\ff$, $\underline{\ff}$, $\tilde{\b}$, $\tilde{\bb}$ be the curvature components of a linear gravitational and electromagnetic perturbation around Reissner-Nordstr{\"o}m spacetime for $|Q| \ll M$. Then for every $k$ we have:
\begin{enumerate}
\item The following energy estimates hold true:
\bea\label{L^2-estimates-theorem}
\begin{split}
\int_{\Si_{r \geq r_2}(\tau)} r^{2+\de} |\a|^2+r^{4+\de}|\nabb_4\a|^2+r^{4+\de} |\nabb\a|^2+r^{2+\de} |\nabb_3\a|^2 &\leq \frac{(\text{initial data for $\a$, $\ff$, $\psi_1$, $\qf$, $\qf^\F$})}{u^{2-2\de}} \\
\int_{\Si_{r \geq r_2}(\tau)} r^{2+\de} |\ff|^2+r^{4+\de}|\nabb_4\ff|^2+r^{4+\de} |\nabb\ff|^2+r^{2+\de} |\nabb_3\ff|^2 &\leq \frac{(\text{initial data for $\a$, $\ff$, $\psi_1$, $\qf$, $\qf^\F$})}{u^{2-2\de}}\\
\int_{\Si_{r \geq r_2}(\tau)} r^{6+\de} |\tilde{\b}|^2+r^{8+\de}|\nabb_4\tilde{\b}|^2+r^{8+\de} |\nabb\tilde{\b}|^2+r^{6+\de} |\nabb_3\tilde{\b}|^2 &\leq \frac{(\text{initial data for $\a$, $\ff$, $\psi_1$, $\tilde{\b}$, $\qf$, $\qf^\F$, $\pf$})}{u^{2-2\de}}
\end{split}
\eea
and 
\bea\label{L^2-estimates-theorem-underline}
\begin{split}
\int_{\Si_{r \geq r_2}(\tau)}  |\aa|^2+r^{2}|\nabb_4\aa|^2+r^{2} |\nabb\aa|^2+ |\nabb_3\aa|^2 &\leq \frac{(\text{initial data for $\aa$, $\underline{\ff}$, $\underline{\psi}_1$, $\underline{\qf}$, $\underline{\qf}^\F$})}{u^{2-2\de}} \\
\int_{\Si_{r \geq r_2}(\tau)} r^{2} |\underline{\ff}|^2+r^{4}|\nabb_4\underline{\ff}|^2+r^{4} |\nabb\underline{\ff}|^2+r^{2} |\nabb_3\underline{\ff}|^2 &\leq  \frac{(\text{initial data for $\aa$, $\underline{\ff}$, $\underline{\psi}_1$, $\underline{\qf}$, $\underline{\qf}^\F$})}{u^{2-2\de}}\\
\int_{\Si_{r \geq r_2}(\tau)} r^{6} |\tilde{\bb}|^2+r^{8}|\nabb_4\tilde{\bb}|^2+r^{8} |\nabb\tilde{\bb}|^2+r^{6} |\nabb_3\tilde{\bb}|^2 &\leq  \frac{(\text{initial data for $\aa$, $\underline{\ff}$, $\underline{\psi}_1$, $\tilde{\bb}$,  $\underline{\qf}$, $\underline{\qf}^\F$, $\underline{\pf}$})}{u^{2-2\de}}
\end{split}
\eea
where $\Sigma_\tau$ are spacelike hypersurfaces entering the horizon and null infinity, and $\Si_{r \geq r_2}(\tau)$ denotes its intersection with $\{ r \geq r_2\}$ for some $r_2>r_{\mathcal{H}}$. 
\item The following pointwise estimates for $\qf$, $\qf^\F$ and $\pf$ hold true, for $r \geq r_2$:
\bea
|| \qf||_{\infty, k} &\les& \min\{ r^{-1} u^{-1/2+\de},  u^{-1+\de}\} \label{estimate-qf}\\
|| \qf^\F||_{\infty, k}&\les& \min\{ r^{-1} u^{-1/2+\de},  u^{-1+\de}\} \label{estimate-qf-F} \\
|| \pf||_{\infty, k}&\les& \min\{ r^{-1} u^{-1/2+\de},  u^{-1+\de}\} \label{estimate-pf}
\eea
\item The following pointwise estimates for $\a$, $\ff$ and $\tilde{\b}$ hold true, for $r \geq r_2$:
\bea
||\a||_{\infty, k}\les \min\{r^{-3-\de}u^{-1/2+\de}, r^{-2-\de} u^{-1+\de} \} \label{estimate-a} \label{estimate-derivatives-a}\\
||\ff||_{\infty, k}\les \min\{r^{-3-\de}u^{-1/2+\de}, r^{-2-\de} u^{-1+\de} \} \label{estimate-ff}\label{estimate-derivatives-ff}\\
||\tilde{\b}||_{\infty, k}\les \min\{r^{-5-\de}u^{-1/2+\de}, r^{-4-\de} u^{-1+\de} \} \label{estimate-tilde-b}\label{estimate-derivatives-tilde-b}
\eea
\item The following pointwise estimates for $\aa$, $\ffb$ and $\tilde{\bb}$ hold true, for $r \geq r_2$:
\bea
||\aa||_{\infty, k}\les  r^{-1} u^{-1+\de}  \label{estimate-aa}\\
||\ffb||_{\infty, k}\les  r^{-2} u^{-1+\de}  \label{estimate-underline-ff}\\
||\tilde{\bb}||_{\infty, k}\les  r^{-4} u^{-1+\de}  \label{estimate-tilde-bb}
\eea
\item The above quantities are all bounded by $v^{-1+\de}$ in the region $r_{\mathcal{H}} \leq r \leq r_2$. 
\end{enumerate}

\end{theorem}

The proof of the theorem is based on the vectorfield method applied to the generalized Regge-Wheeler system and the generalized Fackerell-Ipser equation. See \cite{Giorgi4} and \cite{Giorgi5}.

\section{Initial data and well-posedness}\label{initial-data-well-posedness-chapter}
In this section, we consider the well-posedness of the system of linearized gravitational and electromagnetic perturbations. 
We first describe how to prescribe initial data in Section \ref{initial-data-seed} and we define what it means for data to be asymptotically flat in Section \ref{asymptotically-flat-initial-data}. Finally we formulate the well-posedness theorem in Section \ref{section-well-posedness}. 

\subsection{Seed data on an initial cone}\label{initial-data-seed}
We describe here how to prescribe initial data for the linearized Einstein-Maxwell equations of Section \ref{all-equations}. 

We present a characteristic initial value problem. We fix a sphere $S_0:=S_{u_0, r_0}$ in $\mathcal{M}$, obtained as intersection of two hypersurfaces for some values $\{r=r_0\}$ and $\{u=u_0\}$. Consider the outgoing Reissner-Nordstr{\"o}m light cone $C_0:=C_{u_0}$ originating from $S_0$, and the ingoing Reissner-Nordstr{\"o}m light cone $\underline{C}_0$ originating from $S_0$ on which the data are being prescribed. Initial data are prescribed by so-called \textit{seed data} that can be prescribed freely.

\begin{definition}\label{initial-data-set} Given a sphere $S_0$  with corresponding null cones $C_{0}$ and $\underline{C}_0$, a smooth seed initial data set consists of prescribing 
\begin{itemize}
\item along $C_0$: a smooth symmetric traceless $2$-tensor $\hat{\slashed{g}}_{0, out}$ and a smooth $1$-form $\bF_0$,
\item along $\underline{C}_0$: a smooth symmetric traceless 2-tensor $\hat{\slashed{g}}_{0, in}$, which coincides with $\hat{\slashed{g}}_{0, in}$ on $S_0$, 
\item along $\underline{C}_0$: smooth $1$-forms $\underline{b}_0$, $\xib_0$, and $\bbF_0$ 
\item along $\underline{C}_0$: smooth functions $\check{\Omegab}_0$, $\check{\omb}_0$, $\vsilin_0$, $\Omegablin_0$, $\omblin_0$
\item on the sphere $S_0$: a smooth 1-form $\ze_0$, 
\item on the sphere $S_0$: smooth functions $\check{\tr_{\gamma} \slashed{g}}_0$, $\check{\ka}_0$, $\check{\kab}_0$, $\check{\rhoF}_0$, $\check{\sigmaF}_0$, $\kalin_0$, $\rhoFlin_0$, $\sigmaFlin_0$.
\end{itemize}
\end{definition}

We will show in Theorem \ref{well-posedness-theorem} that the above freely prescribed tensors uniquely determine a solution to the linear gravitational and electromagnetic perturbation of Reissner-Nordstr\"om. 
\subsection{Asymptotic flatness of initial data}\label{asymptotically-flat-initial-data}

We first define the following derived quantities along $C_0$ from a smooth seed initial data as in Definition \ref{initial-data-set}:
\beaa
\chih_{0, out}&=& \frac 12 \nabb_4\hat{\slashed{g}}_{0, out}, \qquad \a_{0, out}= -\frac 12r^{-2}\nabb_4(r^2 \nabb_4\hat{\slashed{g}}_{0, out})
\eeaa
Note that these quantities are uniquely determined in terms of the seed data. 

For a tensor $\xi$ we define for $n_1\geq 0$, $n_2 \geq 0$:
\beaa
\mathcal{D}_{n_1, n_2} \xi=(r\nabb)^{n_1} (r\nabb_4)^{n_2} \xi 
\eeaa

We define the following notion of asymptotic flatness of initial data.

\begin{definition}\label{def-asympt-flat} We call a seed data set asymptotically flat with weight $s$ to order $n$ if the seed data satisfies the following estimates along $C_{0}$ for some $0<s\leq 1$ and any $n_1\geq 0$, $n_2 \geq 0$ with $n_1+n_2\leq n$:
\bea
|\mathcal{D}_{n_1, n_2}(r^2 \chih_{0, out})|+|\mathcal{D}_{n_1, n_2}(r^{3+s}\a_{0, out})|+|\mathcal{D}_{n_1, n_2}(r^{2+s}\bF_{0})| \leq C_{0, n_1, n_2}
\eea
for some constant $C_{0, n_1, n_2}$ depending on $n_1$ and $n_2$. 
\end{definition}

We will show in Theorem \ref{well-posedness-theorem} that asymptotically flat seed data lead in particular to a hierarchy of decay for all quantities on the initial data.

\subsection{The well-posedness theorem}\label{section-well-posedness}

We can now state the fundamental well-posedness theorem for linear gravitational and electromagnetic perturbations of Reissner-Nordstr{\"o}m.

\begin{theorem}\label{well-posedness-theorem} Fix a sphere $S_0$ and consider a smooth seed initial data set as in Definition \ref{initial-data-set}. Then there exists a unique smooth solution $\mathscr{S}$ of linear gravitational and electromagnetic perturbations around Reissner-Nordstr{\"o}m spacetime defined in $\mathcal{M} \cap I^{+}(S_0)$ which agrees with the seed data on $C_{0}$ and $\underline{C}_0$.

Moreover, suppose the smooth seed initial data set is asymptotically flat with weight $s>0$ to order $n$. Then on the initial cone $C_0$, the following estimates hold:
\beaa
|\mathcal{D}_{k}(r^{3+s}\a) |+|\mathcal{D}_{k}(r^{3+s}\b)|+|\mathcal{D}_{k}(r^{3}\check{\rho})|+|\mathcal{D}_{k}(r^{3}\check{\sigma})|+|\mathcal{D}_{k}(r^{2}\bb)| +|\mathcal{D}_{k}(r\aa) |&\leq& C\\
|\mathcal{D}_{k}(r^{2+s}\bF)|+|\mathcal{D}_{k}(r^{2}\check{\rhoF})|+|\mathcal{D}_{k}(r^{2}\check{\sigmaF})| +|\mathcal{D}_{k}(r\bbF)|&\leq& C\\
|\mathcal{D}_{k}(r^2 \chih)|+|\mathcal{D}_{k}(r \chibh)| +|\mathcal{D}_{k}(r^2 \ze)|+|\mathcal{D}_{k}(r \eta)|+|\mathcal{D}_{k}(r \xib)|+|\mathcal{D}_{k}(r^{2}\check{\ka})|+|\mathcal{D}_{k}(r\check{\kab})| &\leq& C
\eeaa
for any $k \leq n-3$ and a constant which can be computed explicitly from initial data. 
\end{theorem}

\section{Gauge-normalized solutions and identification of the Kerr-Newman parameters}\label{gauge-normalized-solutions-chapter}

In this section we define the gauge normalizations that will play a fundamental role in the proof of linear stability. We also identify the correct Kerr-Newman parameters of a solution to the linearized Einstein-Maxwell equations. 

We first define in Section \ref{definition-initial-normalization-section} what it means for a solution $\mathscr{S}$ to be initial data normalized. Such a solution will be used to prove that any solution is bounded, upon a choice of a gauge solution which can be expressed in terms of initial data. Moreover, this choice of gauge implies decay for most components of the solution, but it leads to an incomplete result in terms of decay for some of the components. To overcome this difficulty, in Section \ref{definition-normalization-section} we define what it means for a solution $\mathscr{S}$ to be $S_{U, R}$-normalized. Such normalization takes place at a sphere ``far-away'' in the spacetime, and is reminiscent of the choice of gauge at the last slice in \cite{stabilitySchwarzschild}. This new normalization will be used in the next chapter to obtain the complete optimal decay for all the components of a solution $\mathscr{S}$. 

We then show in Section \ref{achieving-initial-section} and Section \ref{achieving-normalization-section} that given a solution $\mathscr{S}$ of the linearized Einstein-Maxwell equations we can indeed associate to it an initial data normalized solution, and if $\mathscr{S}$ is bounded we can associate a $S_{U, R}$-normalized solution. They are respectively denoted $\mathscr{S}^{id}$ and $\mathscr{S}^{U, R}$, and are realized by adding to $\mathscr{S}$ a pure gauge solution $\mathscr{G}$ of the form described in Section \ref{sec:ssgauge}. 

The pure gauge solution used to obtain the initial data normalization is explicitly computable from initial data. On the other hand, the one used to obtain the $S_{U, R}$-normalization is not. Only in the proof of the decay in the next chapter, in Section \ref{decay-pure-gauge}, we will show that the pure gauge solution used to achieve the $S_{U, R}$-normalization is itself bounded by initial data. 

Finally, in Section \ref{Kerr-Newman-parameters-section} we identify the Kerr-Newman parameters out of the projection of the solution to the $\ell=0,1$ modes of the initial data normalized solution. Those parameters are explicitly computable from the initial data.

\subsection{The initial data normalization}\label{definition-initial-normalization-section}
In this section, we define the notion of initial data normalized solution. As we will show in Theorem \ref{achieving-initial-normalization-theorem}, given a seed initial data set and its associated solution $\mathscr{S}$, we can find a pure gauge solution $\mathscr{G}^{id}$ such that the initial data for $\mathscr{S}-\mathscr{G}^{id}$ satisfies all these conditions.

\begin{definition}\label{definition-initial-data-normalization} Consider a seed data set as in Definition \ref{initial-data-set} and let $\mathscr{S}$ be the resulting solution given by Theorem \ref{well-posedness-theorem}. We say that $\mathscr{S}$ is {\bf{initial data normalized}} if 
\begin{itemize}
\item the following conditions hold along the null hypersurface $\underline{C}_0$:
\begin{enumerate}
\item For the projection to the $\ell=0$ spherical harmonics:
\bea
\kalin&=&0 \label{kalin-initial-data} \\
\kablin&=& \ka \Omegablin \label{Omegablin-initial-data}\\
\vsilin&=&0 \label{vsilin-initial-data} 
\eea
\item For the projection to the $\ell=1$ spherical harmonics: 
\bea
\check{\ka}_{\ell=1}&=& 0 \label{condition-check-ka-l1=initial=data}\\
\divv \bF_{\ell=1}&=& 0 \label{condition-divv-bF-l1=initial=data}\\
\divv \bbF_{\ell=1}&=& 0 \label{condition-divv-bbF-l1=initial=data}
\eea
\item For the projections to the $\ell\geq1$ spherical harmonics:
\bea
\underline{b}^A&=& \frac 1 3 r^3 \ep^{AB} \nabb_B \left(2\check{\sigma}_{\ell=1}+\rhoF \check{\sigmaF}_{\ell=1} \right) \label{underline-b-initial-data}
\eea
\item For the projection to the $\ell\geq 2$ spherical harmonics:
\bea
\chih&=&0 \label{condition-chih-initial-data}\\
\chibh&=&0 \label{condition-chibh-initial-data}\\
\DDs_2\DDs_1(\check{\ka}, 0)&=& 0 \label{condition-check-ka-l2=initial=data}
\eea
\end{enumerate}
\item the following conditions hold on the sphere $S_0$:
\bea
\check{\tr_{\gamma} \slashed{g}}_{\ell=1}&=& 0 \label{tr-slashed-g-l1-initial-data}\\
\hat{\slashed{g}}&=& 0 \label{condition-hat-slashed-g}
\eea
\end{itemize}
We denote such solutions by $\mathscr{S}^{id}$. 
\end{definition}

We note that the above conditions can all be written explicitly in terms of the seed data. 
We also immediately note by straightforward computation:

\begin{proposition}\label{kerr-newman-are-initial-normalized} The linearized Kerr-Newman solutions $\mathscr{K}$ of Definition \ref{def:kerrl} are initial data normalized. 
\end{proposition}
\begin{proof} Conditions \eqref{kalin-initial-data}, \eqref{Omegablin-initial-data}, \eqref{vsilin-initial-data} are verified for the linearized Kerr-Newman solutions supported in $\ell=0$ spherical harmonics. For example
\beaa
\kablin&=& \left(\frac{4\frak{M}}{r^2}-\frac{4Q\frak{Q}}{r^3}\right)=\frac 2 r \left(\frac{2\frak{M}}{r}-\frac{2Q\frak{Q}}{r^2}\right)=\ka\Omegablin
\eeaa
Condition \eqref{underline-b-initial-data} is verified for the linearized Kerr-Newman solutions supported in $\ell=1$ spherical harmonics. Indeed, according to Proposition \ref{lem:explicitkerr}
\beaa
\underline{b}^A&=&\left(-\frac{8M}{r}+\frac{4Q^2}{r^2} \right) \mathfrak{a}\slashed{\epsilon}^{AB} \partial_B Y^{\ell=1}_m= \frac 13 r^3 \slashed{\ep}^{AB} \pr_B \left(2\check{\sigma}_{\ell=1}+\rhoF \check{\sigmaF}_{\ell=1} \right)
\eeaa
The remaining conditions are trivially verified by any linearized Kerr-Newman solution $\mathscr{K}$. 
\end{proof}

\subsection{The $S_{U, R}$-normalization}\label{definition-normalization-section}

In this section we define another normalization for a linear perturbation of Reissner-Nordstr\"om $\mathscr{S}$. The need for a normalization different than the initial data one will become clear in Section \ref{decay-from-initial-data}, when the decay of the components of an initial data normalized solution will result incomplete (see Remark \ref{result-incomplete}). 

In order to obtain a complete and optimal decay for all the components of the solution, we will need to pick gauge conditions ``far away" in the spacetime. In particular, we construct a gauge solution starting from a sphere $S_{U,R}$ for big $U$ and big $R$. 

We describe here the $S_{U, R}$-normalization, and show in Theorem \ref{achieving-normalization-theorem} that for any given solution $\mathscr{S}$ which is bounded in the past of $S_{U, R}$ we can find a pure gauge solution $\mathscr{G}^{U, R}$ such that $\mathscr{S}-\mathscr{G}^{U, R}$ satisfies all these conditions.
\begin{figure}[h]
\centering
\includegraphics[width=7cm]{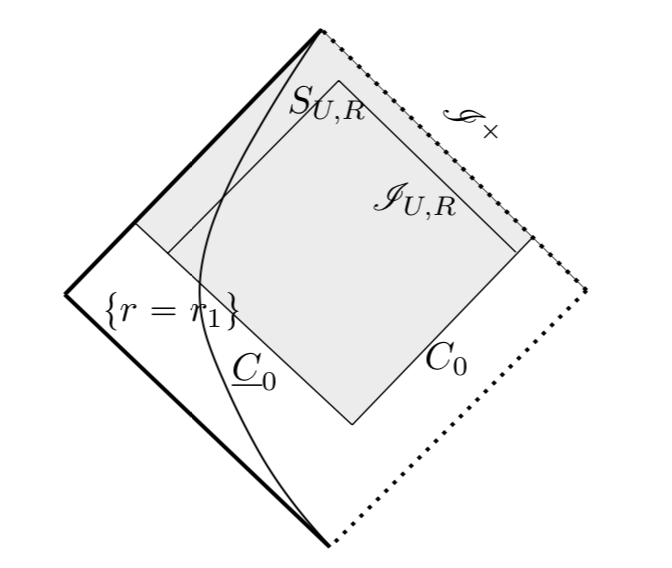}
\caption{Penrose diagram of Reissner-Nordstr{\"o}m spacetime with the sphere $S_{U, R}$ and the null hypersurface $\mathscr{I}_{U, R}$}
\end{figure}

Consider Reissner-Nordstr\"om spacetime with Bondi coordinates $(u, r, \th, \phi)$. Let $S_{U, R}$ be the sphere obtained as the intersection of the hypersurfaces $\{u=U\}$ and $\{r=R\}$ for some $U$ and $R$ such that $R \gg r_0$, where $r_0$ is the radius of $S_0$, and $U \gg u_0$. Denote $\mathscr{I}_{U, R}$ the null hypersurface obtained as the ingoing past of $S_{U, R}$.

The characterization of this gauge normalization is related to two new quantities that we define here.
We define the \textit{charge aspect function} of a solution $\mathscr{S}$ as the scalar function obtained in the following way:
\bea\label{definition-of-nu}
\check{\nu}=r^4 \left(\divv \ze+ 2\rhoF \check{\rhoF} \right)
\eea
The above definition has a similar structure than the mass aspect function in vacuum spacetimes, but only depends on the charge of the spacetime (encoded in $\check{\rhoF}$). This quantity plays a fundamental role in the derivation of the decay for the projection to the $\ell=1$ spherical harmonics, which the electromagnetic tensor is responsible for. 

We also define the \textit{mass-charge aspect function} of a solution $\mathscr{S}$ as the scalar function obtained in the following way:
\bea\label{definition-of-mu}
\check{\mu}=r^3 \left(\divv \ze+ \check{\rho}-4\rhoF \check{\rhoF} \right)-2r^4\rhoF \divv\bF
\eea
Notice that the above definition reduces to the mass aspect function in the absence of an electromagnetic tensor, i.e. $\rhoF=\bF=0$. In the case of an electrovacuum spacetime, the function $\check{\mu}$ depends on both the mass (encoded into $\check{\rho}$) and the charge (encoded into $\check{\rhoF}$).  This generalization plays a fundamental role in the derivation of decay in Section \ref{mega-section-optimal-decay-u}.

Both quantities, $\check{\nu}$ and $\check{\mu}$ verify well-behaved transport equations in the $e_4$ direction, which is the main reason why they are crucial in the derivation of decay. The derivation of the equations is obtained in Section \ref{section-transport-eqs-appendix} of the Appendix.

We can now define the notion of $S_{U, R}$-normalization. 

\begin{definition}\label{definition-SUR-normalized} Consider a seed data set as in Definition \ref{initial-data-set} and let $\mathscr{S}$ be the resulting solution given by Theorem \ref{well-posedness-theorem}. Suppose that $\mathscr{S}$ is bounded at $S_{U, R}$ for some $U$ and $R$.  We say that $\mathscr{S}$ is {\bf{$S_{U, R}$-normalized}} if
\begin{itemize} 
\item the following conditions hold along the null hypersurface $\mathscr{I}_{U, R}$:
\begin{enumerate}
\item For the projection to the $\ell=1$ spherical harmonics:
\bea
\check{\ka}_{\ell=1}&=& 0 \label{condition-check-ka-l1-last-slice}\\
\check{\kab}_{\ell=1}&=& 0 \label{condition-check-kab-l1-last-slice}\\
\check{\nu}_{\ell=1}&=& 0 \label{condition-check-nu-l1-last-slice}\\
\divv\underline{b}_{\ell=1}&=& 0 \label{condition-divv-underline-b-l1-last-slice}
\eea
where $\check{\nu}$ is the charge aspect function, as defined in \eqref{definition-of-nu}.
\item For the projection to the $\ell\geq 2$ spherical harmonics:
\bea
\DDs_2\DDs_1(\check{\ka}, 0)&=& 0 \label{condition-check-ka-definition}\label{condition-check-ka-l2}\\
\DDs_2\DDs_1(\check{\kab}, 0)&=& 0 \label{condition-check-kab-definition}\label{condition-check-kab-l2}\\
\DDs_2\DDs_1(\check{\mu}, 0)&=& 0 \label{condition-check-mu-definition}\label{condition-nu} \\
\DDs_2\underline{b}&=& 0 \label{condition-underline-b-last-slice}
\eea
where $\check{\mu}$ is the mass-charge aspect function, as defined in \eqref{definition-of-mu}.
\end{enumerate}
\item the following conditions hold on the sphere $S_{U, R}$:
\bea
\check{\tr_{\gamma} \slashed{g}}_{\ell=1}&=& 0 \label{tr-slashed-g-l1-last-slice}\\
\hat{\slashed{g}}&=& 0 \label{condition-hat-slashed-g-last-slice}
\eea
\end{itemize}
We denote such solutions by $\mathscr{S}^{U, R}$.
\end{definition}

We immediately note the following.

\begin{proposition}\label{kerr-newman-are-normalized} The linearized Kerr-Newman solutions $\mathscr{K}$ of Definition \ref{def:kerrl} are $S_{U, R}$-normalized for every $U$ and $R$. 
\end{proposition}

\subsection{Achieving the initial data normalization for a general $\mathscr{S}$}\label{achieving-initial-section}
In this section, we prove the existence of a pure gauge solution $\mathscr{G}^{id}$ such that upon subtracting this to a given $\mathscr{S}$ arising from smooth seed data, the resulting solution is generated by data satisfying all conditions of Definition \ref{definition-initial-data-normalization}.

\begin{theorem}\label{achieving-initial-normalization-theorem} Consider a seed data set as in Definition \ref{initial-data-set} and let $\mathscr{S}$ be the resulting solution given by Theorem \ref{well-posedness-theorem}.
Then there exists a pure gauge solution $\mathscr{G}^{id}$, explicitly computable from the seed data of $\mathscr{S}$, such that 
$$\mathscr{S}^{id}:= \mathscr{S}-\mathscr{G}^{id}$$
is initial data normalized. The pure gauge solution $\mathscr{G}^{id}$ is unique and arises itself from the seed data.
\end{theorem}
\begin{proof} To identify the pure gauge solution $\mathscr{G}^{id}$, it suffices to identify the functions $h, \underline{h}, a, c, d, l, q_1, q_2$ as in Lemmas \ref{pure-gauge-solution} and \ref{pure-gauge-2}. In particular, we will make use of the conditions in Definition \ref{definition-initial-data-normalization} to determine those functions and their derivative along the $e_3$ direction, and then we make use of the transport equations required in Lemmas \ref{pure-gauge-solution} and \ref{pure-gauge-2} to extend those functions in the whole spacetime along the $e_4$ direction. Using the orthogonal decomposition in spherical harmonics, we can treat the projection to the $\ell=0$, $\ell=1$ and $\ell\geq 2$ spherical harmonics separately. This procedure uniquely determines the pure gauge solution  $\mathscr{G}^{id}$ globally in the spacetime. 

{\bf{Projection to the $\ell=0$ spherical harmonics - achieving \eqref{kalin-initial-data}, \eqref{Omegablin-initial-data}, \eqref{vsilin-initial-data}.}} 
Recall that the functions in the definition of pure gauge solutions which are supported in $\ell=0$ spherical harmonics are $c$, $d$ and $l$. 
We denote $\kalin_0$, $\kablin_0$, $\vsilin_0$ and $\Omegablin_0$ the functions determined on $\underline{C}_0$ by the seed initial data set. 

We define along $\underline{C}_0$ the following function:
\beaa
c&:=&\frac{r}{2} \kalin_0, \qquad d:= \vsilin_0, \qquad l:=\frac r 2 \left(-\kab c -(\kablin_0-\kab \Omegablin_0)\right)
\eeaa
Using Lemma \ref{pure-gauge-solution}, we see that the above conditions imply that $\mathscr{S}_1:=\mathscr{S}-\mathscr{G}_{0, 0, c, d, l, 0, 0}$ verifies conditions \eqref{kalin-initial-data}, \eqref{Omegablin-initial-data}, \eqref{vsilin-initial-data} on $\underline{C}_0$. The transport equations \eqref{nabb-4-c} and \eqref{nabb-4-d} imply that $\pr_r c(u, r)=\pr_r d(u, r)=0$, therefore $c(u, r)=c(u)$ and $d(u, r)=d(r)$. This implies that the above definition of $c$ and $d$ along $\underline{C}_0$ determines $c$ and $d$ globally. Once $c$ is globally determined, the transport equation \eqref{nabb-4-l} globally determines $l$. Therefore $c, d, l$ are uniquely determined.

The following choice of pure gauge solutions will be supported in $\ell\geq1$, and therefore will not change conditions \eqref{kalin-initial-data}, \eqref{Omegablin-initial-data}, \eqref{vsilin-initial-data}.

{\bf{Projection to the $\ell=1$ spherical harmonics - achieving \eqref{condition-check-ka-l1=initial=data}, \eqref{condition-divv-bF-l1=initial=data} and \eqref{condition-divv-bbF-l1=initial=data}.}}  
We identify globally the projection to the $\ell=1$ spherical harmonics of $a$, $h$ and $\underline{h}$.
We denote $\check{\ka}_0$, $\bF_{0}$, $\bbF_{0}$ the functions on $S_0$ which are part of the seed initial data. 

We define on $S_0$ the following functions:
\beaa
h_{\ell=1}&:=&\frac{r_0^4}{2Q} {\divv\bF_0}_{\ell=1} \\
\underline{h}_{\ell=1}&:=&-\frac{r_0^4}{2Q} {\divv\bbF_0}_{\ell=1} \\
 a_{\ell=1} &=& \frac{r_0}{2} \left( {\check{\ka}}_{\ell=1}-\frac{2}{r_0^2} h_{\ell=1} -\left( \frac 1 4 \ka(r_0)\kab(r_0)\right) h_{\ell=1}-\frac 1 4 \ka(r_0)^2\underline{h}_{\ell=1} \right)
\eeaa
Using Lemma \ref{pure-gauge-solution}, we see that the above conditions imply that $\mathscr{S}_2:=\mathscr{S}_1-\mathscr{G}_{h_{\ell=1}, \underline{h}_{\ell=1}, a_{\ell=1}, 0, 0}$ verifies conditions $\check{\ka}_{\ell=1}=\divv\bF_{\ell=1}=\divv\bbF_{\ell=1}=0$ on $S_0$. 

In order to obtain the cancellation along $\underline{C}_0$, we impose transport equations for $h_{\ell=1}$, $\underline{h}_{\ell=1}$ and $a_{\ell=1}$. In particular, we have along $\underline{C}_0$:
\beaa
\nabb_3\left(\frac{2Q}{r^4} h_{\ell=1}\right)=\nabb_3\divv\bF_{\ell=1}[\mathscr{S}] \\
\nabb_3\left(\frac{2Q}{r^4} \underline{h}_{\ell=1}\right)=-\nabb_3\divv\bbF_{\ell=1}[\mathscr{S}]
\eeaa
The above transport equations together with the above initial conditions uniquely determine $h_{\ell=1}$ and $\underline{h}_{\ell=1}$ along $\underline{C}_0$. We similarly impose the $e_3$ derivative of $a_{\ell=1}$ to coincide with the derivative of the right hand side of the above definition.

Transport equation \eqref{transport-h} then uniquely determines $h_{\ell=1}$ globally.  The transport equation \eqref{transport-a} determines $a_{\ell=1}$ globally. Then transport equation \eqref{transport-hb} uniquely determines the value of $\underline{h}_{\ell=1}$ globally. This implies that $\mathscr{S}_2$ verifies conditions \eqref{condition-check-ka-l1=initial=data}, \eqref{condition-divv-bF-l1=initial=data} and \eqref{condition-divv-bbF-l1=initial=data}.

{\bf{Projection to the $\ell\geq 2$ spherical harmonics - achieving \eqref{condition-chih-initial-data}, \eqref{condition-chibh-initial-data} and \eqref{condition-check-ka-l2=initial=data}.}} 
We identify globally the projection to the $\ell\geq2$ spherical harmonics of $h$, $\underline{h}$ and $a$. 
We denote $\DDs_2\DDs_1(\check{\ka}_0, 0)$, $\chih_{0}$, $\chibh_{0}$ the symmetric traceless $2$-tensors on $S_0$ which are determined by the seed initial data. 

We define on $S_0$ the following symmetric traceless $2$-tensors:
\beaa
\DDs_2\DDs_1(h, 0)&:=& -\chih_0 \\
\DDs_2\DDs_1(\underline{h}, 0)&:=& -\chibh_0 \\
 \DDs_2\DDs_1(a, 0) &:=& \frac{r_0}{2} \big( \DDs_2\DDs_1(\check{\ka}_0, 0)-2\DDs_2\DDd_2(\DDs_2\DDs_1(h, 0))+\frac{2}{r_0^2}\DDs_2\DDs_1(h, 0) \\
 &&-\left( \frac 1 4 \ka(r_0)\kab(r_0)\right) \DDs_2\DDs_1(h, 0)-\frac 1 4 \ka(r_0)^2\DDs_2\DDs_1(\underline{h}, 0) \big)
\eeaa 
By Lemma \ref{lemma-kernel-DDs2}, the above conditions uniquely determine the projection to the $\ell\geq2$ spherical modes of $h$, $\underline{h}$ and $a$ on $S_0$. Using Lemma \ref{pure-gauge-solution}, we see that the above conditions imply that $\mathscr{S}_3:=\mathscr{S}_2-\mathscr{G}_{h_{\ell\geq2}, \underline{h}_{\ell\geq2}, a_{\ell\geq2}, 0, 0}$ verifies conditions $\DDs_2\DDs_1(\check{\ka}, 0)=\chi=\chib=0$ on $S_0$. 

Applying the $e_3$ derivative to the above definitions we derive transport equations for $\DDs_2\DDs_1(h, 0)$, $\DDs_2\DDs_1(\underline{h}, 0)$ and $\DDs_2\DDs_1(a, 0)$ which uniquely determine them on $\underline{C}_0$. Then transport equation \eqref{transport-h} then uniquely determines $\DDs_2\DDs_1(h, 0)$ globally and transport equation \eqref{transport-a} determines $\DDs_2\DDs_1(a, 0)$ globally. Then transport equation \eqref{transport-hb} uniquely determines the value of $\DDs_2\DDs_1(\underline{h}, 0)$ globally. Again by Lemma \ref{lemma-kernel-DDs2}, the above conditions uniquely determine the projection to the $\ell\geq2$ spherical modes of $h$, $\underline{h}$ and $a$ globally. This implies that $\mathscr{S}_3$ verifies in addition conditions \eqref{condition-chih-initial-data}, \eqref{condition-chibh-initial-data} and \eqref{condition-check-ka-l2=initial=data}.

{\bf{Conditions on the metric coefficients - achieving \eqref{underline-b-initial-data}, \eqref{tr-slashed-g-l1-initial-data} and \eqref{condition-hat-slashed-g}.}} 
With the above, we exhausted the freedom of using Lemma \ref{pure-gauge-solution}, since we globally determined the functions $h$, $\underline{h}$, $a$ and $\lambda$. In particular, the above choices also modified $\hat{\slashed{g}}$, $\underline{b}$ and $\check{\tr_{\gamma} \slashed{g}}$ on $\underline{C}_0$, which now differ from the one given by the seed initial data. In what follows we will achieve the remaining conditions of Definition \ref{definition-initial-data-normalization} using pure gauge solutions of the form $\mathscr{G}_{0, 0, 0, q_1, q_2}$ as in Lemma \ref{pure-gauge-2}. Observe that these solutions have all components different from $\hat{\slashed{g}}$, $\underline{b}$ and $\check{\tr_{\gamma} \slashed{g}}$ vanishing, and therefore do not modify the achieved conditions above.

 We identify $q_1$ and $q_2$ on $S_0$. 
We define on $S_0$ the following:
\beaa
(q_1)_{\ell=1}&=& -\frac{1}{4r_0^2} \check{\tr_{\gamma} \slashed{g}}_{l=1}[\mathscr{S}_3]|_{S_0}, \qquad \DDs_2\DDs_1(q_1, q_2)= \frac{1}{2r_0^2} \hat{\slashed{g}}[\mathscr{S}_3]|_{S_0}
\eeaa
where we denote $\check{\tr_{\gamma} \slashed{g}}_{\ell=1}[\mathscr{S}_3]|_{S_0}$ and $\hat{\slashed{g}}[\mathscr{S}_3]|_{S_0}$ the respective value of $\check{\tr_{\gamma} \slashed{g}}_{\ell=1}$ and $\hat{\slashed{g}}$ of the solution $\mathscr{S}_3$ defined above on the initial sphere $S_0$. By the above discussion, these values only depend on initial seed data. By Lemma \ref{lemma-kernel-DDs2}, the above conditions uniquely determine $q_1$ and $q_2$ on $S_0$.

Condition \eqref{underline-b-initial-data} on $\underline{C}_0$ determines a transport equation along $\underline{C}_0$ for $q_1$ and $q_2$:
\beaa
r^2 \DDs_1(\nabb_3 q_1, \nabb_3 q_2)^A&=& \underline{b}^A[\mathscr{S}_3]- \frac 1 3 r^3 \ep^{AB} \pr_B \left(2\check{\sigma}_{\ell=1}[\mathscr{S}_3]+\rhoF \check{\sigmaF}_{\ell=1}[\mathscr{S}_3] \right)
\eeaa
which together with \eqref{nabb-4-q1} and \eqref{nabb-4-q2} globally determine $q_1$ and $q_2$. The above conditions imply that $\mathscr{S}_4:=\mathscr{S}_3-\mathscr{G}_{0, 0, 0, q_1, q_2}$ verifies in addition conditions \eqref{underline-b-initial-data}, \eqref{tr-slashed-g-l1-initial-data} and \eqref{condition-hat-slashed-g}.

Define $\mathscr{G}^{id}:=\mathscr{G}_{h, \underline{h}, a, c, d, l, 0, 0}+\mathscr{G}_{0, 0, 0, q_1, q_2}$ with $h, \underline{h}, a, q_1, q_2$ determined as above. Then 
\beaa
\mathscr{S}^{id}:=\mathscr{S}-\mathscr{G}^{id}=\mathscr{S}_4
\eeaa
verifies all conditions of Definition \ref{definition-initial-data-normalization} and is therefore initial data normalized. By construction, $\mathscr{G}^{id}$ is also uniquely determined.
\end{proof}

\subsection{Achieving the $S_{U, R}$-normalization for a bounded $\mathscr{S}$}\label{achieving-normalization-section}
In this section, we prove the existence of a pure gauge solution $\mathscr{G}^{U, R}$ such that upon subtracting this to a given $\mathscr{S}$, which is assumed to be bounded at the sphere $S_{U, R}$ for some $U$ and $R$, the resulting solution satisfies all conditions of Definition \ref{definition-SUR-normalized}. Observe that we do not need to modify the projection to the $\ell=0$ spherical harmonics because such projection is proved to vanish in Section \ref{proj-l=0-vanishes}. 

\begin{theorem}\label{achieving-normalization-theorem} Consider a seed data set as in Definition \ref{initial-data-set} and let $\mathscr{S}$ be the resulting solution given by Theorem \ref{well-posedness-theorem}. Suppose that the solution $\mathscr{S}$ is bounded at the sphere $S_{U, R}$.

Then there exists a pure gauge solution $\mathscr{G}^{U, R}$ supported in $\ell\geq 1$ spherical harmonics such that 
$$\mathscr{S}^{U,R}:= \mathscr{S}-\mathscr{G}^{U,R}$$
is $S_{U, R}$-normalized. The pure gauge solution $\mathscr{G}^{U,R}$ is unique and is bounded by the initial data seed.
\end{theorem}
\begin{proof} We follow the same pattern as in the proof of Theorem \ref{achieving-initial-normalization-theorem}. Since $\mathscr{G}^{U, R}$ is taken to be supported in $\ell\geq 1$ spherical harmonics it suffices to identify globally the functions $h$, $\underline{h}$ and $a$ with $a$ supported in $\ell\geq1$, and $q_1$ and $q_2$. Again, we treat the projection to the $\ell=1$ and $\ell\geq2$ separately.

{\bf{Projection to the $\ell=1$ spherical harmonics - achieving \eqref{condition-check-ka-l1-last-slice}, \eqref{condition-check-kab-l1-last-slice},  \eqref{condition-check-nu-l1-last-slice}.}}  
We identify globally the projection to the $\ell=1$ spherical harmonics of $a$, $h$ and $\underline{h}$.
According to Lemma \ref{pure-gauge-solution}, for a pure gauge solution the components $\check{\ka}$ and $\check{\kab}$ verify
\bea
\check{\ka}&=& \ka a +\DDd_1\DDs_1 (h, 0) + \frac 1 4 \ka\kab h+\frac 1 4 \ka^2\underline{h}\label{imposing-check-ka} \\
\check{\kab}&=& -\kab a+\DDd_1\DDs_1(\underline{h}, 0)+\left( \frac 1 4 \kab^2+\omb\kab \right) h+\left( \frac 1 4 \ka \kab-\rho \right) \underline{h}\label{imposing-check-kab}
\eea
Multiplying \eqref{imposing-check-ka} by $\kab$ and \eqref{imposing-check-kab} by $\ka$ and summing them we obtain:
\beaa
\DDd_1\DDs_1 (\kab h+\ka\underline{h}, 0)+\left( \frac 1 2 \ka\kab+\omb\ka \right) \kab h +\left( \frac 1 2 \ka \kab-\rho \right) \ka\underline{h}&=& \kab\check{\ka}+\ka\check{\kab}
\eeaa
Setting $z=\kab h +\ka \underline{h}$ and observing that $\omb \ka=-\rho$ in the background we obtain for a pure gauge solution
\bea\label{equation-for-z}
\DDd_1\DDs_1 (z, 0) +\left( \frac 1 2 \ka \kab-\rho \right)z&=& \kab\check{\ka}+\ka\check{\kab}
\eea
Projecting the above equation to the $\ell=1$ spherical harmonics, using that $\DDd_1\DDs_1=-\lapp$ and using the Gauss equation \eqref{Gauss-general}, we obtain an equation for $z_{\ell=1}$:
\bea\label{determine-z-l1}
\left(-3\rho+2\rhoF^2 \right)z_{\ell=1}&=& \kab\check{\ka}[\mathscr{S}]_{\ell=1}+\ka\check{\kab}[\mathscr{S}]_{l=1}
\eea
The above determines the value of $z_{\ell=1}$ along the null hypersurface $\mathscr{I}_{U, R}$. 
Multiplying \eqref{imposing-check-ka} by $\kab$ and \eqref{imposing-check-kab} by $\ka$ and subtracting them we obtain:
\beaa
2\kab\ka a +\DDd_1\DDs_1 (\kab h-\ka \underline{h}, 0) -\omb\ka\kab h+\rho  \ka\underline{h}&=& \kab\check{\ka} -\ka\check{\kab}
\eeaa
Setting $\ov{z}=\kab h -\ka \underline{h}$ and observing that $\omb \ka=-\rho$ we obtain for a pure gauge solution
\bea\label{relation-ov-z-z-a}
\DDd_1\DDs_1 (\ov{z}, 0) +2\ka \kab  a  &=& \kab\check{\ka} -\ka\check{\kab}-\rho z
\eea
According to Lemma \ref{pure-gauge-solution}, for a pure gauge solution the component $\check{\nu}$ verifies
\beaa
r^{-4}\check{\nu}&=&\left(-\frac 1 4 \kab-\omb \right)\DDd_1\DDs_1(h, 0)+\frac 1 4 \ka \DDd_1\DDs_1(\underline{h}, 0)+\DDd_1\DDs_1(a, 0)+ \rhoF^2 (\kab h +\ka \underline{h})
\eeaa
which can be written in terms of $z$ and $\ov{z}$ as
\beaa
r^{-4}\check{\nu}&=&\left(-\frac 1 4 \kab+\frac 1 4 \rho r \right)\kab^{-1} \DDd_1\DDs_1(\ov{z}, 0)+\DDd_1\DDs_1(a, 0)+\frac 1 4 \rho r \kab^{-1} \DDd_1\DDs_1(z, 0)+ \rhoF^2 z 
\eeaa
Multiplying the above by $\kab$ and observing that $-\frac 1 4\kab +\frac 1 4 r \rho=\frac{1}{2r}\left(1-\frac{3M}{r} +\frac{2Q^2}{r^2}\right)$ we obtain
\bea\label{imposing-check-nu}
\begin{split}
\kab r^{-4}\check{\nu}&=\left(-\frac 1 4 \kab+\frac 1 4 \rho r \right) \DDd_1\DDs_1(\ov{z}, 0)+\kab \DDd_1\DDs_1(a, 0)+\frac 1 4 \rho r  \DDd_1\DDs_1(z, 0)+ \kab\rhoF^2 z 
\end{split}
\eea
If $R\gg 3M$, then, along $\mathscr{I}_{U, R}$, $r\gg 3M$ and therefore we can safely multiply \eqref{relation-ov-z-z-a} by $\frac{1}{2r}\left(1-\frac{3M}{r} +\frac{2Q^2}{r^2}\right)$ and subtract \eqref{imposing-check-nu} to obtain:
\bea\label{expression-a-in-terms-z-l1}
\begin{split}
& \DDd_1\DDs_1(a, 0) + \left(\frac 1 2 \ka\kab -\rho+4\rhoF^2\right)  a \\
&=  r^{-4}\check{\nu}-\frac{1}{2r}\left(1-\frac{3M}{r} +\frac{2Q^2}{r^2}\right)\check{\ka} +\frac{1}{2r}\left(1-\frac{3M}{r} +\frac{2Q^2}{r^2}\right)\kab^{-1}\ka\check{\kab}\\
&-\left(\frac 1 4 r \rho \right)\kab^{-1}\DDd_1\DDs_1(z, 0)-\left( \rhoF^2 -\frac{1}{2r}\left(1-\frac{3M}{r} +\frac{2Q^2}{r^2}\right)\rho\kab^{-1}\right) z 
\end{split}
\eea
Projecting to the $\ell=1$ spherical harmonics we see that the right hand side of \eqref{expression-a-in-terms-z-l1} is already determined by \eqref{determine-z-l1}. This gives in particular
\bea\label{determine-a-l1}
\begin{split}
 \left(-3\rho+6\rhoF^2\right)  a_{\ell=1} &=  r^{-4}\check{\nu}[\mathscr{S}]_{\ell=1}-\frac{1}{2r}\left(1-\frac{3M}{r} +\frac{2Q^2}{r^2}\right)\check{\ka}[\mathscr{S}]_{\ell=1} \\
 &+\frac{1}{2r}\left(1-\frac{3M}{r} +\frac{2Q^2}{r^2}\right)\kab^{-1}\ka\check{\kab}[\mathscr{S}]_{\ell=1}\\
&-\left(\frac 1 4 r \rho \right)\kab^{-1}2K z_{\ell=1}-\left( \rhoF^2 -\frac{1}{2r}\left(1-\frac{3M}{r} +\frac{2Q^2}{r^2}\right)\rho\kab^{-1}\right) z_{\ell=1} 
\end{split}
\eea
with known right hand side along $\mathscr{I}_{U, R}$. This determines $a_{\ell=1}$ along $\mathscr{I}_{U, R}$. 
Finally the projection of \eqref{relation-ov-z-z-a} to the $\ell=1$ spherical harmonics determines $\ov{z}_{\ell=1}$:
\beaa
\frac{2}{r^2}\ov{z}_{\ell=1}  &=&-2\ka \kab  a_{\ell=1}  \kab\check{\ka}[\mathscr{S}]_{\ell=1} -\ka\check{\kab}[\mathscr{S}]_{\ell=1}-\rho z_{\ell=1}
\eeaa
The value of $z_{\ell=1}$ and $\ov{z}_{\ell=1}$ uniquely determine the value of $h_{\ell=1}$ and $\underline{h}_{\ell=1}$ along $\mathscr{I}_{U, R}$. Transport equation \eqref{transport-h} uniquely determines $h_{\ell=1}$ globally, transport equation \eqref{transport-a} determines $a_{\ell=1}$ globally and finally transport equation \eqref{transport-hb} uniquely determines the value of $\underline{h}_{\ell=1}$ globally. This implies that $\mathscr{S}-\mathscr{G}_{h_{\ell=1}, \underline{h}_{\ell=1}, a_{\ell=1}, 0, 0}$ verifies conditions \eqref{condition-check-ka-l1-last-slice}, \eqref{condition-check-kab-l1-last-slice},  \eqref{condition-check-nu-l1-last-slice}.

{\bf{Projection to the $\ell\geq2$ spherical harmonics - achieving \eqref{condition-check-ka-definition}, \eqref{condition-check-kab-definition} and \eqref{condition-check-mu-definition}.}} 
We identify globally the projection to the $\ell\geq2$ spherical harmonics of $a$, $h$ and $\underline{h}$. We first derive some preliminary relations. 

According to Lemma \ref{pure-gauge-solution}, for a pure gauge solution the component $\check{\mu}$ verifies
\beaa
r^{-3}\check{\mu}&=&\left(-\frac 1 4 \kab-\omb-2r\rhoF^2 \right)\DDd_1\DDs_1(h, 0)+\frac 1 4 \ka \DDd_1\DDs_1(\underline{h}, 0)+\DDd_1\DDs_1(a, 0)\\
&&+\left(\frac 3 4 \rho-\frac 3 2 \rhoF^2 \right)(\kab h +\ka \underline{h})
\eeaa
which can be written in terms of $z$ and $\ov{z}$ as
\beaa
r^{-3}\check{\mu}&=&\left(-\frac 1 4\kab +\frac 1 4 r \rho-r\rhoF^2 \right)\kab^{-1}\DDd_1\DDs_1(\ov{z}, 0)+\left(\frac 1 4 r \rho-r\rhoF^2 \right)\kab^{-1}\DDd_1\DDs_1(z, 0)\\
&&+\left(\frac 3 4 \rho-\frac 3 2 \rhoF^2 \right)z+\DDd_1\DDs_1(a, 0)
\eeaa
Multiplying the above by $\kab$ and observing that $-\frac 1 4\kab +\frac 1 4 r \rho-r\rhoF^2=\frac{1}{2r}\left(1-\frac{3M}{r} \right)$ we obtain
\bea\label{imposing-check-mu}
\begin{split}
\kab r^{-3}\check{\mu}&=\frac{1}{2r}\left(1-\frac{3M}{r} \right)\DDd_1\DDs_1(\ov{z}, 0)+\left(\frac 1 4 r \rho-r\rhoF^2 \right)\DDd_1\DDs_1(z, 0)+\left(\frac 3 4 \rho-\frac 3 2 \rhoF^2 \right)\kab z\\
&+\kab \DDd_1\DDs_1(a, 0)
\end{split}
\eea
If $R\gg 3M$, then, along $\mathscr{I}_{U, R}$, $r\gg 3M$ and therefore we can safely multiply \eqref{relation-ov-z-z-a} by $\frac{1}{2r}\left(1-\frac{3M}{r} \right)$ and subtract \eqref{imposing-check-mu} to obtain:
\bea\label{expression-a-in-terms-z}
\begin{split}
 \DDd_1\DDs_1(a, 0) + \left(\frac 1 2 \ka\kab -\rho+4\rhoF^2\right)  a &=  r^{-3}\check{\mu}-\frac{1}{2r}\left(1-\frac{3M}{r} \right)\check{\ka} +\frac{1}{2r}\left(1-\frac{3M}{r} \right)\kab^{-1}\ka\check{\kab}\\
&-\left(\frac 1 4 r \rho-r\rhoF^2 \right)\kab^{-1}\DDd_1\DDs_1(z, 0)\\
&-\left(\frac 3 4 \rho-\frac 3 2 \rhoF^2 -\frac{1}{2r}\left(1-\frac{3M}{r} \right)\rho\kab^{-1}\right) z 
\end{split}
\eea

To verify conditions \eqref{condition-check-ka-definition}, \eqref{condition-check-kab-definition} and \eqref{condition-check-mu-definition} we are interested in determining $\DDs_2\DDs_1(z, 0)$, $\DDs_2\DDs_1(\ov{z}, 0)$ and $\DDs_2\DDs_1(a, 0)$ on $\mathscr{I}_{U, R}$. 

We apply the operator $\DDs_2\DDs_1$ to \eqref{equation-for-z}, which translates in the following relation along $\mathscr{I}_{U, R}$:
\beaa
\DDs_2\DDs_1\DDd_1\DDs_1 (z, 0) +\left( \frac 1 2 \ka \kab-\rho \right)\DDs_2\DDs_1(z, 0)&=& \kab\DDs_2\DDs_1(\check{\ka}[\mathscr{S}], 0)+\ka\DDs_2\DDs_1(\check{\kab}[\mathscr{S}], 0)
\eeaa
By \eqref{angular-operators} we have that 
\bea\label{relation-DDs-DDs-DDd}
\DDs_2\DDs_1\DDd_1&=& (2\DDs_2\DDd_2+2K) \DDs_2
\eea
which therefore implies the following relation for $\DDs_2\DDs_1(z, 0)$:
\beaa
\left(2\DDs_2\DDd_2+2K+ \frac 1 2 \ka \kab-\rho \right) \DDs_2\DDs_1 (z, 0) &=& \kab\DDs_2\DDs_1(\check{\ka}[\mathscr{S}], 0)+\ka\DDs_2\DDs_1(\check{\kab}[\mathscr{S}], 0)
\eeaa
By Gauss equation, we have
\beaa
2\DDs_2\DDd_2+2K+ \frac 1 2 \ka \kab-\rho&=& 2\DDs_2\DDd_2+2\left(-\frac 1 4 \ka\kab-\rho+\rhoF^2\right)+ \frac 1 2 \ka \kab-\rho=2\DDs_2\DDd_2-3\rho+2\rhoF^2
\eeaa
Define $\mathcal{E}$ to be the operator $\mathcal{E}:=2 \DDs_2\DDd_2- 3\rho+2\rhoF^2$ on symmetric traceless $2$-tensors. Then the above relation gives
\bea\label{operator-E-z}
\mathcal{E}( \DDs_2\DDs_1 (z, 0) )&=& \kab\DDs_2\DDs_1(\check{\ka}[\mathscr{S}], 0)+\ka\DDs_2\DDs_1(\check{\kab}[\mathscr{S}], 0)
\eea

We show that the operator $\mathcal{E}$ is coercive.

\begin{lemma}\label{first-poincare-inequality-2-tensor} Let $\mathcal{E}$ be the operator defined as $\mathcal{E}:=2 \DDs_2\DDd_2- 3\rho+2\rhoF^2$. For any symmetric traceless two tensor $\theta$ we have
\beaa
\int_{S} \theta \c \mathcal{E}\theta \geq \frac{4}{r^2}\int_{S}|\theta|^2
\eeaa
\end{lemma}
\begin{proof} We compute 
\beaa
\int_{S} \theta \c \mathcal{E}\theta = \int_{S} \theta \c (2 \DDs_2\DDd_2- 3\rho+2\rhoF^2)\theta=\int_{S}2 |\DDd_2\theta|^2+(- 3\rho+2\rhoF^2)|\theta|^2
\eeaa
Using the standard Poincar\'e inequality on spheres and $\int_S |\nabb \theta|^2+2K |\theta|^2=2\int_S |\DDd_2\theta|^2$, we have that $\int_S |\DDd_2\theta|^2\geq \int_S 2K |\theta|^2$, and therefore
\beaa
\int_{S} \theta \c \mathcal{E}\theta \geq \int_{S} \left(\frac{4}{r^2}+\frac{6M}{r^3}-\frac{4Q^2}{r^4}\right)|\theta|^2
\eeaa
Observe that $\frac{6M}{r^3}-\frac{4Q^2}{r^4}\geq \frac{2M^2}{r^3}$ for all $r>M$ and $|Q|<M$. We therefore obtain the inequality.
\end{proof}
The above Lemma shows that \eqref{operator-E-z} uniquely determines $\DDs_2\DDs_1 (z, 0)$ along $\mathscr{I}_{U, R}$. 

Applying the operator $\DDs_2\DDs_1$ to \eqref{expression-a-in-terms-z} we obtain on $\mathscr{I}_{U, R}$
\beaa
 \DDs_2\DDs_1\DDd_1\DDs_1(a, 0) + \left(\frac 1 2 \ka\kab -\rho+4\rhoF^2\right) \DDs_2\DDs_1( a, 0) &=& \text{RHS}(\check{\mu}[\mathscr{S}], \check{\ka}[\mathscr{S}], \check{\kab}[\mathscr{S}], \DDs_2\DDs_1(z, 0))  
\eeaa
where the right hand side depends on the argument, which are determined along $\mathscr{I}_{U, R}$. Using \eqref{relation-DDs-DDs-DDd}, we obtain
\bea\label{derive-a-from-last-slice}
\left(2\DDs_2\DDd_2-3\rho+6\rhoF^2\right)\DDs_2\DDs_1( a, 0) &=& \text{RHS}(\check{\mu}[\mathscr{S}], \check{\ka}[\mathscr{S}], \check{\kab}[\mathscr{S}], \DDs_2\DDs_1(z, 0))
\eea
The above operator is a slight modification of $\mathcal{E}$ (which is even more positive) and possesses an identical Poincar\'e inequality as in Lemma \ref{first-poincare-inequality-2-tensor}. The above relation therefore implies that $ \DDs_2\DDs_1( a, 0)$ is uniquely determined along $\mathscr{I}_{U, R}$ by the above imposition. 

Finally, applying the operator $\DDs_2\DDs_1$ to \eqref{relation-ov-z-z-a} and using \eqref{relation-DDs-DDs-DDd} we obtain on $\mathscr{I}_{U, R}$:
\bea
(2\DDs_2\DDd_2+2K) \DDs_2\DDs_1 (\ov{z}, 0)   &=& \kab \DDs_2\DDs_1(\check{\ka}[\mathscr{S}], 0) -\ka\DDs_2\DDs_1(\check{\kab}[\mathscr{S}], 0)-\rho z-2\ka \kab  \DDs_2\DDs_1(a, 0)
\eea\label{derive-ov-z-last-slice}
where the right hand side has already been determined above.
The above operator is clearly coercive. Indeed, using elliptic estimate \eqref{second-elliptic-estimate} we have
\beaa
\int_{S} \th \cdot (2\DDs_2\DDd_2+2K) \th&=& \int_{S} 2|\DDd_2\th|^2+2K |\th|^2 \geq \frac{4}{r^2} \int_{S} |\th|^2
\eeaa
The above relation therefore implies that $ \DDs_2\DDs_1( z, 0)$ is uniquely determined along $\mathscr{I}_{U, R}$. 

To check that these tensors are smooth along $\mathscr{I}_{U, R}$, we show that their $e_3$ derivative is smooth along the null hypersurface. For instance, suppose $\th$ is a symmetric traceless $2$-tensor which verifies 
\beaa
\mathcal{E}(\th)=\mathcal{F}
\eeaa
where $\mathcal{F}$ is a smooth known function on $\mathscr{I}_{U, R}$. We compute
\beaa
[\mathcal{E}, \nabb_3]&=& (2 \DDs_2\DDd_2- 3\rho+2\rhoF^2)(\nabb_3\th)-\nabb_3((2 \DDs_2\DDd_2- 3\rho+2\rhoF^2)\th)\\
&=& \kab (2 \DDs_2\DDd_2- \frac 9 2 \rho+\rhoF^2)\th= \kab \mathcal{F}-\kab \left(\frac 3 2 \rho +\rhoF^2\right) \th
\eeaa
Therefore we obtain 
\beaa
\mathcal{E}(\nabb_3\th)+\kab \left(\frac 3 2 \rho +\rhoF^2\right) \th&=& \nabb_3 \mathcal{F}+\kab \mathcal{F}
\eeaa
Applying the operator $\mathcal{E}$ to the above we obtain
\beaa
\mathcal{E}^2(\nabb_3\th)&=& \mathcal{E}(\nabb_3 \mathcal{F}+\kab \mathcal{F}) -\kab \left(\frac 3 2 \rho +\rhoF^2\right) \mathcal{F}
\eeaa
This shows that $\nabb_3\th$ is smooth along $\mathscr{I}_{U, R}$. A similar computation applied for the modified version of $\mathcal{E}$ applied to $\DDs_2\DDs_1(a, 0)$ and the operator applied to $\DDs_2\DDs_1(\ov{z}, 0)$.

The above choices uniquely determine the projection to the $\ell\geq2$ spherical harmonics of $h$, $\underline{h}$ and $a$ on $\mathscr{I}_{U, R}$ and as above, integrating in order the transport equation for $h$, $a$ and then for $\underline{h}$, we show that they are globally uniquely determined. 

 By construction,
 \beaa
 \mathscr{S}_{1}&=& \mathscr{S}-\mathscr{G}_{h, \underline{h}, a, 0, 0}
 \eeaa
 verifies conditions  \eqref{condition-check-ka-l1-last-slice}, \eqref{condition-check-kab-l1-last-slice},  \eqref{condition-check-nu-l1-last-slice},  \eqref{condition-check-ka-definition}, \eqref{condition-check-kab-definition} and \eqref{condition-check-mu-definition}.

{\bf{Conditions on the metric coefficients - achieving \eqref{condition-divv-underline-b-l1-last-slice}, \eqref{condition-underline-b-last-slice}, \eqref{tr-slashed-g-l1-last-slice} and \eqref{condition-hat-slashed-g-last-slice}.}} 
With the above, we exhausted the freedom of using Lemma \ref{pure-gauge-solution}, since we globally determined the functions $h$, $\underline{h}$, $a$. In particular, the above choices also modified $\hat{\slashed{g}}$, $\underline{b}$ and $\check{\tr_{\gamma} \slashed{g}}$ on $\mathscr{I}_{U, R}$. In what follows we will achieve the remaining conditions using pure gauge solutions of the form $\mathscr{G}_{0, 0, 0, q_1, q_2}$ as in Lemma \ref{pure-gauge-2}. Observe that these solutions have all components different from $\hat{\slashed{g}}$, $\underline{b}$ and $\check{\tr_{\gamma} \slashed{g}}$ vanishing, and therefore do not modify the achieved conditions above. 

The conditions here imposed are almost identical to the one imposed to the metric coefficient in the initial data normalization. The procedure to find $q_1$ and $q_2$ is identical. We sketch it here.

We define on $S_{U, R}$ the following:
\beaa
(q_1)_{\ell=1}&=& -\frac{1}{4r_0^2} \check{\tr_{\gamma} \slashed{g}}_{\ell=1}[\mathscr{S}_1]|_{S_{U, R}}, \qquad \DDs_2\DDs_1(q_1, q_2)= \frac{1}{2r_0^2} \hat{\slashed{g}}[\mathscr{S}_1]|_{S_{U, R}}
\eeaa
 By Lemma \ref{lemma-kernel-DDs2}, the above conditions uniquely determine $q_1$ and $q_2$ on $S_{U, R}$. 
Conditions \eqref{condition-divv-underline-b-l1-last-slice} and \eqref{condition-underline-b-last-slice} on $\mathscr{I}_{U, R}$ determine a transport equation along $\mathscr{I}_{U, R}$ for $q_1$ and $q_2$:
\beaa
r^2 \DDs_2\DDs_1(\nabb_3 q_1, \nabb_3 q_2)&=& \DDs_2\underline{b}[\mathscr{S}_1] \\
2(\nabb_3q_1)_{\ell=1}&=& \divv\underline{b}[\mathscr{S}_1]_{\ell=1}
\eeaa
which together with \eqref{nabb-4-q1} and \eqref{nabb-4-q2} globally determine $q_1$ and $q_2$. 

Define $\mathscr{G}^{U, R}:=\mathscr{G}_{h, \underline{h}, a, q_1, q_2}$ with $h, \underline{h}, a, q_1, q_2$ determined as above. Then 
\beaa
\mathscr{S}^{U, R}:=\mathscr{S}-\mathscr{G}^{U, R}
\eeaa
verifies all conditions of Definition \ref{definition-SUR-normalized} and is therefore $S_{U, R}$-normalized. By construction, $\mathscr{G}^{U, R}$ is supported in $\ell\geq1$ and is uniquely determined.
\end{proof}

\subsection{The Kerr-Newman parameters in $\ell=0,1$ modes}\label{Kerr-Newman-parameters-section}

The initial data normalization allows to identify the Kerr-Newman parameters from the initial data seed. Notice that, in contrast with the linear stability of Schwarzschild \cite{DHR}, the projection of the initial data normalized solution to the $\ell=0, 1$ modes is not exhausted by the linearized Kerr-Newman solution. Because of the presence of the electromagnetic radiation, there is decay of the components at the level of the $\ell=1$ mode (see Section \ref{decay-l-1mode}).

We define the Kerr-Newman parameters, which are read off at the initial sphere $S_0$ of radius $r_0$.

\begin{definition}\label{definition-kerr-newman-normalized} Let $\mathscr{S}$ and $\mathscr{S}^{id}$ as in Theorem \ref{achieving-initial-normalization-theorem}. 
 We denote $\mathscr{K}^{id}$ the linearized Kerr-Newman solution $\mathscr{K}_{(\mathfrak{M}, \mathfrak{Q}, \mathfrak{b}, \mathfrak{a})}$ where the parameters $\mathfrak{M}, \mathfrak{Q}, \mathfrak{b}, \mathfrak{a}_i$ are given by
\beaa
\mathfrak{Q}=r_0^2 \rhoFlin|_{S_0}, \qquad  \mathfrak{M}= -\frac {r_0^3}{ 2} \rlin|_{S_0}+2Q r_0 \rhoFlin|_{S_0},\qquad \mathfrak{b}=r_0^2 \sigmaFlin|_{S_0}, \\ 
 \mathfrak{a}_{-1}={\check{\sigmaF}_{\ell=1, m=-1}}|_{S_0}, \qquad  \mathfrak{a}_{0}={\check{\sigmaF}_{\ell=1, m=0}}|_{S_0}, \qquad  \mathfrak{a}_{1}={\check{\sigmaF}_{\ell=1, m=1}}|_{S_0}
\eeaa
where the above quantities refer to the initial data normalized solution $\mathscr{S}^{id}$. 
\end{definition}

Observe that the Kerr-Newman parameters in this definition are explicitably computable from the seed initial data.

\section{Proof of boundedness}\label{chapter-linear-stability}

In this section, we prove boundedness of a linear perturbation of Reissner-Nordstr\"om spacetime $\mathscr{S}$ which is initial data-normalized, upon subtracting a member of the linearized Kerr family.

In Section \ref{decay-from-initial-data}, we use the initial data normalization to prove boundedness. In the process of obtaining boundedness, we obtain decay for some components, and non optimal decay for other components. We make use of this boundedness statement to define an associated $S_{U, R}$-normalization, and in Section \ref{decay-pure-gauge} we prove that the gauge solution decays and is controlled by initial data. 

\subsection{Initial data normalization and boundedness}\label{decay-from-initial-data}

Here we state the result proved in this chapter. 

Let $\mathscr{S}$ be a linear gravitational and electromagnetic perturbation around Reissner-Nordstr\"om spacetime $(\MM, g_{M, Q})$, with $|Q| \ll M$, arising from regular asymptotically flat initial data. Let $\mathscr{S}^{id}$ be the initial data normalized solution associated to $\mathscr{S}$ by Theorem \ref{achieving-initial-normalization-theorem} and let $\mathscr{K}^{id}$ the linearized Kerr-Newman solution associated to $\mathscr{S}^{id}$ as in Definition \ref{definition-kerr-newman-normalized}. Define  
\bea\label{definition-S-id-K}
\mathscr{S}^{id, K}:= \mathscr{S}^{id}-\mathscr{K}^{id}=\mathscr{S}-\mathscr{G}^{id}-\mathscr{K}^{id}
\eea

\begin{proposition}\label{prop-boundedness}  
All the components of the solution $\mathscr{S}^{id, K}$ are (polynomially) bounded in the region $\{r \geq r_1\} \cap \{ u_0 \leq u \leq u_1\}$ where $u_1$ is the $u$-coordinate of the sphere of intersection between $\underline{C}_0$ and $\{ r=r_1\}$, for some $r_1>r_{\mathcal{H}}$. Moreover, its projection to the $\ell=0$ spherical harmonics vanishes. 
\end{proposition}

\begin{remark}\label{remark-not-uniform-r} It follows from the proof of Proposition \ref{prop-boundedness} that the estimates hereby obtained (for example \eqref{decay-curl-b-l1}-\eqref{decay-check-sigmaF-l1-optimal-r}) are not uniform in the parameter $r_1$. More precisely, the constants in the bounds tend to infinity as $r_1$ approaches $r_{\mathcal{H}}$ (as one can see from the fact that otherwise those estimates would imply the vanishing of the bounded quantities along the horizon).  Note however that the implicit constants in those estimates can be made uniform in the parameter $r_1$ if one does not require $u$ decay, but only $r$ decay (see also Remark \ref{replace-uniform}).

The boundedness proved in Proposition  \ref{prop-boundedness} is used to apply Theorem \ref{achieving-normalization-theorem} and initialize the $S_{U,R}$-normalization. If one were to use the estimates which depend on $r_1$, which itself depends on the choice of $U$ and $R$, then the decay estimates obtained in Theorem \ref{linear-stability} would eventually depend on $U$ and $R$. 
One can use instead the bounds which are uniform in $r_1$ to obtain the initialization of $S_{U, R}$-normalization, for which the decay in $u$ is not required. 
This is sufficient to close the arguments and obtain the final estimates which do not depend on $U$ and $R$.

\end{remark}

The proof of the Proposition makes use of the gauge conditions imposed in the initial data normalization to then integrate forward the transport equations. The process of integrating forward from a bounded $r$ necessarily implies that some components (namely $\xib$ and $\check{\omb}$) would not decay in $r$ (some would polynomially grow in $r$). For this reason, with this procedure we fail to obtain decay in $r$ for $\xib$ and $\check{\omb}$. In the next chapter we introduce the $S_{U, R}$-normalization, through which we can integrate backward from an unbounded $r$. This allows to obtain the optimal decay in $r$ and $u$ (i.e. as given in Theorem \ref{linear-stability}).

 In Section \ref{proj-l=0-vanishes} we prove that the projection to the $\ell=0$ spherical mode of such a solution vanishes.
  In Section \ref{projection-l1-boundedness} and in Section \ref{proj-l2=boundedness} we prove boundedness and some decay for the projection to the $\ell=1$ mode and $\ell\geq2$ modes respectively.  
  Finally in Section \ref{terms-e3-boundedness} we derive boundedness statements for the quantities involved in the $e_3$ direction, i.e. $\xib$, $\eta$ and $\check{\omb}$ for which the decay is not optimal. This lack of optimality is the reason to use the $S_{U, R}$-normalization later.

We first summarize the main properties of the solution $\mathscr{S}^{id, K}$. 
\begin{enumerate}
\item Since by Proposition \ref{kerr-newman-are-initial-normalized}, the linearized Kerr-Newman solution is initial data normalized, then $\mathscr{S}^{id, K}$ is initial data normalized, i.e. conditions \eqref{kalin-initial-data}-\eqref{condition-hat-slashed-g} are verified. 
\item By Definition \ref{definition-kerr-newman-normalized} of the linearized Kerr-Newman solution $\mathscr{K}^{id}$, the solution $\mathscr{S}^{id, K}$ verifies in addition (recall condition \eqref{underline-b-initial-data}):
\bea
\rhoFlin=\rlin=\sigmaFlin=\check{\sigmaF}_{\ell=1}&=&0 \qquad  \text{on $S_0$} \label{vanishing-l0-after-kerr}\\
(\divv\underline{b})_{\ell=1}=\DDs_2\underline{b}&=&0, \qquad \text{on $\underline{C}_0$} \label{divv-b-0-initial}\\
(\curll \underline{b})_{\ell=1}=-\frac 2 3r^3 \lapp \check{\sigma}_{\ell=1}&=&\frac 4 3r  \check{\sigma}_{\ell=1}  \qquad \text{on $\underline{C}_0$} \label{curll-b-initial}
\eea
\end{enumerate}

We show that the above conditions imply the vanishing of the projection to the $\ell=0$ spherical mode, and the boundedness of the projection to the $\ell\geq 1$ spherical mode. 

\subsubsection{The projection to the $\ell=0$ mode}\label{proj-l=0-vanishes}
We prove here that the projection to the $\ell=0$ spherical mode of the solution $\mathscr{S}^{id, K}$ defined in \eqref{definition-S-id-K} vanishes.

We first prove the vanishing of the $\ell=0$ mode on $S_0$:
\begin{enumerate}
\item From \eqref{kalin-initial-data}, \eqref{Omegablin-initial-data}, \eqref{vsilin-initial-data} and \eqref{vanishing-l0-after-kerr}, we have on $S_0$:
\beaa
\kalin=\kablin-\ka\Omegablin=\vsilin=\rhoFlin=\rlin=\sigmaFlin=0
\eeaa
\item Applying Gauss equation \eqref{Gauss-lin} to the sphere $S_0$ we obtain $\kablin=0$, and therefore from  \eqref{Omegablin-initial-data} we obtain $\Omegablin=0$. 
\item Condition \eqref{kalin-initial-data} holds on $\underline{C}_0$, therefore it implies $\nabb_3\kalin=0$ on $S_0$. Restricting \eqref{nabb-3-kalin-ze-eta} to $S_0$ we obtain that $\omblin=0$ on $S_0$. 
\end{enumerate}

We prove vanishing of the $\ell=0$ mode on $\underline{C}_0$:
\begin{enumerate}
\item From \eqref{kalin-initial-data}, \eqref{Omegablin-initial-data}, \eqref{vsilin-initial-data} we have on $\underline{C}_0$:
\beaa
\kalin=\kablin-\ka\Omegablin=\vsilin=0
\eeaa
\item From the transport equations \eqref{nabb-3-sigmaFlin} and \eqref{nabb-3-rhoFlin} and the vanishing initial data on $S_0$ for $\sigmaFlin$ and $\rhoFlin$, we obtain $\sigmaFlin=\rhoFlin=0$ on $\underline{C}_0$.
\item From the transport equation \eqref{nabb-3-rlin} and the vanishing initial data on $S_0$ for $\rlin$, we obtain $\rlin=0$ on $\underline{C}_0$.
\item From Gauss equation \eqref{Gauss-lin} and \eqref{kalin-initial-data} we obtain $\kablin=0$ on $\underline{C}_0$, and therefore from \eqref{Omegablin-initial-data} we obtain $\Omegablin=0$ on $\underline{C}_0$. 
\item  Restricting \eqref{nabb-3-kalin-ze-eta} to $\underline{C}_0$ we obtain that $\omblin=0$ on $\underline{C}_0$. 
\end{enumerate}

We prove vanishing of the $\ell=0$ mode everywhere:
\begin{enumerate}
\item From \eqref{kalin-initial-data} and \eqref{nabb-4-kalin-ze-eta} we have globally:
\beaa
\kalin=0
\eeaa
\item From the transport equations \eqref{nabb-4-sigmaFlin} and \eqref{nabb-4-rhoFlin} and the vanishing initial data on $\underline{C}_0$ for $\sigmaFlin$ and $\rhoFlin$, we obtain $\sigmaFlin=\rhoFlin=0$ globally.
\item From the transport equation \eqref{nabb-4-rlin} and the vanishing initial data on $\underline{C}_0$ for $\rlin$, we obtain $\rlin=0$ globally.
\item From Gauss equation \eqref{Gauss-lin} we obtain $\kablin=0$ globally. 
\item From \eqref{nabb-4-omblin-ze-eta} and the vanishing initial data on $\underline{C}_0$ we obtain that $\omblin=0$ globally. 
\item From \eqref{nabb-4-vsilin} and the vanishing initial data on $\underline{C}_0$ we obtain that $\vsilin=0$ globally. 
\item From \eqref{nabb-4-Omegablin} and the vanishing initial data on $\underline{C}_0$ we obtain that $\Omegablin=0$ globally. 
\end{enumerate}

The projection to the $\ell=0$ spherical mode of an initial data normalized solution is therefore exhausted by a linearized Reissner-Nordstr\"om solution, with no non-trivial decay supported in this spherical mode.

\subsubsection{The projection to the $\ell=1$ mode}\label{projection-l1-boundedness}
In contrast with the case of linear stability of Schwarzschild \cite{DHR}, in the linear stability of Reissner-Nordstr\"om, because of the presence of the electromagnetic radiation, we expect the projection to the $\ell=1$ mode of the solution not to be exhausted by a pure gauge and a linearized Kerr-Newman solution. We indeed show that there is also decay at the level of the projection to the $\ell=1$ mode.

To obtain decay for the projection to the $\ell=1$ spherical mode, we make use of Theorem \ref{estimates-theorem} stating the decay for the gauge-invariant quantities $\tilde{\b}$, $\tilde{\bb}$, $\frak{p}$. In particular, we show that we can express all the remaining quantities in terms of only $\bF$, $\bbF$ and $\check{\ka}$ and the gauge-invariant quantities already estimated. This will simplify the computations and the derivation of the estimates for all quantities. 

\textbf{Notation} We denote that a quantity $\xi$ is $O(r^{-p-1} u^{-1/2+\de}, r^{-p} u^{-1+\de})$ if 
\beaa
|\xi|\les \min\{r^{-p-1} u^{-1/2+\de}, r^{-p}u^{-1+\de}\} \qquad \text{for all $u \geq u_0$ and $r\geq r_1$}
\eeaa
We write $\xi_1=\xi_2+O(r^{-p-1} u^{-1/2+\de}, r^{-p} u^{-1+\de})$ if $\xi_1=\xi_2+\xi_3$ with $\xi_3=O(r^{-p-1} u^{-1/2+\de}, r^{-p} u^{-1+\de})$.

Since the following relations will be used later in the proof of the optimal decay, we summarize them in the following Proposition. To derive those, we only use elliptic relations, and not transport equations which will be exploited later.

\begin{proposition}\label{expression-all-quantities-l1} The following relations hold true for all $u \geq u_0$ and $r\geq r_1$:
\bea
(\divv \b)_{\ell=1}&=&\frac{3\rho}{2\rhoF} (\divv \bF)_{\ell=1}+O(r^{-4-\de}u^{-1/2+\de}, r^{-3-\de} u^{-1+\de} )  \label{divv-b-in-terms-of-bF-l1}\\
(\divv \bb)_{\ell=1}&=& \frac{3\rho}{2\rhoF} (\divv \bbF)_{\ell=1} +O(r^{-3} u^{-1+\de}) \label{divv-bb-in-terms-of-bbF-l1} \\
(\divv\ze)_{\ell=1} &=&  \frac 1 r \check{\ka}_{\ell=1} +r\left(\frac{3\rho}{2\rhoF} -\rhoF\right)(\divv\bF)_{\ell=1}+O(r^{-3-\de} u^{-1/2+\de}, r^{-2-\de} u^{-1+\de}) \label{divv-ze-in-terms-of-bF-ka-l1} 
\eea
\bea
\begin{split}
 \check{\kab}_{\ell=1}&=-\frac 1 2 r\kab \check{\ka}_{\ell=1} -\frac 1 2 r^3\kab\left(\frac{3\rho}{2\rhoF} -\rhoF\right)(\divv\bF)_{\ell=1}
 +r^2\left(\frac{3\rho}{2\rhoF}-\rhoF\right)(\divv\bbF)_{\ell=1}\\
 &+O( r^{-1} u^{-1+\de}) \label{divv-check-kab-in-terms-of-bF-bbF-ka-l1} 
 \end{split}
 \eea
 \bea
 \begin{split}
 \check{\rho}_{\ell=1}&=\frac 1 4 r^2\kab\left(\frac{3\rho}{2\rhoF} +\rhoF\right)(\divv\bF)_{\ell=1}
 -\frac 1 2 r\left(\frac{3\rho}{2\rhoF}+\rhoF\right)(\divv\bbF)_{\ell=1}\\
 &+O( r^{-2} u^{-1+\de})\label{check-rho-in-terms-of-bF-bbF-l1}
 \end{split}
 \eea
 \bea
  \check{\rhoF}_{\ell=1}  &=& \frac 1 4  r^2\kab(\divv\bF)_{\ell=1}
 - \frac 1 2 r(\divv\bbF)_{\ell=1}+ O( r^{-1} u^{-1+\de}) \label{check-rhoF-in-terms-of-bF-bbF-l1}
\eea
\end{proposition}
\begin{proof} Recall the definition of the gauge-invariant quantities $\tilde{\b}$ and $\tilde{\bb}$ as defined in \eqref{definition-tilde-b}. By taking the divergence and projecting to the $\ell=1$ spherical mode, we have
\beaa
(\divv\tilde{\b})_{\ell=1}&=& 2\rhoF (\divv \b)_{\ell=1} -3\rho (\divv \bF)_{\ell=1} \\
(\divv\tilde{\bb})_{\ell=1}&=& 2\rhoF (\divv \bb)_{\ell=1} -3\rho (\divv \bbF)_{\ell=1}
\eeaa
We can in particular isolate $\b$ and $\bb$ and obtain:
\beaa
(\divv \b)_{\ell=1}&=&\frac{3\rho}{2\rhoF} (\divv \bF)_{\ell=1}+\frac{1}{2\rhoF}(\divv\tilde{\b})_{\ell=1}  \\
(\divv \bb)_{\ell=1}&=& \frac{3\rho}{2\rhoF} (\divv \bbF)_{\ell=1} +\frac{1}{2\rhoF}(\divv\tilde{\bb})_{\ell=1}
\eeaa
Using the estimates for $\tilde{\b}$ and $\tilde{\bb}$ as in \eqref{estimate-tilde-b} and \eqref{estimate-tilde-bb}, we obtain \eqref{divv-b-in-terms-of-bF-l1} and \eqref{divv-bb-in-terms-of-bbF-l1}.

Applying $\divv$ to Codazzi equation \eqref{Codazzi-chi-ze-eta} and projecting to the $\ell=1$ spherical harmonics, since $\divv\divv\chih_{\ell=1}=0$ and $\ka=\frac{2}{r}$ we obtain
\beaa
  (\divv\ze)_{\ell=1}&=& \frac 1 2 r\DDd_1\DDs_1(\check{\ka}, 0)_{\ell=1} +r(\divv\b)_{\ell=1}-r\rhoF(\divv\bF)_{\ell=1}
  \eeaa
  Using \eqref{angular-operators} to project the laplacian to the $\ell=1$, and using \eqref{divv-b-in-terms-of-bF-l1}, we obtain \eqref{divv-ze-in-terms-of-bF-ka-l1}.

Applying $\divv$ to Codazzi equation \eqref{Codazzi-chib-ze-eta} and projecting to the $\ell=1$ spherical harmonics, since $\divv\divv\chibh_{\ell=1}=0$  we obtain
\beaa
\frac 1 2 \DDd_1\DDs_1( \check{\kab}, 0)_{\ell=1}&=&-\frac 1 2 \kab (\divv\ze)_{\ell=1} +(\divv\bb)_{\ell=1}-\rhoF(\divv\bbF)_{\ell=1}
\eeaa
  Using \eqref{angular-operators} to project the laplacian to the $\ell=1$, and using \eqref{divv-bb-in-terms-of-bbF-l1}, we obtain \eqref{divv-check-kab-in-terms-of-bF-bbF-ka-l1}.

  Projecting Gauss equation \eqref{Gauss-check} to the $\ell=1$ spherical harmonics and using \eqref{check-K-l=1}, we have 
\beaa
\check{\rho}_{\ell=1}&=&- \frac 1 4 \kab \check{\ka}_{\ell=1}- \frac 1 4 \ka \check{\kab}_{\ell=1} +2\rhoF\check{\rhoF}_{\ell=1}\\
&=&\frac 1 4 r^2\kab\left(\frac{3\rho}{2\rhoF} -\rhoF\right)(\divv\bF)_{\ell=1}
 -\frac 1 2 r\left(\frac{3\rho}{2\rhoF}-\rhoF\right)(\divv\bbF)_{\ell=1}\\
 &&+2\rhoF\check{\rhoF}_{\ell=1}+O( r^{-2} u^{-1+\de})
\eeaa

Commuting the expression for $\pf$ given by \eqref{relation-pf-rho-rhoF-bF-bbF}, with $\divv$ and projecting to the $\ell=1$ spherical harmonics we obtain
\beaa
\frac{(\divv\pf)_{\ell=1}}{r^5} &=& 2\rhoF \DDd_1\DDs_1(-\check{\rho}, \check{\sigma})_{\ell=1} +(3\rho-2\rhoF^2) \DDd_1\DDs_1(\check{\rhoF}, \check{\sigmaF})_{\ell=1}\\
&&+2 \rhoF^2 (\kab (\divv\bF)_{\ell=1}-\ka(\divv\bbF)_{\ell=1}) \\
&=& 2\rhoF \left(-\frac{2}{r^2} \check{\rho}_{\ell=1}\right) +(3\rho-2\rhoF^2) \left(\frac{2}{r^2} \check{\rhoF}_{\ell=1}\right)+2 \rhoF^2 (\kab (\divv\bF)_{\ell=1}-\ka(\divv\bbF)_{\ell=1})
\eeaa 
Making use of the estimate \eqref{estimate-pf} for $\frak{p}$ and of the above relation for $\check{\rho}_{\ell=1}$ we have
\beaa
  (3\rho-2\rhoF^2) (2 \check{\rhoF}_{\ell=1})&=&4\rhoF (\check{\rho}_{\ell=1})-2 r^2\rhoF^2 (\kab (\divv\bF)_{\ell=1}-\ka(\divv\bbF)_{\ell=1})+ O( r^{-4} u^{-1+\de}) \\
  &=& \frac 1 2  r^2\kab\left(3\rho -2\rhoF^2\right)(\divv\bF)_{\ell=1}
 - r\left(3\rho-2\rhoF^2\right)(\divv\bbF)_{\ell=1}\\
 &&+8\rhoF^2\check{\rhoF}_{\ell=1}-2 r^2\rhoF^2 (\kab (\divv\bF)_{\ell=1}-\ka(\divv\bbF)_{\ell=1})+ O( r^{-4} u^{-1+\de})
\eeaa
which gives
\beaa
 (3\rho-6\rhoF^2) (2 \check{\rhoF}_{\ell=1})  &=& \frac 1 2  r^2\kab\left(3\rho -6\rhoF^2\right)(\divv\bF)_{\ell=1}
 - r\left(3\rho-6\rhoF^2\right)(\divv\bbF)_{\ell=1}+ O(r^{-4} u^{-1+\de})
\eeaa
Putting together the above, we finally obtain \eqref{check-rho-in-terms-of-bF-bbF-l1} and \eqref{check-rhoF-in-terms-of-bF-bbF-l1}.
\end{proof}

We now obtain control over $\check{\ka}_{\ell=1}$, $(\divv\bF)_{\ell=1}$, $(\divv\bbF)_{\ell=1}$ by using transport equations. In integrating transport equation along the $e_4$ direction from the initial data $\underline{C}_0 \cap \{r \geq r_1\}$ we obtain control in the region $\{ r\geq r_1\} \cap \{ u_0 \leq u \leq u_1 \}$, where $u_1$ is the $u$-coordinate of the sphere of intersection between $\underline{C}_0$ and $\{ r=r_1\}$. 

\begin{enumerate}
\item Recall Lemma \ref{commutation-projection-l1} for the commutation of the $\nabb_4$ derivative with the projection to the $\ell=1$ spherical harmonics. From \eqref{nabb-4-check-ka-ze-eta} we obtain
\beaa
\nabb_4(\check{\ka}_{\ell=1})&=& (\nabb_4\check{\ka})_{\ell=1}=-\ka \check{\ka}_{\ell=1}
\eeaa 
Together with condition \eqref{condition-check-ka-l1=initial=data}, this implies 
\bea\label{check-ka-l1=0-initial-data-everywhere}
\check{\ka}_{\ell=1}=0 \qquad \text{in $\{ r\geq r_1\} \cap \{ u_0 \leq u \leq u_1 \}$}
\eea
\item From equation \eqref{nabb-4-bF-tilde-b} commuted with $\divv$ and the estimates \eqref{estimate-tilde-b} for $\tilde{\b}$ we have for $r \geq r_1$:
\beaa
\nabb_4 \divv\bF+ 2 \ka\divv\bF&=&O(r^{-4-\de}u^{-1/2+\de}, r^{-3-\de}u^{-1+\de})
\eeaa
We now integrate the above equation from $\underline{C}_0\cap \{ r\geq r_1\} $ using the condition on the initial data \eqref{condition-divv-bF-l1=initial=data}. We obtain
\beaa
r^4  (\divv\bF)_{\ell=1} =  \int_{r_0(u)}^r  O(r^{-\de}u^{-1/2+\de}, r^{1-\de} u^{-1+\de}) d \tilde{r}
\eeaa
where $r_0(u)$ is the $r$-coordinate of the intersection of $\underline{C}_0$ and the cone of constant $u$ coordinate. Recall that the pointwise estimate on the right hand side of the above is only valid in the region $r \geq r_1$. In particular, we make use of the above integration only in the region\footnote{Indeed the right hand side of the equation cannot be estimated in terms of an inverse power of $u$ in the region $r \leq r_1$. On other hand the estimates in this region near the horizon are obtained in terms of the coordinate $v$ as in \cite{stabilitySchwarzschild}, while the $r$ decay is irrelevant (since $r$ is bounded). See also Section \ref{section-close-horizon}.} given by the intersection of $u \geq u_0$ and $r \geq r_1$ (or also, only for $r_0(u) \geq r_1$). 
This implies
\bea\label{decay-bF-from-initial-data}
|(\divv\bF)_{\ell=1} |\les \min\{r^{-3-\de} u^{-1/2+\de}, r^{-2-\de} u^{-1+\de} \}\qquad  \text{in $\{ r\geq r_1\} \cap \{ u_0 \leq u \leq u_1 \}$}
\eea
\item Using \eqref{divv-ze-in-terms-of-bF-ka-l1} we deduce
\bea\label{decay-ze-from-initial-data}
|(\divv\ze)_{\ell=1}| &\les &  \min\{r^{-3-\de} u^{-1/2+\de}, r^{-2-\de} u^{-1+\de} \}\qquad  \text{in $\{ r\geq r_1\} \cap \{ u_0 \leq u \leq u_1 \}$}
\eea
\item From commuting \eqref{nabb-4-bbF-ze-eta} with $\divv$ we obtain
\beaa
\nabb_4 \divv\bbF+ \ka \divv\bbF&=&\DDd_1\DDs_1(\check{\rhoF}, -\check{\sigmaF})+2\rhoF \divv\ze 
\eeaa
Projecting the above to the $\ell=1$ spherical harmonics we have 
\beaa
\nabb_4( (\divv\bbF)_{\ell=1})+ \ka (\divv\bbF)_{\ell=1}&=&\frac{2}{r^2} \check{\rhoF}_{\ell=1}+2\rhoF (\divv\ze)_{\ell=1}
\eeaa
Using \eqref{decay-ze-from-initial-data} and \eqref{check-rhoF-in-terms-of-bF-bbF-l1} we simplify it to
\beaa
\nabb_4( (\divv\bbF)_{\ell=1})+ \ka (\divv\bbF)_{\ell=1}&=&\frac{2}{r^2} ( - \frac 1 2 r(\divv\bbF)_{\ell=1})+ O( r^{-3} u^{-1+\de}) 
\eeaa
which gives
\beaa
\nabb_4( (\divv\bbF)_{\ell=1})+\frac 3 2\ka (\divv\bbF)_{\ell=1}&=& O( r^{-3} u^{-1+\de})
\eeaa
Integrating the above equation as before from $\underline{C}_0\cap \{ r\geq r_1\} $ in the region where $r_0(u) \geq r_1$, using the initial condition \eqref{condition-divv-bbF-l1=initial=data} we obtain
\bea\label{decay-bbF-from-initial-data}
|(\divv\bbF)_{\ell=1} |\les  r^{-2} u^{-1+\de} \qquad  \text{in $\{ r\geq r_1\} \cap \{ u_0 \leq u \leq u_1 \}$}
\eea
\item The decay obtained for $\check{\ka}_{\ell=1}$, $(\divv\bF)_{\ell=1}$, $(\divv\bbF)_{\ell=1}$  allows to deduce the following decays in $\{ r\geq r_1\} \cap \{ u_0 \leq u \leq u_1 \}$ using Proposition \ref{expression-all-quantities-l1}:
\beaa
|(\divv \b)_{\ell=1}|&\les &\min\{r^{-4-\de}u^{-1/2+\de}, r^{-3-\de} u^{-1+\de} \}\\
|(\divv \bb)_{\ell=1}|&\les& r^{-3} u^{-1+\de}\\
 |\check{\kab}_{\ell=1}|&\les & r^{-1}u^{-1+\de}  \\
 |\check{\rho}_{\ell=1}|&\les& r^{-2} u^{-1+\de} \\
 |\check{\rhoF}_{\ell=1}|&\les& r^{-1} u^{-1+\de} 
\eeaa
\end{enumerate}

We now obtain decay for the curl part of the projection to the $\ell=1$ spherical harmonics. Observe that is in this part that the linearized Kerr-Newman solutions are supported. We derive general elliptic relations for the curl part in the following proposition, where we express all the relevant quantities in terms of the $(\curll \bF)_{\ell=1}$.

\begin{proposition}\label{expression-all-quantities-curl-l1} The following relations hold true for all $u \geq u_0$ and $r \geq r_1$:
\bea
(\curll \b)_{\ell=1}&=&\frac{3\rho}{2\rhoF} (\curll \bF)_{\ell=1}+O(r^{-4-\de}u^{-1/2+\de}, r^{-3-\de} u^{-1+\de} )  \label{curll-b-in-terms-of-bF-l1}\\
(\curll\ze)_{\ell=1} &=& r\left(\frac{3\rho}{2\rhoF} -\rhoF\right)(\curll\bF)_{\ell=1}+O(r^{-3-\de} u^{-1/2+\de}, r^{-2-\de} u^{-1+\de}) \label{curll-ze-in-terms-of-bF-ka-l1} \\
(\curll \bbF)_{\ell=1}&=&\frac 1 2 \kab r(\curll\bF)_{\ell=1}   +O(r^{-2} u^{-1+\de}) \label{curll-bbF-in-terms-of-bF-l1} \\
(\curll \bb)_{\ell=1}&=& \frac{3\rho}{4\rhoF} \kab r(\curll\bF)_{\ell=1}   +O(r^{-3} u^{-1+\de}) \label{curll-bb-in-terms-of-bbF-l1} \\
(\curl\eta)_{\ell=1}&=& r\left(\frac{3\rho}{2\rhoF} -\rhoF\right)(\curll\bF)_{\ell=1}+O(r^{-3-\de} u^{-1/2+\de}, r^{-2-\de} u^{-1+\de}) \label{curll-eta-in-terms-of-bF-ka-l1} \\
 \check{\sigma}_{\ell=1}&=&r\left(\frac{3\rho}{2\rhoF} -\rhoF\right)(\curll\bF)_{\ell=1}+O(r^{-3-\de} u^{-1/2+\de}, r^{-2-\de} u^{-1+\de}) \label{check-sigma-in-terms-of-bF-bbF-l1}\\
  \check{\sigmaF}_{\ell=1}&=&-r(\curll\bF)_{\ell=1}+O(r^{-1} u^{-1+\de})  \label{check-sigmaF-in-terms-of-bF-bbF-l1}
\eea
\end{proposition}
\begin{proof} By taking the $\curll$ of the definition of the gauge-invariant quantity $\tilde{\b}$ and using their estimates we obtain \eqref{curll-b-in-terms-of-bF-l1}.
Applying $\curll$ to Codazzi equation \eqref{Codazzi-chi-ze-eta} and projecting to the $\ell=1$ spherical harmonics, since $\curll\divv\chih_{\ell=1}=0$ and $\curll\DDs_1(\check{\ka}, 0)=0$ we obtain
\beaa
  (\curll\ze)_{\ell=1}&=&  r(\curll\b)_{\ell=1}-r\rhoF(\curll\bF)_{\ell=1}
  \eeaa
 Using \eqref{curll-b-in-terms-of-bF-l1}, we obtain \eqref{curll-ze-in-terms-of-bF-ka-l1}.
Applying $\curll$ to Codazzi equation \eqref{Codazzi-chib-ze-eta} and projecting to the $\ell=1$ spherical harmonics, we obtain
\beaa
0&=&-\frac 1 2 \kab (\curll\ze)_{\ell=1} +(\curll\bb)_{\ell=1}-\rhoF(\curll\bbF)_{\ell=1}
\eeaa
 Using the definition of the gauge-invariant quantity $\tilde{\bb}$ and its estimates and \eqref{curll-ze-in-terms-of-bF-ka-l1}, we obtain 
 \beaa
 0&=&-\frac 1 2 \kab r\left(\frac{3\rho}{2\rhoF} -\rhoF\right)(\curll\bF)_{\ell=1}  +\frac{3\rho}{2\rhoF} (\curll \bbF)_{\ell=1} -\rhoF(\curll\bbF)_{\ell=1}+O(r^{-3} u^{-1+\de})
 \eeaa
 which gives \eqref{curll-bbF-in-terms-of-bF-l1}.
  By taking the $\curll$ of the definition of the gauge-invariant quantity $\tilde{\bb}$ and using \eqref{curll-bbF-in-terms-of-bF-l1} we obtain \eqref{curll-bb-in-terms-of-bbF-l1}.
  
  Projecting \eqref{sigma-curl-ze-eta} and \eqref{curl-ze-eta} to the $\ell=1$ spherical harmonics, we obtain \eqref{curll-eta-in-terms-of-bF-ka-l1} and \eqref{check-sigma-in-terms-of-bF-bbF-l1}.

Commuting the expression for $\pf$ given by \eqref{relation-pf-rho-rhoF-bF-bbF}, with $\curll$ and projecting to the $\ell=1$ spherical harmonics we obtain
\beaa
\frac{(\curll\pf)_{\ell=1}}{r^5} &=& 2\rhoF \curll\DDs_1(-\check{\rho}, \check{\sigma})_{\ell=1} +(3\rho-2\rhoF^2) \curll\DDs_1(\check{\rhoF}, \check{\sigmaF})_{\ell=1}\\
&&+2 \rhoF^2 (\kab (\curll\bF)_{\ell=1}-\ka(\curll\bbF)_{\ell=1}) 
\eeaa
Using that $\curll\DDs_1(-\check{\rho}, \check{\sigma})_{\ell=1}=-\lapp \check{\sigma}_{\ell=1}=\frac{2}{r^2} \check{\sigma}_{\ell=1}$, and making use of the estimate \eqref{estimate-pf} for $\frak{p}$ and of the above relation for $\check{\sigma}_{\ell=1}$ and $\curll \bbF_{\ell=1}$ we have
\beaa
 2\rhoF(r\left(\frac{3\rho}{2\rhoF} -\rhoF\right)(\curll\bF)_{\ell=1}) +(3\rho-2\rhoF^2) \check{\sigmaF}_{\ell=1}&=&O(r^{-4} u^{-1+\de}) 
\eeaa
which gives \eqref{check-sigmaF-in-terms-of-bF-bbF-l1}.
\end{proof}

We now obtain control over $(\curll\bF)_{\ell=1}$ by using transport equations.
\begin{enumerate}
\item Recall that we have by \eqref{vanishing-l0-after-kerr} that $\check{\sigmaF}_{\ell=1}=0$ on $S_0$. We now consider the projection to the $\ell=1$ spherical harmonics of the equation \eqref{nabb-3-check-sigmaF-ze-eta}. According to Lemma \ref{commutation-projection-l1}, we obtain
\beaa
\nabb_3(\check{\sigmaF}_{\ell=1})+\kab\check{\sigmaF}_{\ell=1} &=& (\curll \bbF)_{\ell=1}
\eeaa
Using \eqref{curll-bbF-in-terms-of-bF-l1} and \eqref{check-sigmaF-in-terms-of-bF-bbF-l1} we obtain 
\beaa
\nabb_3(\check{\sigmaF}_{\ell=1})+\kab\check{\sigmaF}_{\ell=1} &=& \frac 1 2 \kab r\curll\bF_{\ell=1} +\mathscr{A}=-\frac 1 2 \kab \check{\sigmaF}_{\ell=1} +\mathscr{A}
\eeaa
where $\mathscr{A}$ is a gauge-invariant quantity\footnote{It is in fact an expression in terms of $\tilde{\b}$, see Proposition \ref{expression-all-quantities-curl-l1}. } that has the pointwise estimate $O(r^{-2} u^{-1+\de}) $ for $r \geq r_1$ and is bounded close to the horizon, as indicated in Proposition \ref{expression-all-quantities-curl-l1}.  The above equation reduces then to
 \beaa
 \nabb_3(r^3\check{\sigmaF}_{\ell=1})&=& r^3\mathscr{A}
 \eeaa
 According to the estimates obtained in Theorem \ref{estimates-theorem}, the gauge invariant set of quantities $\mathscr{A}$ also have a consistent $L^2$ estimates on spacelike hypersurfaces and along null hypersurfaces as in \eqref{L^2-estimates-theorem}. In particular, we have $\int_{\underline{C}_0 \cap \{ r\geq r_1\}} |\mathscr{A}|^2 \leq u^{-2+2\de}$.
 Integrating then the above equation over $\underline{C}_0$ from $S_0$ and using the vanishing of $\check{\sigmaF}$ on $S_0$ till $\{ r=r_1\}$, we can bound the right hand side by its $L^2$ norm and obtain decay in $u$ along $\underline{C}_0  \cap \{ r\geq r_1\} $\footnote{In order to have decay in $u$ the integration should be stopped away from the horizon, where $u$ degenerates. When integrating transport equations in the $e_4$ direction on an outgoing null line of very large $u$-coordinate, one picks up a boundary term close to the horizon which is not decaying as a function of $u$.} for $\check{\sigmaF}_{\ell=1}$, i.e.
 \beaa
 |\check{\sigmaF}_{\ell=1}|\les u^{-1+\de} \text{along $\underline{C}_0  \cap \{ r\geq r_1\}$}
 \eeaa
and boundedness of $\check{\sigmaF}_{\ell=1}$ along $\underline{C}_0$. 
 \item By \eqref{check-sigmaF-in-terms-of-bF-bbF-l1} restricted to $\underline{C}_0$ we obtain $|\curll\bF_{l=1}|\les u^{-1+\de}$ along $\underline{C}_0\cap \{ r\geq r_1\}$ and bounded along $\underline{C}_0$.  From equation \eqref{nabb-4-bF-tilde-b} commuted with $\curll$ and the estimates \eqref{estimate-tilde-b} for $\tilde{\b}$ we have
\beaa
\nabb_4 \curll\bF+ 2 \ka\curll\bF&=&O(r^{-4-\de}u^{-1/2+\de}, r^{-3-\de}u^{-1+\de}) \qquad \text{for $r \geq r_1$}
\eeaa
Projecting the equation to the $\ell=1$ spherical harmonics we obtain
\beaa
\nabb_4 ((\curll\bF)_{\ell=1})+ 2\ka(\curll\bF)_{\ell=1}&=&O(r^{-4-\de}u^{-1/2+\de}, r^{-3-\de}u^{-1+\de})
\eeaa
Integrating the above equation from $\underline{C}_0\cap \{ r\geq r_1\} $, we obtain
\bea\label{decay-curll-bF-from-initial-data}
|(\curll\bF)_{\ell=1} |\les  \min\{r^{-3-\de} u^{-1/2+\de}, r^{-2-\de} u^{-1+\de}\}  \qquad  \text{in $\{ r\geq r_1\} \cap \{ u_0 \leq u \leq u_1 \}$}
\eea
\item Using \eqref{decay-curll-bF-from-initial-data} and Proposition \ref{expression-all-quantities-curl-l1} we obtain in $\{ r\geq r_1\} \cap \{ u_0 \leq u \leq u_1 \}$:
\bea
|(\curll \b)_{\ell=1}|&\les& \min\{r^{-4-\de} u^{-1/2+\de}, r^{-3-\de} u^{-1+\de}\}  \label{decay-curl-b-l1}\\
|(\curll\ze)_{\ell=1}| &\les&   \min\{r^{-3-\de} u^{-1/2+\de}, r^{-2-\de} u^{-1+\de}\}  \label{decay-curl-ze-l1} \\
|(\curll \bbF)_{\ell=1}|&\les& r^{-2} u^{-1+\de} \\
|(\curll \bb)_{\ell=1}|&\les& r^{-3} u^{-1+\de} \\
|(\curl\eta)_{\ell=1}|&\les&  r^{-2} u^{-1+\de} \\
 |\check{\sigma}_{\ell=1}|&\les& r^{-2} u^{-1+\de}\label{decay-check-sigma-l1-optimal-u}\\
  |\check{\sigmaF}_{\ell=1}|&\les&r^{-1} u^{-1+\de}\label{decay-check-sigmaF-l1-optimal-u}
\eea
\item Using \eqref{decay-curll-bF-from-initial-data} and \eqref{decay-curl-b-l1} and integrating \eqref{nabb-4-check-sigmaF-ze-eta} and \eqref{nabb-4-check-sigma-ze-eta}, we have in $\{ r\geq r_1\} \cap \{ u_0 \leq u \leq u_1 \}$:
\bea
|\check{\sigma}_{\ell=1}|&\les& r^{-3} u^{-1/2+\de}\label{decay-check-sigma-l1-optimal-r}\\
  |\check{\sigmaF}_{\ell=1}|&\les&r^{-2} u^{-1/2+\de}\label{decay-check-sigmaF-l1-optimal-r}
\eea

\end{enumerate}

\subsubsection{The projection to the $\ell\geq 2$ modes}\label{proj-l2=boundedness}
To obtain decay for the projection to the $\ell\geq 2$ spherical modes, we make use of the decay obtained in Theorem \ref{estimates-theorem} for the gauge-invariant quantities $\ff$, $\underline{\ff}$, $\tilde{\b}$, $\tilde{\bb}$, $\qf^\F$, $\frak{p}$. In particular, we can express all the remaining quantities in terms of only $\chih$, $\chibh$ and $\DDs_2\DDs_1(\check{\ka}, 0)$ and the gauge-invariant quantities already estimated.

As for the $\ell=1$ projection, we summarize those relations in the following Proposition. We only use elliptic relations, and not transport equations which will be exploited later.

\begin{proposition}\label{expression-all-quantities-l2} The following relations hold true for all $u \geq u_0$ and $r\geq r_1$:
\bea
\DDs_2\bF&=& -\rhoF \chih+O( r^{-3-\de} u^{-1/2+\de}, r^{-2-\de} u^{-1+\de})\label{write-bF-in-terms-of-chih-general}\\
\DDs_2\bbF&=& \rhoF \chibh +O( r^{-2} u^{-1+\de}) \label{write-bbF-in-terms-of-chibh-general} \\
\DDs_2\b&=& -\frac 3 2\rho \chih+ O(r^{-4-\de} u^{-1/2+\de}, r^{-3-\de} u^{-1+\de} )\label{write-b-in-terms-of-chih-general}\\
\DDs_2\bb&=& \frac 3 2\rho \chibh+O( r^{-3} u^{-1+\de} ) \label{write-bb-in-terms-of-chibh-general}
\eea
\bea
\begin{split}
\DDs_2\ze&= r\left(\DDs_2\DDd_2 \chih+\left(-\frac 3 2\rho+\rhoF^2\right) \chih\right)+\frac 1 2 r \DDs_2\DDs_1(\check{\ka}, 0)+  O(r^{-3-\de} u^{-1/2+\de}, r^{-2-\de} u^{-1+\de} ) \label{write-ze-in-terms-of-chih-general} 
\end{split}
\eea
\bea
\DDs_2\DDs_1(\check{\rhoF}, \check{\sigmaF})&=&  -\frac 12 \rhoF \left(  \kab \chih + \ka \chibh \right)+O(r^{-4} u^{-1/2+\de}, r^{-3} u^{-1+\de}) \label{write-rhoF-in-terms-of-chih-chibh-general} \\
 \DDs_2\DDs_1(-\check{\rho}, \check{\sigma}) &=&\left(\frac 3 4 \rho+\frac 1 2 \rhoF^2\right)\left(  \kab \chih + \ka \chibh \right) +O( r^{-4} u^{-1+\de}) \label{check-rho-in-terms-of=chih-chib-general}
 \eea
\bea\label{check-kab-in-terms-of=chih-chib-general}
\begin{split}
 \DDs_2\DDs_1( \check{\kab}, 0) &=-2\left(\DDs_2\DDd_2 \chibh+ \left(-\frac 3 2 \rho +\rhoF^2\right) \chibh\right)- r\kab \left(\DDs_2\DDd_2 \chih+\left(-\frac 3 2\rho+\rhoF^2\right) \chih\right)\\
 & -\frac 1 2 r \kab \DDs_2\DDs_1(\check{\ka}, 0)+O( r^{-3} u^{-1+\de})
 \end{split}
\eea
Observe that the terms in $O$ are given by gauge-invariant quantities, and are therefore independent of the gauge. 
\end{proposition}
\begin{proof}  Recall the definition of $\ff$ and $\underline{\ff}$ as defined in \eqref{definition-ff}. We can then express $\bF$ and $\bbF$ in terms of $\chih$ and $\chibh$ respectively:
\beaa
\DDs_2\bF&=& -\rhoF \chih+\ff, \qquad \DDs_2\bbF= \rhoF \chibh +\underline{\ff} 
\eeaa
Using the estimates for $\ff$ and $\underline{\ff}$ as in \eqref{estimate-ff} and \eqref{estimate-underline-ff}, we obtain \eqref{write-bF-in-terms-of-chih-general} and \eqref{write-bbF-in-terms-of-chibh-general}.

Using the definition of $\tilde{\b}$ and $\tilde{\bb}$ and the above relations to express $\b$ and $\bb$ in terms of $\chih$ and $\chibh$, we obtain
\beaa
2\rhoF \DDs_2\b&=& \DDs_2\tilde{\b}+3\rho \DDs_2\bF= \DDs_2\tilde{\b}+3\rho (-\rhoF \chih+\ff)\\
2\rhoF \DDs_2\bb&=& \DDs_2\tilde{\bb}+3\rho \DDs_2\bbF=\DDs_2\tilde{\bb}+3\rho (\rhoF \chibh +\underline{\ff})
\eeaa
which gives
\beaa
\DDs_2\b&=& -\frac 3 2\rho \chih+ \rhoF^{-1}(\frac 1 2 \DDs_2\tilde{\b}+\frac 3 2\rho\ff)\\
\DDs_2\bb&=& \frac 3 2\rho \chibh+\rhoF^{-1}(\frac 1 2 \DDs_2\tilde{\bb} +\frac 3 2\rho\underline{\ff})
\eeaa
Using the estimates for $\tilde{\b}$, $\tilde{\bb}$ as in \eqref{estimate-underline-ff}, \eqref{estimate-tilde-b}, \eqref{estimate-tilde-bb}, we obtain \eqref{write-b-in-terms-of-chih-general} and \eqref{write-bb-in-terms-of-chibh-general}.

Using the Codazzi equation \eqref{Codazzi-chi-ze-eta}, we express $\ze$ in terms of $\check{\ka}$ and $\chih$. Indeed we obtain
\beaa
 \ze&=& r\DDd_2 \chih+\frac 1 2 r \DDs_1(\check{\ka}, 0)+ r\b -r\rhoF\bF 
\eeaa
Applying the operator $\DDs_2$ to the above and using \eqref{write-bF-in-terms-of-chih-general} and \eqref{write-b-in-terms-of-chih-general} we obtain
\beaa
\DDs_2\ze&=& r\DDs_2\DDd_2 \chih+\frac 1 2 r \DDs_2\DDs_1(\check{\ka}, 0)+ r\DDs_2\b -r\rhoF\DDs_2\bF \\
&=& r\DDs_2\DDd_2 \chih+\frac 1 2 r \DDs_2\DDs_1(\check{\ka}, 0)+ r(-\frac 3 2\rho \chih+ \min\{r^{-4-\de} u^{-1/2+\de}, r^{-3-\de} u^{-1+\de} \})\\
&& -r\rhoF(-\rhoF \chih+\min\{ r^{-3-\de} u^{-1/2+\de}, r^{-2-\de} u^{-1+\de}\})
\eeaa
which gives \eqref{write-ze-in-terms-of-chih-general}.

We use the alternative expression for $\qf^\F$ given by \eqref{relation-qf-rhoF-chih-chibh} to express $\check{\rhoF}$ and $\check{\sigmaF}$ in terms of $\chih$ and $\chibh$:
\beaa
\DDs_2\DDs_1(\check{\rhoF}, \check{\sigmaF})&=&  -\frac 12 \rhoF \left(  \kab \chih + \ka \chibh \right)-\frac{\qf^\F}{r^3}
\eeaa
Using the estimate for $\qf^\F$ given by \eqref{estimate-qf-F}, we obtain \eqref{write-rhoF-in-terms-of-chih-chibh-general}.

We use the alternative expression for $\pf$ given by \eqref{relation-pf-rho-rhoF-bF-bbF}, and the above relations, to express $\check{\rho}$ and $\check{\sigma}$ in terms of $\chih$ and $\chibh$. Using the estimate for $\pf$ given by \eqref{estimate-pf} we have
\beaa
 2\rhoF \DDs_1(-\check{\rho}, \check{\sigma}) &=&(-3\rho+2\rhoF^2) \DDs_1(\check{\rhoF}, \check{\sigmaF})+2 \rhoF^2 (\ka\bbF-\kab\bF) +\min\{r^{-6} u^{-1/2+\de}, r^{-5} u^{-1+\de}\}
\eeaa
Applying the operator $\DDs_2$ to the above and using \eqref{write-rhoF-in-terms-of-chih-chibh-general}, \eqref{write-bF-in-terms-of-chih-general} and \eqref{write-bbF-in-terms-of-chibh-general} we obtain
\beaa
 &&2\rhoF \DDs_2\DDs_1(-\check{\rho}, \check{\sigma}) =(-3\rho+2\rhoF^2) \DDs_2\DDs_1(\check{\rhoF}, \check{\sigmaF})+2 \rhoF^2 (\ka \DDs_2\bbF-\kab \DDs_2\bF) \\
 && +\min\{r^{-7} u^{-1/2+\de}, r^{-6} u^{-1+\de}\}\\
 &=&(-3\rho+2\rhoF^2)\left( -\frac 12 \rhoF \left(  \kab \chih + \ka \chibh \right)+\min\{r^{-4} u^{-1/2+\de}, r^{-3} u^{-1+\de}\}\right)\\
 &&+2 \rhoF^2 (\ka (\rhoF \chibh +O( r^{-2} u^{-1+\de}))-\kab (-\rhoF \chih+\min\{r^{-3-\de} u^{-1/2+\de}, r^{-2-\de} u^{-1+\de}\})) \\
&&+\min\{r^{-7} u^{-1/2+\de}, r^{-6} u^{-1+\de}\}
\eeaa
which gives \eqref{check-rho-in-terms-of=chih-chib-general}.

Using \eqref{Codazzi-chib-ze-eta}, we express $\check{\kab}$ in terms of $\chih$ and $\chibh$. Applying the operator $\DDs_2$ to \eqref{Codazzi-chib-ze-eta} and using \eqref{write-ze-in-terms-of-chih-general}, \eqref{write-bb-in-terms-of-chibh-general} and \eqref{write-bbF-in-terms-of-chibh-general}, we have
\beaa
&& \DDs_2\DDs_1( \check{\kab}, 0)=-2\DDs_2\DDd_2 \chibh- \kab \DDs_2\ze +2\DDs_2\bb-2\rhoF\DDs_2\bbF\\
 &=&-2\DDs_2\DDd_2 \chibh\\
 &&- \kab \left(r\left(\DDs_2\DDd_2 \chih+\left(-\frac 3 2\rho+\rhoF^2\right) \chih\right)+\frac 1 2 r \DDs_2\DDs_1(\check{\ka}, 0)+  \min\{r^{-3-\de} u^{-1/2+\de}, r^{-2-\de} u^{-1+\de} \}\right)\\
 &&+2\left( \frac 3 2\rho \chibh+O( r^{-3} u^{-1+\de} )\right)-2\rhoF\left(\rhoF \chibh +O( r^{-2} u^{-1+\de}) \right)
\eeaa
which gives \eqref{check-kab-in-terms-of=chih-chib-general}.
\end{proof}

We now obtain control over $\DDs_2\DDs_1(\check{\ka}, 0)$, $\chih$ and $\chibh$ by using transport equations.

\begin{enumerate}
\item Commuting \eqref{nabb-4-check-ka-ze-eta} with $\DDs_2\DDs_1$ we obtain 
\beaa
\nabb_4(\DDs_2\DDs_1(\check{\ka}, 0))+2\ka \DDs_2\DDs_1(\check{\ka}, 0)&=& 0
\eeaa
Together with condition \eqref{condition-check-ka-l2=initial=data}, this implies 
\bea\label{check-ka-l2=0-initial-data-everywhere}
\DDs_2\DDs_1(\check{\ka}, 0)=0 \qquad \text{in $\{ r\geq r_1\} \cap \{ u_0 \leq u \leq u_1 \}$}
\eea
\item From equation \eqref{nabb-4-chih-ze-eta} and the estimate \eqref{estimate-a} for $\a$ we have
\beaa
\nabb_4\chih+\ka\chih&=& O(r^{-3-\de}u^{-1/2+\de}, r^{-2-\de} u^{-1+\de})
\eeaa
We now integrate the above equation from $\underline{C}_0\cap \{ r\geq r_1\} $ using the condition on the initial data \eqref{condition-chih-initial-data}. We obtain
\beaa
r^2 \chih =  \int_{r_0(u)}^r  O(r^{-1-\de}u^{-1/2+\de}, r^{-\de} u^{-1+\de}) d \tilde{r}
\eeaa
where $r_0(u)$ is the $r$-coordinate of the intersection of $\underline{C}_0$ and the cone of constant $u$ coordinate. Recall that the pointwise estimate on $\a$, which constitutes the right hand side of the above equation, is only valid in the region $r \geq r_1$. In particular, we make use of the above integration only in the region\footnote{Recall from Theorem \ref{estimates-theorem}, point 5, that $\alpha$ cannot be estimated in terms of an inverse power of $u$ in the region $r \leq r_1$. On other hand the estimates in this region near the horizon are obtained in terms of the coordinate $v$ as in \cite{stabilitySchwarzschild}, while the $r$ decay is irrelevant (since $r$ is bounded). See also Section \ref{section-close-horizon}.} given by the intersection of $u \geq u_0$ and $r \geq r_1$ (or also, only for $r_0(u) \geq r_1$). 
This implies
\bea
|\chih|\les \min\{r^{-2-\de}u^{-1/2+\de}, r^{-1-\de} u^{-1+\de} \}\qquad \text{in $\{ r\geq r_1\} \cap \{ u_0 \leq u \leq u_1 \}$}
\eea
\item Using \eqref{write-ze-in-terms-of-chih-general} we deduce
\beaa
\DDs_2\ze&=&\min\{r^{-3-\de} u^{-1/2+\de}, r^{-2-\de} u^{-1+\de} \} \qquad \text{in $\{ r\geq r_1\} \cap \{ u_0 \leq u \leq u_1 \}$}
\eeaa
\item From equation \eqref{nabb-4-chibh-ze-eta} and the previous estimates we obtain
\beaa
\nabb_4\chibh+\frac 1 2 \ka \chibh&=& O(r^{-3-\de} u^{-1/2+\de}, r^{-2-\de} u^{-1+\de} )
\eeaa
Integrating the above equation from $\underline{C}_0 \cap \{r \geq r_1\}$ and using condition \eqref{condition-chibh-initial-data} on the initial data as above, we obtain
\bea
|\chibh|\les  r^{-1} u^{-1+\de} \qquad \text{in $\{ r\geq r_1\} \cap \{ u_0 \leq u \leq u_1 \}$}
\eea
\item The decay obtained for $\DDs_2\DDs_1(\check{\ka}, 0)$, $\chih$ and $\chibh$ allows to deduce the following decays in $\{ r\geq r_1\} \cap \{ u_0 \leq u \leq u_1 \}$ using Proposition \ref{expression-all-quantities-l2}:
\beaa
|\DDs_2\bF|&\les& \min\{ r^{-3-\de} u^{-1/2+\de}, r^{-2-\de} u^{-1+\de}\}\\
|\DDs_2\bbF|&\les&  r^{-2} u^{-1+\de}  \\
|\DDs_2\b|&\les&  \min\{r^{-4-\de} u^{-1/2+\de}, r^{-3-\de} u^{-1+\de} \}\\
|\DDs_2\bb|&\les&  r^{-3} u^{-1+\de}  \\
|\DDs_2\DDs_1(\check{\rhoF}, \check{\sigmaF})|&\les&  r^{-3} u^{-1+\de}\\
 |\DDs_2\DDs_1(-\check{\rho}, \check{\sigma}) |&\les& r^{-4} u^{-1+\de}\\
 |\DDs_2\DDs_1( \check{\kab}, 0)| &\les&  r^{-3} u^{-1+\de}
\eeaa
\end{enumerate}
Combining the estimates for the projection to the $\ell=1$ spherical harmonics and the estimates for the above using elliptic estimates as in Lemma \ref{main-elliptic-estimate}, we in particular obtain
\bea
 |\check{\kab}| &\les&  r^{-1} u^{-1+\de} \label{estimate-check-kab-initia}\\
 |\check{K}|&\les& r^{-2} u^{-1+\de} \label{estimate-check-K-initial}
\eea 
where we obtain decay for $\check{K}$ using Gauss equation \eqref{Gauss-check}.

\begin{remark}\label{replace-uniform} As pointed out in Remark \ref{remark-not-uniform-r}, the bounds here obtained are not uniform in $r_1$. Nevertheless, one can perform the same above procedure, integrating forward in the $e_4=\partial_r$ direction, and replace all the above estimates by the ones with only $r$ decay. For example, the bound \eqref{estimate-check-kab-initia} can be replaced by 
\beaa
 |\check{\kab}| &\les&  r^{-1}
\eeaa
where the constant implicit in $\les$ is now uniform in $r_1$. This uniform decay is used to deduce boundedness and apply Theorem  \ref{achieving-normalization-theorem} to initialize the $S_{U,R}$-normalization.
\end{remark}

\subsubsection{The terms involved in the $e_3$ direction}\label{terms-e3-boundedness}

We derive here boundedness for the terms involved in the $e_3$ directions, i.e. for $\eta$, $\xib$ and  $\check{\omb}$.

\begin{enumerate}
\item  Restricting \eqref{nabb-3-bF-ze-eta} to $\underline{C}_0$ implies that $\eta$ is bounded along $\underline{C}_0$. Integrating \eqref{nabb-4-eta-ze-eta} along $e_4$ we obtain 
\bea
| \eta|\les r^{-1}  \qquad \text{in $\{ r\geq r_1\} \cap \{ u_0 \leq u \leq u_1 \}$}
\eea
\item Restricting \eqref{nabb-3-check-ka-ze-eta} and \eqref{nabb-3-check-kab-ze-eta} to $\underline{C}_0$ we obtain that $\check{\omb}$, $\check{\Omegab}$ and $ \xib$ are bounded. Integrating \eqref{nabb-4-check-omb-ze-eta} along $e_4$ we obtain 
\bea
| \check{\omb}|\les  C \qquad \text{in $\{ r\geq r_1\} \cap \{ u_0 \leq u \leq u_1 \}$}
\eea
\item Integrating \eqref{nabb-4-xib-ze-eta} along $e_4$ we obtain
\bea
| \xib|\les  C \qquad \text{in $\{ r\geq r_1\} \cap \{ u_0 \leq u \leq u_1 \}$}
\eea
\end{enumerate}

\begin{remark}\label{result-incomplete} Observe that the decays hereby obtained are significantly worse than the optimal decay we are aiming to prove, as stated in Theorem \ref{linear-stability}. In particular, $\xib$ and $\check{\omb}$ do not present decay in $r$. Because of this loss of decay in integrating the equations from the initial data for the quantities $\check{\omb}$, $\xib$ and $\check{\Omegab}$, we will use a different approach through the $S_{U, R}$-normalization to prove the optimal decay. From the initial data normalization we will make use of the boundedness of the solution to apply Theorem \ref{achieving-normalization-theorem} later. 
\end{remark}

\subsubsection{The metric coefficients}
We derive here boundedness for the metric coefficients $\hat{\slashed{g}}$, $\underline{b}$, $\check{\Omegab}$ and $\check{\vsi}$.

\begin{enumerate}
\item By condition \eqref{condition-chibh-initial-data} and \eqref{divv-b-0-initial}, equation \eqref{nabb-3-g} restricted to $\underline{C}_0$ reads
\beaa
\nabb_3 \hat{\slashed{g}}&=& 0
\eeaa
Using condition \eqref{condition-hat-slashed-g}, and integrating the above along $\underline{C}_0$, we obtain that $\hat{\slashed{g}}=0$ on $\underline{C}_0$. By integrating \eqref{nabb-4-g} from $\underline{C}_0\cap \{ r\geq r_1\} $ we obtain
\beaa
\hat{\slashed{g}}(u, r) = 2 \int_{r_0(u)}^r \chih(u, \tilde{r}) d \tilde{r}
\eeaa
Here, $r_0(u)$ is the $r$-coordinate of the intersection of $\underline{C}_0$ and the cone of constant $u$ coordinate. Recall that the pointwise estimates on $\chih$ in terms of inverse powers of $u$ are only valid in the region $r \geq r_1$. In particular, we make use of the above integration only in the region\footnote{This is because the term $2 \int_{r_0(u)}^{r_1} \chih(u, \tilde{r}) d \tilde{r}$ cannot be estimates in terms of an inverse power of $u$, since $\chih$ does not decay as an inverse power or $u$ in this region. The estimates on this region are obtained in terms of the coordinate $v$ as in \cite{stabilitySchwarzschild}. See also Section \ref{section-close-horizon}.} given by the intersection of $u \geq u_0$ and $r \geq r_1$ (or also, only for $r_0(u) \geq r_1$). In this region, the pointwise estimate for $\chih$ can be used, and it implies
\beaa
|\hat{\slashed{g}}|\les  u^{-1+\de} \qquad \text{in $\{ r\geq r_1\} \cap \{ u_0 \leq u \leq u_1 \}$}
\eeaa
\item Using \eqref{check-K-tr} and the estimate \eqref{estimate-check-K-initial} we can estimate the projection to the $\ell\geq2$ spherical harmonics of $\check{\tr_\gamma \slashed{g}}$:
\bea\label{estimate-tr-l2-initial}
| \DDs_2\DDs_1(\check{\tr_\gamma \slashed{g}}, 0)| &\les&  r^{-2}u^{-1+\de}\qquad \text{in $\{ r\geq r_1\} \cap \{ u_0 \leq u \leq u_1 \}$}
\eea
On the other hand, integrating \eqref{nabb-3-check-tr} along $\underline{C}_0$ and using condition \eqref{tr-slashed-g-l1-initial-data} we obtain
\beaa
|\check{\tr_\gamma \slashed{g}}_{\ell=1}|\les  C \qquad \text{along $\underline{C}_0$}
\eeaa
Consequently we have that $|\check{\tr_\gamma \slashed{g}}|\les  C $ along $\underline{C}_0$. Integrating \eqref{nabb-4-check-tr} we then obtain
\bea\label{estimate-tr-l1-initial}
|\check{\tr_\gamma \slashed{g}}|\les  r^2  \qquad \text{in $\{ r\geq r_1\} \cap \{ u_0 \leq u \leq u_1 \}$}
\eea

\item Condition \eqref{curll-b-initial} and the boundedness for $\check{\sigma}_{\ell=1}$ implies that $\curll\underline{b}_{\ell=1}$ is bounded along $\underline{C}_0$. Combining this with conditions \eqref{divv-b-0-initial} and standard elliptic estimates we obtain that $\underline{b}$ is bounded  along $\underline{C}_0$. Integrating \eqref{derivatives-b-ze-eta}, we obtain
\bea\label{estimate-underline-b-initial-data}
|\underline{b}|\les r  \qquad \text{in $\{ r\geq r_1\} \cap \{ u_0 \leq u \leq u_1 \}$}
\eea
\item  Equation \eqref{nabb-check-vsi} implies 
\beaa
|\check{\vsi}|\les C \qquad \text{in $\{ r\geq r_1\} \cap \{ u_0 \leq u \leq u_1 \}$}
\eeaa
\item Equation \eqref{xib-Omegab-ze-eta} implies
\bea\label{estimate-check-Omegab-initial-data}
| \check{\Omegab}|\les r  \qquad \text{in $\{ r\geq r_1\} \cap \{ u_0 \leq u \leq u_1 \}$}
\eea
\end{enumerate}

\begin{remark}\label{result-incomplete-2} Observe that the decays obtained for the metric coefficients are also significantly worse than the optimal decay we are aiming to prove. 
\end{remark}
The growth in $r$ and the non optimal decay for many components forces us to consider a different normalization in order to obtain the optimal decay stated in Theorem \ref{linear-stability}.

\subsection{Decay of the pure gauge solution $\mathscr{G}^{U, R}$}\label{decay-pure-gauge}
In the previous section, we provide bounds for the solution $\mathscr{S}^{id, K}$ in the entire exterior region. Fix $U$, $R$ with $R \gg \max\{r_0, 3M\}$ and $U \gg u_0$. 
Consider the sphere of intersection of $\{ u=U\} $ and the ingoing initial data $\underline{C}_0$ and let $\tilde{r}$ be the $r$-coordinate of such sphere. By construction $r_{\mathcal{H}} < \tilde{r} < r_0$. Let $r_1$ be such that $r_{\mathcal{H}} < r_1 < \tilde{r}$. Observe that by construction, $r_1$ depends on $U$ and $R$. 

By Proposition \ref{prop-boundedness}, the solution $\mathscr{S}^{id, K}$ is bounded in the region $\{ r\geq r_1\} \cap \{ u_0 \leq u \leq u_1 \}$ and by construction the sphere $S_{U, R}$ belongs to this region. 
Therefore, by Theorem \ref{achieving-normalization-theorem}, we can associate to $\mathscr{S}^{id, K}$ a $S_{U, R}$-normalized solution 
\beaa
\mathscr{S}^{U, R}:=\mathscr{S}^{id, K}-\mathscr{G}^{U, R}
\eeaa

Observe that according to Theorem \ref{achieving-normalization-theorem}, the pure gauge solution $\mathscr{G}^{U, R}$ is supported in $\ell\geq 1$ spherical harmonics, therefore the projection to the $\ell=0$ spherical harmonics of $\mathscr{S}^{U, R}$ still vanishes. 

We derive here decay statements for the functions $h, \underline{h}, a, z, \ov{z}$ constructed in the proof of Theorem \ref{achieving-normalization-theorem}. 
From \eqref{operator-E-z}, we have on $\mathscr{I}_{U, R}$, using \eqref{estimate-check-kab-initia}
\beaa
\mathcal{E} (\DDs_2\DDs_1(z, 0))&=& \kab \DDs_2\DDs_1(\check{\ka}[\mathscr{S}^{id}], 0)+\ka \DDs_2\DDs_1(\check{\kab}[\mathscr{S}^{id}], 0)\les r^{-4} u^{-1+\de}
\eeaa
By Lemma \ref{first-poincare-inequality-2-tensor}, we have on $\mathscr{I}_{U, R}$
\beaa
|\DDs_2\DDs_1(z, 0)|&\les& r^{-2} u^{-1+\de}
\eeaa
Similarly from \eqref{derive-a-from-last-slice} and \eqref{derive-ov-z-last-slice}, we obtain
\beaa
|\DDs_2\DDs_1(a, 0)|&\les& r^{-2} u^{-1+\de}, \qquad |\DDs_2\DDs_1(\ov{z}, 0)|\les r^{-2} u^{-1+\de}
\eeaa
which implies on $\mathscr{I}_{U, R}$:
\bea
|\DDs_2\DDs_1(h, 0)|\les r^{-1} u^{-1+\de} \qquad |\DDs_2\DDs_1(\underline{h}, 0)|\les r^{-1} u^{-1+\de} 
\eea
Integrating \eqref{transport-h}, \eqref{transport-a} and \eqref{transport-fb} we see that those bounds hold in the whole spacetime. Similarly for the projection to the $\ell=1$ spherical harmonics.  This implies decay for all the components of the gauge solution, at a rate which is consistent with the decay for $\mathscr{S}^{id, K}$. 

\section{Proof of linear stability: decay}\label{decay-chapter}

In this section, we prove decay in $r$ and $u$ of a linear perturbation of Reissner-Nordstr\"om spacetime $\mathscr{S}$ to the sum of a pure gauge solution and a linearized Kerr-Newman solution.

In Section \ref{statement-theorem} we give an outline of the proof. 
 In Section \ref{normalization-far-away-decay} we prove decay of the solution along the null hypersurface $\mathscr{I}_{U, R}$ making use of the $S_{U, R}$-normalization, and finally in Section \ref{final-decay} we prove the optimal decay as stated in the Theorem.

\subsection{Outline of the proof of Theorem \ref{linear-stability}}\label{statement-theorem}

Recall the statement of Theorem  \ref{linear-stability} in Section \ref{statement-section}. We outline here the proof.
\begin{enumerate}
\item {\bf{Achieving $S_{U, R}$-normalization:}} By Proposition \ref{prop-boundedness}, we have that the solution $\mathscr{S}^{id, K}$ as defined in \eqref{definition-S-id-K} verifies explicit bounds in $r$ and $u$ in the entire exterior region and by Theorem \ref{achieving-normalization-theorem}, we can associate to $\mathscr{S}^{id, K}$ a $S_{U, R}$-normalized solution $$\mathscr{S}^{U, R}:=\mathscr{S}^{id, K}-\mathscr{G}^{U, R}=\mathscr{S}-(\mathscr{G}^{id}+\mathscr{G}^{U, R})-\mathscr{K}^{id}$$ Observe that according to Theorem \ref{achieving-normalization-theorem}, the pure gauge solution $\mathscr{G}^{U, R}$ is supported in $\ell\geq 1$ spherical harmonics, therefore the projection to the $\ell=0$ spherical harmonics of $\mathscr{S}^{U, R}$ still vanishes. 
\item {\bf{Decay along the null hypersurface $\mathscr{I}_{U, R}$:}} By definition, the solution $\mathscr{S}^{U, R}$ verifies the conditions of $S_{U, R}$-normalized solution, which hold along the null hypersurface $\mathscr{I}_{U, R}$. These conditions imply elliptic relations along $\mathscr{I}_{U, R}$ which, together with the elliptic relations implied by the decay of the gauge-invariant quantities as summarized in Proposition \ref{expression-all-quantities-l1} and \ref{expression-all-quantities-l2}, imply decay for all quantities along $\mathscr{I}_{U, R}$.  

Recall the charge aspect function $\check{\nu}$ (defined in \eqref{definition-of-nu}) and the mass-charge aspect function $\check{\mu}$ (defined in \eqref{definition-of-mu}).
The gauge condition for the $\ell=1$ spherical harmonics along the null hypersurface $\mathscr{I}_{U, R}$ imposes the vanishing of $\check{\nu}_{\ell=1}$. Such condition implies a better decay for $\check{\mu}_{\ell=1}$ along the null hypersurface $\mathscr{I}_{U, R}$. In particular, the quantity  $\check{\mu}_{\ell=1}$, which in principle would not be bounded in $r$ along  $\mathscr{I}_{U, R}$, is implied to be bounded by the condition for $\nu$, giving 
\bea\label{improved-mu}
\check{\mu}_{\ell=1} \les u^{-1+\de} \qquad \text{along the null hypersurface $\mathscr{I}_{U, R}$}
\eea
On the other hand, for the projection to the $\ell\geq 2$ spherical harmonics, the gauge conditions impose the vanishing of $\check{\mu}_{\ell\geq 2}$ along the null hypersurface $\mathscr{I}_{U, R}$.

 Finally, the condition $\check{\kab}=0$ is necessary to obtain decay in $r$ for the quantities involved in the $e_3$ direction, namely $\check{\omb}$ and $\xib$. Observe that also the improved decay of $\check{\mu}_{\ell=1} $ \eqref{improved-mu} is necessary to obtain the decay in $r$ for $\check{\omb}_{\ell=1}$ and $\xib_{\ell=1}$.
 This decay is crucial to improve the result obtained in the previous section about boundedness.

We show decay in $u$ and $r$ for all components of $\mathscr{S}^{U, R}$ along the null hypersurface $\mathscr{I}_{U, R}$. The case $\ell=1$ and $\ell\geq 2$ spherical harmonics will be treated separately. This is done in Section \ref{normalization-far-away-decay}.

\item {\bf{Optimal decay in $r$ and $u$:}} Once obtained the decay for all the components along $\mathscr{I}_{U, R}$, we integrate transport equations from $\mathscr{I}_{U, R}$ towards the past to obtain decay in the past of $S_{U, R}$. We derive the optimal decay in $r$ as given by the Theorem, with $u^{-1/2+\de}$ decay. The case $\ell=1$ and $\ell\geq 2$ spherical harmonics will be treated separately. This is done in Section \ref{mega-section-optimal-decay-r}. Obtaining the optimal decay in $u$ is more delicate, because only one transport equation as given in the Einstein-Maxwell equations is integrable from $\mathscr{I}_{U, R}$ to obtain decay in $u^{-1+\de}$, i.e. the equation for $\check{\ka}$. Imposing the vanishing of $\check{\ka}$ at the hypersurface $\mathscr{I}_{U, R}$ then implies the vanishing of it everywhere. This condition by itself is not enough to obtain optimal decay for all the remaining quantities.  It is therefore crucial to make use of quantities which verify integrable transport equations along the $e_4$ direction. 

The charge aspect function $\check{\nu}$, the mass-charge aspect function $\check{\mu}$ and a new quantity $\Xi$ which is a $2$-tensor (defined in Lemma \ref{estimate-for-Gamma}) satisfy the following:
\bea
\nabb_4(\check{\nu}_{\ell=1})&=& 0 \label{eq=int-1}\\
\nabb_4(\check{\mu})&=& O(r^{-1-\de} u^{-1+\de}) \label{eq=int-2}\\
\nabb_4(\Xi)&=& O(r^{-1-\de} u^{-1+\de}) \label{eq=int-3}
\eea
which all have integrable (in $r$) right hand side.

The first two equations will be used in the projection to the $\ell=1$ spherical harmonics and the last two equations in the projection to the $\ell\geq 2$ spherical harmonics. Observe that the equation for $\check{\mu}$ will be used differently in the two cases: indeed in the projection to the $\ell=1$ spherical harmonics we do not need to impose the vanishing of $\check{\mu}$ along $\mathscr{I}_{U, R}$, while for the projection to the $\ell\geq 2$ spherical harmonics we do. We outline the procedure in each case here. 

\begin{itemize}

\item {\bf{The projection to the $\ell=1$ spherical harmonics:}} Equation \eqref{eq=int-1} together with the gauge condition $\check{\nu}_{\ell=1}=0$ implies the vanishing of $\check{\nu}_{\ell=1}$ everywhere in the spacetime. The improved decay for $\check{\mu}_{\ell=1}$ \eqref{improved-mu} allows to make use of equation \eqref{eq=int-2} to transport the decay for $\check{\mu}_{\ell=1}$ in the whole spacetime exterior. The above implies optimal decay for all the quantities in the exterior.  

\item {\bf{The projection to the $\ell\geq 2$ spherical harmonics:}} Equation \eqref{eq=int-2} together with the gauge condition $\check{\mu}_{\ell\geq 2}=0$ implies the decay of $\check{\mu}_{\ell\geq2}$ everywhere in the spacetime. In particular, in this part of the solution the quantity $\check{\mu}$ plays the same role as $\check{\nu}$ in the projection to the $\ell=1$ spherical harmonics. The new quantity $\Xi$, only supported in $\ell\geq2$ modes, has the property that it is bounded along the null hypersurface $\mathscr{I}_{U, R}$. Therefore equation \eqref{eq=int-3} implies decay for $\Xi$ everywhere, and this is enough to obtain optimal decay for all the components. In this part of the solution, $\Xi$ plays the same role as $\check{\mu}$ in the projection to the $\ell=1$ spherical harmonics.  Finally, the quantities involved in the $e_3$ direction can be directly integrated from the null hypersurface $\mathscr{I}_{U, R}$, obtaining optimal decay for them.
\end{itemize}
Observe that the estimates here derived do not depend on $U$ or $R$, i.e. they do not depend on the initial far-away sphere $S_{U, R}$. Therefore the sphere can be taken to be arbitrarily far away. Also observe that the above decay in inverse powers of $u$ cannot hold in the region close to the horizon. The integration backward from $\mathscr{I}_{U, R}$ is performed up to a timelike hypersurface $\{ r=r_2\}$ for $r_2 < R$.

\item  {\bf{Decay in $v$ close to the horizon:}} To obtain decay as inverse power of $v$ in the region close to the horizon, we translate the decay in $\{ r\geq r_2\}$ as a decay in inverse powers of $v$ on $\{ r=r_2\}$ and then integrate forward the horizon. We therefore obtain estimates in the gray region of Figure 1. This is done in Section \ref{section-close-horizon}.

\item  {\bf{Conclusion of the proof:}} Observe that the above estimates are all independent of $U$ and $R$, and as $U$ and $R$ vary for arbitrarily large values, the region where the estimates hold cover the whole exterior region. 
By defining $\tilde{\mathscr{G}}$ as the limit as $U$ and $R$ go to infinity, we obtain that the pure gauge solution $\mathscr{G}=\mathscr{G}^{id}+\tilde{\mathscr{G}}$ and the linearized Kerr-Newman solution $\mathscr{K}=\mathscr{K}^{id}$ as defined from the initial data verify the estimates for $\mathscr{S}-\mathscr{G}-\mathscr{K}$ given by Theorem \ref{linear-stability} in the entire exterior region of Reissner-Nordstr\"om spacetime. This is done in Section \ref{conclusions}. 
\end{enumerate}

\subsection{Decay of the solution along the null hypersurface $\mathscr{I}_{U, R}$}\label{normalization-far-away-decay}

The aim of this subsection is to obtain decay in $u$ and $r$ for all components of $\mathscr{S}^{U, R}$ on $\mathscr{I}_{U, R}$ for all the curvature components. The case $\ell=1$ and $\ell\geq 2$ spherical harmonics will be treated separately. 

\subsubsection{The projection to the $\ell=1$ mode}
Recall conditions \eqref{condition-check-ka-l1-last-slice}, \eqref{condition-check-kab-l1-last-slice}, \eqref{condition-check-nu-l1-last-slice} verified by the solution $\mathscr{S}^{U, R}$. 
In this subsection we show how the above conditions imply the decay for every component in the projection to the $\ell=1$ mode along $\mathscr{I}_{U, R}$. We combine the relations summarized in Proposition \ref{expression-all-quantities-l1} to the above gauge conditions. We first compute $\check{\nu}_{\ell=1}$. We have, using \eqref{divv-ze-in-terms-of-bF-ka-l1} and \eqref{check-rhoF-in-terms-of-bF-bbF-l1}, 
\bea\label{condition-nu-last-slice-implies}
\begin{split}
\check{\nu}_{\ell=1}&= r^4 \left((\divv\ze)_{\ell=1}+2\rhoF \check{\rhoF}_{\ell=1}\right)\\
&= r^4 \Big(\frac 1 r \check{\ka}_{\ell=1} +r\left(\frac{3\rho}{2\rhoF} -\rhoF\right)(\divv\bF)_{\ell=1}+O(r^{-3-\de} u^{-1/2+\de}, r^{-2-\de} u^{-1+\de})\\
&+2\rhoF \left(\frac 1 4  r^2\kab(\divv\bF)_{\ell=1}
 - \frac 1 2 r(\divv\bbF)_{\ell=1}+ O( r^{-1} u^{-1+\de}) \right)\Big)\\
 &= r^3 \check{\ka}_{\ell=1} +r^5\left(\frac{3\rho}{2\rhoF} -\rhoF+\frac 1 2 \rhoF r \kab \right)(\divv\bF)_{\ell=1}-r^5\rhoF  (\divv\bbF)_{\ell=1} \\
 &+O(r^{1-\de} u^{-1/2+\de}, r^{2-\de} u^{-1+\de})
 \end{split}
\eea
Conditions \eqref{condition-check-ka-l1-last-slice}, \eqref{condition-check-kab-l1-last-slice} imply, using \eqref{divv-check-kab-in-terms-of-bF-bbF-ka-l1}
\bea\label{blablabla}
(\divv\bbF)_{\ell=1}&=& \frac 1 2 r\kab(\divv\bF)_{\ell=1}+O( r^{-2} u^{-1+\de})
 \eea
On the other hand, the relation \eqref{condition-nu-last-slice-implies} restricted to $\mathscr{I}_{U, R}$ implies 
\bea\label{check-nu-translated}
\begin{split}
 &\left(\frac{3\rho}{2\rhoF} -\rhoF+\frac 1 2 \rhoF r \kab \right)(\divv\bF)_{\ell=1}-\rhoF  (\divv\bbF)_{\ell=1} \\
 &=O(r^{-4-\de} u^{-1/2+\de}, r^{-3-\de} u^{-1+\de})
 \end{split}
\eea
Combining \eqref{blablabla} and \eqref{check-nu-translated}, we obtain
\beaa
 \left(\frac{3\rho}{2\rhoF} -\rhoF \right)(\divv\bF)_{\ell=1} &=&O(r^{-4-\de} u^{-1/2+\de}, r^{-3-\de} u^{-1+\de})
\eeaa
which therefore implies 
\bea
|(\divv\bF)_{\ell=1}|&\les& \min\{r^{-3-\de} u^{-1/2+\de}, r^{-2-\de} u^{-1+\de}\}\label{estimate-bF-last-slice-l1}
\eea
Proposition \ref{expression-all-quantities-l1} then implies on $\mathscr{I}_{U, R}$:
\bea
|(\divv\bbF)_{\ell=1}|&\les& r^{-2} u^{-1+\de} \\
|(\divv\b)_{\ell=1}|&\les&  \min\{r^{-4-\de} u^{-1/2+\de}, r^{-3-\de} u^{-1+\de}\}\label{estimate-b-last-slice-l1}\\
|(\divv\bb)_{\ell=1}|&\les&  r^{-3} u^{-1+\de}  \\
|\check{\rhoF}_{\ell=1}|&\les&   r^{-1} u^{-1+\de} \label{estimate-check-rhoF-last-slice-l1}
  \eea
  Projecting \eqref{relation-pf-rho-rhoF-bF-bbF}  to the $\ell=1$ spherical harmonics and making use of the fact that along $\mathscr{I}_{U, R}$ $r \gg u$, we obtain
\bea\label{estimate-check-rhoF-last-slice-l1-optimal-r}
|\check{\rhoF}_{\ell=1}|\les r^{-2} u^{-1/2+\de}
\eea
Using the definition of $\check{\nu}$ and condition \eqref{condition-check-nu-l1-last-slice} we observe that $\divv\ze_{\ell=1}=-2\rhoF \check{\rhoF}_{\ell=1}$ along $\mathscr{I}^{U, R}$, and therefore \eqref{estimate-check-rhoF-last-slice-l1} implies the enhanced\footnote{By enhanced estimate we mean an estimate that is better than the final one valid everywhere on the spacetime.} estimate for $\divv\ze$ on $\mathscr{I}_{U, R}$:
\bea
|(\divv\ze)_{\ell=1}|&\les&  r^{-3} u^{-1+\de}  \label{estimate-ze-last-slice-l1}
\eea
Using Gauss equation \eqref{Gauss-check} and conditions \eqref{condition-check-ka-l1-last-slice} and \eqref{condition-check-kab-l1-last-slice} we observe that $\check{\rho}_{\ell=1}=2\rhoF \check{\rhoF}_{\ell=1}$ and therefore \eqref{estimate-check-rhoF-last-slice-l1} implies the enhanced estimate for $\check{\rho}$:
\bea
 |\check{\rho}_{\ell=1}| &\les&  r^{-3} u^{-1+\de} \label{estimate-check-rho-last-slice-l1}
\eea
The above then implies a better decay than expected for $\check{\mu}$ at $\mathscr{I}^{U, R}$. Indeed, 
\bea\label{estimate-check-mu-last-slice-l1}
\begin{split}
|\check{\mu}_{\ell=1}|&=|r^3 \left((\divv \ze)_{\ell=1}+ \check{\rho}_{\ell=1}-4\rhoF \check{\rhoF}_{\ell=1} \right)-2r^4\rhoF (\divv\bF)_{\ell=1} |\\
&\les r^3 \left(r^{-3} u^{-1+\de} \right)+r^4r^{-2}  r^{-2-\de} u^{-1+\de}\les u^{-1+\de}
\end{split}
\eea
This improved decay is necessary to use the transport equation for $\check{\mu}_{\ell=1}$ later.

\subsection{The projection to the $\ell\geq 2$ modes}

Recall conditions \eqref{condition-check-ka-definition}, \eqref{condition-check-kab-definition} and \eqref{condition-check-mu-definition} verified by $\mathscr{S}^{U, R}$ on $\mathscr{I}_{U, R}$.

We combine the relations summarized in Proposition \ref{expression-all-quantities-l2} to the above gauge conditions. In order to do so, we first compute $\DDs_2\DDs_1(\check{\mu}, 0)$. We have 
 \beaa
 \DDs_2\DDs_1(\check{\mu}, 0)&=& r^3 \DDs_2\DDs_1\DDd_1 \ze-2r^4\rhoF \DDs_2\DDs_1\DDd_1\bF+ r^3\DDs_2\DDs_1(\check{\rho}, \check{\sigma})-4r^3\rhoF \DDs_2\DDs_1(\check{\rhoF}, \check{\sigmaF}) \\
 &=& r^3 \left(2\DDs_2\DDd_2+2K\right)\DDs_2 \ze-2r^4\rhoF \left(2\DDs_2\DDd_2+2K\right)\DDs_2\bF\\
 &&+ r^3\DDs_2\DDs_1(\check{\rho}, \check{\sigma})-4r^3\rhoF \DDs_2\DDs_1(\check{\rhoF}, \check{\sigmaF}) 
 \eeaa
 where we used \eqref{relation-DDs-DDs-DDd}. 
 Using relations \eqref{write-rhoF-in-terms-of-chih-chibh-general} and \eqref {check-rho-in-terms-of=chih-chib-general}, we obtain
 \beaa
\DDs_2\DDs_1(\check{\mu}, 0)&=& r^3 \left(2\DDs_2\DDd_2+2K\right)\DDs_2 \ze-2r^4\rhoF \left(2\DDs_2\DDd_2+2K\right)\DDs_2\bF\\
 &&+ r^3\left(\left(-\frac 3 4 \rho-\frac 1 2 \rhoF^2\right)\left(  \kab \chih + \ka \chibh \right) +O( r^{-4} u^{-1+\de})\right)\\
 &&-4r^3\rhoF \left(-\frac 12 \rhoF \left(  \kab \chih + \ka \chibh \right)+\min\{r^{-4} u^{-1/2+\de}, r^{-3} u^{-1+\de}\}\right)\\
 &=& r^3 \left(2\DDs_2\DDd_2+2K\right)\DDs_2 \ze-2r^4\rhoF \left(2\DDs_2\DDd_2+2K\right)\DDs_2\bF\\
 &&+ r^3\left(-\frac 3 4 \rho+\frac 3 2 \rhoF^2\right)\left(  \kab \chih + \ka \chibh \right) +O( r^{-1} u^{-1+\de})
 \eeaa
 Using relation \eqref{write-bF-in-terms-of-chih-general}, we obtain
 \beaa
 \DDs_2\DDs_1(\check{\mu}, 0) &=& r^3 \left(2\DDs_2\DDd_2+2K\right)\DDs_2 \ze+ \left(2\DDs_2\DDd_2+2K\right)(2r^4\rhoF^2 \chih)+ r^3\left(-\frac 3 4 \rho+\frac 3 2 \rhoF^2\right)\left(  \kab \chih + \ka \chibh \right) \\
 &&+O( r^{-1} u^{-1+\de})
 \eeaa
Using relation \eqref{write-ze-in-terms-of-chih-general}, we obtain
 \beaa
\DDs_2\DDs_1(\check{\mu}, 0) &=& r^4 \left(2\DDs_2\DDd_2+2K\right)\left(\DDs_2\DDd_2 \chih+\left(-\frac 3 2\rho+\rhoF^2\right) \chih+\frac 1 2  \DDs_2\DDs_1(\check{\ka}, 0)\right) \\
&&+ \left(2\DDs_2\DDd_2+2K\right)(2r^4\rhoF^2 \chih)+ r^3\left(-\frac 3 4 \rho+\frac 3 2 \rhoF^2\right)\left(  \kab \chih + \ka \chibh \right) +O( r^{-1} u^{-1+\de})
 \eeaa
 which finally gives
 \bea\label{nu-in-terms-of-chi-chib-general}
 \begin{split}
\DDs_2\DDs_1(\check{\mu}, 0) &= r^4 \left(2\DDs_2\DDd_2+2K\right)\left(\DDs_2\DDd_2 \chih\right)+r^4 \left(2\DDs_2\DDd_2+2K\right)\left(-\frac 3 2\rho+3\rhoF^2\right) \chih\\
& +\frac 1 2 r^4 \left(2\DDs_2\DDd_2+2K\right)\DDs_2\DDs_1(\check{\ka}, 0)+ r^3\left(-\frac 3 4 \rho+\frac 3 2 \rhoF^2\right) \kab \chih + r^3\left(-\frac 3 4 \rho+\frac 3 2 \rhoF^2\right) \ka \chibh \\
& +O( r^{-1} u^{-1+\de})
 \end{split}
 \eea
which will be used later.

Condition \eqref{condition-check-kab-l2}, \eqref{condition-check-ka-l2} and relation \eqref{check-kab-in-terms-of=chih-chib-general} restricted to $\mathscr{I}_{U, R}$ imply
\beaa
2\DDs_2\DDd_2 \chibh+ \left(- 3  \rho +2\rhoF^2\right) \chibh &=&-\frac r 2 \kab \left(2\DDs_2\DDd_2 \chih+\left(- 3 \rho+2\rhoF^2\right) \chih\right) +O( r^{-3} u^{-1+\de})
\eeaa
Recall the operator  $\mathcal{E}=2 \DDs_2\DDd_2- 3\rho+2\rhoF^2$, then the above relation becomes 
\bea\label{E-relation-chi-chib}
\mathcal{E}\left(\chibh+\frac r 2 \kab \chih\right)&=& O( r^{-3} u^{-1+\de})
\eea
Applying Lemma \ref{first-poincare-inequality-2-tensor}  to \eqref{E-relation-chi-chib} we obtain
\bea\label{chibh-in-terms-of-chi-last-slice-E}
\chibh&=&-\frac r 2 \kab \chih +O( r^{-1} u^{-1+\de})
\eea
On the other hand, relation \eqref{nu-in-terms-of-chi-chib-general} restricted to $\mathscr{I}_{U, R}$ becomes:
\beaa
&& r^2 \left(2\DDs_2\DDd_2+2K\right)\left(\DDs_2\DDd_2 \chih\right)+r^2 \left(2\DDs_2\DDd_2+2K\right)\left(-\frac 3 2\rho+3\rhoF^2\right) \chih \\
 &&+ r\left(-\frac 3 4 \rho+\frac 3 2 \rhoF^2\right) \kab \chih +\left(-\frac 3 2 \rho+3  \rhoF^2\right)  \chibh  =O( r^{-3} u^{-1+\de})
\eeaa
Applying \eqref{chibh-in-terms-of-chi-last-slice-E}, we finally obtain on $\mathscr{I}_{U, R}$:
\beaa
&&  \left(2\DDs_2\DDd_2+2K\right)\left(2\DDs_2\DDd_2 \chih+\left(- 3 \rho+6\rhoF^2\right) \chih\right) =O( r^{-5} u^{-1+\de})
\eeaa
Recall that $2\DDs_2\DDd_2+2K=\DDs_1\DDd_1$, therefore the standard Poincar\'e inequality implies
\bea
2\DDs_2\DDd_2 \chih+\left(- 3 \rho+6\rhoF^2\right) \chih \les O( r^{-3} u^{-1+\de})
\eea
A slight modification of Lemma \ref{first-poincare-inequality-2-tensor} gives an identical Poincar\'e inequality applied to the operator $2\DDs_2\DDd_2 +\left(- 3 \rho+6\rhoF^2\right) $ which finally implies
\bea\label{estimate-chih-last-slice}
|\chih| \les O( r^{-1} u^{-1+\de})
\eea
Relation \eqref{chibh-in-terms-of-chi-last-slice-E} implies
\bea\label{estimate-chibh-last-slice}
|\chibh| \les  O( r^{-1} u^{-1+\de})
\eea
Proposition \ref{expression-all-quantities-l2} implies on $\mathscr{I}_{U, R}$, using that $r \gg u$:
\bea
|\DDs_2\bF|&\les &\min\{ r^{-3-\de} u^{-1/2+\de}, r^{-2-\de} u^{-1+\de}\}\label{estimate-bf-last-slice}\\
|\DDs_2\bbF|&\les&O( r^{-2} u^{-1+\de})\\
|\DDs_2\b|&\les& \min\{r^{-4-\de} u^{-1/2+\de}, r^{-3-\de} u^{-1+\de} \}\label{estimate-b-last-slice}\\
|\DDs_2\bb|&\les& O( r^{-3} u^{-1+\de} ) \\
|\DDs_2\ze|&\les& \min\{r^{-3-\de} u^{-1/2+\de}, r^{-2-\de} u^{-1+\de} \} \label{estimate-ze-last-slice}\\
|\DDs_2\DDs_1(\check{\rhoF}, \check{\sigmaF})|&\les&  \min\{r^{-4} u^{-1/2+\de}, r^{-3} u^{-1+\de}\} \label{estimate-check-rhoF-last-slice}\\
 |\DDs_2\DDs_1(-\check{\rho}, \check{\sigma})| &\les&\min\{r^{-5} u^{-1/2+\de}, r^{-4} u^{-1+\de}\} \label{estimate-check-rho-last-slice} \\
 |\DDs_2\DDs_1(\check{K}, 0)|&\les& \min\{r^{-5-\de} u^{-1/2+\de}, r^{-4} u^{-1+\de} \}\label{estimate-check-K-last-slice}
  \eea
  Using the above in Codazzi equation \eqref{Codazzi-chi-ze-eta} we finally have on $\mathscr{I}_{U, R}$
  \bea
  |\chih| \les O(r^{-2} u^{-1/2+\de}) \label{estimate-chih-last-slice-optimal-r}
  \eea
\subsubsection{The terms involved in the $e_3$ direction}
 We obtain here decay along $\mathscr{I}_{U, R}$ for the quantities $\eta$,  $\xi$, $\check{\omb}$.

 Conditions \eqref{condition-check-ka-l1-last-slice},  \eqref{condition-check-kab-l1-last-slice} and \eqref{condition-check-nu-l1-last-slice} imply $\nabb_3\check{\ka}_{\ell=1}=\check{\ka}_{\ell=1}=0$, $\nabb_3\check{\kab}_{\ell=1}=\check{\kab}_{\ell=1}=0$ and $\nabb_3\check{\nu}_{\ell=1}=\check{\nu}_{\ell=1}=0$  on $\mathscr{I}_{U, R}$. Therefore, using \eqref{nabb-3-check-ka-ze-eta}, \eqref{nabb-3-check-kab-ze-eta} and \eqref{nabb-3-check-nu} projected to the $\ell=1$ spherical harmonics we have
 \bea
0&=&\frac 1 2\ka^2\check{\underline{\Omega}}_{\ell=1}+ 2\ka\check{\omb}_{\ell=1} +2 \divv \eta_{\ell=1} +2\check{\rho}_{\ell=1} ,\label{relation-coming-fromka}\\
0&=&-2\kab\check{\omb}_{\ell=1}+2\divv\xib_{\ell=1}+\left(\frac 1 2 \ka\kab-2\rho\right)\check{\underline{\Omega}}_{\ell=1}\label{relation-coming-fromkab}\\
 0&=&\frac{4}{r^2} \check{\omb}_{\ell=1}+\left(\frac 1 2 \kab+2\omb \right)\divv \left(\ze_{\ell=1}-\eta_{\ell=1}\right) +\frac 1 2 \ka \divv\xib_{\ell=1} +2\rhoF^2 \ka\check{\Omegab}_{\ell=1}\label{relation-coming-fromnu}
 \eea
 The projection of \eqref{xib-Omegab-ze-eta} to the $\ell=1$ spherical harmonics implies $\frac{2}{r^2} \check{\Omegab}_{\ell=1}=\divv\xib_{\ell=1}+\Omegab (\divv\eta_{\ell=1}-\divv\ze_{\ell=1})$, therefore relations \eqref{relation-coming-fromka} and \eqref{relation-coming-fromkab} simplify to
 \bea
 -2\ka\check{\omb}_{\ell=1}&=&(\divv\xib_{\ell=1}+\Omegab\divv\eta_{\ell=1})+2 \divv\eta_{\ell=1}  +O(r^{-3} u^{-1+\de})\label{relation-1-in-between}\\
 2\kab\check{\omb}_{\ell=1}&=&2\divv\xib_{\ell=1}+r\left(\frac 1 2\kab-\rho\right)(\divv\xib_{\ell=1}+\Omegab\divv\eta_{\ell=1})+O(r^{-3} u^{-1+\de})\label{relation=in-between}
 \eea
 because of the enhanced estimate for $\divv\ze_{\ell=1}$ \eqref{estimate-ze-last-slice-l1} and $\check{\rho}_{\ell=1}$ \eqref{estimate-check-rho-last-slice-l1}.
 Multiplying \eqref{relation-1-in-between} by $\kab$ and \eqref{relation=in-between} by $\ka$ and summing the two we obtain
 \beaa
0&=&\left(\frac{4}{r}+2\kab-2r\rho\right)\divv\xib_{\ell=1}+\left(\frac{4}{r}+2\kab-2r\rho\right)\Omegab\divv\eta_{\ell=1}+O(r^{-4} u^{-1+\de})
 \eeaa
 Observe that $\frac{4}{r}+2\kab-2r\rho=\frac{12M}{r^2}-\frac{8Q^2}{r^2}$, therefore the above gives
 \bea\label{eq100}
 |\divv\xib_{\ell=1}+\Omegab\divv\eta_{\ell=1}|&\les&O(r^{-2} u^{-1+\de})
 \eea
 Using \eqref{xib-Omegab-ze-eta} we obtain
 \bea\label{estmate-Omegab-l1}
 |\check{\Omegab}|\les O( u^{-1+\de})
 \eea
  Multiplying \eqref{relation-1-in-between} by $r\left(\frac 1 2\kab-\rho\right)$ and subtracting \eqref{relation=in-between} we obtain
 \bea\label{relation-save-ass}
 \left(-4\kab+4r\rho\right)\check{\omb}_{\ell=1}&=&2 r\left(\frac 1 2\kab-r\rho\right)\divv\eta_{\ell=1} -2\divv\xib_{\ell=1} +O(r^{-3} u^{-1+\de})
 \eea
Relation \eqref{relation-coming-fromnu} simplifies, using \eqref{estmate-Omegab-l1} and enhanced estimate for $\divv\ze_{\ell=1}$ \eqref{estimate-ze-last-slice-l1} to
 \beaa
 \frac{4}{r^2} \check{\omb}_{\ell=1}&=&\left(\frac 1 2 \kab+2\omb \right)\divv\eta_{\ell=1} -\frac 1 2 \ka \divv\xib_{\ell=1}+O(r^{-4} u^{-1+\de})
 \eeaa
 Recall that $2\omb=-r\rho$, so it can be written as
 \bea\label{relation-save-ass-2}
 \frac{8}{r} \check{\omb}_{\ell=1}&=&2r\left(\frac 1 2 \kab-r\rho \right)\divv\eta_{\ell=1} - 2\divv\xib_{\ell=1}+O(r^{-3} u^{-1+\de})
 \eea
 Subtracting \eqref{relation-save-ass-2} from \eqref{relation-save-ass} we then obtain
 \beaa
 \left(-4\kab+4r\rho-\frac{8}{r}\right)\check{\omb}_{\ell=1}&=&O(r^{-3} u^{-1+\de})
 \eeaa
 and since $-4\kab+4r\rho-\frac{8}{r}=-\frac{24M}{r^2}+\frac{16Q^2}{r^3}$ we obtain
 \beaa
 |\check{\omb}_{\ell=1}|\les O(r^{-1} u^{-1+\de})
 \eeaa
 Observe that the enhanced estimates obtained for $\divv\ze_{\ell=1}$ and $\check{\rho}_{\ell=1}$ along $\mathscr{I}_{U, R}$ allowed us to obtain the optimal decay for $\check{\omb}$, $\xib$ and $\divv\eta_{\ell=1}$ and compensate for the loss of a power of $r$ in the degenerate Poincar\'e inequality above, where the leading term in $\frac 1 r$ in the expression $-4\kab+4r\rho-\frac{8}{r}$ cancels out, and the final leading term is in $\frac{1}{r^2}$. 
 
 Using \eqref{eq100}, relation \eqref{relation=in-between} simplifies to
 \beaa
  2\kab\check{\omb}_{\ell=1}&=&2\divv\xib_{\ell=1}+O(r^{-2} u^{-1+\de})=-2\Omegab\divv\eta_{\ell=1}+O(r^{-2} u^{-1+\de})
  \eeaa
  which gives
  \beaa
   \check{\omb}_{\ell=1}  &=&-\frac{r}{2}\divv\eta_{\ell=1}+O(r^{-1} u^{-1+\de})
  \eeaa
 Using the above we obtain
 \beaa
 |\divv\eta_{\ell=1}, \divv\xib_{\ell=1}|\les O(r^{-1} u^{-1+\de})
 \eeaa
 This gives all the desired decay for the projection to the $\ell=1$ spherical harmonics of $\xib$, $\eta$, $\check{\omb}$ and $\check{\Omegab}$.

Conditions \eqref{condition-check-ka-l2}, \eqref{condition-check-kab-l2} and \eqref{condition-nu} imply $\nabb_3\DDs_2\DDs_1(\check{\ka}, 0)=\DDs_2\DDs_1(\check{\ka}, 0)=0$, $\nabb_3\DDs_2\DDs_1(\check{\kab}, 0)=\DDs_2\DDs_1(\check{\kab}, 0)=0$ and $\nabb_3\DDs_2\DDs_1(\check{\mu}, 0)=\DDs_2\DDs_1(\check{\mu}, 0)=0$ on $\mathscr{I}_{U, R}$. Therefore, using \eqref{nabb-3-check-ka-ze-eta}, \eqref{nabb-3-check-kab-ze-eta} and Lemma \ref{nabb-3-mu} commuted with $\DDs_2\DDs_1$ along $\mathscr{I}_{U, R}$ we obtain the following relations:
\bea
0&=&\frac 1 2 \ka^2\DDs_2\DDs_1(\check{\underline{\Omega}}, 0)+ 2\ka \DDs_2\DDs_1(\check{\omb}, 0) +2 \DDs_2\DDs_1\DDd_1 \eta +2\DDs_2\DDs_1(\check{\rho}, 0), \label{first-relation-last-slice-}\\
0&=&-2\kab\DDs_2\DDs_1(\check{\omb}, 0)+2\DDs_2\DDs_1\DDd_1\xib+\left(\frac 1 2 \ka\kab-2\rho\right)\DDs_2\DDs_1(\check{\underline{\Omega}}, 0)\label{second-relation-last-slice-}
\eea
and
\bea\label{relation-nabb-3-nu-last-slice}
\begin{split}
0&=2 \DDs_2\DDs_1\DDd_1\DDs_1( \check{\omb}, 0)+\left(\frac 1 2 \kab+2\omb \right)\DDs_2\DDs_1\DDd_1(\ze-\eta) +\frac 1 2 \ka \DDs_2\DDs_1\DDd_1\xib - 2\DDs_2\DDs_1\DDd_1\bb \\
& +2\rhoF\DDs_2\DDs_1\DDd_1\bbF -\left(\frac 3 2  \rho-3\rhoF^2\right)(-\ka \DDs_2\DDs_1(\check{\Omegab}, 0)) -4\omb r\rhoF \DDs_2\DDs_1\DDd_1\bF \\
&+2r\rhoF \DDs_2\DDs_1\DDd_1\DDs_1(\check{\rhoF}, \check{\sigmaF}) -4r\rhoF^2 \DDs_2\DDs_1\DDd_1\eta
\end{split}
\eea
Taking into account the estimates already obtained above we can simplify \eqref{first-relation-last-slice-} and \eqref{relation-nabb-3-nu-last-slice} as 
\bea\label{relation-last-slice-whatever}
\frac 1 2 \ka^2\DDs_2\DDs_1(\check{\underline{\Omega}}, 0)+ 2\ka\DDs_2\DDs_1(\check{\omb}, 0) +2 \DDs_2\DDs_1\DDd_1 \eta &\les& O( r^{-4} u^{-1+\de})
\eea
and
\bea\label{relation-last-slice-from-nu}
\begin{split}
&2 \DDs_2\DDs_1\DDd_1\DDs_1( \check{\omb}, 0)-\left(\frac 1 2 \kab+2\omb + 4r\rhoF^2\right)\DDs_2\DDs_1\DDd_1\eta +\frac 1 2 \ka \DDs_2\DDs_1\DDd_1\xib \\
& +\left(\frac 3 2  \rho-3\rhoF^2\right)\ka \DDs_2\DDs_1(\check{\Omegab}, 0) \les O(r^{-5} u^{-1+\de} )
\end{split}
\eea
Using \eqref{xib-Omegab-ze-eta}, relation \eqref{second-relation-last-slice-} simplifies to
\bea\label{check-omb-in-terms-of-xib-eta}
2\kab\DDs_2\DDs_1(\check{\omb}, 0)&=&2\DDs_2\DDs_1\DDd_1\xib+\left(\frac 1 2 \ka\kab-2\rho\right)\DDs_2(\xib+\Omegab\eta)+O(r^{-4} u^{-1+\de})   
\eea
Again using \eqref{xib-Omegab-ze-eta}, relation \eqref{relation-last-slice-whatever} simplifies to
\beaa
\frac 1 2 \ka^2\DDs_2(\xib+\Omegab\eta)+ 2\ka\DDs_2\DDs_1(\check{\omb}, 0) +2 \DDs_2\DDs_1\DDd_1 \eta &\les& O( r^{-4} u^{-1+\de})
\eeaa
Using \eqref{check-omb-in-terms-of-xib-eta} to substitute in the above relation upon multiplying by $\kab$, we obtain
\beaa
 &&\ka\left(2\left(2\DDs_2\DDd_2+2K\right)\DDs_2\xib+\left(\ka\kab-2\rho\right)\DDs_2\xib\right)  +2 \kab\left(2\DDs_2\DDd_2+2K\right)\DDs_2 \eta+\left(\ka^2\kab-2\ka\rho\right)\DDs_2(\Omegab\eta)\\
  &\les& O( r^{-5} u^{-1+\de})
\eeaa
where we used \eqref{relation-DDs-DDs-DDd}.
Recalling that $\kab=\ka\Omegab$, we obtain
\beaa
 \left(4\DDs_2\DDd_2+4K+\ka\kab-2\rho\right) \DDs_2(\xib+\Omegab\eta)  &\les& O( r^{-4} u^{-1+\de})
\eeaa
Using Gauss equation the above relation gives
\beaa
\mathcal{E} (\DDs_2\xib+\Omegab\DDs_2\eta)  &\les& O( r^{-4} u^{-1+\de})
\eeaa
where $\mathcal{E}$ is the operator defined in Lemma \ref{first-poincare-inequality-2-tensor}. By Lemma \ref{first-poincare-inequality-2-tensor}, we obtain 
\beaa
\DDs_2\xib+\Omegab\DDs_2\eta\les  O( r^{-2} u^{-1+\de})
\eeaa
Using \eqref{xib-Omegab-ze-eta}, we obtain
\bea
|\DDs_2\DDs_1(\check{\underline{\Omega}}, 0)|\les  O( r^{-2} u^{-1+\de})
\eea
Relation \eqref{check-omb-in-terms-of-xib-eta} simplifies to
\beaa
2\kab\DDs_2\DDs_1(\check{\omb}, 0)&=&2\DDs_2\DDs_1\DDd_1\xib+O( r^{-4} u^{-1+\de} ) =2\left(2\DDs_2\DDd_2+2K\right)\DDs_2\xib+O( r^{-4} u^{-1+\de} ) \\
&=&-2\Omegab\left(2\DDs_2\DDd_2+2K\right)\DDs_2\eta+O( r^{-4} u^{-1+\de} ) 
\eeaa
which gives
\beaa
\ka\DDs_2\DDs_1(\check{\omb}, 0)&=&-\left(2\DDs_2\DDd_2+2K\right)\DDs_2\eta+O( r^{-4} u^{-1+\de} )
\eeaa
Using the above to simplify \eqref{relation-last-slice-from-nu}, we finally obtain
\beaa
&&2 \left(2\DDs_2\DDd_2+2K\right)\ka\DDs_2\DDs_1( \check{\omb}, 0)-\left( \ka\kab+2\omb\ka + 4r\ka\rhoF^2\right)\left(2\DDs_2\DDd_2+2K\right)\DDs_2\eta   \\
&=&- \left(4\DDs_2\DDd_2+4K+\ka\kab+2\omb\ka + 8\rhoF^2\right)\left(2\DDs_2\DDd_2+2K\right)\DDs_2\eta+O( r^{-6} u^{-1+\de} )\\
&=&- 2\left(2\DDs_2\DDd_2-3\rho + 6\rhoF^2\right)\left(2\DDs_2\DDd_2+2K\right)\DDs_2\eta+O( r^{-6} u^{-1+\de} )\les O( r^{-6} u^{-1+\de} )
\eeaa
where we recall that $2\omb=-r\rho$. The above operator is then a slight modification of the operator $\mathcal{E}$, for which a Poincar\'e inequality holds. We therefore have
\beaa
\left(2\DDs_2\DDd_2+2K\right)\DDs_2\eta &\les& O( r^{-4} u^{-1+\de} )
\eeaa
The standard Poincar\'e inequality applied to the above then gives
\beaa
|\DDs_2\eta| &\les& O( r^{-2} u^{-1+\de} )
\eeaa
The above relations imply
\beaa
|\DDs_2\xib|&\les&  O( r^{-2} u^{-1+\de}) \\
|\DDs_2\DDs_1( \check{\omb}, 0)|&\les&  O( r^{-3} u^{-1+\de})
\eeaa

Combining the above with the estimates obtained for the projection to the $\ell=1$ mode we obtain
\bea
| \xib|, |\eta|, |\check{\omb}|&\les& r^{-1} u^{-1+\de} \qquad \text{along $\mathscr{I}_{U, R}$}\label{estimate-xib-last-slice-1}\label{estimate-eta-last-slice-1}\label{estimate-omb-last-slice-1}
\eea

\subsubsection{The metric coefficients}

We derive here decay along $\mathscr{I}_{U, R}$ for the metric coefficients $\hat{\slashed{g}}$, $\underline{b}$ and $\check{\vsi}$.

\begin{enumerate}
\item By condition \eqref{condition-hat-slashed-g-last-slice} and \eqref{condition-underline-b-last-slice}, integrating equation \eqref{nabb-3-g} along $\mathscr{I}_{U, R}$ gives
\bea\label{estimate-slashed-g-last-slice}
|\hat{\slashed{g}}|&\les& \min\{r^{-1} u^{-1/2+\de},  u^{-1+\de} \}\qquad \text{along $\mathscr{I}_{U, R}$}
\eea
\item Using \eqref{check-K-tr} and the estimate \eqref{estimate-check-K-last-slice} we can estimate the projection to the $\ell\geq2$ spherical harmonics of $\check{\tr_\gamma \slashed{g}}$:
\bea\label{estimate-tr-l2-last-slice}
| \DDs_2\DDs_1(\check{\tr_\gamma \slashed{g}}, 0)| &\les& \min\{r^{-3-\de} u^{-1/2+\de},  r^{-2}u^{-1+\de}\}
\eea
On the other hand, integrating the projection to the $\ell=1$ spherical harmonics of \eqref{nabb-3-check-tr} along $\mathscr{I}_{U, R}$ and using conditions \eqref{tr-slashed-g-l1-last-slice} and \eqref{condition-divv-underline-b-l1-last-slice} we obtain
\beaa
|\check{\tr_\gamma \slashed{g}}_{\ell=1}|\les  \min\{r^{-1-\de} u^{-1/2+\de}, u^{-1+\de}\}
\eeaa
Consequently we have \bea\label{estimate-tr-l1-last-slice}
|\check{\tr_\gamma \slashed{g}}|\les   \min\{r^{-1-\de} u^{-1/2+\de}, u^{-1+\de}\} \qquad \text{along $\mathscr{I}_{U, R}$}
\eea

\item Conditions \eqref{condition-divv-underline-b-l1-last-slice} and \eqref{condition-underline-b-last-slice} imply that $\underline{b}=0$ along $\mathscr{I}_{U, R}$. 
\item  Equation \eqref{nabb-check-vsi} implies 
\beaa
|\check{\vsi}|\les u^{-1+\de} \qquad \text{along $\mathscr{I}_{U, R}$}
\eeaa
\end{enumerate}

\subsection{Decay of the solution $\mathscr{S}^{U, R}$ in the exterior}\label{final-decay}

Using the decay in $u$ and $r$ along the null hypersurface $\mathscr{I}_{U, R}$ as obtained in the previous section, we will transport it to the past of it using transport equations along the $e_4$ direction. We summarize the standard computation involved in the integration along the $e_4$ direction in the following Lemma, where we fix $r_1>r_{\mathcal{H}}$. We denote $A \les B$ if there exists an universal constant $C$ depending on the initial data such that $A \leq C B$. 

\begin{lemma}\label{optimal-decay-in-r} If $f$ verifies the transport equation 
\beaa
\nabb_4 f+\frac p 2 \ka f&=& F
\eeaa
and $f$ and $F$ satisfy the following estimates:
\bea
|f|&\les& r^{-p-q_1}u^{-1/2+\de} \  \text{on $\mathscr{I}_{U, R}$} \label{decay-last-slice-necessary} \\
|F|&\les& r^{-p-1-q_2}u^{-1/2+\de} \  \text{on $\{ r\geq r_2\}$} \label{decay-F-last-slice-necessary} 
\eea
with $q_1, q_2\geq0$, we have in  $\{ r \geq r_2\} \cap \{ u \geq u_0 \} \cap I^{-}(S_{U, R})$
\beaa
 |f|&\les& r^{-p-\min\{q_1, q_2\}}u^{-1/2+\de}
\eeaa
\end{lemma}
\begin{proof} According to Lemma \ref{nabb-4-int-f}, the transport equation verified by $f$ is equivalent to 
\beaa
\nabb_4(r^p f)&=& r^p F
\eeaa
Using \eqref{null-frame-RN}, the transport equation becomes
\beaa
\pr_r(r^p f)&=& r^p F
\eeaa
Consider now a fixed $u<U$. The null hypersurface of fixed $u$ intersects $\mathscr{I}_{U, R}$ at a certain $r=r_*(u)$ in the sphere $S_{u, r_*(u)}$. We now integrate the above equation along the fixed $u$ hypersurface from the sphere $S_{u, r_*(u)}$ on $\mathscr{I}_{U, R}$ to the sphere $S_{u, r}$ for any $r_1\leq r\leq R $. We obtain
\beaa
r^p f(u,r)&=&r_*(u)^p f(u,r_*(u))+\int_{r_*}^r \lambda^p F(u, \lambda)d\lambda
\eeaa
If condition \eqref{decay-last-slice-necessary} is satisfied, then $|f(u,r_*(u))| \les r_*(u)^{-p-q_1}u^{-1/2+\de}$ and $|F(u,\lambda)|\les \lambda^{-p-1-q_2}u^{-1/2+\de}$, which gives
\beaa
r^p |f(u,r)|&\les& r_*(u)^{-q_1}u^{-1/2+\de}+\int_{r_*}^r  \lambda^{-1-q_2}u^{-1/2+\de}d\lambda \\
&\les& r_*(u)^{-q_1}u^{-1/2+\de}+  r^{-q_2}u^{-1/2+\de}+  r_*(u)^{-q_2}u^{-1/2+\de}
\eeaa
Since by construction $r_*(u)\geq R$ for every $u<U$, and $q_1\geq0$, we can bound the right hand side by
\beaa
r^p |f(u,r)|&\les& r^{-q_2}u^{-1/2+\de}+R^{-\min\{q_1, q_2\}}u^{-1/2+\de}
\eeaa
Since $R\geq r$ we can bound $R^{-\min\{q_1, q_2\}}\leq r^{-\min\{q_1, q_2\}}$, and finally obtain an estimate which is independent of $R$:
\beaa
r^p |f(u,r)|&\les& r^{-\min\{q_1, q_2\}}u^{-1/2+\de}
\eeaa
Diving by $r^p$, we obtain the desired estimate.
\end{proof} 

\subsubsection{The projection to the $\ell=1$ mode: optimal decay in $r$ and in $u$}\label{decay-l-1mode}
\begin{enumerate}
\item Integrating \eqref{nabb-4-check-ka-ze-eta} and using condition \eqref{condition-check-ka-l1-last-slice} we obtain 
\bea\label{check-ka-l1=0-last-slice-everywhere}
\check{\ka}_{\ell=1}=0 \qquad \text{in  $\{ r \geq r_2\} \cap \{ u \geq u_0 \} \cap I^{-}(S_{U, R})$}
\eea

\item The projection of \eqref{nabb-4-check-nu} to the $\ell=1$ spherical harmonics gives, using \eqref{check-ka-l1=0-last-slice-everywhere},
\beaa
\nabb_4\check{\nu}_{\ell=1}&=& r^2\check{\ka}_{\ell=1}-2r^4\rhoF^2\check{\ka}_{\ell=1}+  r^4\divv \divv \chih_{\ell=1}=0
\eeaa
Integrating the above from $\mathscr{I}_{U, R}$ and using condition \eqref{condition-check-nu-l1-last-slice} we obtain 
\bea\label{check-nu-l1=0-last-slice-everywhere}
\check{\nu}_{\ell=1}=0 \qquad \text{in in  $\{ r \geq r_2\} \cap \{ u \geq u_0 \} \cap I^{-}(S_{U, R})$}
\eea

Using the relation implied by $\check{\nu}$ \eqref{check-nu-translated} and \eqref{check-nu-l1=0-last-slice-everywhere} we obtain for $r \geq r_1$:
\bea\label{non-iela}
\begin{split}
& \left(\frac{3\rho}{2\rhoF} -\rhoF+\frac 1 2 \rhoF r \kab \right)(\divv\bF)_{\ell=1}-\rhoF  (\divv\bbF)_{\ell=1}\\
&=O(r^{-4-\de} u^{-1/2+\de}, r^{-3-\de} u^{-1+\de})
\end{split}
 \eea

\item The projection of \eqref{nabb-4-check-mu} to the $\ell=1$ spherical harmonics gives, using \eqref{check-ka-l1=0-last-slice-everywhere}
\beaa
\nabb_4\check{\mu}_{\ell=1}&=& O(r^{-1-\de}u^{-1+\de}) 
\eeaa
Consider now a fixed $u<U$. As in Lemma \ref{optimal-decay-in-r}, we integrate the above equation along the fixed $u$ hypersurface from the sphere $S_{u, r_*(u)}$ on $\mathscr{I}_{U, R}$ to the sphere $S_{u, r}$ for any $r_0\leq r\leq R $. We obtain
\beaa
|\check{\mu}_{\ell=1}|&\les & |\check{\mu}_{l=1}(u,r_*(u))|+\int_{r_*}^r \lambda^{-1-\de} u^{-1+\de} d\lambda
\eeaa
Because of estimate \eqref{estimate-check-mu-last-slice-l1}, we obtain
\beaa
|\check{\mu}_{\ell=1}|&\les & u^{-1+\de}+\int_{r_*}^r \lambda^{-1-\de} u^{-1+\de} d\lambda \les  u^{-1+\de}
\eeaa
We have, using \eqref{divv-ze-in-terms-of-bF-ka-l1}, \eqref{check-rho-in-terms-of-bF-bbF-l1} and \eqref{check-rhoF-in-terms-of-bF-bbF-l1}, 
\bea\label{condition-mu-implies}
\begin{split}
\check{\mu}_{\ell=1}&= r^3 (\left(\divv\ze)_{\ell=1}+\check{\rho}_{\ell=1}-4\rhoF \check{\rhoF}_{\ell=1}\right)-2r^4\rhoF (\divv\bF)_{\ell=1}\\
&= r^3 \Big(\frac 1 r \check{\ka}_{\ell=1} +r\left(\frac{3\rho}{2\rhoF} -\rhoF\right)(\divv\bF)_{\ell=1}+O(r^{-3-\de} u^{-1/2+\de}, r^{-2-\de} u^{-1+\de})\\
&+\frac 1 4 r^2\kab\left(\frac{3\rho}{2\rhoF} +\rhoF\right)(\divv\bF)_{\ell=1}
 -\frac 1 2 r\left(\frac{3\rho}{2\rhoF}+\rhoF\right)(\divv\bbF)_{\ell=1}+O( r^{-2} u^{-1+\de})\\
&-4\rhoF \left(\frac 1 4  r^2\kab(\divv\bF)_{\ell=1}
 - \frac 1 2 r(\divv\bbF)_{\ell=1}+ O( r^{-1} u^{-1+\de}) \right)\Big)-2r^4\rhoF (\divv\bF)_{\ell=1}\\
 &= r^2 \check{\ka}_{\ell=1} +r^4\left(\frac{3\rho}{2\rhoF} -3\rhoF\right)(\divv\bF)_{\ell=1}\\
&+\frac 1 4 r^5\kab\left(\frac{3\rho}{2\rhoF} -3\rhoF\right)(\divv\bF)_{\ell=1}
 -\frac 1 2 r^4\left(\frac{3\rho}{2\rhoF}-3\rhoF\right)(\divv\bbF)_{\ell=1}+O( r u^{-1+\de})
 \end{split}
\eea
Using the estimate for $\check{\mu}$ obtained above we have
\bea\label{non-iela=2}
 (\divv\bbF)_{\ell=1}&=& \left(2+\frac 1 2 r\kab\right)(\divv\bF)_{l=1}+ O( r^{-2} u^{-1+\de})
\eea
Substituting the above into \eqref{non-iela} we finally obtain
\beaa
 \left(\frac{3\rho}{2\rhoF} -3\rhoF \right)(\divv\bF)_{\ell=1} &=&O(r^{-4-\de} u^{-1/2+\de}, r^{-3-\de} u^{-1+\de})
\eeaa
which gives in  $\{ r \geq r_2\} \cap \{ u \geq u_0 \} \cap I^{-}(S_{U, R})$
\bea\label{final-decay-divv-bF-l1}
|(\divv\bF)_{\ell=1}| &\les&O(r^{-3-\de} u^{-1/2+\de}, r^{-2} u^{-1+\de})
\eea
\item Relation \eqref{non-iela=2} implies then
\bea\label{final-decay-divv-bbF-l1}
|(\divv\bbF)_{\ell=1}| &\les&O( r^{-2} u^{-1+\de})
\eea
\item The decay obtained for $\check{\ka}_{\ell=1}$, $(\divv\bF)_{\ell=1}$, $(\divv\bbF)_{\ell=1}$  allows to deduce the following decays in  $\{ r \geq r_2\} \cap \{ u \geq u_0 \} \cap I^{-}(S_{U, R})$ using Proposition \ref{expression-all-quantities-l1}:
\bea
|(\divv \ze)_{\ell=1}|&\les &\min\{r^{-3-\de}u^{-1/2+\de}, r^{-2-\de} u^{-1+\de} \}\label{final-decay-divv-ze-l1} \\
|(\divv \b)_{\ell=1}|&\les &\min\{r^{-4-\de}u^{-1/2+\de}, r^{-3-\de} u^{-1+\de} \}\label{final-decay-divv-b-l1}\\
|(\divv \bb)_{\ell=1}|&\les& r^{-3} u^{-1+\de}\label{final-decay-divv-bb-l1}\\
 |\check{\kab}_{\ell=1}|&\les & r^{-1}u^{-1+\de} \label{final-decay-kab-l1} \\
 |\check{\rho}_{\ell=1}|&\les& r^{-2} u^{-1+\de} \label{final-decay-rho-l1}\\
 |\check{\rhoF}_{\ell=1}|&\les& r^{-1} u^{-1+\de} \label{final-decay-rhoF-l1}
 \eea
 \item Using \eqref{final-decay-divv-bF-l1} and \eqref{estimate-check-rhoF-last-slice-l1-optimal-r}, we apply Lemma \ref{optimal-decay-in-r} to \eqref{nabb-4-check-rhoF-ze-eta} with $p=2$, $q_1=0$ and $q_2=\de$ and obtain
 \bea\label{final-decay-check-rhof-optimal-r}
  |\check{\rhoF}_{\ell=1}|&\les& r^{-2} u^{-1/2+\de} 
 \eea
 Similarly, using \eqref{final-decay-divv-bF-l1}, \eqref{final-decay-divv-b-l1}, \eqref{final-decay-check-rhof-optimal-r}, \eqref{estimate-check-rho-last-slice-l1}, we apply Lemma \ref{optimal-decay-in-r} to \eqref{nabb-4-check-rho-ze-eta} with $p=3$, $q_1=0$ and $q_2=\de$ and obtain
 \bea\label{final-decay-check-rho-optimal-r}
  |\check{\rho}_{\ell=1}|&\les& r^{-3} u^{-1/2+\de} 
 \eea
Similarly, using \eqref{final-decay-divv-ze-l1}, \eqref{final-decay-check-rho-optimal-r} and condition $\check{\kab}_{\ell=1}=0$ on $\mathscr{I}_{U, R}$, and integrating \eqref{nabb-4-check-kab-ze-eta} we obtain
\bea\label{final-decay-check-kab-optimal-r}
  |\check{\kab}_{\ell=1}|&\les& r^{-2} u^{-1/2+\de} 
\eea
in  $\{ r \geq r_2\} \cap \{ u \geq u_0 \} \cap I^{-}(S_{U, R})$.
\end{enumerate}

\subsubsection{The projection to the $\ell\geq 2$ modes: optimal decay in $r$ }\label{mega-section-optimal-decay-r}

 Using the decay in $u$ and $r$ along $\mathscr{I}_{U, R}$, we will transport it to the past of $S_{U,R}$ using transport equations. 
We apply Lemma \ref{optimal-decay-in-r} to obtain optimal decay in $r$ (and decay in $u$ as $u^{- 1/ 2 +\de}$), for some of the components.

\begin{enumerate}
\item Applying Lemma \ref{optimal-decay-in-r} to \eqref{nabb-4-check-ka-ze-eta} for $f=\DDs_2\DDs_1(\check{\ka}, 0)$, $F=0$, and using condition \eqref{condition-check-ka-l2}, we obtain
\bea\label{check-ka-=-0}
\DDs_2\DDs_1(\check{\ka}, 0)&=&0 \qquad \qquad \text{in  $\{ r \geq r_2\} \cap \{ u \geq u_0 \} \cap I^{-}(S_{U, R})$}
\eea
\item Using \eqref{estimate-chih-last-slice-optimal-r} and \eqref{estimate-a}, we apply Lemma \ref{optimal-decay-in-r} to \eqref{nabb-4-chih-ze-eta}   for $f=\chih$, $F=-\alpha$, $p=2$, $q_1=0$, $q_2=\de$. We obtain
\bea\label{final-decay-chih}
|\chih|&\les& r^{-2} u^{-1/2+\de}\qquad \text{in  $\{ r \geq r_2\} \cap \{ u \geq u_0 \} \cap I^{-}(S_{U, R})$}
\eea

\item Using \eqref{write-ze-in-terms-of-chih-general}, \eqref{write-bF-in-terms-of-chih-general} and \eqref{write-b-in-terms-of-chih-general}  and the above we obtain 
\bea
|\DDs_2\bF|&\les&  r^{-3-\de} u^{-1/2+\de} \label{decay-dds2-bF-optimal-r}\\
|\DDs_2\b|&\les & r^{-4-\de} u^{-1/2+\de} \label{decay-dds2-b-optimal-r}\\
|\DDs_2\ze|&\les& r^{-3} u^{-1/2+\de} \label{decay-dds2-ze-optimal-r}
\eea

\item Commuting \eqref{nabb-4-check-rhoF-ze-eta} by $\DDs_2\DDs_1$ we obtain
\beaa
\nabb_4(\DDs_2\DDs_1(\check{\rhoF}, \check{\sigmaF}))+2\ka \DDs_2\DDs_1(\check{\rhoF}, \check{\sigmaF})&=& \DDs_2\DDs_1\DDd_1\bF\\
&=& \left(2\DDs_2\DDd_2+2K\right)\DDs_2\bF
\eeaa
We apply Lemma \ref{optimal-decay-in-r} to the above   for $f=\DDs_2\DDs_1(\check{\rhoF}, \check{\sigmaF})$, $p=4$, $q_1=0$, $q_2=\de$. Using \eqref{estimate-check-rhoF-last-slice} we obtain
\bea\label{final-decay-rhoF-l2-optimal-r}
|\DDs_2\DDs_1(\check{\rhoF}, \check{\sigmaF})|&\les& r^{-4} u^{-1/2+\de}
\eea
A similar procedure applies to $\DDs_2\DDs_1(-\check{\rho}, \check{\sigma})$ and $\DDs_2\DDs_1(\check{\kab}, 0)$. We obtain
\bea
|\DDs_2\DDs_1(\check{\rho}, \check{\sigma})|&\les& r^{-5} u^{-1/2+\de} \label{final-decay-rho-l2-optimal-r}\\
|\DDs_2\DDs_1(\check{\kab}, 0)|&\les& r^{-4} u^{-1/2+\de}\label{final-decay-kab-l2-optimal-r}
\eea
in  $\{ r \geq r_2\} \cap \{ u \geq u_0 \} \cap I^{-}(S_{U, R})$.
\end{enumerate}

This completes the proof of the optimal decay in $r$ for the above quantities.

\subsubsection{The projection to the $\ell\geq 2$ modes: optimal decay in $u$}\label{mega-section-optimal-decay-u}

We derive now the decay estimates which imply decay of order $u^{-1+\de}$ for all the curvature components in the linear stability. In order to do so, we make use of two quantities: the first is $\check{\mu}$ as introduced in the gauge normalization, and the second one is a new quantity $\Xi$. Both these quantities have the property that they transport in a very good way from $\mathscr{I}_{U, R}$ in the $e_4$ direction.

\begin{lemma}\label{estimate-for-nu} The mass-charge aspect function $\check{\mu}$ defined by  scalar defined in \eqref{definition-of-mu}
verifies in  $\{ r \geq r_2\} \cap \{ u \geq u_0 \} \cap I^{-}(S_{U, R})$ the following estimate:
\beaa
|\DDs_2\DDs_1(\check{\mu}, 0)|&\les& r^{-2-\de} u^{-1+\de}
\eeaa
\end{lemma}
\begin{proof} Commuting \eqref{nabb-4-check-mu} with $r^2\DDs_2\DDs_1$ and using \eqref{check-ka-=-0}, we obtain
\beaa
\nabb_4(r^2\DDs_2\DDs_1(\check{\mu}, 0))&=& O(r^{-1-\de}u^{-1+\de}) 
\eeaa
Therefore $r^2\DDs_2\DDs_1(\check{\mu}, 0)$ verifies the transport equation 
\beaa
\pr_r(r^2\DDs_2\DDs_1(\check{\mu}, 0))&=& O(r^{-1-\de}u^{-1+\de}) 
\eeaa 
Consider now a fixed $u<U$. As in Lemma \ref{optimal-decay-in-r}, we integrate the above equation along the fixed $u$ hypersurface from the sphere $S_{u, r_*(u)}$ on $\mathscr{I}_{U, R}$ to the sphere $S_{u, r}$ for any $r_0\leq r\leq R $. We obtain
\beaa
r^2\DDs_2\DDs_1(\check{\mu}, 0)&\les & r_*(u)^2\DDs_2\DDs_1(\check{\mu}, 0)(u,r_*(u))+\int_{r_*}^r \lambda^{-1-\de} u^{-1+\de} d\lambda
\eeaa
Because of condition \eqref{condition-nu}, we obtain
\beaa
|r^2\DDs_2\DDs_1(\check{\mu}, 0)(u,r)|&\les & \int_{r_*}^r \lambda^{-1-\de} u^{-1+\de} d\lambda \les  r^{-\de} u^{-1+\de}
\eeaa
as desired.
\end{proof}

In addition to $\check{\mu}$ we define the following quantity, which has the property that verifies a good transport equation in the $e_4$ direction. This quantity generalizes a quantity, also denoted $\Xi$, in \cite{stabilitySchwarzschild}, which had the same purpose.

\begin{lemma}\label{estimate-for-Gamma} The traceless symmetric two tensor defined as 
\bea\label{definition-Xi}
\Xi&:=& \left(\kab r^2 +\rho r^3-\frac 2 3 r^3\rhoF^2 \right)\chih-\ka r^2 \chibh -4r^2 \DDs_2\ze+2r^3 \DDs_2 \b-\frac 2 3 r^3\rhoF \DDs_2 \bF 
\eea
verifies in  $\{ r \geq r_2\} \cap \{ u \geq u_0 \} \cap I^{-}(S_{U, R})$ the following estimate:
\beaa
|\Xi|&\les& u^{-1+\de}
\eeaa
\end{lemma}
\begin{proof} We compute $\nabb_4\Xi$. 
We start by computing $\nabb_4(\kab r^2\chih)$.
\beaa
\nabb_4(\kab r^2\chih)&=& \nabb_4(\kab) r^2\chih+\kab \nabb_4(r^2)\chih+\kab r^2\nabb_4(\chih)\\
&=& (-\frac 1 2 \ka\kab +2\rho) r^2\chih+\ka\kab r^2\chih+\kab r^2(-\ka\chih-\a)= \left(-\frac 1 2 \ka\kab +2\rho\right) r^2\chih-\kab r^2\a
\eeaa
We compute $\nabb_4(\rho r^3\chih)$:
\beaa
\nabb_4(\rho r^3\chih)&=& \nabb_4(\rho) r^3\chih+\rho \nabb_4(r^3)\chih+\rho r^3\nabb_4(\chih)\\
&=& (-\frac 3 2 \ka\rho-\ka\rhoF^2) r^3\chih+\rho \frac 3 2 r^3\ka\chih+\rho r^3(-\ka\chih-\a)= (-2\rho-2\rhoF^2) r^2\chih-\rho r^3\a
\eeaa
We compute $\nabb_4(\frac 2 3 r^3\rhoF^2\chih)$.
\beaa
\nabb_4(\frac 2 3 r^3\rhoF^2\chih)&=& \frac 2 3 \nabb_4(r^3)\rhoF^2\chih+\frac 2 3 r^3\nabb_4(\rhoF^2)\chih+\frac 2 3 r^3\rhoF^2\nabb_4(\chih)\\
&=& r^3 \ka\rhoF^2\chih-\frac 4 3 r^3\ka\rhoF^2\chih+\frac 2 3 r^3\rhoF^2(-\ka\chih-\a)= -2r^2 \rhoF^2\chih-\frac 2 3 r^3\rhoF^2\a
\eeaa
We therefore obtain
\bea\label{nabb-4-gamma-first-part}
\nabb_4\left(\left(\kab r^2 +\rho r^3-\frac 2 3 r^3\rhoF^2 \right)\chih\right)&=& -\frac 1 2 \ka\kab r^2\chih+\left(-\kab r^2-\rho r^3+\frac 2 3 r^3\rhoF^2\right)\a
\eea
We compute $\nabb_4(\ka r^2\chibh)$.
\beaa
\nabb_4(\ka r^2\chibh)&=& \nabb_4(\ka) r^2\chibh+\ka \nabb_4(r^2)\chibh+\ka r^2\nabb_4(\chibh)\\
&=& (-\frac 1 2 \ka^2) r^2\chibh+\ka^2 r^2\chibh+\ka r^2(-\frac 1 2 \ka\chibh+2\DDs_2\ze-\frac 1 2 \kab \chih)=-\frac 1 2 r^2\ka\kab \chih+ 2\ka r^2 \DDs_2\ze
\eeaa
Putting it together with \eqref{nabb-4-gamma-first-part}, we obtain
\bea\label{nabb-4-gamma-second-part}
\nabb_4\left(\left(\kab r^2 +\rho r^3-\frac 2 3 r^3\rhoF^2 \right)\chih-\ka r^2\chibh\right)&=& - 2\ka r^2 \DDs_2\ze+\left(-\kab r^2-\rho r^3+\frac 2 3 r^3\rhoF^2\right)\a
\eea

Commuting \eqref{nabb-4-ze-ze-eta} with $r\DDs_2$ we obtain
\beaa
\nabb_4 (r\DDs_2\ze)+\ka r\DDs_2\ze &=&-r\DDs_2\b  -\rhoF r\DDs_2\bF\\
\nabb_4 (r\DDs_2\ze)+\frac 1 2 \ka r\DDs_2\ze &=&-\frac 1 2 \ka r\DDs_2\ze-r\DDs_2\b  -\rhoF r\DDs_2\bF
\eeaa
which can be written as 
\bea\label{nabb-4-gamma-third-part}
\nabb_4 (r^2\DDs_2\ze) &=&-\frac 1 2 \ka r^2\DDs_2\ze-r^2\DDs_2\b  -\rhoF r^2\DDs_2\bF
\eea
Combining \eqref{nabb-4-gamma-second-part} and \eqref{nabb-4-gamma-third-part} we obtain
\beaa
\nabb_4\left(\left(\kab r^2 +\rho r^3-\frac 2 3 r^3\rhoF^2 \right)\chih-\ka r^2\chibh-4r^2\DDs_2\ze\right)&=&4r^2\DDs_2\b  +4\rhoF r^2\DDs_2\bF\\
&&+ \left(-\kab r^2-\rho r^3+\frac 2 3 r^3\rhoF^2\right)\a
\eeaa

Commuting \eqref{nabb-4-b-ze-eta} with $r\DDs_2$ we obtain
\beaa
\nabb_4(r\DDs_2 \b)+ 2\ka r\DDs_2\b &=&r\DDs_2\DDd_2\a +\rhoF \nabb_4(r\DDs_2\bF)\\
\nabb_4(r\DDs_2 \b)+ \ka r\DDs_2\b &=&-\ka r\DDs_2\b+r\DDs_2\DDd_2\a +\rhoF \nabb_4(r\DDs_2\bF)
\eeaa
which can be written as 
\beaa
\nabb_4(r^3\DDs_2 \b) &=&-\ka r^3\DDs_2\b+r^3\DDs_2\DDd_2\a +r^2\rhoF \nabb_4(r\DDs_2\bF)
\eeaa
Since $r^2\rhoF=Q$, we obtain
\bea\label{nabb-4-gamma-fourth-part}
\nabb_4(r^3\DDs_2 \b-r^3\rhoF\DDs_2\bF) &=&-2 r^2\DDs_2\b+r^3\DDs_2\DDd_2\a  
\eea
Commuting \eqref{nabb-4-bF-tilde-b} with $r\DDs_2$ we obtain
\beaa
\nabb_4(r\DDs_2 \bF)+ \frac 3 2 \ka r\DDs_2\bF&=&( -3\rho+2\rhoF^2)^{-1}\left(\nabb_4 (r\DDs_2\tilde{\b}) +3\ka r\DDs_2\tilde{\b}- 2\rhoF r\DDs_2\DDd_2 \a\right)
\eeaa
We therefore have
\beaa
\nabb_4(r^3\rhoF\DDs_2\bF)&=& Q\nabb_4( r\DDs_2\bF)\\
&=&Q\left(- \frac 3 2 \ka r\DDs_2\bF+( -3\rho+2\rhoF^2)^{-1}\left(\nabb_4 (r\DDs_2\tilde{\b}) +3\ka r\DDs_2\tilde{\b}- 2\rhoF r\DDs_2\DDd_2 \a\right)\right)\\
&=&- 3 r^2\rhoF\DDs_2\bF+Q( -3\rho+2\rhoF^2)^{-1}\left(\nabb_4 (r\DDs_2\tilde{\b}) +3\ka r\DDs_2\tilde{\b}- 2\rhoF r\DDs_2\DDd_2 \a\right)
\eeaa
We can finally put together the above computations, and obtain
\beaa
\nabb_4\Xi&=&  \left(-\kab r^2-\rho r^3+\frac 2 3 r^3\rhoF^2\right)\a+2(r^3\DDs_2\DDd_2\a)\\
&&+\frac 4 3Q( -3\rho+2\rhoF^2)^{-1}\left(\nabb_4 (r\DDs_2\tilde{\b}) +3\ka r\DDs_2\tilde{\b}- 2\rhoF r\DDs_2\DDd_2 \a\right)
\eeaa
Using the estimate for $\a$ and $\tilde{\b}$ given by \eqref{estimate-a} and \eqref{estimate-derivatives-tilde-b}, we can bound the right hand side of the above by 
\beaa
 |r\a|+|r^2\tilde{\b}| \les r^{-1-\de}u^{-1+\de}
\eeaa
We integrate the above equation along the fixed $u$ hypersurface from the sphere $S_{u, r_*(u)}$ on $\mathscr{I}_{U, R}$ to the sphere $S_{u, r}$ for any $r_0\leq r\leq R $. We obtain
\beaa
\Xi(u,r)&\les & \Xi(u,r_*(u))+\int_{r_*}^r \lambda^{-1-\de}u^{-1+\de}d\lambda
\eeaa
On $\mathscr{I}_{U, R}$, the quantity $\Gamma$ verifies the following estimate
\beaa
|\Xi(u,r_*(u))|&\leq& |r\chih|+|r\chibh|+|r^2\DDs_2\ze|+|r^3\DDs_2\b|+|r^{-1}\DDs_2\bF|\les  u^{-1+\de}
\eeaa
where we used \eqref{estimate-chih-last-slice}, \eqref{estimate-chibh-last-slice}, \eqref{estimate-ze-last-slice}, \eqref{estimate-b-last-slice} and \eqref{estimate-bf-last-slice}. 
Therefore we finally obtain
\beaa
|\Xi(u,r)|&\les &  u^{-1+\de}+\int_{r_*}^r \lambda^{-1-\de} u^{-1+\de} d\lambda \les   u^{-1+\de}
\eeaa
as desired.
\end{proof}

We now show that the decay for $\DDs_2\DDs_1(\check{\mu}, 0)$ and $\Xi$ implies decay for all the remaining quantities.
We write $\Xi$ using the expressions given by Proposition \ref{expression-all-quantities-l2}. We have, using \eqref{check-ka-=-0}:
 \beaa
 \Xi&=&\left(\kab r^2 +\rho r^3-\frac 2 3 r^3\rhoF^2 \right)\chih-\ka r^2 \chibh -4r^2 \DDs_2\ze+2r^3 \DDs_2 \b-\frac 2 3 r^3\rhoF \DDs_2 \bF\\
 &=& \left(\kab r^2 +\rho r^3-\frac 2 3 r^3\rhoF^2 \right)\chih-\ka r^2 \chibh -4r^3\left(\DDs_2\DDd_2 \chih+\left(-\frac 3 2\rho+\rhoF^2\right) \chih\right) \\
 &&+2r^3 \left(-\frac 3 2\rho \chih\right)-\frac 2 3 r^3\rhoF (-\rhoF \chih)+  O(r^{-\de} u^{-1+\de})\\
 &=& \left(\kab r^2 +4\rho r^3 -4\rhoF^2 r^3\right)\chih-\ka r^2 \chibh -4r^3\left(\DDs_2\DDd_2 \chih\right) +  O(r^{-\de} u^{-1+\de})
 \eeaa
  Recalling the estimate for $\Xi$ given by Lemma \ref{estimate-for-Gamma}, we can write $\chibh$ in terms of $\chih$:
  \beaa
     \chibh&=& r^2\left(-2\DDs_2\DDd_2 +\frac 1 4 \ka\kab  +2\rho  -2\rhoF^2 \right)\chih  +  O(r^{-1} u^{-1+\de})\\
     &=& r^2\left(-2\DDs_2\DDd_2 -K+\rho-\rhoF^2 \right)\chih  +  O(r^{-1} u^{-1+\de})
     \eeaa
     finally giving 
     \bea
     \chibh   &=& r^2\left(-2\DDs_2\DDd_2   -2K-\frac 1 4 \ka\kab \right)\chih  +  O(r^{-1} u^{-1+\de}) \label{write-chibh-in-terms-of-chih}
        \eea
Recall relation \eqref{nu-in-terms-of-chi-chib-general}. We use the expressions given by Proposition \ref{expression-all-quantities-l2}. We have, using \eqref{check-ka-=-0} and \eqref{write-chibh-in-terms-of-chih}:
  \beaa
\DDs_2\DDs_1(\check{\mu}, 0) &=& r^4 \left(2\DDs_2\DDd_2+2K\right)\left(\DDs_2\DDd_2 \chih\right)+r^4 \left(2\DDs_2\DDd_2+2K\right)\left(-\frac 3 2\rho+3\rhoF^2\right) \chih \\
 &&+ r^3\left(-\frac 3 4 \rho+\frac 3 2 \rhoF^2\right) \kab \chih + r^4\left(\frac 3 2 \rho- 3  \rhoF^2\right)    \left(2\DDs_2\DDd_2   +2K+\frac 1 4 \ka\kab \right)\chih   +O(r^{-1} u^{-1+\de})\\
 &=& r^4 \left(2\DDs_2\DDd_2+2K\right)\left(\DDs_2\DDd_2 \chih\right)    +O(r^{-1} u^{-1+\de})
 \eeaa
 Using the estimate for $\DDs_2\DDs_1(\check{\mu}, 0)$ obtained in Lemma \ref{estimate-for-nu}, we have
 \beaa
  \left(2\DDs_2\DDd_2+2K\right)\left(\DDs_2\DDd_2 \chih\right)   &=& \DDs_2\DDs_1\DDd_1\DDd_2\chih=O(r^{-5} u^{-1+\de})
 \eeaa
The above operator is clearly coercive. Indeed,
\beaa
\int_{S} \chih \cdot \DDs_2\DDs_1\DDd_1\DDd_2\chih&=& \int_{S} |\DDd_1\DDd_2\chih|^2 \geq \frac{1}{r^4} \int_{S} |\chih|^2
\eeaa
This gives 
\bea\label{final-decay-chih-optimal-u}
|\chih| \les r^{-1} u^{-1+\de} \qquad \text{in  $\{ r \geq r_2\} \cap \{ u \geq u_0 \} \cap I^{-}(S_{U, R})$}
\eea
Relation \eqref{write-chibh-in-terms-of-chih} then implies
\bea
|\chibh| \les r^{-1} u^{-1+\de} \qquad \text{in  $\{ r \geq r_2\} \cap \{ u \geq u_0 \} \cap I^{-}(S_{U, R})$}
\eea

Proposition \ref{expression-all-quantities-l2} then implies in  $\{ r \geq r_2\} \cap \{ u \geq u_0 \} \cap I^{-}(S_{U, R})$:
\bea
|\DDs_2\bF|&\les&  r^{-2-\de} u^{-1+\de}\label{decay-dds2-bF}\\
|\DDs_2\bbF|&\les&  r^{-2} u^{-1+\de}\label{decay-dds2-bbF}\\
|\DDs_2\b|&\les& r^{-3-\de} u^{-1+\de} \label{decay-dds2-b}\\
|\DDs_2\bb|&\les&  r^{-3} u^{-1+\de}  \label{decay-dds2-bb}\\
|\DDs_2\ze|&\les& r^{-2} u^{-1+\de}   \label{decay-dds2-ze}\\
|\DDs_2\DDs_1(\check{\rhoF}, \check{\sigmaF})|&\les&   r^{-3} u^{-1+\de} \label{decay-dds2-rhoF}\\
| \DDs_2\DDs_1(-\check{\rho}, \check{\sigma})| &\les& r^{-4} u^{-1+\de} \label{final-decay-rho-l2-optimal-u}\\
 |\DDs_2\DDs_1( \check{\kab}, 0)| &\les& r^{-3} u^{-1+\de}\label{final-decay-kab-l2-optimal-u}
\eea

We combine the estimates obtained in the separated case of projection to the $\ell=1$ and $\ell\geq 2$ spherical harmonics using the elliptic estimates given by Lemma \ref{main-elliptic-estimate}, to obtain the following estimates in the region in  $\{ r \geq r_2\} \cap \{ u \geq u_0 \} \cap I^{-}(S_{U, R})$. Recall the estimates for the curl part obtained in Section \ref{projection-l1-boundedness}.

\begin{enumerate}
\item Combining \eqref{final-decay-divv-b-l1}, \eqref{decay-curl-b-l1}, \eqref{decay-dds2-b-optimal-r} and \eqref{decay-dds2-b} we obtain
\beaa
|\b|&\les& \min\{r^{-3-\de} u^{-1/2+\de}, r^{-2-\de} u^{-1+\de} \}
\eeaa
\item Combining \eqref{final-decay-check-rho-optimal-r}, \eqref{final-decay-rho-l1}, \eqref{final-decay-rho-l2-optimal-r} and \eqref{final-decay-rho-l2-optimal-u} we obtain
\beaa
|\check{\rho}|&\les& \min\{r^{-3} u^{-1/2+\de}, r^{-2} u^{-1+\de} \}
\eeaa
\item Combining \eqref{decay-check-sigma-l1-optimal-u}, \eqref{decay-check-sigma-l1-optimal-r}, \eqref{final-decay-rho-l2-optimal-r} and \eqref{final-decay-rho-l2-optimal-u} we obtain
\beaa
|\check{\sigma}|&\les& \min\{r^{-3} u^{-1/2+\de}, r^{-2} u^{-1+\de} \}
\eeaa
\item Combining \eqref{final-decay-divv-bF-l1}, \eqref{decay-curll-bF-from-initial-data}, \eqref{decay-dds2-bF-optimal-r} and \eqref{decay-dds2-bF} we obtain
\beaa
|\bF|&\les& \min\{r^{-2-\de} u^{-1/2+\de}, r^{-1-\de} u^{-1+\de} \}
\eeaa
\item Combining \eqref{final-decay-chih} and \eqref{final-decay-chih-optimal-u} we obtain
\beaa
|\chih|&\les& \min\{r^{-2} u^{-1/2+\de}, r^{-1} u^{-1+\de} \}
\eeaa
\item Combining \eqref{final-decay-divv-ze-l1}, \eqref{decay-curl-ze-l1}, \eqref{decay-dds2-ze-optimal-r} and \eqref{decay-dds2-b} we obtain
\beaa
|\ze|&\les& \min\{r^{-2} u^{-1/2+\de}, r^{-1} u^{-1+\de} \}
\eeaa
\item Combining \eqref{final-decay-rhoF-l1}, \eqref{final-decay-check-rhof-optimal-r}, \eqref{final-decay-rhoF-l2-optimal-r} and \eqref{decay-dds2-rhoF} we obtain
\beaa
|\check{\rhoF}|&\les& \min\{r^{-2} u^{-1/2+\de}, r^{-1} u^{-1+\de} \}
\eeaa
\item Combining \eqref{decay-check-sigmaF-l1-optimal-u}, \eqref{decay-check-sigmaF-l1-optimal-r}, \eqref{final-decay-rhoF-l2-optimal-r} and \eqref{decay-dds2-rhoF} we obtain
\beaa
|\check{\sigmaF}|&\les& \min\{r^{-2} u^{-1/2+\de}, r^{-1} u^{-1+\de} \}
\eeaa
\item Combining \eqref{final-decay-kab-l1}, \eqref{final-decay-check-kab-optimal-r}, \eqref{final-decay-kab-l2-optimal-r} and \eqref{final-decay-kab-l2-optimal-u} we obtain
\beaa
|\check{\kab}|&\les& \min\{r^{-2} u^{-1/2+\de}, r^{-1} u^{-1+\de} \}
\eeaa
\item Using Gauss equation \eqref{Gauss-check} and the above estimates for $\check{\kab}$, $\check{\rho}$, $\check{\rhoF}$ we obtain
\beaa
|\check{K}|&\les& \min\{r^{-3} u^{-1/2+\de}, r^{-2} u^{-1+\de} \}
\eeaa
\item Combining \eqref{final-decay-divv-bb-l1} and \eqref{decay-dds2-bb} we obtain
\beaa
|\bb|\les r^{-2} u^{-1+\de}
\eeaa
\item Combining \eqref{final-decay-divv-bbF-l1} and \eqref{decay-dds2-bbF} we obtain
\beaa
|\bbF|\les r^{-1} u^{-1+\de}
\eeaa
\end{enumerate}
in  $\{ r \geq r_2\} \cap \{ u \geq u_0 \} \cap I^{-}(S_{U, R})$.

 \subsubsection{The terms involved in the $e_3$ direction}
 We obtain optimal decay for $\check{\omb}$, $\eta$, $\xib$.

\begin{enumerate}
\item Integrating \eqref{nabb-4-eta-ze-eta} and using \eqref{estimate-eta-last-slice-1} we obtain
\bea
|\eta|\les r^{-1} u^{-1+\de}
\eea
\item Integrating \eqref{nabb-4-check-omb-ze-eta} and using \eqref{estimate-omb-last-slice-1} we obtain
\bea
|\check{\omb}|\les r^{-1} u^{-1+\de}
\eea
\item Integrating \eqref{nabb-4-xib-ze-eta} and using \eqref{estimate-xib-last-slice-1} we obtain 
\bea
|\xib|\les r^{-1} u^{-1+\de}
\eea
\end{enumerate}
in  $\{ r \geq r_2\} \cap \{ u \geq u_0 \} \cap I^{-}(S_{U, R})$.
 
 \subsubsection{The metric coefficients}
 We finally derive here the decay for the metric coefficients $\hat{\slashed{g}}$, $\underline{b}$, $\check{\Omegab}$ and $\check{\vsi}$.

\begin{enumerate}
\item Integrating \eqref{nabb-4-g} and using \eqref{estimate-slashed-g-last-slice} we obtain
\beaa
|\hat{\slashed{g}}|&\les& \min\{r^{-1} u^{-1/2+\de},  u^{-1+\de} \}
\eeaa
\item Integrating \eqref{nabb-4-check-tr} and using \eqref{estimate-tr-l1-last-slice} we obtain
 \beaa
|\check{\tr_\gamma \slashed{g}}|\les   \min\{r^{-1-\de} u^{-1/2+\de}, u^{-1+\de}\} 
\eeaa

\item Integrating \eqref{derivatives-b-ze-eta} we obtain
\beaa
|\underline{b}|\les u^{-1+\de}
\eeaa
\item  Equation \eqref{nabb-check-vsi} implies 
\beaa
|\check{\vsi}|\les u^{-1+\de} 
\eeaa
\item Equation \eqref{xib-Omegab-ze-eta} implies
\beaa
| \check{\Omegab}|\les  u^{-1+\de} 
\eeaa
\end{enumerate}
in  $\{ r \geq r_2\} \cap \{ u \geq u_0 \} \cap I^{-}(S_{U, R})$.

\subsubsection{Decay close to the horizon}\label{section-close-horizon}
Recall that we fixed $r_2> r_{\mathcal{H}}$, for some $r_2> r_{\mathcal{H}}$. As explained in Section \ref{Outgoing-coords}, the Bondi coordinates $(u, r, \th, \phi)$ do not cover the boundary of the exterior of the spacetime, and therefore they do not cover the horizon. Moreover, the null frame $\mathscr{N}$ is not regular towards the horizon, while $\mathscr{N}_*=\{ \Omegab^{-1}e_3, \Omegab e_4\}$ is regular towards the horizon.

As a consequence of the proof of decay in the previous section, we obtained bounds for the components defined with respect to the null frame $\mathscr{N}$ in the region in  $\{ r \geq r_2\} \cap \{ u \geq u_0 \} \cap I^{-}(S_{U, R})$, which can be translated into bounds in inverse powers of $u$ along the time-like hypersurface $\{ r = r_2 \}$.

In order to make sense of similar bounds in the region close to the horizon, we consider the rescaling of the components as defined in terms of the frame $\mathscr{N}_*$, which is regular towards the horizon.  Such rescaling is a conformal one depending on the lapse function $\Omegab$.  
For example, the following quantities are regular towards the horizon:
\beaa
\Omegab^2\a,\quad \Omegab^{-2}\aa, \quad \Omegab\chih, \quad \Omegab^{-1}\chibh, \qquad  \hat{\slashed{g}}, \quad \ze, \quad \eta,  \quad \Omegab^{-2}\xib, \quad \Omegab\b, \quad \Omegab^{-1}\bb, \quad \Omegab\bF, \quad \Omegab^{-1}\bbF, \quad \underline{b} 
\eeaa
which can be interpreted as the respective quantities defined in terms of the regular null frame $\mathscr{N}_*$. 

In the region $\{ r_{\mathcal{H} }\leq r \leq r_2 \}$ for bounded $r$, all the quantities will be estimated in terms of $v^{-1+\de}$ where $(v, r)$ are the ingoing Eddington-Finkelstein coordinates as defined in Section \ref{ingoing-coords}. The estimates can be obtained by integration forward the horizon from the timelike hypersurface $\mathcal{T}=\{ r=r_2\}$, in the same manner as in Theorem M5 of \cite{stabilitySchwarzschild}. 

We sketch here the proof. Recall that we obtain the above decay up to the timelike hypersurface $\mathcal{T}$ with respect to the null frame $\mathscr{N}$. The regular null frame $\mathscr{N}_*$ is obtained as a conformal renormalization of the null frame $\mathscr{N}$, and the conformal factor is regular away from the horizon at $r=r_2$. In particular, we can transfer the decay estimates obtained for all the quantities defined with respect to $\mathscr{N}$ to all the quantities defined with respect to $\mathscr{N}_*$ along $\mathcal{T}$. In ingoing Eddington-Finkelstein coordinates, this translates on $\{ r=r_2\}$ as decay of all quantities as $v^{-1+\de}$. Using the decay as $v^{-1+\de}$ of the gauge-invariant quantities in the region $\{ r_{\mathcal{H}} \leq r \leq r_2\}$ as given by Theorem \ref{estimates-theorem}, and integrating toward the horizon along the $e^*_3$ direction on an ingoing null line, we obtain a uniform decay in $v^{-1+\de}$ for all the quantities. The hierarchical structure of the transport equations which are used is similar to the one used in the previous section. See also the Theorem M5 of \cite{stabilitySchwarzschild}. 

Note that since the region close to the horizon is bounded in $r$, the decay in inverse powers of $r$, which was fundamental in the unbounded region, is now completely irrelevant. In particular, all the components present the same decay in inverse powers of $v$, which is consistent with non-linear applications.

One may wonder if the use of the redshift vector field plays any role in closing non-degenerate decay estimates along the event horizon. The redshift vector field estimates have been used in deriving non-degenerate estimates for the gauge invariant quantities $\a$, $\aa$, $\ff$, $\underline{\ff}$, $\tilde{\b}$, $\tilde{\bb}$ as recalled in Theorem \ref{estimates-theorem}. The non-degeneracy of those estimates along the event horizon allows the integration forward from $\{ r=r_2 \}$ up to and including the horizon.

\subsubsection{Conclusion of the proof}\label{conclusions}

We now summarize the previous steps and conclude the proof of the Theorem. 

Recall that for $(U, R)$ such that $U \gg u_0$, $R \gg r_0$, one can choose $r_1$ such that $r_{\mathcal{H}} < r_1 < \tilde{r}$ where $\tilde{r}$ is the radius of the sphere of intersection of $\{ u=U\} $ and the ingoing initial data $\underline{C}_0$.  Apply the proof of boundedness of Section 11 which holds in the region $\{ r\geq r_1\} \cap \{ u_0 \leq u \leq u_1 \}$, where $u_1$ is the $u$-coordinate of the sphere of intersection between $\underline{C}_0$ and $\{ r=r_1\}$. By construction $S_{U,R}$ belongs to this region, and therefore we can associate to it the $S_{U, R}$-normalized solution. Using the proof of decay from the far-away normalized solution at $S_{U, R}$ integrating backward till $\{ r = r_2\}$ for some $ \tilde{r} < r_2 < R$, one obtains decay (independent of $U$ and $R$) in the region  $\{ r \geq r_2\} \cap \{ u \geq u_0 \} \cap I^{-}(S_{U, R})$. By integrating forward from $r_2$ toward the horizon, one obtains decay in $v$ in the region close to the horizon and past the ingoing null cone of $S_{U, r_2}$. In Figure 1, the region where the estimates are derived is colored in gray. 
Observe that the region where the estimates hold is finally independent of the choice of $r_1$. 

Since the estimate on this region do not depend on $U$ and $R$, in the above construction $U$ and $R$ can be taken to be arbitrarily large. As $U$ and $R$ vary to arbitrarily large values, the regions constructed as the above cover the whole exterior region in the future of $C_0 \cup \underline{C}_0$, and therefore the above proves the validity of the estimates in the whole exterior region. 

\appendix
 
 \section{Explicit computations}\label{appendix}

 \subsection{Proof of Lemma \ref{lemma-coords}}\label{app-sec}
 Recall the general coordinate transformation
\beaa
\tilde{u}&=& u+\ep g_1(u, r, \th, \phi), \qquad \tilde{r}= r+\ep g_2(u, r, \th, \phi), \qquad \tilde{\th}= \th+\ep g_3(u, r, \th, \phi), \qquad \tilde{\phi}= \phi+\ep g_4(u, r, \th, \phi) 
\eeaa
We compute the differentials:
\beaa
d\tilde{u}&=& du+\ep \left( (g_1)_u du + (g_1)_r dr+(g_1)_\th d\th +(g_1)_\phi d\phi  \right)\\
d\tilde{r}&=& dr+\ep \left( (g_2)_u du + (g_2)_r dr+(g_2)_\th d\th +(g_2)_\phi d\phi  \right) \\
d\tilde{\th}&=& d\th+\ep \left( (g_3)_u du + (g_3)_r dr+(g_3)_\th d\th +(g_3)_\phi d\phi  \right) \\
d\tilde{\phi}&=& d\phi+\ep \left( (g_4)_u du + (g_4)_r dr+(g_4)_\th d\th +(g_4)_\phi d\phi  \right)
\eeaa
We also linearize the functions
\beaa
\tilde{r}^2&=& r^2+\ep \left(2 r g_2 \right), \qquad \underline{\Omega}(\tilde{r})= \underline{\Omega}(r) +\ep \left(\partial_r \underline{\Omega} \cdot g_2\right), \qquad \sin^2\tilde{\th}= \sin^2\th+\ep \left(2\sin\th\cos\th g_3 \right)
\eeaa
The linear expansion of the metric \eqref{Schw:EF-coordinates-out} becomes
\beaa
g_{M, Q}&=&- 2 du dr+\ep\Big( -2(g_2)_u du^2 -2 (g_2)_r du dr-2(g_2)_\th du d\th -2(g_2)_\phi du d\phi\\
&&-2 (g_1)_u dudr -2 (g_1)_r dr^2-2(g_1)_\th drd\th -2(g_1)_\phi drd\phi \Big)  \\
&&+\underline{\Omega}(r)  du^2+  \ep \Big( 2\underline{\Omega}(r)(g_1)_u du^2 + 2\underline{\Omega}(r)(g_1)_r dudr+2\underline{\Omega}(r)(g_1)_\th dud\th \\
&&+2\underline{\Omega}(r)(g_1)_\phi dud\phi  \Big) +\ep \left(\partial_r \underline{\Omega} \cdot g_2\right) (du^2) \\
&& + r^2(d\th^2+\sin^2\th d\phi^2+\ep \left( 2(g_3)_u dud\th + 2(g_3)_r drd\th+2(g_3)_\th d\th^2 +2(g_3)_\phi d\th d\phi  \right)\\
&&+\ep\Big( 2\sin^2\th(g_4)_u dud\phi + 2\sin^2\th(g_4)_r drd\phi+2\sin^2\th(g_4)_\th d\th d\phi \\
&&+2\sin^2\th(g_4)_\phi d\phi^2 +2\sin\th\cos\th g_3 d\phi^2 \Big))+\ep \left(2 r g_2 \right)(d\th^2+\sin^2\th d\phi^2)
\eeaa

The metric in Bondi gauge is of the form \eqref{double-null-metric}. 
  In particular the terms  $dr^2$, $dr d\th$, $dr d\phi$ do not appear in the Bondi form.
    \begin{itemize}
  \item The only term in $dr^2$ is  $\ep (-2 (g_1)_r dr^2)=0$
  which then implies 
  \bea\label{g1-not-r}
  \partial_r(g_1)&=& 0, \qquad g_1=g_1(u, \th, \phi)
  \eea
  \item The terms in $dr d\th$ are $  -2(g_1)_\th drd\th+ 2r^2(g_3)_r drd\th=0$
  which gives $  (g_3)_r =\frac{1}{r^2}(g_1)_\th $
  and therefore, using \eqref{g1-not-r}
  \bea\label{expression-g3}
   g_3 =-\frac{1}{r}(g_1)_\th(u, \th, \phi)+j_3(u, \th, \phi) \qquad \text{for some function $j_3(u, \th, \phi)$.}
  \eea
  \item The terms in $dr d\phi$ are $  -2(g_1)_\phi drd\phi+2r^2\sin^2\th(g_4)_r drd\phi=0$
  which gives $   (g_4)_r =\frac{1}{r^2\sin^2\th}(g_1)_\phi$
  and therefore
  \bea\label{expression-g4}
   g_4 =-\frac{1}{r\sin^2\th}(g_1)_\phi(u, \th, \phi)+j_4(u, \th, \phi)
  \eea
  \end{itemize}
  The Bondi form of the metric also imposes $\nabb_4(\varsigma)=0$. We first compute $\varsigma$, which we can read off from the term $du dr$. The terms with $du dr$ are 
\beaa
-2du dr+\ep (-2(g_2)_r-2(g_1)_u+2\underline{\Omega}(r)(g_1)_r)du dr=(-2+\ep (-2(g_2)_r-2(g_1)_u))du dr
\eeaa
which gives
\bea\label{expression-varsigma}
\varsigma&=& 1+\ep ((g_2)_r+(g_1)_u)
\eea
 From \eqref{expression-varsigma}, we obtain 
    \beaa
\partial_r\varsigma&=& \partial_r( 1+\ep ((g_2)_r+(g_1)_u))=\ep \partial_r^2 (g_2)=0
\eeaa
which implies
\bea\label{expres-g2}
g_2&=& r \cdot w_1(u, \th, \phi)+w_2(u, \th, \phi)
\eea
Putting together \eqref{g1-not-r}, \eqref{expression-g3}, \eqref{expression-g4} and \eqref{expres-g2}, we obtain the general expression for coordinate transformations preserving the Bondi metric, proving Lemma \ref{lemma-coords}.

\subsection{Alternative expressions for $\qf^\F$ and $\pf$}
We summarize in the following lemma alternative expressions for $\qf^\F$ and $\pf$ which differ from their definitions given in Section \ref{transformation-theory-1}.

\begin{lemma} The following relations hold true:
\bea
\frac{\qf^\F}{r^3}&=& -\DDs_2\DDs_1(\check{\rhoF}, \check{\sigmaF}) -\frac 12 \rhoF \left(  \kab \chih + \ka \chibh \right), \label{relation-qf-rhoF-chih-chibh} \\
\frac{\pf}{r^5} &=& 2\rhoF \DDs_1(-\check{\rho}, \check{\sigma}) +(3\rho-2\rhoF^2) \DDs_1(\check{\rhoF}, \check{\sigmaF})+2 \rhoF^2 (\kab\bF-\ka\bbF) \label{relation-pf-rho-rhoF-bF-bbF} 
\eea
\end{lemma}
\begin{proof} The gauge-invariant quantity $\qf^\F$ is by definition \eqref{quantities-1} given by
\beaa
\qf^\F&=& \frac 1 2  r (r^2 \kab\ff)+\frac{1}{\kab} r\nabb_3(r^2 \kab\ff)=\frac 1 2  r (r^2 \kab\ff)+\frac{1}{\kab} r\left(r^2 \kab \nabb_3(\ff)+\left(\frac {1}{ 2}  \kab -2 \omb\right)r^2\kab \ff \right)\\
&=& r^3\left( \nabb_3(\ff)+\left( \kab -2 \omb\right)\ff \right)
\eeaa
Using the definition of $\ff$ \eqref{definition-ff}, we compute
\beaa
\nabb_3(\ff)+\left( \kab -2 \omb\right)\ff&=& \nabb_3(\DDs_2 \bF+\rhoF \chih)+\left( \kab -2 \omb\right)(\DDs_2 \bF+\rhoF \chih)\\
&=& \DDs_2\nabb_3 \bF-\frac 1 2 \kab \DDs_2 \bF+\nabb_3\rhoF \chih+\rhoF \nabb_3\chih+\left( \kab -2 \omb\right)(\DDs_2 \bF+\rhoF \chih)
\eeaa
and using \eqref{nabb-3-bF-ze-eta} and \eqref{nabb-3-chih-ze-eta}, we obtain
\beaa
\nabb_3(\ff)+\left( \kab -2 \omb\right)\ff&=& -\left(\frac 1 2 \kab-2\omb\right) \DDs_2\bF -\DDs_2\DDs_1(\check{\rhoF}, \check{\sigmaF}) +2\rhoF \DDs_2\eta-\frac 1 2 \kab \DDs_2 \bF-\kab \rhoF \chih\\
&&+\rhoF (-\left(\frac 12  \kab -2 \omb\right) \chih   -2 \slashed{\mathcal{D}}_2^\star \eta-\frac 1 2 \ka \chibh)+\left( \kab -2 \omb\right)(\DDs_2 \bF+\rhoF \chih)\\
&=&  -\DDs_2\DDs_1(\check{\rhoF}, \check{\sigmaF}) -\frac 12\rhoF \left(  \kab \chih + \ka \chibh \right)
\eeaa
which proves \eqref{relation-qf-rhoF-chih-chibh}.

The gauge-invariant quantity $\pf$ is by definition \eqref{quantities-1} given by 
\beaa
\pf&=& \frac 1 2  r (r^4 \kab\tilde{\b})+\frac{1}{\kab} r\nabb_3(r^4 \kab\tilde{\b})=\frac 1 2  r (r^4 \kab\tilde{\b})+\frac{1}{\kab} r(\frac 3 2 \kab -2\omb)r^4\kab\tilde{\b}+\frac{1}{\kab} rr^4 \kab \nabb_3(\tilde{\b})\\
&=&  r^5 \left(\nabb_3(\tilde{\b})+(2 \kab -2\omb)\tilde{\b}\right)
\eeaa
Using the definition of $\tilde{\b}$ \eqref{definition-tilde-b}, we compute 
\beaa
\nabb_3(\tilde{\b})+(2 \kab -2\omb)\tilde{\b}&=& \nabb_3(2\rhoF \b-3\rho\bF)+(2 \kab -2\omb)(2\rhoF \b-3\rho\bF)\\
&=& 2\nabb_3(\rhoF) \b+2\rhoF \nabb_3(\b)-3\nabb_3(\rho)\bF-3\rho\nabb_3(\bF)\\
&&+(2 \kab -2\omb)(2\rhoF \b-3\rho\bF)\\
&=& -2\kab\rhoF \b+2\rhoF(-\left(\kab -2\omb\right) \b +\DDs_1(-\check{\rho}, \check{\sigma}) +3 \rho \ze\\
&& +\rhoF\left(-\DDs_1(\check{\rhoF}, \check{\sigmaF})  - \ka\bbF-\frac 1 2 \kab \bF    \right))\\
&&-3(-\frac 3 2 \kab \rho-\kab \rhoF^2)\bF-3\rho(-\left(\frac 1 2 \kab-2\omb\right) \bF -\DDs_1(\check{\rhoF}, \check{\sigmaF}) +2\rhoF \ze)\\
&&+(2 \kab -2\omb)(2\rhoF \b-3\rho\bF)\\
&=& 2\rhoF \DDs_1(-\check{\rho}, \check{\sigma}) +(3\rho-2\rhoF^2) \DDs_1(\check{\rhoF}, \check{\sigmaF})+2 \rhoF^2 (\kab\bF-\ka\bbF) 
\eeaa
which proves \eqref{relation-pf-rho-rhoF-bF-bbF}. 

\end{proof}

\subsection{Remarkable transport equations}\label{section-transport-eqs-appendix}
We summarize here some transport equations which are important in the proof of linear stability.  

\begin{lemma} The following transport equations hold:
\bea
\nabb_4 \bF+ \frac 3 2 \ka\bF&=&( -3\rho+2\rhoF^2)^{-1}\left(\nabb_4 \tilde{\b} +3\ka \tilde{\b}- 2\rhoF \divv \a\right) \label{nabb-4-bF-tilde-b}
\eea
\beaa
\nabb_3 \bbF+ \left(\frac 3 2 \kab+2\omb\right)\bbF +2\rhoF \xib&=& ( -3\rho+2\rhoF^2)^{-1}\left(\nabb_3 \tilde{\bb} +(3\kab+2\omb) \tilde{\bb}+2\rhoF \divv \aa \right) 
\eeaa
\end{lemma}
\begin{proof} We compute, using Maxwell equations and Bianchi identities:
\beaa
\nabb_4 \tilde{\b} +3\ka \tilde{\b}&=& \nabb_4 (2\rhoF \b-3\rho \bF) +3\ka (2\rhoF \b-3\rho \bF)\\
&=&  2\nabb_4\rhoF \b+2\rhoF \nabb_4\b-3\nabb_4\rho \bF-3\rho \nabb_4\bF +3\ka (2\rhoF \b-3\rho \bF)\\
&=&  2(-\ka \rhoF) \b+2\rhoF (-2\ka \b +\divv \a +\rhoF \nabb_4 \bF)-3(-\frac 3 2 \ka \rho - \ka \rhoF^2) \bF-3\rho \nabb_4\bF \\
&&+3\ka (2\rhoF \b-3\rho \bF)\\
&=&2\rhoF \divv \a +(2\rhoF^2 -3\rho)(\nabb_4 \bF+ \frac 3 2 \ka\bF )
\eeaa
which proves \eqref{nabb-4-bF-tilde-b}. Similarly,
\beaa
\nabb_3 \tilde{\bb} +(3\kab+2\omb) \tilde{\bb}&=&  2\nabb_3\rhoF \bb+2\rhoF \nabb_3\bb-3\nabb_3\rho \bbF-3\rho \nabb_3\bbF +(3\kab+2\omb) (2\rhoF \bb-3\rho \bbF)\\
&=&  2(-\kab \rhoF) \bb+2\rhoF (-2\ka \b-2\omb \bb -\divv \aa-3\rho \xib +\rhoF (\nabb_3 \bbF+2\omb \bbF+2\rhoF \xib))\\
&&-3(-\frac 3 2 \kab \rho - \kab \rhoF^2) \bF-3\rho \nabb_3\bbF +(3\kab +2\omb)(2\rhoF \bb-3\rho \bbF)\\
&=&-2\rhoF \divv \aa +(2\rhoF^2 -3\rho)(\nabb_3 \bbF+ (\frac 3 2 \kab+2\omb)\bbF +2\rhoF \xib)
\eeaa  
\end{proof}

Recall the definition of the the charge aspect function $\check{\nu}$ in \eqref{definition-of-nu}. We compute in the following lemma the transport equation for the charge aspect function. 

\begin{lemma}\label{nabb-3-nu} The function $\check{\nu}$ verifies the following transport equation:
\bea
\nabb_4\check{\nu}&=&\frac 1 2 r^4\DDd_1\DDs_1(\check{\ka}, 0)-2r^4\rhoF^2\check{\ka}+  r^4\divv \divv \chih \label{nabb-4-check-nu}
\eea
\bea
\begin{split}
\nabb_3\check{\nu}=r^4 &\left(2 \DDd_1\DDs_1( \check{\omb}, 0)+\Big(\frac 1 2 \kab+2\omb \right)\divv \left(\ze-\eta\right) +\frac 1 2 \ka \divv\xib -\divv\divv \chibh\\
&-\frac 1 2 \DDd_1\DDs_1( \check{\kab}, 0)-2\rhoF^2 \left(\check{\kab}-\ka\check{\Omegab}\right) \Big)\label{nabb-3-check-nu}
\end{split}
\eea
\end{lemma}
\begin{proof} We compute, using \eqref{nabb-4-ze-ze-eta} and \eqref{nabb-4-check-rhoF-ze-eta}:
\beaa
\nabb_4\left(\frac{\check{\nu}}{r^4}\right)&=& \nabb_4\left( \divv\ze+2\rhoF\check{\rhoF}\right)\\
&=& \divv(\nabb_4\ze)-\frac 1 2 \ka \divv \ze+2(\nabb_4\rhoF)\check{\rhoF}+2\rhoF\nabb_4\check{\rhoF} \\
&=& \divv(-  \ka\ze - \b  -\rhoF\bF)-\frac 1 2 \ka \divv \ze-2\ka\rhoF\check{\rhoF}+2\rhoF(- \ka \check{\rhoF}- \rhoF\check{\ka} +\divv\bF)
\eeaa
We obtain
\beaa
\nabb_4\left(\frac{\check{\nu}}{r^4}\right)&=& -\divv \b  +\rhoF\divv\bF-\frac 3 2 \ka \divv \ze-4\ka\rhoF\check{\rhoF}-2\rhoF^2\check{\ka}
\eeaa
Taking the divergence of Codazzi equation \eqref{Codazzi-chi-ze-eta}, we can write
\beaa
 -\divv\b +\rhoF\divv\bF&=&\divv \divv \chih-\frac 1 2 \ka \divv\ze+\frac 1 2 \DDd_1\DDs_1(\ka, 0)
\eeaa
Substituting in the above we finally obtain
\beaa
\nabb_4\left(\frac{\check{\nu}}{r^4}\right)&=& \divv \divv \chih-\frac 1 2 \ka \divv\ze+\frac 1 2 \DDd_1\DDs_1(\ka, 0)-\frac 3 2 \ka \divv \ze-4\ka\rhoF\check{\rhoF}-2\rhoF^2\check{\ka}\\
&=&-2 \ka \divv \ze-4\ka\rhoF\check{\rhoF}+ \divv \divv \chih+\frac 1 2 \DDd_1\DDs_1(\ka, 0)-2\rhoF^2\check{\ka}\\
&=& -2 \ka \left(\frac{\check{\mu}}{r^4}\right)+ \divv \divv \chih+\frac 1 2 \DDd_1\DDs_1(\check{\ka}, 0)-2\rhoF^2\check{\ka}
\eeaa
as desired.
We also compute, using \eqref{nabb-3-ze-ze-eta} and \eqref{nabb-3-check-rhoF-ze-eta}:
\beaa
\nabb_3\left(\frac{\check{\nu}}{r^4}\right)&=& \nabb_3\left( \divv\ze+2\rhoF\check{\rhoF}\right)\\
&=& \divv(\nabb_3\ze)-\frac 1 2 \kab \divv \ze+2(\nabb_3\rhoF)\check{\rhoF}+2\rhoF\nabb_3\check{\rhoF}\\
&=& \divv(- \left(\frac 1 2 \kab-2\omb \right)\ze+2 \DDs_1( \check{\omb}, 0)-\left(\frac 1 2 \kab+2\omb \right)\eta +\frac 1 2 \ka \xib - \bb  -\rhoF\bbF)\\
&&-\frac 1 2 \kab \divv \ze-2\kab \rhoF\check{\rhoF}+2\rhoF (- \kab \check{\rhoF}- \rhoF\left(\check{\kab}-\ka\check{\Omegab}\right)  -\divv\bbF)\\
&=& - \left( \kab-2\omb \right)\divv\ze+2 \DDd_1\DDs_1( \check{\omb}, 0)-\left(\frac 1 2 \kab+2\omb \right)\divv\eta +\frac 1 2 \ka \divv\xib - \divv\bb  +\rhoF\divv\bbF\\
&&-4\kab \rhoF\check{\rhoF}-2\rhoF^2 \left(\check{\kab}-\ka\check{\Omegab}\right) 
\eeaa
Taking the divergence of Codazzi equation \eqref{Codazzi-chib-ze-eta}, we can write
\beaa
-\divv\bb+\rhoF\divv\bbF&=& -\divv\divv \chibh-\frac 1 2 \kab \divv\ze-\frac 1 2 \DDd_1\DDs_1( \check{\kab}, 0) 
\eeaa
Substituting in the above, we finally have
\beaa
\nabb_3\left(\frac{\check{\nu}}{r^4}\right)&=& - \left( \frac 3 2 \kab-2\omb \right)\divv\ze+2 \DDd_1\DDs_1( \check{\omb}, 0)-\left(\frac 1 2 \kab+2\omb \right)\divv\eta +\frac 1 2 \ka \divv\xib -\divv\divv \chibh-\frac 1 2 \DDd_1\DDs_1( \check{\kab}, 0)\\
&&-4\kab \rhoF\check{\rhoF}-2\rhoF^2 \left(\check{\kab}-\ka\check{\Omegab}\right) \\
&=& -  2 \kab \divv\ze-4\kab \rhoF\check{\rhoF}+2 \DDd_1\DDs_1( \check{\omb}, 0)+\left(\frac 1 2 \kab+2\omb \right) \left(\divv \ze-\divv\eta\right) +\frac 1 2 \ka \divv\xib -\divv\divv \chibh\\
&&-\frac 1 2 \DDd_1\DDs_1( \check{\kab}, 0)-2\rhoF^2 \left(\check{\kab}-\ka\check{\Omegab}\right) 
\eeaa
as desired.
\end{proof}

Recall the definition of the the mass-charge aspect function $\check{\mu}$ in \eqref{definition-of-mu}. We compute in the following lemma the transport equation for the mass-charge aspect function. 

\begin{lemma}\label{nabb-3-mu} The function $\check{\mu}$ verifies the following transport equations:
\bea\label{nabb-4-check-mu}
\nabb_4\check{\mu}&=&-r^3\left(\frac 3 2 \rho-7\rhoF^2\right)\check{\ka}+ O(r^{-1-\de}u^{-1+\de}) 
\eea
\bea\label{nabb-3-check-mu}
\begin{split}
\nabb_3\check{\mu}&=r^3\Big(2 \DDd_1\DDs_1( \check{\omb}, 0)+\left(\frac 1 2 \kab+2\omb \right)(\divv\ze-\divv\eta) +\frac 1 2 \ka \divv\xib - 2\divv\bb  +2\rhoF\divv\bbF\\
& -\left(\frac 3 2  \rho-3\rhoF^2\right)(\check{\kab}-\ka \check{\Omegab}) -4\omb r\rhoF \divv\bF +2r\rhoF \DDd_1\DDs_1(\check{\rhoF}, \check{\sigmaF}) -4r\rhoF^2 \divv\eta\Big)
\end{split}
\eea
\end{lemma}
\begin{proof} We compute $\nabb_4\check{\mu}$.

Commuting \eqref{nabb-4-ze-ze-eta} with $r\divv$ we obtain
\beaa
\nabb_4 (r\divv\ze)+\ka r\divv\ze &=&-r\divv\b  -\rhoF r\divv\bF
\eeaa
which can be written as 
\bea\label{nabb-4-nu-part-1}
\nabb_4 (r^3\divv\ze)&=&-r^3\divv\b  -\rhoF r^3\divv\bF
\eea
Equation \eqref{nabb-4-check-rho-ze-eta} can be written as 
\bea\label{nabb-4-nu-part-2}
\nabb_4(r^3\check{\rho})&=&-r^3\left(\frac 3 2 \rho+\rhoF^2\right)\check{\ka}-4r^2\rhoF\check{\rhoF} +r^3\divv\b+r^3\rhoF \ \divv\bF
\eea
Equation \eqref{nabb-4-check-rhoF-ze-eta} implies
\beaa
\nabb_4(\rhoF\check{ \rhoF})&=& -\ka\rhoF \check{ \rhoF}+\rhoF (- \ka \check{\rhoF}-\rhoF \check{\ka}+\divv\bF)\\
&=& -2\ka\rhoF \check{ \rhoF}-\rhoF^2 \check{\ka}+\rhoF \divv\bF
\eeaa
which can be written as 
\bea\label{nabb-4-nu-part-3}
\nabb_4(-4r^3\rhoF\check{ \rhoF})&=& 4 r^2 \rhoF \check{ \rhoF}-4r^3\rhoF \divv\bF+4r^3\rhoF^2 \check{\ka}
\eea
Commuting \eqref{nabb-4-bF-tilde-b} with $\divv$ and using estimates for $\tilde{\b}$ we obtain
\beaa
\nabb_4 \divv\bF+ 2 \ka\divv\bF&=&O(r^{-3-\de}u^{-1+\de}) 
\eeaa
which implies 
\beaa
\nabb_4(\rhoF \divv \bF)&=& -\ka\rhoF \divv \bF+\rhoF (- 2 \ka\divv\bF+O(r^{-3-\de}u^{-1+\de}) )\\
&=& -3\ka\rhoF \divv \bF+ O(r^{-5-\de}u^{-1+\de}) 
\eeaa
which can be written as
\bea\label{nabb-4-nu-part-4}
\nabb_4(-2r^4\rhoF \divv \bF)&=& 4r^3\rhoF \divv \bF+ O(r^{-1-\de}u^{-1+\de}) 
\eea
Summing \eqref{nabb-4-nu-part-1}, \eqref{nabb-4-nu-part-2}, \eqref{nabb-4-nu-part-3}, \eqref{nabb-4-nu-part-4}, we obtain
\beaa
\nabb_4\check{\mu}&=& -r^3\divv\b  -\rhoF r^3\divv\bF-4r^2\rhoF\check{\rhoF} +r^3\divv\b+r^3\rhoF \ \divv\bF\\
&&+4 r^2 \rhoF \check{ \rhoF}-4r^3\rhoF \divv\bF+4r^3\rhoF^2 \check{\ka}+4r^3\rhoF \divv \bF-r^3\left(\frac 3 2 \rho-3\rhoF^2\right)\check{\ka}\\
&&+ O(r^{-1-\de}u^{-1+\de}) \\
&=&-r^3\left(\frac 3 2 \rho-7\rhoF^2\right)\check{\ka}+ O(r^{-1-\de}u^{-1+\de}) 
\eeaa
Using the definition \eqref{definition-of-mu}, we compute
\beaa
\nabb_3\left(\frac{\check{\mu}}{r^3}\right)&=&\nabb_3( \divv \ze+ \check{\rho}-4\rhoF \check{\rhoF} -2r\rhoF \divv\bF)\\
&=&- \left( \kab-2\omb \right)\divv\ze+2 \DDd_1\DDs_1( \check{\omb}, 0)-\left(\frac 1 2 \kab+2\omb \right)\divv\eta +\frac 1 2 \ka \divv\xib - \divv\bb  -\rhoF\divv\bbF\\
&&-\frac 3 2 \kab \check{\rho} -\left(\frac 3 2  \rho+\rhoF^2\right)(\check{\kab}-\ka \check{\Omegab})-2\kab\rhoF\check{\rhoF}-\divv\bb-\rhoF \ \divv\bbF\\
&&+4\kab\rhoF \check{\rhoF}-4\rhoF (- \kab \check{\rhoF}- \rhoF\left(\check{\kab}-\ka\check{\Omegab}\right)  -\divv\bbF ) \\
&&-\kab r\rhoF \divv\bF+2r\kab\rhoF \divv\bF
-2r\rhoF (-\left( \kab-2\omb\right)\divv \bF -\DDd_1\DDs_1(\check{\rhoF}, \check{\sigmaF}) +2\rhoF \divv\eta)
\eeaa
which gives
\beaa
\nabb_3\left(\frac{\check{\mu}}{r^3}\right)&=&- \frac 3 2 \kab\divv\ze+2 \DDd_1\DDs_1( \check{\omb}, 0)+\left(\frac 1 2 \kab+2\omb \right)(\divv\ze-\divv\eta) +\frac 1 2 \ka \divv\xib - 2\divv\bb  +2\rhoF\divv\bbF\\
&&-\frac 3 2 \kab \check{\rho} -\left(\frac 3 2  \rho-3\rhoF^2\right)(\check{\kab}-\ka \check{\Omegab})+6\kab\rhoF\check{\rhoF}\\
&&+\left( 3\kab-4\omb\right) r\rhoF \divv\bF +2r\rhoF \DDd_1\DDs_1(\check{\rhoF}, \check{\sigmaF}) -4r\rhoF^2 \divv\eta
\eeaa
Using again the definition of $\check{\mu}$ in the above, we obtain \eqref{nabb-3-check-mu}. 
\end{proof}



\end{document}